\documentclass[sigplan]{acmart}
\acmConference[]{Report version}
\acmISBN{} 
\acmDOI{} 
\startPage{1}

\setcopyright{none}

\usepackage{booktabs}   
\usepackage{subcaption} 

\usepackage{afterpage}
\usepackage{amsmath}
\usepackage{amssymb,float}
\usepackage[american]{babel}
\usepackage{bussproofs}
\usepackage{calc}
\usepackage{centernot}
\usepackage{chngcntr}
\usepackage[dmyyyy]{datetime}
\usepackage[inline]{enumitem}
\usepackage{esvect}
\usepackage{fancybox}
\usepackage{hyperref}
\usepackage{multirow}
\usepackage{mathtools}
\usepackage{placeins}
\usepackage{soul}   
\usepackage{stmaryrd}
\usepackage{url}
\usepackage{ushort}
\let\savebigtimes\bigtimes
\let\bigtimes\relax
\usepackage{mathabx}
\let\bigtimes\savebigtimes
\usepackage{graphicx}
\usepackage{tikz,fp}
\usetikzlibrary{arrows.meta,arrows,calc,automata,angles,quotes,matrix,
                positioning,shapes.geometric,shapes.symbols,shapes.multipart,through,trees,
                decorations.markings,decorations.text
                }

\pgfdeclaredecoration{funbisim}{final}{
  \state{final}[width=\pgfdecoratedpathlength]{
    \draw[->]
      (0,\pgfdecorationsegmentamplitude+0.1mm) -- +(\pgfdecoratedpathlength,0);
    \draw[-]
      (0,-\pgfdecorationsegmentamplitude-0.1mm) -- +(\pgfdecoratedpathlength,0);}}
\tikzset{
  funbisim/.style={
    decoration={funbisim, amplitude=0.25ex},
    decorate,
    funbisim options/.style={#1}    
  }}

\pgfdeclaredecoration{bisim}{final}{
  \state{final}[width=\pgfdecoratedpathlength]{
    \draw[<->]
      (0,\pgfdecorationsegmentamplitude+0.1mm) -- +(\pgfdecoratedpathlength,0);
    \draw[-]
      (0,-\pgfdecorationsegmentamplitude-0.1mm) -- +(\pgfdecoratedpathlength,0);}}
\tikzset{
  bisim/.style={
    decoration={bisim, amplitude=0.25ex},
    decorate,
    bisim options/.style={#1}    
  }}

\makeatletter
\def\calcLength(#1,#2)#3{%
\pgfpointdiff{\pgfpointanchor{#1}{center}}%
             {\pgfpointanchor{#2}{center}}%
\pgf@xa=\pgf@x%
\pgf@ya=\pgf@y%
\FPeval\@temp@a{\pgfmath@tonumber{\pgf@xa}}%
\FPeval\@temp@b{\pgfmath@tonumber{\pgf@ya}}%
\FPeval\@temp@sum{(\@temp@a*\@temp@a+\@temp@b*\@temp@b)}%
\FProot{\FPMathLen}{\@temp@sum}{2}%
\FPround\FPMathLen\FPMathLen5\relax
\global\expandafter\edef\csname #3\endcsname{\FPMathLen}
}
\makeatother

\tikzset{
  my dash/.style={dash pattern=on 5pt off 2pt}
         }

\numberwithin{equation}{section}

\newtheorem*{repeatedlem}{Lemma}{\normalfont\bfseries}{\itshape}
\newtheorem*{repeatedprop}{Proposition}{\normalfont\bfseries}{\itshape}
\setul{1pt}{.4pt}  

\DeclareFontFamily{U}{mathx}{\hyphenchar\font45}   
\DeclareFontShape{U}{mathx}{m}{n}{
      <5> <6> <7> <8> <9> <10>
      <10.95> <12> <14.4> <17.28> <20.74> <24.88>
      mathx10
      }{}
\DeclareSymbolFont{mathx}{U}{mathx}{m}{n}
\DeclareFontSubstitution{U}{mathx}{m}{n}
\DeclareMathAccent{\widecheck}{0}{mathx}{"71}
\DeclareMathAccent{\wideparen}{0}{mathx}{"75}

\definecolor{azure}{rgb}{0.94,1.00,1.00}
\definecolor{brown}{rgb}{.75,.25,.25}
\definecolor{cyan}{rgb}{0.25,0.88,0.82}
\definecolor{chocolate}{rgb}{0.82,0.41,0.12}
\definecolor{darkcyan}{rgb}{0.5,0,1}
\definecolor{darkgreen}{rgb}{0,0.39,0}
\definecolor{darkmagenta}{rgb}{0.5,0,0.5}
\definecolor{darkgoldenrod}{RGB}{184,134,11}
\definecolor{firebrick}{RGB}{175,25,25}
\definecolor{forestgreen}{rgb}{0.13,0.55,0.13}
\definecolor{goldenrod}{RGB}{218,165,32}
\definecolor{lightcyan}{rgb}{0.88,1.00,1.00}
\definecolor{lightpink}{rgb}{1.00,0.71,0.76}
\definecolor{myyellow}{RGB}{235,235,0}
\definecolor{lightyellow}{rgb}{1.00,1.00,0.88}
\definecolor{lightgoldenrod}{rgb}{0.83,0.97,0.51}
\definecolor{lightgoldenrodyellow}{rgb}{0.98,0.98,0.82}
\definecolor{lightskyblue}{rgb}{0.53,0.81,0.98}
\definecolor{moccasin}{rgb}{1.00,0.89,0.71}
\definecolor{magenta}{rgb}{1,0,1}
\definecolor{navyblue}{rgb}{0,0,0.5}
\definecolor{orange}{rgb}{1.0,0.65,0.0}
\definecolor{orangered}{rgb}{1.0,0.27,0.0}
\definecolor{palegreen}{rgb}{0.60,0.98,0.60}
\definecolor{powderblue}{rgb}{0.69,0.88,0.90}
\definecolor{purple}{rgb}{1,0.5,1}
\definecolor{royalblue}{RGB}{65,105,225}
\definecolor{mediumblue}{RGB}{0,0,205}
\definecolor{cornflowerblue}{RGB}{100,149,237}
\definecolor{springgreen}{rgb}{0.0,1.0,0.5}
\definecolor{turquoise}{rgb}{0.25,0.88,0.82}
\definecolor{snow}{rgb}{1.00,0.98,0.98}
\definecolor{tan}{rgb}{0.82,0.71,0.55}
\definecolor{red}{rgb}{1,0,0}
\definecolor{violetred}{RGB}{208,32,144}

\newcommand{\colorin}[1]{\textcolor{#1}}

\newcommand{\black}{\colorin{black}}

\newcommand{\colorred}{\colorin{red}}

\newcommand{\darkcyan}{\colorin{darkcyan}}

\newcommand{\forestgreen}{\colorin{forestgreen}}

\newcommand{\nb}{\nobreakdash}
\newcommand{\punc}[1]{\ensuremath{\hspace*{1.5pt}{#1}}}

\newcommand{\nf}{\normalfont}

\newenvironment{new}{\color{chocolate}}{\color{black}}
\newenvironment{newer}{\color{firebrick}}{\color{black}}
\newenvironment{newest}{\color{red}}{\color{black}}
\newenvironment{revised}{\color{violetred}}{\color{black}}
\newenvironment{change}{\color{violetred}}{\color{black}}
\newcommand{\funin}{\mathrel{:}}
\newcommand{\fap}[2]{{#1}(\hspace*{-0.5pt}{#2}\hspace*{-0.5pt})}

\newcommand{\bfap}[3]{{#1}({#2},\hspace*{0.5pt}{#3})}

\newcommand{\iap}[2]{#1_{#2}}
\newcommand{\bap}[2]{#1_{#2}}
\newcommand{\pap}[2]{#1^{#2}}
\newcommand{\bpap}[3]{#1_{#2}^{#3}}
\newcommand{\pbap}[3]{#1_{#3}^{#2}}
\newcommand{\sidfunon}{\iap{\textrm{\nf id}}}
\newcommand{\idfunon}[1]{\fap{\sidfunon{#1}}}

\newcommand{\sdefdby}{{:=}}
\newcommand{\defdby}{\mathrel{\sdefdby}}

\newcommand\tuple[1]{\langle #1 \rangle}

\newcommand\tuplespace{\hspace*{0.5pt}}
\newcommand\pair[2]{\tuple{#1, \tuplespace #2}}
\newcommand\triple[3]{\tuple{#1, \tuplespace #2, \tuplespace #3}}

\newcommand{\nat}{\mathbb{N}}
\newcommand{\natplus}{\pap{\nat}{+}} 

\newcommand{\BNFor}{\:\mid\:}
\newcommand{\BNFdefdby}{\:{::=}\:}

\newcommand{\ssyntequal}{{\equiv}}
\newcommand{\syntequal}{\mathrel{\ssyntequal}}

\newcommand{\sred}{{\to}}
\newcommand{\red}{\mathrel{\sred}}
\newcommand{\sredi}[1]{{\iap{\sred}{#1}}}
\newcommand{\redi}[1]{\mathrel{\sredi{#1}}}

\newcommand{\sconvredi}[1]{{\iap{\leftarrow}{#1}}}

\newcommand{\sredtc}{\sred^{+}}
\newcommand{\redtc}{\mathrel{\sredtc}}
\newcommand{\sredtci}[1]{{\iap{\sredtc}{#1}}}
\newcommand{\redtci}[1]{\mathrel{\sredtci{#1}}}

\newcommand{\sredrtc}{\sred^{*}}
\newcommand{\redrtc}{\mathrel{\sredrtc}}
\newcommand{\sredrtci}[1]{{\iap{\sredrtc}{#1}}}
\newcommand{\redrtci}[1]{\mathrel{\sredrtci{#1}}}

\newcommand{\scomprewrels}[2]{{#1}\cdot{#2}}
\newcommand{\comprewrels}[2]{\mathrel{\scomprewrels{#1}{#2}}}

\newcommand{\stavoidsv}{\text{\nf\bf\st{$\hspace*{1.5pt}$t$\hspace*{1.5pt}$}}}
\newcommand{\tavoidsv}{\fap{\stavoidsv}}
\newcommand{\sredtavoidsv}[1]{{\xrightarrow[\raisebox{0pt}{\scriptsize {$\tavoidsv{#1}$}}]{}}}
\newcommand{\redtavoidsv}[1]{\mathrel{\sredtavoidsv{#1}}}
\newcommand{\sredtavoidsvrtc}[1]{{\xrightarrow[\raisebox{0pt}{\scriptsize {$\tavoidsv{#1}$}}]{}}{^{*}}}
\newcommand{\redtavoidsvrtc}[1]{\mathrel{\sredtavoidsvrtc{#1}}}
\newcommand{\sredtavoidsvtc}[1]{{\xrightarrow[\raisebox{0pt}{\scriptsize {$\tavoidsv{#1}$}}]{}}{^{+}}}
\newcommand{\redtavoidsvtc}[1]{\mathrel{\sredtavoidsvtc{#1}}}
\newcommand{\sredtavoidsvi}[2]{{\sredtavoidsv{#1}{}_{#2}}}
\newcommand{\redtavoidsvi}[2]{\mathrel{\sredtavoidsvi{#1}{#2}}}
\newcommand{\sredtavoidsvrtci}[2]{{\xrightarrow[\raisebox{0pt}{\scriptsize {$\tavoidsv{#1}$}}]{}}{^{*}_{#2}}}
\newcommand{\redtavoidsvrtci}[2]{\mathrel{\sredtavoidsvrtci{#1}{#2}}}
\newcommand{\sredtavoidsvtci}[2]{{\xrightarrow[\raisebox{0pt}{\scriptsize {$\tavoidsv{#1}$}}]{}}{^{+}_{#2}}}
\newcommand{\redtavoidsvtci}[2]{\mathrel{\sredtavoidsvtci{#1}{#2}}}
  \newcommand{\specfontsize}{\fontsize{5}{6}\selectfont} 
  \newcommand{\subosr}{\hspace*{-1pt}\mbox{\specfontsize $(\hspace*{-0.6pt}\sone\hspace*{-0.85pt})$}}

\newcommand{\descsetexpmid}{\mathrel{\vert}}
\newcommand{\descsetexpbigmid}{\mathrel{\big\vert}}

\newcommand{\descsetexp}[2]{\left\{{#1}\descsetexpmid{#2}\right\}}

\newcommand{\descsetexpbig}[2]{\bigl\{{#1}\descsetexpbigmid{#2}\bigr\}}

\newcommand{\sphifun}{\phi}
\newcommand{\phifun}{\fap{\sphifun}}

\newcommand{\scompfuns}[2]{{#1}\circ{#2}}
\newcommand{\compfuns}[2]{\fap{\scompfuns{#1}{#2}}}

\newcommand{\setexp}[1]{\left\{{#1}\right\}}

\newcommand{\setexpbig}[1]{\bigl\{{#1}\bigr\}}

\renewcommand{\emptyset}{\varnothing}

\newcommand{\slogand}{\wedge}

\newcommand{\logand}{\mathrel{\slogand}}
 
\newcommand{\slognot}{\neg}
\newcommand{\lognot}[1]{\slognot{#1}}

\newcommand{\existsstzero}[1]{\exists{\hspace*{1pt}#1}}

\newcommand{\actionderivative}{action derivative}
\newcommand{\actionderivatives}{\actionderivative{s}}

\newcommand{\languageequivalent}{lan\-guage-equiv\-a\-lent}

\newcommand{\aderivative}[1]{{$#1$}\nb-de\-riv\-a\-tive}
\newcommand{\aderivatives}[1]{\aderivative{#1}s}

\newcommand{\entrytransition}{en\-try tran\-si\-tion}
\newcommand{\entrytransitions}{\entrytransition{s}}

\newcommand{\generatedby}[1]{${#1}$\nb-ge\-ne\-ra\-ted}

\newcommand{\startconnected}{start-ver\-tex con\-nect\-ed}

\newcommand{\loopelimination}{loop\nb-elim\-i\-nat\-ion}

\newcommand{\loopentry}{loop-en\-try}
\newcommand{\loopbody}{loop-body}

\newcommand{\txtloopsbackto}{loops-back-to}

\newcommand{\entrybodylabeling}{en\-try\discretionary{/}{}{/}body-la\-be\-ling}
\newcommand{\entrybodylabelings}{\entrybodylabeling{s}}
\newcommand{\looplabeling}{loop-la\-be\-ling}

\newcommand{\LEEshaped}{\LEE\nb-shaped}

\newcommand{\LEEwitness}{$\LEE$\nb-wit\-ness}

\newcommand{\LEEwitnesses}{$\LEE$\hspace*{1.25pt}\nb-wit\-nes\-ses}

\newcommand{\LLEEwitness}{{\nf LLEE}\nb-wit\-ness}
\newcommand{\LLEEwitnesses}{{\nf LLEE}\nb-wit\-nesses}

\newcommand{\LLEEchart}{{\nf LLEE}\nb-chart}
\newcommand{\LLEEchartemph}{\emph{LLEE\nb-chart}}
\newcommand{\LLEEcharts}{{\nf LLEE}\nb-chart{s}}

\newcommand{\nonzero}{non-zero}

\newcommand{\onetransitions}{\oneTransition{s}}
\newcommand{\oneTransition}{$\sone$\nb-tran\-si\-tion}

\newcommand{\provablein}[1]{{$#1$}\nb-pro\-vable}

\newcommand{\provablyin}[1]{{$#1$}\nb-pro\-vably}

\newcommand{\structureconstrained}{struc\-ture-con\-strained}

\newcommand{\transitionact}[1]{{${#1}$}\nb-tran\-si\-tion}

\newcommand{\welldefined}{well-de\-fined}

\newcommand{\scc}{scc}
\newcommand{\sccs}{\scc's}

\newcommand{\sscc}{\textsf{scc}}
\newcommand{\sccof}{\fap{\sscc}}

\newcommand{\astexp}{e}
\newcommand{\bstexp}{f}
\newcommand{\cstexp}{g}

\newcommand{\astexpi}{\iap{\astexp}}
\newcommand{\bstexpi}{\iap{\bstexp}}

\newcommand{\astexpacc}{\astexp'}

\newcommand{\astexpacci}{\iap{\astexpacc}}

\newcommand{\astexpdacc}{\astexp''}

\newcommand{\astexpdacci}{\iap{\astexpdacc}}

\newcommand{\StExps}{\mathit{StExp}}
\newcommand{\StExpsover}{\fap{\StExps}}

\newcommand{\tickStExpsover}[1]{\bap{\StExpsover{#1}}{\tick}}

\newcommand{\stexpzero}{0}
\newcommand{\stexpone}{1}
\newcommand{\astexpact}{a}

\newcommand{\sstexpit}{\sstar}
\newcommand{\stexpit}[1]{{#1^{\sstexpit}}}
\newcommand{\stexppl}[1]{{#1^{\spl}}}
\newcommand{\sstexpprod}{{\cdot}}
\newcommand{\stexpprod}[2]{{#1}\mathrel{\sstexpprod}{#2}}
\newcommand{\sstexpsum}{+}
\newcommand{\stexpsum}[2]{{#1}\sstexpsum{#2}}

\newcommand{\sstexpbit}{\circledast} 
\newcommand{\stexpbit}[2]{{#1}\hspace*{0.35pt}\pap{}{\sstexpbit}\hspace*{-0.6pt}{#2}}

\newcommand{\atickstexp}{\xi}

\newcommand{\bsth}[1]{|{#1}|_{\scalebox{0.8}{$\sstexpbit$}}}

\newcommand{\sdescrelstexpit}{\sredi{\scriptscriptstyle(\sstar)}}

\newcommand{\descrelstexpit}[1]{\mathrel{\sdescrelstexpit}}

\newcommand{\sconvdescrelstexpit}{\sconvredi{\scriptscriptstyle(\sstar)}}

\newcommand{\convdescrelstexpit}[1]{\mathrel{\sconvdescrelstexpit}}

\newcommand{\tick}{\surd}

\newcommand{\spartderivs}{\partial}

\newcommand{\sterms}{\textit{tm}}
\renewcommand{\terms}{\fap{\sterms}}

\newcommand{\sactderivs}{\textit{A\hspace*{-0.25pt}$\spartderivs$}}
\newcommand{\actderivs}{\fap{\sactderivs}}

\newcommand{\sprocint}{P}
\newcommand{\procint}[1]{\llbracket{#1}\rrbracket_{\sprocint}}
 
\newcommand{\sprocsem}{P}
\newcommand{\procsem}[1]{\llbracket{#1}\rrbracket_{\sprocsem}}

\newcommand{\slt}[1]{{\xrightarrow{#1}}}
\newcommand{\slti}[2]{{\xrightarrow{#1}}{_{#2}}}

\newcommand{\lt}[1]{\mathrel{\slt{#1}}}
\newcommand{\lti}[2]{\mathrel{\slti{#1}{#2}}}

\newcommand{\sone}{1}
\newcommand{\sstar}{*}
\newcommand{\spl}{\omega}
\newcommand{\branchlab}{\text{\nf br}}
\newcommand{\bodylab}{\text{\nf bo}}

\newcommand{\bodytransition}{body tran\-si\-tion}
\newcommand{\bodytransitions}{\bodytransition{s}}

\newcommand{\loopsteplabof}[1]{[{#1}]}

\newcommand{\loopnsteplab}[1]{[{#1}]}
\newcommand{\sloopnstepto}[1]{{\iap{\rightarrow}{\loopnsteplab{#1}}}}
\newcommand{\loopnstepto}[1]{\mathrel{\sloopnstepto{#1}}}

\newcommand{\loopentrytransition}{loop-entry tran\-si\-tion}
\newcommand{\loopentrytransitions}{\loopentrytransition{s}}

\newcommand{\aLname}{\alpha}
\newcommand{\aLnamei}{\iap{\aLname}}
\newcommand{\bLname}{\beta}
\newcommand{\bLnamei}{\iap{\bLname}}
\newcommand{\cLname}{\gamma}
\newcommand{\cLnamei}{\iap{\aLname}}
\newcommand{\dLname}{\delta}
\newcommand{\dLnamei}{\iap{\dLname}}
\newcommand{\eLname}{\epsilon}

\newcommand{\aLnameacc}{\alpha'}

\renewcommand{\ll}[1]{\left\lvert{#1}\right\rvert}

\newcommand{\bosn}[1]{\left\lVert{#1}\right\rVert_{\bodylab}} 
\newcommand{\bosnnf}[1]{\lVert{#1}\rVert_{\bodylab}}

\newcommand{\sen}{\text{\nf en}}
\newcommand{\enl}[1]{\left\lvert{#1}\right\rvert_{\sen}}
\newcommand{\enlnf}[1]{|{#1}|_{\sen}}

\newcommand{\slex}{\mathit{lex}}
\newcommand{\sltlex}{{\bap{<}{\slex}}}
\newcommand{\ltlex}{\mathrel{\sltlex}}

\newcommand{\sult}{{\to}}
\newcommand{\ult}{\mathrel{\sult}}

\newcommand{\achart}{\mathcal{C}}
\newcommand{\acharti}{\iap{\achart}}
\newcommand{\achartacc}{\mathcal{C}'}

\newcommand{\acharthat}{\hspace*{0.75pt}\Hat{\hspace*{-0.75pt}\achart}\hspace*{-0pt}} 
\newcommand{\acharthati}[1]{\hspace*{0.2pt}\iap{\Hat{\hspace*{-0.75pt}\achart}}{#1}\hspace*{-0.75pt}} 
\newcommand{\acharthatacc}{\acharthat'}

\newcommand{\acharthatdacc}{\hspace*{0.2pt}\acharthat\hspace*{0.2pt}''\hspace*{-0.75pt}}

\newcommand{\chartof}{\fap{\achart}}

\newcommand{\charthatof}[1]{\widehat{\rule{0pt}{6.5pt}\smash{\fap{\achart\hspace*{-1pt}}{{#1}}}}}
\newcommand{\charthighhatof}[1]{\widehat{\rule{0pt}{7.5pt}\smash{\fap{\achart\hspace*{-1pt}}{{#1}}}}}

\newcommand{\aloop}{\mathcal{L}}
\newcommand{\aloopi}{\iap{\aloop}}

\newcommand{\indsubchartinat}[1]{\fap{\acharti{#1}}}

\newcommand{\sinktermination}{sink-ter\-mi\-na\-tion}

\newcommand{\sasol}{s}
\newcommand{\sasoli}{\iap{\sasol}}
\newcommand{\asol}{\fap{\sasol}}
\newcommand{\asoli}[1]{\fap{\iap{\sasol}{#1}}}
\newcommand{\sbsol}{t}

\newcommand{\sextrsol}{\sasol}
\newcommand{\sextrsolof}{\iap{\sextrsol}}
\newcommand{\sextrsoluntil}{\sbsol}
\newcommand{\sextrsoluntilof}{\iap{\sextrsoluntil}}

\newcommand{\extrsolof}[1]{\fap{\sextrsolof{#1}}}
\newcommand{\extrsoluntilof}[1]{\bfap{\sextrsoluntilof{#1}}}

\newcommand{\entries}{E}
\newcommand{\entriesof}{\fap{\entries}}

\newcommand{\actions}{\mathit{A}}

\newcommand{\aact}{a}
\newcommand{\bact}{b}
\newcommand{\cact}{c}
\newcommand{\dact}{d}

\newcommand{\aacti}{\iap{\aact}}
\newcommand{\bacti}{\iap{\bact}}
\newcommand{\cacti}{\iap{\cact}}
\newcommand{\dacti}{\iap{\dact}}

\newcommand{\verts}{V}

\newcommand{\start}{\averti{\hspace*{-0.5pt}\text{\nf s}}}
\newcommand{\transs}{T}

\newcommand{\alab}{l}

\newcommand{\vertsof}{\fap{\verts\hspace*{-1pt}}}

\newcommand{\transsof}{\fap{\transs}}

\newcommand{\vertsi}[1]{\iap{\verts}{\hspace*{-0.25pt}{#1}}}

\newcommand{\starti}[1]{\averti{\text{\nf s},#1}}
\newcommand{\transsi}[1]{\iap{\transs}{\hspace*{-0.25pt}{#1}}}

\newcommand{\transshat}{\widehat{\transs}}

\newcommand{\apath}{\pi}

\newcommand{\apathacc}{\pi'}

\newcommand{\asettranss}{U}

\newcommand{\avert}{v}
\newcommand{\bvert}{w}
\newcommand{\cvert}{u}
\newcommand{\dvert}{x}
\newcommand{\averti}{\iap{\avert}}
\newcommand{\bverti}{\iap{\bvert}}
\newcommand{\cverti}{\iap{\cvert}}
\newcommand{\dverti}{\iap{\dvert}}
\newcommand{\avertacc}{\avert'}
\newcommand{\bvertacc}{\bvert'}
\newcommand{\cvertacc}{\cvert'}

\newcommand{\bverthat}{\widehat{\bvert}}

\newcommand{\avertacci}{\iap{\avertacc}}

\newcommand{\cvertacci}{\iap{\cvertacc}}

\newcommand{\cverttilde}{\tilde{\cvert}}

\newcommand{\avertdacc}{\avert''}

\newcommand{\avertdacci}{\iap{\avertdacc}}

\newcommand{\bverthati}{\iap{\bverthat}}

\newcommand{\atrans}{\tau}

\newcommand{\atransi}{\iap{\atrans}}

\newcommand{\atranshat}{\widehat{\atrans}}

\newcommand{\elimloopfrom}[2]{{#2}{-}{#1}}

\newcommand{\connthroughin}[2]{\pbap{#1}{(#2\hspace*{-1pt})}}

\newcommand{\connectthroughchart}[2]{con\-nect-{$#1$}-through-to-{$#2$} chart}

\newcommand{\avar}{x}
\newcommand{\bvar}{y}
\newcommand{\cvar}{z}

\newcommand{\sfunbisim}{%
    \setbox0=\hbox{\kern-.1ex{$\rightarrow$}\kern-.1ex}
    \setbox1=\vbox{\hbox{\raise .1ex \box0}\hrule}%
    {\hbox{\kern.05ex\box1\kern.1ex}}
  }
\newcommand{\funbisim}{\hspace*{-1.5pt}\mathrel{\sfunbisim}}

\newcommand{\sconvfunbisim}{%
    \setbox0=\hbox{\kern-.1ex{$\leftarrow$}\kern-.1ex}
    \setbox1=\vbox{\hbox{\raise .1ex \box0}\hrule}%
    {\hbox{\kern.05ex\box1\kern.1ex}}
  }
\newcommand{\convfunbisim}{\mathrel{\sconvfunbisim}}

\newcommand{\sbisim}{%
    \setbox0=\hbox{\kern-.1ex{$\leftrightarrow$}\kern-.1ex}
    \setbox1=\vbox{\hbox{\raise .1ex \box0}\hrule}%
    \hbox{\kern.1ex\box1\kern.1ex}
  }
\newcommand{\bisim}{\mathrel{\sbisim\hspace*{1pt}}}

\newcommand{\sfunbisimos}{%
    \setbox0=\hbox{\kern-.1ex{$\rightarrow$}\kern-.1ex}
    \setbox1=\vbox{\hbox{\raise .1ex \box0}\hrule}%
    {\pap{\hbox{\kern.05ex\box1\kern.1ex}}{\hspace*{0.5pt}\subosr}}
  }

\newcommand{\sconvfunbisimos}{%
    \setbox0=\hbox{\kern-.1ex{$\leftarrow$}\kern-.1ex}
    \setbox1=\vbox{\hbox{\raise .1ex \box0}\hrule}%
    {\pap{\hbox{\kern.05ex\box1\kern.1ex}}{\hspace*{0.5pt}\subosr}}
  }

\newcommand{\sbisimos}{%
    \setbox0=\hbox{\kern-.1ex{$\leftrightarrow$}\kern-.1ex}
    \setbox1=\vbox{\hbox{\raise .1ex \box0}\hrule}%
    \ensuremath{\pap{\mathrel{\hbox{\kern.1ex\box1\kern.1ex}}}{\hspace*{0.5pt}\subosr}}
  }

\newcommand{\sfunbisimosvia}[1]{%
    \setbox0=\hbox{\kern-.1ex{$\rightarrow$}\kern-.1ex}
    \setbox1=\vbox{\hbox{\raise .1ex \box0}\hrule}%
    {\bpap{\hbox{\kern.05ex\box1\kern.1ex}}{#1}{\hspace*{0.5pt}\subosr}}
  }

\newcommand{\sconvfunbisimosvia}[1]{%
    \setbox0=\hbox{\kern-.1ex{$\leftarrow$}\kern-.1ex}
    \setbox1=\vbox{\hbox{\raise .1ex \box0}\hrule}%
    {\bpap{\hbox{\kern.05ex\box1\kern.1ex}}{#1}{\hspace*{0.5pt}\subosr}}
  }

\newcommand{\abisim}{B}
\newcommand{\abisimi}{\iap{\abisim}}

\newcommand{\sbehinc}{{\sqsubseteq}}

\newcommand{\sbehinca}[1]{{\prescript{#1}{}{\sbehinc}}}
\newcommand{\behinca}[1]{\mathrel{\sbehinca}}

\newcommand{\sonebehinc}{{\pap{\sbehinc}{\subosr}}}

\newcommand{\sonebehinca}[1]{{{}_{#1}\sonebehinc}}
\newcommand{\onebehinca}[1]{\mathrel{\sonebehinca}}
\newcommand{\sRSP}{\textrm{\nf RSP}}

\newcommand{\RSPpl}{\sRSP^{\spl}\hspace*{-1pt}}
\newcommand{\RSPbit}{\sRSP^{\sstexpbit}\hspace*{-1pt}}

\newcommand{\assocstexpsum}{\text{\nf B2}} 
\newcommand{\assocstexpprod}{\text{\nf B5}}  
\newcommand{\commstexpsum}{\text{\nf B1}} 
\newcommand{\neutralstexpsum}{\text{\nf B6}} 
\newcommand{\idempotstexpsum}{\text{\nf B3}} 
\newcommand{\distr}{\text{\nf B4}}  

\newcommand{\stexpzerostexpprod}{\text{\nf B7}}

\newcommand{\ACI}{\text{\sf ACI}}

\newcommand{\BPA}{\text{\sf BPA}}
\newcommand{\BPAzeropl}{\text{$\pbap{\BPA}{\spl}{\text{\nf\sf 0}}$}}

\newcommand{\seqin}[1]{{\iap{=}{#1}\hspace*{1pt}}}
\newcommand{\eqin}[1]{\mathrel{\seqin{#1}}}
\newcommand{\seqinsol}[1]{{\pbap{=}{\text{\scriptsize (sol)}}{#1}\hspace*{1pt}}}
\newcommand{\eqinsol}[1]{\mathrel{\seqinsol{#1}}}

\newcommand{\ACIeq}{\eqin{\ACI}}

\newcommand{\BBPeq}{\eqin{\BBP}}

\newcommand{\sdescendsinloopto}{{\curvearrowright}}  
\newcommand{\descendsinloopto}{\mathrel{\sdescendsinloopto}}

\newcommand{\sdescendsinlooptotc}{{\pap{\sdescendsinloopto}{\hspace*{-0.8pt}+}}}
\newcommand{\descendsinlooptotc}{\mathrel{\sdescendsinlooptotc}}
\newcommand{\sdescendsinlooptortc}{{\pap{\sdescendsinloopto}{\hspace*{-0.8pt}*}}}
\newcommand{\descendsinlooptortc}{\mathrel{\sdescendsinlooptortc}}

\newcommand{\sdescendsinlooplto}[1]{{\pap{}{#1}\hspace*{-1pt}{\sdescendsinloopto}}}
\newcommand{\descendsinlooplto}[1]{\mathrel{\sdescendsinlooplto{#1}}}

\newcommand{\txtdescendsinloopto}{de\-scends-in-loop-to}

\newcommand{\sloopsbackto}{{\lefttorightarrow}} 
\newcommand{\loopsbackto}{\mathrel{\sloopsbackto}}

\newcommand{\sloopsbacktotc}{{\sloopsbackto^{+}}}
\newcommand{\loopsbacktotc}{\mathrel{\sloopsbacktotc\hspace*{-1pt}}}
\newcommand{\sloopsbacktortc}{{\sloopsbackto^{*}}}
\newcommand{\loopsbacktortc}{\mathrel{\sloopsbacktortc\hspace*{-1pt}}}

\newcommand{\sconvloopsbackto}{{\righttoleftarrow}} 
\newcommand{\convloopsbackto}{\mathrel{\sconvloopsbackto}}
\newcommand{\sconvloopsbacktotc}{{\pap{\sconvloopsbackto}{\hspace*{-1pt}+}}}
\newcommand{\convloopsbacktotc}{\mathrel{\sconvloopsbacktotc\hspace*{-1pt}}}
\newcommand{\sconvloopsbacktortc}{{\pap{\sconvloopsbackto}{\hspace*{-1pt}*}}}
\newcommand{\convloopsbacktortc}{\mathrel{\sconvloopsbacktortc\hspace*{-1pt}}}

\newcommand{\sdloopsbackto}{\hspace*{-1pt}\prescript{}{\mathit{d}}{\hspace*{-0.75pt}\sloopsbackto}}  

\newcommand{\dloopsbackto}{\mathrel{\sdloopsbackto}}
\newcommand{\sdloopsbacktotc}{{\prescript{}{\mathit{d}}{\hspace*{-0.75pt}\sloopsbacktotc}}}
\newcommand{\dloopsbacktotc}{\mathrel{\sdloopsbacktotc}}

\newcommand{\sconvdloopsbackto}{\hspace*{-1pt}\prescript{}{\mathit{d}}{\hspace*{-0.75pt}{\sconvloopsbackto}}} 
\newcommand{\convdloopsbackto}{\mathrel{\sconvdloopsbackto}}

\newcommand{\milnersysmin}{\text{\smash{$\text{\sf Mil}^{\boldsymbol{-}}$}}}

\newcommand{\BBP}{\text{$\text{\sf BBP}$}}

\newcommand{\sLEE}{\text{\nf LEE}}

\newcommand{\LEE}{\sLEE}

\newcommand{\LEEexists}{\text{$\sLEE^{\exists}$}}

\newcommand{\LEEforall}{\text{$\sLEE^{\forall}$}}

\newcommand{\thplus}[2]{{#1}{+}{#2}}

\newcommand{\slbs}{\textit{lb}}
\newcommand{\slbsred}{\sredi{\slbs}}
\newcommand{\lbsred}{\mathrel{\slbsred}}

\newcommand{\lbsminn}[1]{\left\lVert{#1}\right\rVert^{\text{\nf min}}_{\slbs}}

\begin{document}

%
\title{A Complete Proof System for 1-Free~Regular~Expressions~Modulo~Bisimilarity}
%
%
%



\author{Clemens Grabmayer}
\orcid{nnnn-nnnn-nnnn-nnnn}             
\affiliation{
  \department{Department of Computer Science}              
  \institution{Gran Sasso Science Institute}            
  \streetaddress{Viale F.\ Crispi, 7}
  \city{L'Aquila}
  \postcode{67100 AQ}
  \country{Italy}                    
}
\email{clemens.grabmayer@gssi.it}          

\author{Wan Fokkink}
\orcid{nnnn-nnnn-nnnn-nnnn}             
\affiliation{
  \department{Department of Computer Science}             
  \institution{Vrije Universiteit Amsterdam}           
  \streetaddress{De Boelelaan 1111}
  \city{Amsterdam}
  \postcode{1081 HV}
  \country{The Netherlands}                   
}
\email{w.j.fokkink@vu.nl}         

\begin{abstract}
  Robin Milner (1984)
gave a sound proof system for bisimilarity of regular expressions interpreted as processes:
Basic Process Algebra with unary Kleene star iteration, deadlock 0, successful termination 1, and a fixed-point rule.
He asked whether this system is complete. Despite intensive research over the last 35 years, the problem is still open.

This paper gives a partial positive answer to Milner's problem. 
We prove that the adaptation of Milner's system over the subclass of regular expressions that arises by dropping 
the constant 1, and by changing to binary Kleene star iteration is complete. 
The crucial tool we use is a graph structure property that guarantees expressibility of a process graph
by a regular expression, and is preserved by going over from a process graph to its bisimulation collapse.


%

\end{abstract}



\maketitle

\section{Introduction}
  \label{intro}
%

Regular expressions, introduced by Kleene \cite{klee:1951}, are widely studied in formal language theory, 
notably for string searching \cite{thom:1968}. 
They are constructed from constants 0 (no strings), 1 (the empty string), and $a$ (a single letter) from some alphabet; 
binary operators $+$ and $\cdot$ (union and concatenation); and the unary Kleene star ${}^{\sstexpit}$ (zero or more iterations). 

Their interpretations are Kleene algebras with as prime example the algebra of regular events,
the language semantics of regular expressions,
which is closely linked with deterministic finite state automata. 
Aanderaa \cite{aand:1965} and Salomaa \cite{salo:1966} gave complete axiomatizations 
for the language semantics of regular expressions, with a non-algebraic fixed-point rule that has a non-empty-word property as side~condition.
Krob \cite{krob:1991} gave an infinitary, and then Kozen \cite{koze:1994} a finitary algebraic axiomatization involving equational implications.


Regular expressions also received significant attention in the process algebra community \cite{berg-fokk-pons:2001}, 
where they are interpreted modulo the bisimulation process semantics \cite{park:1981}.
Robin Milner \cite{miln:1984} was the first to study regular expressions in this setting,
where he called them star expressions. 
Here the interpretation of 
0 is deadlock, 
1 is (successful) termination, 
$a$ is an atomic action, and 
$+$ and $\cdot$ are alternative and sequential composition of two processes, respectively. 
Milner adapted Salomaa's axiomatization to obtain a sound proof system for this setting, 
and posed the (still open) question whether this axiomatization is complete, 
meaning that if the process graphs of two star expressions are bisimilar, then they can be proven equal. 

Milner's axiomatization contains a fixed-point rule,
which is inevitable because due to the presence of $\stexpzero$
the underlying equational theory is not finitely based \cite{sewe:1994,sewe:1997}.
Bergstra, Bethke, and Ponse \cite{berg-beth-pons:1994} studied star expressions without 0 and 1, 
replaced the unary by the binary Kleene star ${}^{\sstexpbit}$, which represents an iteration of the first argument, 
possibly eventually followed by the execution of the second argument.
They obtained an axiomatization by basically omitting the axioms for 0 and 1 as well as the fixed-point rule from Milner's axiomatization, 
and adding Troeger's axiom~\cite{troe:1993}.
This purely equational axiomatization was proven complete in \cite{fokk:zant:1994,fokk:1996:kleene:star:AMAST}. 
A sound and complete axiomatization for star expressions without 
unary Kleene star, 
but with $\stexpzero$ and $\stexpone$ and a unary perpetual loop operator ${}^{\sstexpit}$0
(equivalently, unary star is restricted to terms $\stexpprod{\stexpit{\astexp}}{\stexpzero}$),
was given in~\cite{fokk:1996:term:cycle:LGPS,fokk:1997:pl:ICALP}.

In contrast to the formal languages setting, not all finite-state process graphs can be expressed by a star expression modulo bisimilarity. 
Milner posed a second question in \cite{miln:1984}, namely, 
to characterize which finite-state process graphs can be expressed. 
This was shown to be decidable in \cite{baet:corr:grab:2007} by defining and using `well-behaved' specifications.

In this paper we prove completeness of Milner's axiomatization (tailored to the adapted setting) for star expressions with 0, 
but without 1 and with the binary Kleene star. 

While earlier completeness proofs focus on manipulation of terms, we follow Milner's footsteps and focus on their process graphs. 
A key idea is to determine loops in graphs associated to star expressions. 
By a loop we mean a subgraph generated by a set of entry transitions from a vertex $\avert$ in which 
(1) there is an infinite path from $\avert$,
(2) each infinite path eventually returns to $\avert$, and 
(3) termination is not permitted.
A graph is said to satisfy LLEE (Layered Loop Existence and Elimination) 
if repeatedly eliminating the entry transitions of a loop, and performing garbage collection, leads to a graph without infinite paths.
LLEE offers a generalization (and more elegant definition) of the notion of a well-behaved specification.

Our completeness proof roughly works as follows (for more details see Sect.~\ref{compl:proof}). 
Let $e_1$ and $e_2$ be star expressions that have bisimilar graphs process graph interpretations $g_1$ and $g_2$.
We show that $g_1$ and $g_2$ satisfy LLEE.
We moreover prove that LLEE is preserved under bisimulation collapse.
And we construct for each graph that satisfies LLEE a star expression that corresponds to this graph, modulo bisimilarity. 
In particular such a star expression $f$ can be constructed for the bisimulation collapse of $g_1$ and $g_2$. 
We show that both $e_1$ and $e_2$ can be proven equal to $f$, 
by a pull-back over the functional bisimulations from the bisimulation collapse back to $g_1$~and~$g_2$. 
This yields the desired completeness result.

In our proof, the minimization of terms (and thereby of the associated process graphs)
in the left-hand side of a binary Kleene star modulo bisimilarity is partly inspired by \cite{fokk:1996:term:cycle:LGPS,fokk:1997:pl:ICALP}. 
Interestingly, we will be able to use as running example
the process graph interpretation of the star expression that at the end of \cite{fokk:1997:pl:ICALP} is mentioned as problematic for a completeness proof.
Our crucial use of witnesses for the graph property LLEE
borrows from the representation of cyclic $\lambda$-terms \cite{grab:roch:2013}
as \structureconstrained\ term graphs, 
as used for defining and implementing maximal sharing in the $\lambda$-calculus with {\tt letrec} \cite{grab:roch:2014} (see also \cite{grab:2018}).

The completeness result for star expressions with 0 but without 1 and with the binary Kleene star 
settles a natural question. 
We are also hopeful that the property LLEE provides a strong conceptual tool for approaching
Milner's long-standing open question regarding the class of all star expressions.
The presence of \onetransitions\ in graphs presents new challenges, such as
that LLEE is not always preserved under bisimulation collapse.
In order to be able to still work with this concept, we will need workarounds. 

This is a report version of the article \cite{grab:fokk:2020}
in the proceedings of the conference LICS~2020.
It was compiled from the submission version, containing a technical appendix.

{\em Please see the appendix for details of proofs that have been omitted or that are only sketched.}

\section{Preliminaries}%
  \label{prelims}

In this section we define star expressions,
their process semantics as `charts',
the proof system \BBP\ for bisimilarity of their chart interpretations,
and provable solutions of charts.

\begin{definition}\label{def:StExps}
  Given a set $A$ of \emph{actions}, the set $\StExpsover{\actions}$ of \emph{star expressions over $\actions$} 
  is generated by the grammar:
  \begin{center}
    $
    \astexp 
      \:\BNFdefdby\:
        \stexpzero 
          \BNFor
        \astexpact
          \BNFor 
        (\stexpsum{\astexpi{1}}{\astexpi{2}})
          \BNFor
        (\stexpprod{\astexpi{1}}{\astexpi{2}})
          \BNFor
        (\stexpbit{\astexpi{1}}{\astexpi{2}}) 
      \quad  
        \text{\nf (with $a\in A$).  } 
    $
  \end{center}
  0 represents deadlock (i.e., does not perform any action), $a$ an atomic action, $+$ alternative and $\cdot$ sequential composition, and ${}^{\sstexpbit}$ the binary Kleene star.
  Note that 1 (for empty steps) is missing from the syntax.
 $\sum_{i=1}^k e_i$ 
  is defined recursively as $0$ if $k = 0$,
  $e_1$ if $k=1$, and $\stexpsum{(\sum_{i=1}^{k-1} e_i)}{e_{k}}$ if $k > 1$.
  
  The \emph{star height} $\bsth{\astexp}$ of a star expression $\astexp\in\StExpsover{\actions}$
  denotes the maximum number of nestings of Kleene stars~in~$\astexp\,$:
  it is defined by $\bsth{\stexpzero} \defdby \bsth{\aact} \defdby 0$, 
                   $\bsth{\stexpsum{\bstexp}{\cstexp}} \defdby \bsth{\stexpprod{\bstexp}{\cstexp}}
                                                       \defdby \max\setexp{\bsth{\bstexp}, \bsth{\cstexp}}$, 
               and $\bsth{\stexpbit{\bstexp}{\cstexp}} \defdby \max\setexp{\bsth{\bstexp}+1,\bsth{\cstexp}}$.
\end{definition}

\begin{definition}\label{def:charts}
  By a (finite \sinktermination) \emph{chart} $\achart$
  we understand a 5\nb-tuple $\tuple{\verts,\tick,\start,\actions,\transs}$
  where $\verts$ is a finite set of \emph{vertices},
  $\tick$ is, in case $\tick\in\verts$,   
    a special vertex with no outgoing transitions (a sink) that indicates termination
    (in case $\tick\notin\verts$, the chart does not admit termination),
  $\start\in\verts\backslash\{\tick\}$ is the \emph{start vertex},
  $\actions$ is a set \emph{actions},
  and $\transs \subseteq \verts\times A\times\verts$ the set of
  \emph{transitions}.
  Since $\actions$ can be reconstructed from $\transs$,
  we will frequently keep $\actions$ implicit,
  denote a chart as a 4-tuple $\tuple{\verts,\tick,\start,\transs}$. 
  A chart is \emph{\startconnected} if
  every vertex is reachable by a path from the start vertex.
  This property can be achieved by removing unreachable vertices (`garbage collection').
  We will assume charts to~be~\mbox{\startconnected}. 
  
  In a chart~$\achart$, let $\avert\in\verts$ and $\asettranss\subseteq\transs$ be a set of transitions from $\avert$.
  By the \emph{\generatedby{\pair{\avert}{\asettranss}} subchart of $\achart$}
  we mean the chart $\acharti{0} = \tuple{\vertsi{0},\tick,\avert,\actions,\transsi{0}}$ with start vertex~$\avert$
  where $\vertsi{0}$ is the set of vertices and $\transsi{0}$ the set of transitions
  that are on paths in $\achart$ from $\avert$
  that first take a transition in $\asettranss$, and then, until $\avert$ is reached again,
  continue with other transitions~of~$\achart$.
  
  \noindent
  We use the standard notation $\avert \lt{\aact} \avertacc$~in~lieu~of~$\triple{\bvert}{\aact}{\bvertacc}\in\transs$. 
\end{definition}

\begin{definition}
  Let $\acharti{i} = \tuple{\vertsi{i},\tick,\starti{i},\transsi{i}}$ for $i\in\setexp{1,2}$ be two charts.
  A \emph{bisimulation} between $\acharti{1}$ and $\acharti{2}$ 
  is a relation $\abisim \subseteq \vertsi{1}\times\vertsi{2}$
  that satisfies the following conditions:  
  \begin{description}
    \item{(\emph{start}) }
       $\starti{1} \,\abisim\, \starti{2}$ (it relates the start vertices),
  \end{description}     
  and for all $\averti{1},\averti{2}\in\verts$ with $\averti{1} \,\abisim\, \averti{2}\,$:
  \begin{description}
    \item{(\emph{forth}) }
      for every transition $\averti{1} \lt{\aact} \avertacci{1}$ in $\acharti{1}$ 
        there is a transition $\averti{2} \lt{\aact} \avertacci{2}$ in $\acharti{2}$ with $\avertacci{1}\,\abisim\,\avertacci{2}$,
    \item{(\emph{back}) }
      for every transition $\averti{2} \lt{\aact} \avertacci{2}$ in $\acharti{2}$ 
        there is a transition $\averti{1} \lt{\aact} \avertacci{1}$ in $\acharti{1}$ with $\avertacci{1}\,\abisim\,\avertacci{2}$,
    \item{(\emph{termination}) }
      $\averti{1} = \tick$ if and only if $\averti{2} = \tick$.
  \end{description}
  If there is a bisimulation between $\acharti{1}$ and $\acharti{2}$,
  then we write $\acharti{1} \bisim \acharti{2}$ and say that $\acharti{1}$ and $\acharti{2}$ are \emph{bisimilar}.
  If a bisimulation is the graph of a function, we say that it is a \emph{functional} bisimulation. 
  We write $\acharti{1} \funbisim \acharti{2}$ if there is a functional bisimulation between $\acharti{1}$ and $\acharti{2}$. 
\end{definition}


\begin{definition}\label{def:chart:interpretation}
  For every star expression $\astexp\in\StExpsover{\actions}$ 
  the \emph{chart interpretation~$\chartof{\astexp} = \tuple{\vertsof{\astexp},\tick,\astexp,\actions,\transsof{\astexp}}$ of $\astexp$} is
  the chart with start vertex $\astexp$ that is specified by iteration via the following transition rules,
  which form a transition system specification (TSS),
  with $\astexp,\astexpi{1},\astexpi{2},\astexpacci{1}\in\StExpsover{\actions}$,
       $\aact\in\actions$:     
  \begin{gather*}
     \begin{aligned}
       &
       \AxiomC{$\phantom{\aact \:\lt{\aact}\: \tick}$}
       \UnaryInfC{$\aact \:\lt{\aact}\: \tick$}
       \DisplayProof\hspace*{2mm}
       & & 
       \AxiomC{$ \astexpi{i} \:\lt{a}\: \xi $}
       \RightLabel{$(i=1,2)$}
       \UnaryInfC{$ \stexpsum{\astexpi{1}}{\astexpi{2}} \:\lt{a}\: \xi $}
       \DisplayProof
     \end{aligned}
     \displaybreak[0]\\
     \begin{aligned}
       &
       \AxiomC{$ \astexpi{1} \:\lt{a}\: \astexpacci{1} $}
       \UnaryInfC{$ \stexpprod{\astexpi{1}}{\astexpi{2}} \:\lt{a}\: \stexpprod{\astexpacci{1}}{\astexpi{2}} $}
       \DisplayProof\hspace*{2mm}
       & &
       \AxiomC{$ \astexpi{1} \:\lt{a}\: \surd$}
       \UnaryInfC{$ \stexpprod{\astexpi{1}}{\astexpi{2}} \:\lt{a}\: \astexpi{2} $}
       \DisplayProof
     \end{aligned}
     \displaybreak[0]\\
     \begin{aligned}
       &
       \AxiomC{$\astexpi{1} \:\lt{a}\: \astexpacci{1}$}
       \UnaryInfC{$\stexpbit{\astexpi{1}}{\astexpi{2}} \:\lt{a}\: \stexpprod{\astexpacci{1}}{\stexpbit{(\astexpi{1}}{\astexpi{2})}}$}
       \DisplayProof
       & &
       \AxiomC{$\astexpi{1} \:\lt{a}\: \surd$}
       \UnaryInfC{$\stexpbit{\astexpi{1}}{\astexpi{2}} \:\lt{a}\: \stexpbit{\astexpi{1}}{\astexpi{2}}$}
       \DisplayProof
       & &
       \AxiomC{$\astexpi{2} \:\lt{a}\: \xi$}
       \UnaryInfC{$\stexpbit{\astexpi{1}}{\astexpi{2}} \:\lt{a}\: \xi$}
       \DisplayProof 
     \end{aligned}
  \end{gather*} 
  with $\atickstexp\in\tickStExpsover{\actions} \defdby \StExpsover{\actions}\cup\setexp{\tick}$,
  where $\tick$ indicates sink termination.
  If $\astexp \lt{\aact} \atickstexp$ can be proved, 
  $\atickstexp$ is called an \emph{\aderivative{\aact}},
  or just \emph{derivative}, of $\astexp$. 
  The set $\vertsof{\astexp} \subseteq \tickStExpsover{\actions}$ consists of the \emph{iterated derivatives} of $\astexp$. 
  To see that $\chartof{\astexp}$ is finite, 
  Antimirov's result \cite{anti:1996},
  that a regular expression has only finitely many iterated derivatives,
  can be adapted.
  
  We say that a star expression $\astexp\in\StExpsover{\actions}$ is \emph{normed} 
  if there is a path of transitions from $\astexp$ to $\tick$ in $\chartof{\astexp}$.
\end{definition}

  %
\begin{center}
\hspace*{-1em}
\begin{tikzpicture}[scale=1,every node/.style={transform shape}]
%
\matrix[anchor=north,row sep=0.9cm,every node/.style={draw,very thick,circle,minimum width=2.5pt,fill,inner sep=0pt,outer sep=2pt}] at (0,-0.5) {
  \node(v_e1_0){};
  \\
  \node(v_e1_1){};
  \\
  \node(v_e1_2){};
  \\
  \node(v_e1_0'){};
  \\
};
\calcLength(v_e1_0,v_e1_1){mylen}
\draw[<-,very thick,>=latex,chocolate](v_e1_0) -- ++ (90:{0.45*\mylen pt});
%
\draw[->] (v_e1_0) to 
                      node[left,pos=0.5]{\small $\aact$} (v_e1_1);
%
\draw[->] (v_e1_1) to 
                      node[left,pos=0.625]{\small $\aact$} (v_e1_2);
\draw[->,shorten <= 4.5pt] (v_e1_1) to[out=170,in=180,distance={0.75*\mylen pt}] node[left,pos=0.5]{\small $\cact$} (v_e1_0);
%
\draw[->] (v_e1_2) to node[right,pos=0.5]{\small $\bact$} (v_e1_0');
\draw[->] (v_e1_2) to[out=180,in=190,distance={0.75*\mylen pt}] node[left,pos=0.5]{\small $\bact$} (v_e1_1);
%
\draw[->] (v_e1_0') to[out=0,in=0,distance={1.5*\mylen pt}] node[right,pos=0.5]{\small $\aact$} (v_e1_1); 
%
\path (v_e1_0') ++ ({0*\mylen pt},{-0.6*\mylen pt}) node{\large $\chartof{\astexpi{1}}$};

\matrix[anchor=north,row sep=0.9cm,every node/.style={draw,very thick,circle,minimum width=2.5pt,fill,inner sep=0pt,outer sep=2pt}] at (2.8,0) {
  \node[draw=none,fill=none](v_1-dummy){};
  \\
  \node(v_0){};
  \\
  \node(v_1){};
  \\
  \node(v_2){};
  \\
};
\calcLength(v_0,v_1){mylen}
\draw[draw=none,<-,very thick,>=latex,chocolate](v_1-dummy) -- ++ (90:{0.45*\mylen pt});
\draw[<-,very thick,>=latex,chocolate](v_0) -- ++ (90:{0.5*\mylen pt});
\path(v_0) ++ ({-0.25*\mylen pt},{0.25*\mylen pt}) node{$\averti{0}$};
\draw[->](v_0) to node[right,xshift={-0.05*\mylen pt},pos=0.45]{\small $\aact$} (v_1); 
\path(v_1) ++ ({0.325*\mylen pt},0cm) node{$\averti{1}$};
\draw[->](v_1) to 
                  node[left,xshift={0.05*\mylen pt},pos=0.45]{\small $\aact$} (v_2);
\draw[->,shorten <= 5pt](v_1) to[out=175,in=180,distance={0.75*\mylen pt}] 
         node[above,yshift={0.05*\mylen pt},pos=0.65]{\small $\cact$} (v_0);
\path(v_2) ++ (-0cm,{-0.275*\mylen pt}) node{$\averti{2}$};
\draw[->](v_2) to[out=180,in=185,distance={0.75*\mylen pt}]  
               node[below,yshift={0.0*\mylen pt},pos=0.2]{\small $\bact$} (v_1);
\draw[->](v_2) to[out=0,in=0,distance={1.3*\mylen pt}] 
               node[below,yshift={0.00*\mylen pt},pos=0.125]{\small $\bact$} (v_0);

\draw[-,magenta,thick,densely dashed] (v_e1_0) to (v_0);
\draw[-,magenta,thick,densely dashed] (v_e1_0') to (v_0);
\draw[-,magenta,thick,densely dashed] (v_e1_1) to (v_1);
\draw[-,magenta,thick,densely dashed] (v_e1_2) to (v_2);

\path (v_2) ++ ({0*\mylen pt},{-0.8*\mylen pt}) node{\large $\chartof{\astexpi{0}}$};

%
\matrix[anchor=north,row sep=0.9cm,every node/.style={draw,very thick,circle,minimum width=2.5pt,fill,inner sep=0pt,outer sep=2pt}] at (5.85,0.1) {
  \node(v_e2_0){};
  \\
  \node(v_e2_1){};
  \\
  \node(v_e2_2){};
  \\
  \node(v_e2_0''){};
  \\
  \node(v_e2_1'){};
  \\
};
\calcLength(v_e2_0,v_e2_1){mylen}
  \draw[draw=none] (v_e2_2) arc (270:205:{\mylen pt}) node[style={draw,very thick,circle,minimum width=2.5pt,fill,inner sep=0pt,outer sep=2pt}](v_e2_0'){};
  \draw[draw=none] (v_e2_0'') arc (90:205:{\mylen pt}) node[style={draw,very thick,circle,minimum width=2.5pt,fill,inner sep=0pt,outer sep=2pt}](v_e2_0'''){};
\draw[<-,very thick,>=latex,chocolate](v_e2_0) -- ++ (90:{0.45*\mylen pt});
\draw[->](v_e2_0) to node[right,xshift={-0.05*\mylen pt},pos=0.4]{\small $\aact$} (v_e2_1);
%
%
\draw[->](v_e2_1) to 
                     node[left,pos=0.45,xshift={0.065*\mylen}]{\small $\aact$} (v_e2_2);
\draw[->](v_e2_1) to 
                     node[above,pos=0.65]{\small $\cact$} (v_e2_0');
%
\draw[->] (v_e2_0') to[out=115,in=150,distance={0.7*\mylen pt}] node[above,pos=0.75]{\small $\aact$} (v_e2_1);
%
\draw[->](v_e2_2) to 
                     node[left,pos=0.425,xshift={0.05*\mylen}]{\small $\bact$} (v_e2_0'');
\draw[->,shorten <= 5pt] (v_e2_2) to[out=10,in=0,distance=0.7cm] 
                                node[right,pos=0.5,xshift={-0.025*\mylen}]{\small $\bact$} (v_e2_1);
%
\draw[->](v_e2_0'') to node[right,xshift={-0.05*\mylen pt},pos=0.4]{\small $\aact$} (v_e2_1');
%
\draw[->](v_e2_1') to 
                      node[above,pos=0.65]{\small $\cact$} (v_e2_0''');
\draw[->](v_e2_1') to[out=0,in=0,distance=1.4cm] node[right,pos=0.6,xshift={-0.025*\mylen}]{\small $\aact$} (v_e2_2);
%
\draw[->] (v_e2_0''') to[out=115,in=150,distance={0.7*\mylen pt}] node[above,pos=0.4]{\small $\aact$} (v_e2_1');

\path (v_e2_1') ++ ({1*\mylen pt},{-0.3*\mylen pt}) node{\large $\chartof{\astexpi{2}}$};

\draw[-,magenta,thick,densely dashed] (v_e2_0) to (v_0);
\draw[-,magenta,thick,densely dashed] (v_e2_0') to (v_0);
\draw[-,magenta,thick,densely dashed] (v_e2_0'') to (v_0);
\draw[-,magenta,thick,densely dashed] (v_e2_0''') to (v_0);
\draw[-,magenta,thick,densely dashed,out=160,bend right,distance={0.75*\mylen pt}] (v_e2_1) to (v_1);
\draw[-,magenta,thick,densely dashed] (v_e2_1') to (v_1);
\draw[-,magenta,thick,densely dashed] (v_e2_2) to (v_2);

\end{tikzpicture} 
\end{center}
  %
\begin{example}\label{ex:chart:interpretation}
  By the rules in Def.~\ref{def:chart:interpretation},
  $\astexpi{0} \defdby \stexpprod{\aact}{\astexpacci{0}}$
  with 
  $\astexpacci{0} \defdby \stexpbit{(\stexpsum{\stexpprod{\cact}{\aact}}
                                              {\stexpprod{\aact}{(\stexpsum{\bact}{\stexpprod{\bact}{\aact}})})})}{\stexpzero}$
  has the chart $\chartof{\astexpi{0}}$ as above,
  with $\averti{0} \defdby \astexpi{0}$,
       $\averti{1} \defdby \astexpacci{0}$
       and 
       $\averti{2} \defdby \stexpprod{(\stexpsum{\bact}{\stexpprod{\bact}{\aact}})}{\astexpacci{0}}$. 
  This chart is the bisimulation collapse of the 
  charts $\chartof{\astexpi{1}}$ and $\chartof{\astexpi{2}}$
  of star expressions
  $\astexpi{1} \defdby \stexpbit{(\stexpprod{\aact}{(\stexpbit{(\stexpprod{\aact}{(\stexpsum{\bact}{\stexpprod{\bact}{\aact}})})}{\cact})})}{\stexpzero}\,$, 
  and 
  $\astexpi{2} \defdby \stexpprod{\aact}{(\stexpbit{(\stexpsum{\stexpprod{\cact}{\aact}}
                                 {\stexpprod{\aact}{\stexpbit{(\stexpprod{\bact}{\stexpprod{\aact}{(\stexpbit{(\stexpprod{\cact}{\aact})}{\aact})}})}{\bact}}})} {\stexpzero})}$.
  Bisimulations between $\chartof{\astexpi{1}}$ and $\chartof{\astexpi{0}}$, and between $\chartof{\astexpi{0}}$ and $\chartof{\astexpi{1}}$
  are indicated by the broken lines.
  The chart~$\chartof{\astexpi{0}}$ was considered problematic in \cite{fokk:1997:pl:ICALP}. 
\end{example}

\begin{example}\label{ex:not:expressible}
  The left chart below does not admit termination. 
  The right chart is a double-exit graph with the sink termination vertex~$\tick$ at the bottom.
  
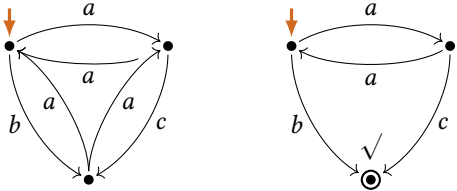
\begin{center}
\begin{tikzpicture}
  
%
\matrix[anchor=north,row sep=0.8cm,column sep=0.924cm,ampersand replacement=\&,
        every node/.style={draw,very thick,circle,minimum width=2.5pt,fill,inner sep=0pt,outer sep=2pt}] at (3.75,0) {
  \node(C_1-0){};  \&                  \&     \node(C_1-1){};
  \\
                   \&                  \&                  
  \\
                   \& \node(C_1-2){};  \&
  \\
};
\draw[<-,very thick,>=latex,color=chocolate](C_1-0) -- ++ (90:0.5cm);

\draw[thick] (C_1-2) circle (0.12cm);
\path (C_1-2) ++ (0cm,0.45cm) node{$\tick$};

\draw[->,bend left,distance=0.6cm] (C_1-0) to node[above]{$a$} (C_1-1); 
\draw[->,bend left,distance=0.5cm] (C_1-1) to node[below]{$a$} (C_1-0); 

\draw[->,bend right,distance=0.6cm,shorten >=2pt] (C_1-0) to node[left]{$b$} (C_1-2);
\draw[->,bend left,distance=0.6cm,shorten >=2pt] (C_1-1) to node[right]{$c$} (C_1-2);

%
\matrix[anchor=north,row sep=0.8cm,column sep=0.924cm,ampersand replacement=\&,
        every node/.style={draw,very thick,circle,minimum width=2.5pt,fill,inner sep=0pt,outer sep=2pt}] at (0,0) {
  \node(C_2-0){};  \&                  \&     \node(C_2-1){};
  \\
                   \&                  \&                  
  \\
                   \& \node(C_2-2){};  \&
  \\
};
\draw[<-,very thick,>=latex,color=chocolate](C_2-0) -- ++ (90:0.5cm);

\draw[->,bend left,distance=0.6cm]  (C_2-0) to node[above]{$a$} (C_2-1); 
\draw[->,bend right,distance=0.6cm] (C_2-0) to node[left]{$b$}  (C_2-2);

\draw[->,bend left,distance=0.5cm,shorten <=9pt] ($(C_2-1)+(-0.125cm,0cm)$) to node[below]{$a$} (C_2-0); 
\draw[->,bend left,distance=0.6cm]               (C_2-1) to node[right]{$c$} (C_2-2);

\draw[->,bend right,distance=0.5cm] (C_2-2) to node[left]{$a$}  (C_2-0);
\draw[->,bend left,distance=0.5cm]  (C_2-2) to node[right]{$a$} (C_2-1);

\end{tikzpicture}
\end{center}
  %
  These charts are not bisimilar to chart interpretations of star expressions.
  For the left chart this was shown by Milner~\cite{miln:1984}, 
  and for the right chart by Bosscher~\cite{boss:1997}.
\end{example}

\begin{definition}\label{def:BBP}
  The proof system \BBP\ or the class of star expressions has the axioms (B1)--(B6), (BKS1), (BKS2), 
  the inference rules of equational logic,
  and the rule $\RSPbit$:
  \begin{alignat*}{2}
    {(\text{B1})} & \quad\; &
      \stexpsum{\avar}{\bvar}
        & \hspace{2mm}=\hspace{2mm}
      \stexpsum{\bvar}{\avar} 
    \displaybreak[0]\\[-0.25ex]
    {(\text{B2})} & & 
      \stexpsum{(\stexpsum{\avar}{\bvar})}
               {\cvar}
        & \hspace{2mm}=\hspace{2mm}
      \stexpsum{\avar}
               {(\stexpsum{\bvar}{\cvar})} 
    \displaybreak[0]\\[-0.25ex]             
    {(\text{B3})} & &
      \stexpsum{\avar}{\avar}
        & \hspace{2mm}=\hspace{2mm}
      \avar 
    \displaybreak[0]\\[-0.25ex]
    {(\text{B4})} & &
      \stexpprod{(\stexpsum{\avar}{\bvar})}
                {\cvar}
      & \hspace{2mm}=\hspace{2mm}
      \stexpsum{\stexpprod{\avar}{\cvar}}
               {\stexpprod{\bvar}{\cvar}}  
    \displaybreak[0]\\[-0.25ex]
    {(\text{B5})} & &
      \stexpprod{(\stexpprod{\avar}{\bvar})}
                {\cvar}
      & \hspace{2mm}=\hspace{2mm}
      \stexpprod{\avar}
                {(\stexpprod{\bvar}{\cvar})}
    \displaybreak[0]\\[-0.25ex]
    {(\text{B6})} & \quad\; &
          \stexpsum{\avar}{\stexpzero}
      & \hspace{2mm}=\hspace{2mm}
      \avar 
    \displaybreak[0]\\[-0.25ex]
    {(\text{B7})} & &
      \stexpprod{\stexpzero}{\avar}
      & \hspace{2mm}=\hspace{2mm}
      \stexpzero
    \displaybreak[0]\\[-0.25ex]     
    {(\text{BKS1})} & &
      \stexpsum{\stexpprod{\avar}{(\stexpbit{\avar}{\bvar})}}
               {\bvar}
      & \hspace{2mm}=\hspace{2mm}
      \stexpbit{\avar}{\bvar} 
    \displaybreak[0]\\[-0.25ex]  
    {(\text{BKS2})} & &
      \stexpprod{(\stexpbit{\avar}{\bvar})}
                {\cvar}
      & \hspace{2mm}=\hspace{2mm}
      \stexpbit{\avar}
               {(\stexpprod{\bvar}{\cvar})}
    \displaybreak[0]\\[-0.25ex]           
    {(\text{$\RSPbit$})} & & & \hspace{-1.925ex}
       \begin{gathered}    
         \Axiom$\avar \hspace{2mm}\fCenter=\hspace{2mm} \stexpsum{\stexpprod{(\bvar}{\avar)\,}}{\,\cvar} $
         \UnaryInf$ \avar \hspace{2mm}\fCenter=\hspace{2mm} \stexpbit{\bvar}{\cvar} $
         \DisplayProof
       \end{gathered}
  \end{alignat*}
  By $\astexpi{1} \BBPeq \astexpi{2}$ we denote that $\astexpi{1} = \astexpi{2}$ is derivable in \BBP.

\end{definition}

\BBP\ is a finite `implicational' proof system \cite{tayl:1977}, because
unlike in Salomaa's and Milner's systems for regular expressions with $\stexpone$
the fixed-point rule does not require any side-condition to ensure `guardedness'.  
%
%

\begin{definition}\label{def:provable-solution}
  For a chart $\achart = \tuple{\verts,\tick,\start,\actions,\transs}$,
  a \emph{provable solution of $\achart$}
    is a function $\sasol \funin \verts\setminus\setexp{\tick} \to \StExpsover{\actions}$ 
    such that:\vspace*{-1mm}
  \begin{alignat*}{2}
    \asol{\avert} & \BBPeq \Bigl(\sum_{i=1}^m a_i\Bigr)+\Bigl(\sum_{j=1}^n b_j\cdot s(w_j)\Bigr)
                  & & \;\;\text{(for all $\avert\in\verts\setminus\{\tick\}$)}\vspace*{-2mm}
  \end{alignat*}
  holds, \vspace*{-.5mm}given that the union of  
  $ \descsetexpbig{ \avert \lt{\aacti{i}} \tick }{ i=1,\ldots,m } $
      and
  $ \descsetexpbig{ \avert \lt{\bacti{j}} \bverti{j} }{ j=1,\ldots,n ,\, \bverti{j}\neq\tick } $ 
  is the set of transitions from $\avert$ in $\achart$.
  We call $\asol{\start}$ the \emph{principal value} of $\sasol$.
\end{definition}

\begin{proposition}[uses \BBP-axioms (B1)--(B7), (BKS1)]\label{prop:id:is:sol:chart:interpretation}
  For every $\astexp\in\StExpsover{\actions}$,
  the identity function $\sidfunon{\vertsof{\astexp}} \funin \vertsof{\astexp} \to \vertsof{\astexp}\subseteq\StExpsover{\actions}$, $\astexpacc \mapsto \astexpacc$,
  is a provable solution of the chart interpretation $\chartof{\astexp}$ of $\astexp$. 
\end{proposition}

\begin{proof}[Proof (Idea)]
  Each $\astexp$ in $\StExpsover{\actions}$ 
  is the \provablein{\BBP} sum of expressions $\aact$ and $\stexpprod{\aact}{\astexpacc}$ over all $\aact\in\actions$
  for \aderivatives{\aact} $\tick$ and $\astexpacc$, respectively, of $\astexp$.
  This `fundamental theorem%
    \footnote{Rutten \cite{rutt:2005} used this name for an analogous result on infinite streams \cite{rutt:2005}.
              The first author \cite{grab:2005}, and Kozen and Silva \cite{koze:silv:2020,silv:2010} used it
              for the provable synthesis of regular expressions from their Brzozowski derivatives.
              The result here can be viewed as stating the provable synthesis of regular expressions from their partial derivatives (due to Antimirov \cite{anti:1996}).}
        of differential calculus for star expressions'
  implies, quite directly, that $\sidfunon{\vertsof{\astexp}}$
  is a provable solution of $\chartof{\astexp}$. 
\end{proof}

\section{Layered loop existence and elimination}%
  \label{LLEE}
\renewcommand{\ll}[1]{#1}

As preparation for the definition of the central concept of `\LLEEwitness', 
we start with an informal explanation of the structural chart property `\LEE'.
It is a necessary condition for a chart to be the chart interpretation of a star expression. 
\LEE\ is defined by a dynamic elimination procedure that analyses the structure of the graph
by peeling off `loop sub\-charts'. Such subcharts capture,
within the chart interpretation of a star expression~$\astexp$,
the behaviour of the iteration of $\bstexpi{1}$ within innermost subterms $\stexpbit{\bstexpi{1}}{\bstexpi{2}}$ in~$\astexp$.
(A weaker form of `loop' by Milner \cite{miln:1984}, which describes the behavior of general iteration subterms,
 is not sufficient for our aims.)  

\begin{definition}\label{def:loop:chart}
  A chart $\aloop = \tuple{\verts,\tick,\start,\transs}$ is a \emph{loop chart} if:
  \begin{enumerate}[label={{\rm (L\arabic*)}},leftmargin=*,align=left,itemsep=0.5ex]
    \item{}\label{loop:1}
      There is an infinite path from the start vertex $\start$.
    \item{}\label{loop:2}  
      Every infinite path from $\start$ returns to $\start$ after a positive number of transitions
      (and so visits $\start$ infinitely often).
    \item{}\label{loop:3}
       $\verts$ does not contain the vertex $\tick$.
  \end{enumerate}
  In such a loop chart we call the transitions from $\start$ \emph{\loopentry\ transitions},
  and all other transitions \emph{\loopbody\ transitions}.
  
  Let $\achart$ be a chart. A loop chart $\aloop$ is called a \emph{loop subchart of $\achart$}
  if $\aloop$ is the \generatedby{\pair{\avert}{\asettranss}} subchart of $\achart$
  for some vertex $\avert$ of $\achart$, and a set $\asettranss$ of transitions of $\achart$ that depart from $\avert$
  (so the transitions in $\asettranss$ are the \loopentrytransitions~of~$\aloop$).
\end{definition}

Note that the two charts in Ex.~\ref{ex:not:expressible} are not loop charts:
the left one violates \ref{loop:2}, and the right one violates \ref{loop:3}.
Moreover, none of these charts contains a loop subchart. 
While the chart $\chartof{\astexpi{0}}$ in Ex.~\ref{ex:chart:interpretation} is not
a loop chart either, as it violates \ref{loop:2}, we will see that it has loop subcharts. 

Let $\aloop$ be a loop subchart of a chart~$\achart$.
Then the result of \emph{eliminating $\aloop$ from $\achart$}
arises by removing all \loopentrytransitions\ of $\aloop$ from $\achart$, 
and then removing all vertices and transitions that get unreachable. 
We say that a chart $\achart$ has the \emph{loop existence and elimination property (LEE)}
if the process, started on~$\achart$, of repeated eliminations of loop subcharts
results in a chart that does not have an infinite path.

For the charts in Ex.~\ref{ex:not:expressible} the procedure stops immediately,
as they do not contain loop subcharts. Since both of them have infinite paths,
it follows that they do not satisfy LEE. 

We consider three runs of the elimination procedure for the
chart~$\chartof{\astexpi{0}}$ in Ex.~\ref{ex:chart:interpretation}. 
The \loopentrytransitions\ of loop subcharts that are removed 
in each step are marked in bold.
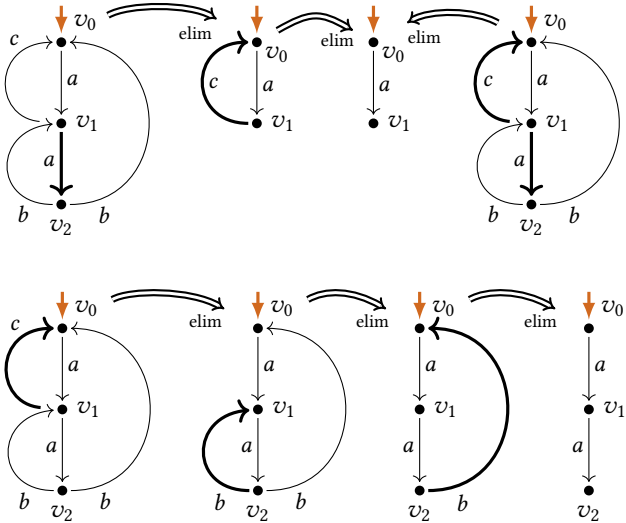
\begin{center}
\begin{tikzpicture}
  
\matrix[anchor=north,row sep=0.9cm,every node/.style={draw,very thick,circle,minimum width=2.5pt,fill,inner sep=0pt,outer sep=2pt}] at (0,0) {
  \node(v_0){};
  \\
  \node(v_1){};
  \\
  \node(v_2){};
  \\
};
\calcLength(v_0,v_1){mylen}
\draw[<-,very thick,>=latex,chocolate](v_0) -- ++ (90:{0.45*\mylen pt});
\path(v_0) ++ ({0.3*\mylen pt},{0.25*\mylen pt}) node{$\averti{0}$};
\draw[->](v_0) to node[right,xshift={-0.05*\mylen pt},pos=0.45]{\small $\aact$} (v_1); 
\path(v_1) ++ ({0.325*\mylen pt},0cm) node{$\averti{1}$};
\draw[->,very thick](v_1) to 
                  node[left,xshift={0.05*\mylen pt},pos=0.45]{\small $\aact$} (v_2);
\draw[->,shorten <= 5pt](v_1) to[out=175,in=180,distance={0.75*\mylen pt}]
                              node[above,yshift={0.05*\mylen pt},pos=0.7]{\small $\cact$} (v_0);
\path(v_2) ++ (-0cm,{-0.275*\mylen pt}) node{$\averti{2}$};
\draw[->](v_2) to[out=180,in=185,distance={0.75*\mylen pt}]  
               node[below,yshift={0.0*\mylen pt},pos=0.2]{\small $\bact$} (v_1);
\draw[->](v_2) to[out=0,in=0,distance={1.3*\mylen pt}] 
               node[below,yshift={0.00*\mylen pt},pos=0.125]{\small $\bact$} (v_0);

\matrix[anchor=north,row sep=0.9cm,every node/.style={draw,very thick,circle,minimum width=2.5pt,fill,inner sep=0pt,outer sep=2pt}] at (2.6,0) {
  \node(v_01){};
  \\
  \node(v_11){};
  \\
  \node[draw=none,fill=none](v_21){};
  \\
};
\calcLength(v_0,v_1){mylen}
\draw[<-,very thick,>=latex,chocolate](v_01) -- ++ (90:{0.45*\mylen pt});
\path(v_01) ++ ({0.25*\mylen pt},{-0.15*\mylen pt}) node{$\averti{0}$};
\draw[->](v_01) to node[right,xshift={-0.05*\mylen pt},pos=0.55]{\small $\aact$} (v_11);
\path(v_11) ++ ({0.325*\mylen pt},0cm) node{$\averti{1}$};
\draw[->,very thick](v_11) to[out=180,in=180,distance={0.75*\mylen pt}]
                           node[right,xshift={-0.05*\mylen pt},pos=0.5]{\small $\cact$} (v_01);

\draw[-implies,thick,double equal sign distance, bend left,distance={0.85*\mylen pt},
               shorten <= 0.6cm,shorten >= 0.5cm
               ] (v_0) to node[below,pos=0.7]{\scriptsize elim} (v_01);

\matrix[anchor=north,row sep=0.9cm,every node/.style={draw,very thick,circle,minimum width=2.5pt,fill,inner sep=0pt,outer sep=2pt}] at (4.15,0) {
  \node(v_02){};
  \\
  \node(v_12){};
  \\
  \node[draw=none,fill=none](v_2){};
  \\
};
\calcLength(v_0,v_1){mylen}
\draw[<-,very thick,>=latex,chocolate](v_02) -- ++ (90:{0.45*\mylen pt});
\path(v_02) ++ ({0.25*\mylen pt},{-0.15*\mylen pt}) node{$\averti{0}$};
\draw[->](v_02) to node[right,xshift={-0.05*\mylen pt},pos=0.55]{\small $\aact$} (v_12); 
\path(v_12) ++ ({0.325*\mylen pt},0cm) node{$\averti{1}$};

\draw[-implies,thick,double equal sign distance, bend left,distance={0.65*\mylen pt},
               shorten <= 0.2cm,shorten >= 0.2cm
               ] (v_01) to node[below,pos=0.75]{\scriptsize elim} (v_02);

\matrix[anchor=north,row sep=0.9cm,every node/.style={draw,very thick,circle,minimum width=2.5pt,fill,inner sep=0pt,outer sep=2pt}] at (6.25,0) {
  \node(v_0-rep){};
  \\
  \node(v_1-rep){};
  \\
  \node(v_2-rep){};
  \\
};
\calcLength(v_0-rep,v_1-rep){mylen}
\draw[<-,very thick,>=latex,chocolate](v_0-rep) -- ++ (90:{0.45*\mylen pt});
\path(v_0-rep) ++ ({0.3*\mylen pt},{0.25*\mylen pt}) node{$\averti{0}$};
\draw[->](v_0-rep) to node[right,xshift={-0.05*\mylen pt},pos=0.45]{\small $\aact$} (v_1-rep); 
\path(v_1-rep) ++ ({0.325*\mylen pt},0cm) node{$\averti{1}$};
\draw[->,very thick](v_1-rep) to node[left,xshift={0.05*\mylen pt},pos=0.45]{\small $\aact$} (v_2-rep);
\draw[->,very thick,shorten <= 5pt](v_1-rep) to[out=175,in=180,distance={0.75*\mylen pt}]
                                             node[right,xshift={-0.05*\mylen pt},pos=0.5]{\small $\cact$} (v_0-rep);
\path(v_2-rep) ++ (-0cm,{-0.275*\mylen pt}) node{$\averti{2}$};
\draw[->](v_2-rep) to[out=180,in=185,distance={0.75*\mylen pt}]  
               node[below,yshift={0.0*\mylen pt},pos=0.2]{\small $\bact$} (v_1-rep);
\draw[->](v_2-rep) to[out=0,in=0,distance={1.3*\mylen pt}] 
               node[below,yshift={0.00*\mylen pt},pos=0.125]{\small $\bact$} (v_0-rep);

\draw[-implies,thick,double equal sign distance, bend right,distance={0.7*\mylen pt},
               shorten <= 0.4cm,shorten >= 0.4cm
               ] (v_0-rep) to node[below,pos=0.7]{\scriptsize elim} (v_02);

\end{tikzpicture}  
\end{center}
  %
\begin{center}  
\begin{tikzpicture}
  %
  %
\matrix[anchor=north,row sep=0.9cm,every node/.style={draw,very thick,circle,minimum width=2.5pt,fill,inner sep=0pt,outer sep=2pt}] at (0,0) {
  \node(v_0){};
  \\
  \node(v_1){};
  \\
  \node(v_2){};
  \\
};
\calcLength(v_0,v_1){mylen}
\draw[<-,very thick,>=latex,chocolate](v_0) -- ++ (90:{0.45*\mylen pt});
\path(v_0) ++ ({0.3*\mylen pt},{0.25*\mylen pt}) node{$\averti{0}$};
\draw[->](v_0) to node[right,xshift={-0.05*\mylen pt},pos=0.45]{\small $\aact$} (v_1); 
\path(v_1) ++ ({0.325*\mylen pt},0cm) node{$\averti{1}$};
\draw[->](v_1) to node[left,xshift={0.05*\mylen pt},pos=0.45]{\small $\aact$} (v_2);
\draw[->,very thick,shorten <= 5pt](v_1) to[out=175,in=180,distance={0.75*\mylen pt}]
         node[above,yshift={0.05*\mylen pt},pos=0.7]{\small $\cact$} (v_0);
\path(v_2) ++ (-0cm,{-0.275*\mylen pt}) node{$\averti{2}$};
\draw[->](v_2) to[out=180,in=185,distance={0.75*\mylen pt}]  
               node[below,yshift={0.0*\mylen pt},pos=0.2]{\small $\bact$} (v_1);
\draw[->](v_2) to[out=0,in=0,distance={1.3*\mylen pt}] 
               node[below,yshift={0.00*\mylen pt},pos=0.125]{\small $\bact$} (v_0);

\matrix[anchor=north,row sep=0.9cm,every node/.style={draw,very thick,circle,minimum width=2.5pt,fill,inner sep=0pt,outer sep=2pt}] at (2.6,0) {
  \node(v_01){};
  \\
  \node(v_11){};
  \\
  \node(v_21){};
  \\
};
\calcLength(v_01,v_11){mylen}
\draw[<-,very thick,>=latex,chocolate](v_01) -- ++ (90:{0.45*\mylen pt});
\path(v_01) ++ ({0.3*\mylen pt},{0.25*\mylen pt}) node{$\averti{0}$};
\draw[->](v_01) to node[right,xshift={-0.05*\mylen pt},pos=0.45]{\small $\aact$} (v_11); 
\path(v_11) ++ ({0.325*\mylen pt},0cm) node{$\averti{1}$};
\draw[->](v_11) to node[left,xshift={0.05*\mylen pt},pos=0.45]{\small $\aact$} (v_21);
\path(v_21) ++ (-0cm,{-0.275*\mylen pt}) node{$\averti{2}$};
\draw[->,very thick](v_21) to[out=180,in=185,distance={0.75*\mylen pt}]  
                           node[below,yshift={0.0*\mylen pt},pos=0.2]{\small $\bact$} (v_11);
\draw[->](v_21) to[out=0,in=0,distance={1.3*\mylen pt}] 
               node[below,yshift={0.00*\mylen pt},pos=0.125]{\small $\bact$} (v_01);

\draw[-implies,thick,double equal sign distance, bend left,distance={0.85*\mylen pt},
               shorten <= 0.65cm,shorten >= 0.4cm
               ] (v_0) to node[below,pos=0.75]{\scriptsize elim} (v_01);

\matrix[anchor=north,row sep=0.9cm,every node/.style={draw,very thick,circle,minimum width=2.5pt,fill,inner sep=0pt,outer sep=2pt}] at (4.75,0) {
  \node(v_02){};
  \\
  \node(v_12){};
  \\
  \node(v_22){};
  \\
};
\calcLength(v_02,v_12){mylen}
\draw[<-,very thick,>=latex,chocolate](v_02) -- ++ (90:{0.45*\mylen pt});
\path(v_02) ++ ({0.3*\mylen pt},{0.25*\mylen pt}) node{$\averti{0}$};
\draw[->](v_02) to node[right,xshift={-0.05*\mylen pt},pos=0.45]{\small $\aact$} (v_12); 
\path(v_12) ++ ({0.325*\mylen pt},0cm) node{$\averti{1}$};
\draw[->](v_12) to node[left,xshift={0.05*\mylen pt},pos=0.45]{\small $\aact$} (v_22);
\path(v_22) ++ (-0cm,{-0.275*\mylen pt}) node{$\averti{2}$};

\draw[->,very thick](v_22) to[out=0,in=0,distance={1.3*\mylen pt}] 
               node[below,yshift={0.00*\mylen pt},pos=0.125]{\small $\bact$} (v_02);

\draw[-implies,thick,double equal sign distance, bend left,distance={0.85*\mylen pt},
               shorten <= 0.65cm,shorten >= 0.4cm
               ] (v_01) to node[below,pos=0.75]{\scriptsize elim} (v_02);

\matrix[anchor=north,row sep=0.9cm,every node/.style={draw,very thick,circle,minimum width=2.5pt,fill,inner sep=0pt,outer sep=2pt}] at (7,0) {
  \node(v_03){};
  \\
  \node(v_13){};
  \\
  \node(v_23){};
  \\
};
\calcLength(v_03,v_13){mylen}
\draw[<-,very thick,>=latex,chocolate](v_03) -- ++ (90:{0.45*\mylen pt});
\path(v_03) ++ ({0.3*\mylen pt},{0.25*\mylen pt}) node{$\averti{0}$};
\draw[->](v_03) to node[right,xshift={-0.05*\mylen pt},pos=0.45]{\small $\aact$} (v_13); 
\path(v_13) ++ ({0.325*\mylen pt},0cm) node{$\averti{1}$};
\draw[->](v_13) to node[left,xshift={0.05*\mylen pt},pos=0.45]{\small $\aact$} (v_23);
\path(v_23) ++ (-0cm,{-0.275*\mylen pt}) node{$\averti{2}$};

\draw[-implies,thick,double equal sign distance, bend left,distance={0.85*\mylen pt},
               shorten <= 0.65cm,shorten >= 0.4cm
               ] (v_02) to node[below,pos=0.75]{\scriptsize elim} (v_03);


 %
 %
\end{tikzpicture}
\end{center} 
  %
Each run witnesses that $\achart$ satisfies \LEE. 
Note that loop elimination does not yield a unique result.%
  \footnote{
    Confluence, and unique normalization, can be shown 
    if a pruning operation is added that permits to drop transitions to deadlocking vertices.}
Runs can be recorded by attaching, in the original chart, to transitions 
that get removed in the elimination procedure as marking label
the sequence number of the appertaining elimination step. 
For the three runs of loop elimination above we get the following 
marking labeled versions of $\achart$, respectively:
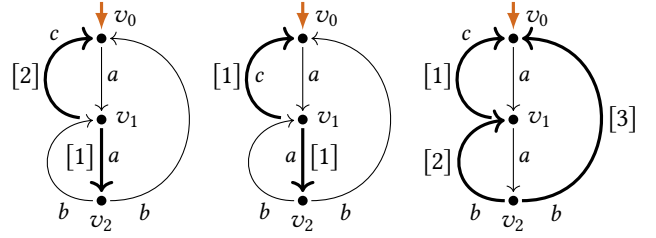
\begin{center}
\begin{tikzpicture}
  %
  %
  
\matrix[anchor=north,row sep=0.9cm,every node/.style={draw,very thick,circle,minimum width=2.5pt,fill,inner sep=0pt,outer sep=2pt}] at (5.475,0) {
  \node(v_0-hat-1){};
  \\
  \node(v_1-hat-1){};
  \\
  \node(v_2-hat-1){};
  \\
};
\calcLength(v_0,v_1){mylen}
\draw[<-,very thick,>=latex,chocolate](v_0-hat-1) -- ++ (90:{0.45*\mylen pt});
\path(v_0-hat-1) ++ ({0.3*\mylen pt},{0.25*\mylen pt}) node{$\averti{0}$};
\draw[->](v_0-hat-1) to node[right,xshift={-0.05*\mylen pt},pos=0.45]{\small $\aact$} (v_1-hat-1); 
\path(v_1-hat-1) ++ ({0.325*\mylen pt},0cm) node{$\averti{1}$};
\draw[->](v_1-hat-1) to node[right,xshift={-0.05*\mylen pt},pos=0.45]{\small $\aact$} (v_2-hat-1);
\draw[->,very thick,shorten <= 5pt](v_1-hat-1) to[out=175,in=180,distance={0.75*\mylen pt}] 
         node[left,pos=0.5,xshift={0.05*\mylen pt}]{$\loopnsteplab{1}$} 
         node[above,yshift={0.05*\mylen pt},pos=0.7]{\small $\cact$} (v_0-hat-1);
\path(v_2-hat-1) ++ (-0cm,{-0.275*\mylen pt}) node{$\averti{2}$};
\draw[->,very thick](v_2-hat-1) to[out=180,in=185,distance={0.75*\mylen pt}] 
               node[left,pos=0.5,xshift={0.05*\mylen pt}]{$\loopnsteplab{2}$} 
               node[below,yshift={0.0*\mylen pt},pos=0.2]{\small $\bact$} (v_1-hat-1);
\draw[->,very thick](v_2-hat-1) to[out=0,in=0,distance={1.3*\mylen pt}] 
               node[right,pos=0.5,xshift={-0.05*\mylen pt}]{$\loopnsteplab{3}$} 
               node[below,yshift={0.00*\mylen pt},pos=0.125]{\small $\bact$} (v_0-hat-1);

\matrix[anchor=north,row sep=0.9cm,every node/.style={draw,very thick,circle,minimum width=2.5pt,fill,inner sep=0pt,outer sep=2pt}] at (0,0) {
  \node(v_0-hat-2){};
  \\
  \node(v_1-hat-2){};
  \\
  \node(v_2-hat-2){};
  \\
};
\calcLength(v_0,v_1){mylen}
\draw[<-,very thick,>=latex,chocolate](v_0-hat-2) -- ++ (90:{0.45*\mylen pt});
\path(v_0-hat-2) ++ ({0.3*\mylen pt},{0.25*\mylen pt}) node{$\averti{0}$};
\draw[->](v_0-hat-2) to node[right,xshift={-0.05*\mylen pt},pos=0.45]{\small $\aact$} (v_1-hat-2); 
\path(v_1-hat-2) ++ ({0.325*\mylen pt},0cm) node{$\averti{1}$};
\draw[->,very thick](v_1-hat-2) to node[right,xshift={-0.05*\mylen pt},pos=0.45]{\small $\aact$} 
                                   node[left,pos=0.45,xshift={0.05*\mylen pt}]{$\loopnsteplab{1}$} (v_2-hat-2);
\draw[->,very thick,shorten <= 5pt](v_1-hat-2) to[out=175,in=180,distance={0.75*\mylen pt}] 
         node[left,pos=0.5,xshift={0.05*\mylen pt}]{$\loopnsteplab{2}$} 
         node[above,yshift={0.05*\mylen pt},pos=0.7]{\small $\cact$} (v_0-hat-2);
\path(v_2-hat-2) ++ (-0cm,{-0.275*\mylen pt}) node{$\averti{2}$};
\draw[->](v_2-hat-2) to[out=180,in=185,distance={0.75*\mylen pt}] 
               node[below,yshift={0.0*\mylen pt},pos=0.2]{\small $\bact$} (v_1-hat-2);
\draw[->](v_2-hat-2) to[out=0,in=0,distance={1.3*\mylen pt}]  
               node[below,yshift={0.00*\mylen pt},pos=0.125]{\small $\bact$} (v_0-hat-2);

\matrix[anchor=north,row sep=0.9cm,every node/.style={draw,very thick,circle,minimum width=2.5pt,fill,inner sep=0pt,outer sep=2pt}] at (2.675,0) {
  \node(v_0-hat-3){};
  \\
  \node(v_1-hat-3){};
  \\
  \node(v_2-hat-3){};
  \\
};
\calcLength(v_0,v_1){mylen}
\draw[<-,very thick,>=latex,chocolate](v_0-hat-3) -- ++ (90:{0.45*\mylen pt});
\path(v_0-hat-3) ++ ({0.3*\mylen pt},{0.25*\mylen pt}) node{$\averti{0}$};
\draw[->](v_0-hat-3) to node[right,xshift={-0.05*\mylen pt},pos=0.45]{\small $\aact$} (v_1-hat-3); 
\path(v_1-hat-3) ++ ({0.325*\mylen pt},0cm) node{$\averti{1}$};
\draw[->,very thick](v_1-hat-3) to node[left,xshift={0.05*\mylen pt},pos=0.45]{\small $\aact$} 
                                   node[right,pos=0.45,xshift={-0.05*\mylen pt}]{$\loopnsteplab{1}$} (v_2-hat-3);
\draw[->,very thick,shorten <= 5pt](v_1-hat-3) to[out=175,in=180,distance={0.75*\mylen pt}] 
                                                 node[left,pos=0.5,xshift={0.05*\mylen pt}]{$\loopnsteplab{1}$} 
                                                 node[right,xshift={-0.05*\mylen pt},pos=0.5]{\small $\cact$} (v_0-hat-3);
\path(v_2-hat-3) ++ (-0cm,{-0.275*\mylen pt}) node{$\averti{2}$};
\draw[->](v_2-hat-3) to[out=180,in=185,distance={0.75*\mylen pt}] 
                     node[below,yshift={0.0*\mylen pt},pos=0.2]{\small $\bact$} (v_1-hat-3);
\draw[->](v_2-hat-3) to[out=0,in=0,distance={1.3*\mylen pt}] 
                     node[below,yshift={0.00*\mylen pt},pos=0.125]{\small $\bact$} (v_0-hat-3);

\end{tikzpicture}  
\end{center}
Since all three runs were successful (as they yield charts without infinite paths), 
these recordings (marking-labeled charts) can be viewed as `\LEEwitnesses'.
We now will define a concept of a `layered \LEEwitness' (\LLEEwitness), i.e., a \LEEwitness\
with the added constraint that in the formulated run of the loop elimination procedure
it never happens that a \loopentrytransition\ is removed from within the body of 
a previously removed loop subchart. This refined concept has simpler properties, and it will fit our purpose.  

Before introducing `\LLEEwitnesses', we first define chart labelings
that mark transitions in a chart as `(loop-)entry' and as `(loop-)body' transitions,
but without safeguarding that these markings refer to actual loops. 


\begin{definition}\label{def:entry-body-labeling}
  Let $\achart = \tuple{\verts,\start,\tick,\actions,\transs}$ be a chart. 
  An \emph{\entrybodylabeling}~$\acharthat = \tuple{\verts,\start,\tick,\actions\times\nat,\transshat}$ of $\achart$ is 
  a chart 
  that arises from $\achart$ by adding, for each transition~$\atrans = \triple{\averti{1}}{\aact}{\averti{2}}\in\transs$, 
  to the action label $\aact$ of $\atrans$ a 
  \emph{marking label} $\aLname\in\nat$,
  yielding \mbox{$\atranshat = \triple{\averti{1}}{\pair{\aact}{\aLname}}{\averti{2}} \in \transshat$}. 
  In such an \entrybodylabeling\ we call transitions with marking label $0$ \emph{body transitions},
  and transitions with marking labels in $\natplus$ \emph{\entrytransitions}.
  
  Let $\acharthat$ be an \entrybodylabeling\ of $\achart$, and let $\avert$ and $\bvert$ be vertices of $\achart$ and $\acharthat$.
  We denote by $\avert \redi{\bodylab} \bvert$ that there is a \bodytransition\ $\avert \lt{\pair{\aact}{0}} \bvert$ in $\acharthat$ for some $\aact\in\actions$,
  and by $\avert \redi{\loopnsteplab{\aLname}} \bvert$, for $\aLname\in\natplus$
  that there is an \entrytransition\ $\avert \lt{\pair{\aact}{\aLname}} \bvert$ in $\acharthat$ for some $\aact\in\actions$.
  We will use $\aLname,\bLname,\cLname,\ldots$ for marking labels in $\natplus$ of \entrytransitions.
  By the set $\entriesof{\acharthat}$ of \emph{\entrytransition\ identifiers} we denote the set of pairs $\pair{\avert}{\aLname}\in\verts\times\natplus$ 
  such that an \entrytransition\ $\sloopnstepto{\aLname}$ departs from $\avert$ in $\acharthat$. 
  For $\pair{\avert}{\aLname}\in\entriesof{\acharthat}$,
  we define by $\indsubchartinat{\acharthat}{\avert,\aLname}$ the subchart of $\achart$ 
  with start vertex $\start$
  that consists of the vertices and transitions which occur on paths in $\achart$ 
  as follows: they start with a $\sloopnstepto{\aLname}$ \entrytransition\ from $\avert$,
  continue with body transitions only, and halt immediately if $\avert$ is revisited. 
\end{definition}

\begin{definition}\label{def:LLEEwitness}
  A \emph{LLEE-witness}~$\acharthat$ of a chart~$\achart$ is an \entrybodylabeling\ of $\achart$ 
  that satisfies the following properties:
  \begin{enumerate}[label={\mbox{\rm (W\arabic*)}},leftmargin=*,align=left,itemsep=0.5ex]
    \item{}\label{LLEEw:1}%
      There is no infinite path of $\sredi{\bodylab}$ transitions from $\start$. 
    \item{}\label{LLEEw:2}%
      For all $\pair{\avert}{\aLname}\in\entriesof{\acharthat}$, 
      \begin{enumerate*}[label={(\alph*)}]
        \item{}\label{LLEEw:2a}%
          $\indsubchartinat{\acharthat}{\avert,\aLname}$ is a loop chart,
          and
        \item{}\label{LLEEw:2b} (\emph{layeredness}) %
           from \ul{no vertex} $\bvert\neq \avert$ of $\indsubchartinat{\acharthat}{\avert,\aLname}$ 
           there departs in $\acharthat$ an \entrytransition\ $\sloopnstepto{\bLname}$ with $\bLname \geq \aLname$.%
      \end{enumerate*}
  \end{enumerate}
  
  The stipulation in \ref{LLEEw:2}\ref{LLEEw:2a} justifies to call \entrytransitions\ in a \LLEEwitness\
  a \emph{\loopentrytransition}.
  For a \loopentrytransition\ $\sredi{\loopnsteplab{\bLname}}$ with $\bLname\in\natplus$, we call $\bLname$ its \emph{loop level}. 
  
  A chart is a \LLEEchartemph\ if it has a \LLEEwitness.
\end{definition}

\begin{example} 
  The three labelings of the chart~$\chartof{\astexpi{0}}$ in Ex.~\ref{ex:chart:interpretation}
  that arose as recordings of runs of the loop elimination procedure
  can be viewed as \entrybodylabelings\ of that chart. 
  There, and below, we dropped the body labels of transitions,
  and instead only indicated the entry labels in boldface together with their levels. 
  By checking conditions \ref{LLEEw:1} and \ref{LLEEw:2},\ref{LLEEw:2a}-\ref{LLEEw:2b},
  it is easy to verify that these \entrybodylabelings\ are \LLEEwitnesses.
  In fact it is not difficult to establish
  that every \LLEEwitness\ of $\chartof{\astexpi{0}}$ in Ex.~\ref{ex:chart:interpretation}
  is of either of the following two forms, with marking labels $\aLname,\bLname,\cLname,\dLname,\eLname\in\natplus$:
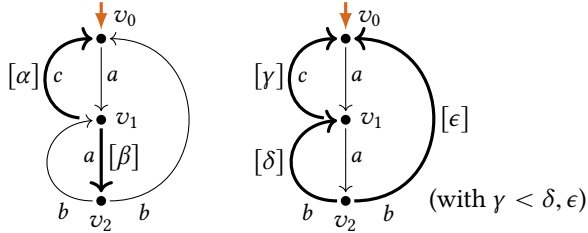
  \begin{center}
\begin{tikzpicture}
  %
  %
  
\matrix[anchor=north,row sep=0.9cm,every node/.style={draw,very thick,circle,minimum width=2.5pt,fill,inner sep=0pt,outer sep=2pt}] at (0,0) {
  \node(v_0-hat-1){};
  \\
  \node(v_1-hat-1){};
  \\
  \node(v_2-hat-1){};
  \\
};
\calcLength(v_0,v_1){mylen}
\draw[<-,very thick,>=latex,chocolate](v_0-hat-1) -- ++ (90:{0.45*\mylen pt});
\path(v_0-hat-1) ++ ({0.3*\mylen pt},{0.25*\mylen pt}) node{$\averti{0}$};
\draw[->](v_0-hat-1) to node[right,xshift={-0.05*\mylen pt},pos=0.45]{\small $\aact$} (v_1-hat-1); 
\path(v_1-hat-1) ++ ({0.325*\mylen pt},0cm) node{$\averti{1}$};
\draw[->,very thick](v_1-hat-1) to node[left,xshift={0.05*\mylen pt},pos=0.45]{\small $\aact$} 
                                   node[right,pos=0.45,xshift={-0.05*\mylen pt}]{$\loopnsteplab{\bLname}$} (v_2-hat-1);
\draw[->,very thick,shorten <= 5pt](v_1-hat-1) to[out=175,in=180,distance={0.75*\mylen pt}] 
                                                 node[left,pos=0.5,xshift={0.05*\mylen pt}]{$\loopnsteplab{\aLname}$} 
                                                 node[right,xshift={-0.05*\mylen pt},pos=0.5]{\small $\cact$} (v_0-hat-1);
\path(v_2-hat-1) ++ (-0cm,{-0.275*\mylen pt}) node{$\averti{2}$};
\draw[->](v_2-hat-1) to[out=180,in=185,distance={0.75*\mylen pt}] 
                     node[below,yshift={0.0*\mylen pt},pos=0.2]{\small $\bact$} (v_1-hat-1);
\draw[->](v_2-hat-1) to[out=0,in=0,distance={1.3*\mylen pt}] 
                     node[below,yshift={0.00*\mylen pt},pos=0.125]{\small $\bact$} (v_0-hat-1);

\matrix[anchor=north,row sep=0.9cm,every node/.style={draw,very thick,circle,minimum width=2.5pt,fill,inner sep=0pt,outer sep=2pt}] at (3.25,0) {
  \node(v_0-hat-3){};
  \\
  \node(v_1-hat-3){};
  \\
  \node(v_2-hat-3){};
  \\
};
\calcLength(v_0,v_1){mylen}
\draw[<-,very thick,>=latex,chocolate](v_0-hat-3) -- ++ (90:{0.45*\mylen pt});
\path(v_0-hat-3) ++ ({0.3*\mylen pt},{0.25*\mylen pt}) node{$\averti{0}$};
\draw[->](v_0-hat-3) to node[right,xshift={-0.05*\mylen pt},pos=0.45]{\small $\aact$} (v_1-hat-3); 
\path(v_1-hat-3) ++ ({0.325*\mylen pt},0cm) node{$\averti{1}$};
\draw[->](v_1-hat-3) to node[right,xshift={-0.05*\mylen pt},pos=0.45]{\small $\aact$} (v_2-hat-3);
\draw[->,very thick,shorten <= 5pt](v_1-hat-3) to[out=175,in=180,distance={0.75*\mylen pt}] 
         node[left,pos=0.5,xshift={0.05*\mylen pt}]{$\loopnsteplab{\cLname}$} 
         node[right,xshift={-0.05*\mylen pt},pos=0.5]{\small $\cact$} (v_0-hat-3);
\path(v_2-hat-3) ++ (-0cm,{-0.275*\mylen pt}) node{$\averti{2}$};
\draw[->,very thick](v_2-hat-3) to[out=180,in=185,distance={0.75*\mylen pt}] 
               node[left,pos=0.5,xshift={0.05*\mylen pt}]{$\loopnsteplab{\dLname}$} 
               node[below,yshift={0.0*\mylen pt},pos=0.2]{\small $\bact$} (v_1-hat-3);
\draw[->,very thick](v_2-hat-3) to[out=0,in=0,distance={1.3*\mylen pt}] 
               node[right,pos=0.5,xshift={-0.05*\mylen pt}]{$\loopnsteplab{\eLname}$} 
               node[below,yshift={0.00*\mylen pt},pos=0.125]{\small $\bact$} (v_0-hat-3);  
               
\path(v_2-hat-3) ++ ({2*\mylen pt},0cm) node{(with $\cLname < \dLname, \eLname$)};              

\end{tikzpicture}  
\end{center}
\end{example}

We now argue  that \LLEEwitnesses\
guarantee the property \LEE.
Let $\acharthat$ be a \LLEEwitness\ of a chart $\achart$. 
Repeatedly pick an \entrytransition\ identifier $\pair{\avert}{\aLname}$ 
with $\aLname\in\natplus$ minimal, remove the loop subchart that is generated by \loopentry\ transitions of level $\aLname$ from $\avert$
(it is indeed a loop by \ref{LLEEw:2}{\ref{LLEEw:2a}, and minimality of $\aLname$ and \ref{LLEEw:2}{\ref{LLEEw:2b} ensure
 the absence of departing \loopentrytransitions\ of lower level), and perform garbage collection. 
Eventually the part of $\achart$ that is reachable by body transitions from the start vertex
is obtained. This subchart does not have an infinite path due to \ref{LLEEw:1}.
Therefore $\achart$ indeed satisfies \LEE, as witnessed by $\acharthat$.

The property \LEE\ and the concept of \LLEEwitness\
are closely linked with the process semantics of star expressions.
In fact, we now define a labeling of the TSS in Def.~\ref{def:chart:interpretation}
that permits to define, for every star expression $\astexp$,
an \entrybodylabeling\ of the chart interpretation $\chartof{\astexp}$ of $\astexp$,
which can then be recognized as a \LLEEwitness\ of $\chartof{\astexp}$. 
 
We refine the TSS rules in Def.~\ref{def:chart:interpretation}
as follows:
A body label is added to transitions that cannot return to the star expression in their left-hand side.
The rule for transitions into the iteration part $\astexpi{1}$ of an iteration $\stexpbit{\astexpi{1}}{\astexpi{2}}$
is split into the cases where $\astexpi{1}$ is normed or not.  
Only in the normed case can $\stexpbit{\astexpi{1}}{\astexpi{2}}$ return to itself,
and then a \loopentrytransition\ with the star height $\bsth{\astexpi{1}}$ of $\astexpi{1}$ as its level is created. 

\begin{definition}\label{def:lbl:chart:translation}
  For every $\astexp\in\StExpsover{\actions}$,
  we define \emph{the \entrybodylabeling\/ $\charthatof{\astexp}$ of} the chart interpretation~$\chartof{\astexp}$ of $\astexp$
  in analogy with $\chartof{\astexp}$ by using the following
  transition rules that refine the rules in Def.~\ref{def:chart:interpretation}
  by adding marking labels:
  \begin{gather*}
    \AxiomC{$\phantom{a \:\lt{\aact}_{\bodylab}\: \tick}$}
    \UnaryInfC{$a \:\lt{\aact}_{\bodylab}\: \tick$}
    \DisplayProof
    \hspace*{1.75em}
    \AxiomC{$ \astexpi{i} \:\lt{\aact}_{\alab}\: \atickstexp $}
    \RightLabel{$i\in\setexp{1,2}$}
    \UnaryInfC{$ \stexpsum{\astexpi{1}}{\astexpi{2}} \:\lt{\aact}_{\bodylab}\: \atickstexp $}
    \DisplayProof
    \displaybreak[0]\\[1ex]
    \AxiomC{$ \astexpi{1} \:\lt{\aact}_{\alab}\: \astexpacci{1} $}
    \UnaryInfC{$ \stexpprod{\astexpi{1}}{\astexpi{2}} \:\lt{\aact}_{\alab}\: \stexpprod{\astexpacci{1}}{\astexpi{2}} $}
    \DisplayProof
    \hspace*{1.75em}
    \AxiomC{$ \astexpi{1} \:\lt{\aact}_{\bodylab}\: \tick$}
    \UnaryInfC{$ \stexpprod{\astexpi{1}}{\astexpi{2}} \:\lt{\aact}_{\bodylab}\: \astexpi{2} $}
    \DisplayProof
    \displaybreak[0]\\[1ex]
    \AxiomC{$\astexpi{1} \:\lt{\aact}_{\alab}\: \astexpacci{1}$}
    \RightLabel{if $\astexpi{1}$ is normed}
    \UnaryInfC{$\stexpbit{\astexpi{1}}{\astexpi{2}} 
                  \:\lt{\aact}_{\loopnsteplab{\bsth{\astexpi{1}}+1}}\: 
                \stexpprod{\astexpacci{1}}{\stexpbit{(\astexpi{1}}{\astexpi{2})}}$}
    \DisplayProof
    \displaybreak[0]\\[1ex]
    \AxiomC{$\astexpi{1} \:\lt{\aact}_{\alab}\: \astexpacci{1}$}
    \RightLabel{if $\astexpi{1}$ is not normed}
    \UnaryInfC{$\stexpbit{\astexpi{1}}{\astexpi{2}} \:\lt{\aact}_{\bodylab}\: \stexpprod{\astexpacci{1}}{\stexpbit{(\astexpi{1}}{\astexpi{2})}}$}
    \DisplayProof
    \displaybreak[0]\\
    \AxiomC{$ \astexpi{1} \:\lt{\aact}_{\bodylab}\: \tick$}
    \UnaryInfC{$\stexpbit{\astexpi{1}}{\astexpi{2}} \:\lt{\aact}_{\loopnsteplab{\bsth{\astexpi{1}}+1}}\: \stexpbit{\astexpi{1}}{\astexpi{2}}$}
    \DisplayProof
    \hspace*{1.75em}
    \AxiomC{$\astexpi{2} \:\lt{\aact}_{\alab}\: \atickstexp$}
    \UnaryInfC{$\stexpbit{\astexpi{1}}{\astexpi{2}} \:\lt{\aact}_{\bodylab}\: \atickstexp$}
    \DisplayProof 
  \end{gather*}
  for $\alab\in\setexp{\bodylab}\cup\descsetexp{\loopnsteplab{\aLname}}{\aLname\in\natplus}$, where we employed notation defined in Def.~\ref{def:chart:interpretation}
  for writing marking labels as subscripts.
\end{definition}

\begin{example}\label{ex:entry-body-labeling}
  In Fig.~\ref{fig-ex-LLEEw-translation}
  we depict the \entrybodylabelings, as defined in Def.~\ref{def:entry-body-labeling},
  for star expressions $\astexpi{1}$, $\astexpi{0}$, and $\astexpi{2}$ in Ex.~\ref{ex:chart:interpretation}.
  \begin{figure}[t!]
  %
\begin{center}
\hspace*{-1em}
\begin{tikzpicture}[scale=1,every node/.style={transform shape}]
%
\matrix[anchor=north,row sep=0.9cm,every node/.style={draw,very thick,circle,minimum width=2.5pt,fill,inner sep=0pt,outer sep=2pt}] at (0,-0.5) {
  \node(v_e1_0){};
  \\
  \node(v_e1_1){};
  \\
  \node(v_e1_2){};
  \\
  \node(v_e1_0'){};
  \\
};
\calcLength(v_e1_0,v_e1_1){mylen}
\draw[<-,very thick,>=latex,chocolate](v_e1_0) -- ++ (90:{0.45*\mylen pt});
%
\draw[->,very thick] (v_e1_0) to node[right,pos=0.45,xshift={-0.05*\mylen pt}]{$\loopnsteplab{2}$} 
                                 node[left,pos=0.45]{\small $\aact$} (v_e1_1);
%
\draw[->,very thick] (v_e1_1) to node[right,pos=0.45,xshift={-0.05*\mylen pt}]{$\loopnsteplab{1}$} 
                                 node[left,pos=0.45]{\small $\aact$} (v_e1_2);
\draw[->,shorten <= 4.5pt] (v_e1_1) to[out=170,in=180,distance={0.75*\mylen pt}] node[left,pos=0.5]{\small $\cact$} (v_e1_0);
%
\draw[->] (v_e1_2) to node[right,pos=0.5]{\small $\bact$} (v_e1_0');
\draw[->] (v_e1_2) to[out=180,in=190,distance={0.75*\mylen pt}] node[left,pos=0.5]{\small $\bact$} (v_e1_1);
%
\draw[->] (v_e1_0') to[out=0,in=0,distance={1.5*\mylen pt}] node[right,pos=0.5]{\small $\aact$} (v_e1_1); 
%
\path (v_e1_0') ++ ({0*\mylen pt},{-0.6*\mylen pt}) node{\large $\charthighhatof{\astexpi{1}}$};

\matrix[anchor=north,row sep=0.9cm,every node/.style={draw,very thick,circle,minimum width=2.5pt,fill,inner sep=0pt,outer sep=2pt}] at (2.9,0) {
  \node[draw=none,fill=none](v_1-dummy){};
  \\
  \node(v_0){};
  \\
  \node(v_1){};
  \\
  \node(v_2){};
  \\
};
\calcLength(v_0,v_1){mylen}
\draw[draw=none,<-,very thick,>=latex,chocolate](v_1-dummy) -- ++ (90:{0.45*\mylen pt});
\draw[<-,very thick,>=latex,chocolate](v_0) -- ++ (90:{0.5*\mylen pt});
\draw[->](v_0) to node[right,xshift={-0.05*\mylen pt},pos=0.45]{\small $\aact$} (v_1); 
%
\draw[->,very thick](v_1) to node[right,xshift={-0.05*\mylen pt},pos=0.45]{$\loopnsteplab{1}$}
                             node[left,xshift={0.05*\mylen pt},pos=0.45]{\small $\aact$} (v_2);
\draw[->,very thick,shorten <= 5pt](v_1) to[out=175,in=180,distance={0.75*\mylen pt}] 
         node[left,pos=0.5,xshift={0.05*\mylen pt}]{$\loopnsteplab{1}$} 
         node[right,yshift={0.05*\mylen pt},pos=0.45,xshift={-0.05*\mylen pt}]{\small $\cact$} (v_0);
%
\draw[->](v_2) to[out=180,in=185,distance={0.75*\mylen pt}]  
               node[below,yshift={0.0*\mylen pt},pos=0.2]{\small $\bact$} (v_1);
\draw[->](v_2) to[out=0,in=0,distance={1.3*\mylen pt}] 
               node[below,yshift={0.00*\mylen pt},pos=0.125]{\small $\bact$} (v_0);


\path (v_2) ++ ({0*\mylen pt},{-0.8*\mylen pt}) node{\large $\charthighhatof{\astexpi{0}}$};

%
\matrix[anchor=north,row sep=0.9cm,every node/.style={draw,very thick,circle,minimum width=2.5pt,fill,inner sep=0pt,outer sep=2pt}] at (5.85,0.1) {
  \node(v_e2_0){};
  \\
  \node(v_e2_1){};
  \\
  \node(v_e2_2){};
  \\
  \node(v_e2_0''){};
  \\
  \node(v_e2_1'){};
  \\
};
\calcLength(v_e2_0,v_e2_1){mylen}
  \draw[draw=none] (v_e2_2) arc (270:205:{\mylen pt}) node[style={draw,very thick,circle,minimum width=2.5pt,fill,inner sep=0pt,outer sep=2pt}](v_e2_0'){};
  \draw[draw=none] (v_e2_0'') arc (90:205:{\mylen pt}) node[style={draw,very thick,circle,minimum width=2.5pt,fill,inner sep=0pt,outer sep=2pt}](v_e2_0'''){};
\draw[<-,very thick,>=latex,chocolate](v_e2_0) -- ++ (90:{0.45*\mylen pt});
\draw[->](v_e2_0) to node[right,xshift={-0.05*\mylen pt},pos=0.4]{\small $\aact$} (v_e2_1);
%
%
\draw[->,very thick](v_e2_1) to node[right,pos=0.45,xshift={-0.075*\mylen}]{$\loopnsteplab{3}$}
                                node[left,pos=0.45,xshift={0.065*\mylen}]{\small $\aact$} (v_e2_2);
\draw[->,very thick](v_e2_1) to node[below,pos=0.45]{$\loopnsteplab{3}$} 
                                node[above,pos=0.65]{\small $\cact$} (v_e2_0');
%
\draw[->] (v_e2_0') to[out=115,in=150,distance={0.7*\mylen pt}] node[above,pos=0.75]{\small $\aact$} (v_e2_1);
%
\draw[->,very thick](v_e2_2) to node[right,pos=0.45,xshift={-0.075*\mylen}]{$\loopnsteplab{2}$} 
                                node[left,pos=0.425,xshift={0.05*\mylen}]{\small $\bact$} (v_e2_0'');
\draw[->,shorten <= 5pt] (v_e2_2) to[out=10,in=0,distance=0.7cm] 
                                node[right,pos=0.5,xshift={-0.025*\mylen}]{\small $\bact$} (v_e2_1);
%
\draw[->](v_e2_0'') to node[right,xshift={-0.05*\mylen pt},pos=0.4]{\small $\aact$} (v_e2_1');
%
\draw[->,very thick](v_e2_1') to node[below,pos=0.45]{$\loopnsteplab{1}$} 
                                 node[above,pos=0.65]{\small $\cact$} (v_e2_0''');
\draw[->](v_e2_1') to[out=0,in=0,distance=1.4cm] node[right,pos=0.6,xshift={-0.025*\mylen}]{\small $\aact$} (v_e2_2);
%
\draw[->] (v_e2_0''') to[out=115,in=150,distance={0.7*\mylen pt}] node[above,pos=0.4]{\small $\aact$} (v_e2_1');

\path (v_e2_1') ++ ({1*\mylen pt},{-0.3*\mylen pt}) node{\large $\charthighhatof{\astexpi{2}}$};


\end{tikzpicture} 
\end{center}
  %
    \vspace*{-2.25ex}
    \caption{\label{fig-ex-LLEEw-translation}%
             \LLEEwitness\ \entrybodylabelings\ as defined by Def.~\ref{def:lbl:chart:translation} for the chart interpretations
             of $\astexpi{0}$, $\astexpi{1}$, and $\astexpi{2}$ in Ex.~\ref{ex:chart:interpretation}.}
  \end{figure}
  It is easy to verify that these labelings
  are \LLEEwitnesses\ of the charts~$\chartof{\astexpi{0}}$, $\chartof{\astexpi{1}}$, and $\chartof{\astexpi{2}}$
  in Ex.~\ref{ex:chart:interpretation}, resp.. 
\end{example}

\begin{proposition}\label{prop:lbl:chart:translation:is:LLEEw}
  For every $\astexp\in\StExpsover{\actions}$,
  the \entrybodylabeling\/ $\charthatof{\astexp}$ of $\chartof{\astexp}$ is a \LLEEwitness\ of $\chartof{\astexp}$.  
\end{proposition} 



For a binary relation $R$, let $R^+$ and $R^*$ be its transitive and transitive-reflexive closures.
$u\,\sredi{l}\,v$ denotes that there is a transition $u\,\lt{a}_l\,v$ for an $a\in A$, and in proofs (but not pictures)
 $u\,\sred\,v$ denotes that $u\,\sredi{l}\,v$ for some label $l$.
By $u \redtavoidsvi{w}{l} v$ we denote that $u \redi{l} v$ and $v\neq w$ (this transition avoids \ul{t}arget~$w$).
\vspace*{-.5mm}Likewise, $u \redtavoidsv{w} v$ denotes that $u \redtavoidsvi{w}{l} v$ for some label~$l$.
By $\sccof{u}$ we denote the strongly connected component (scc) to which $u$ belongs.

\begin{definition}\label{def:descendsinloopto}
  Let $\acharthat$ be a \LLEEwitness\ of chart $\achart$. 
  If there is a path $\avert \comprewrels{\sredtavoidsvi{\avert}{\loopnsteplab{\aLname}}}{\sredtavoidsvrtci{\avert}{\bodylab}} \bvert$,\vspace*{-0.5mm}
  then we write $\avert \descendsinlooplto{\aLname} \bvert$.
  (Note that $\avert \descendsinlooplto{\aLname} \bvert$ holds  
     if and only if 
   $\bvert$ is a vertex $\neq\avert$ of the loop chart $\indsubchartinat{\acharthat}{\avert,\aLname}$
   that is generated by the $\sredi{\loopnsteplab{\aLname}}$ \entrytransitions\ at $\avert$ in $\achart$.)
  We write $\avert \descendsinloopto \bvert$
    and say that $\avert$ \emph{descends in a loop to} $\bvert$ if $\avert \descendsinlooplto{\aLname} \bvert$ for some $\aLname\in\natplus$.
    
  We write $\bvert \loopsbackto \avert$ (or $\avert \convloopsbackto \bvert$), and say that $\bvert$ \emph{loops back to} $\avert$, 
    if $\avert \descendsinloopto \bvert \redtci{\bodylab} \avert$. 
    The \txtloopsbackto~relation $\sloopsbackto$ totally orders its successors (see Lem.~\ref{lem:loop:relations},~\ref{it:loopsbackto:lo:successors}).
    Therefore we define the `direct successor relation' $\sdloopsbackto$ of $\sloopsbackto$ as follows:
  We write $\bvert \dloopsbackto \avert$ (or $\avert \convdloopsbackto \bvert$), and say that $\bvert$ \emph{directly loops back to} $\avert$, 
  if $\bvert \loopsbackto \avert$ and for all $\cvert$ with $\bvert \loopsbackto \cvert$ either $\cvert = \avert$ or $\avert \loopsbackto \cvert$.  
\end{definition}

\begin{lemma}\label{lem:loop:relations}
  The relations $\sredi{\bodylab}$, $\sdescendsinloopto$, $\sloopsbackto$, $\sdloopsbackto$\vspace*{-0.5mm}
  as defined by a \LLEEwitness~$\acharthat$ on a chart~$\achart$ satisfy the following properties:
  \begin{enumerate}[label=(\roman{*})]
    \item{}\label{it:bo:terminating}
      There are no infinite $\sredi{\bodylab}$ paths (so no $\sredi{\bodylab}$ cycles).
    \item{}\label{it:descendsinloopto:scc:loopsbackto}
      If $\sccof{\cvert} = \sccof{\avert}$, then $\cvert \descendsinlooptortc \avert$ implies $\avert \loopsbacktortc \cvert$.
    \item{}\label{it:descendsinloopto:notloopsbackto:not:normed}
      If $\avert \descendsinloopto \bvert$ and $\lognot{(\bvert\,\sloopsbackto)}$, then $\bvert$ is not normed.
    \item{}\label{it:bi:reachable:lLEEw}
      $\sccof{\cvert}=\sccof{\avert}$ if and only if  $\cvert\loopsbacktortc \bvert$ and $\avert\loopsbacktortc \bvert$ for some vertex~$\bvert$.
    \item{}\label{it:least-upper-bound}   
      $\sloopsbacktortc$ is a partial order with the least-upper-bound property: 
      if a non\-empty set of vertices has an upper bound with respect to $\sloopsbacktortc$, then it has a least upper bound.
    \item{}\label{it:loopsbackto:lo:successors}
      $\sloopsbackto$ is a total order on $\sloopsbackto$-successor vertices: 
      if $\bvert \loopsbackto \averti{1}$ and $\bvert \loopsbackto \averti{2}$, 
      then $\averti{1} \loopsbackto \averti{2}$ or $\averti{1} = \averti{2}$ or $\averti{2} \loopsbackto \averti{1}$.  
    \item{}\label{it:direct-subordinates}  
      If $\averti{1}\dloopsbackto \cvert$ and $\averti{2}\dloopsbackto \cvert$ for distinct $\averti{1},\averti{2}$, 
      then there is no vertex $\bvert$ such that both $\bvert\loopsbacktortc \averti{1}$ and $w\loopsbacktortc \averti{2}$.
  \end{enumerate}
\end{lemma}

\section{The completeness proof, anticipated}%
  \label{compl:proof}

After having introduced \LLEEcharts\
as our crucial auxiliary concept,
we now sketch the completeness proof. 
In doing so we need to anticipate four results that will be developed
in the next two sections:
    {\bf (C)} The bisimulation collapse of a \LLEEchart\ is again a \LLEEchart.
    {\bf (E)} From every \LLEEchart\ a provable solution can be extracted.
    {\bf (S)} All provable solutions of \LLEEcharts\ are provably equal.
    {\bf (P)} All provable solutions can be pulled back from the target to the source chart of a functional bisimulation.

Then completeness of $\BBP$ can be argued as follows.
Given two bisimilar star expressions $\astexpi{1}$ and $\astexpi{2}$, 
obtain their chart interpretations $\chartof{\astexpi{1}}$ and $\chartof{\astexpi{2}}$, which are \LLEEcharts\ due to Prop.~\ref{prop:lbl:chart:translation:is:LLEEw}.
By Prop.~\ref{prop:id:is:sol:chart:interpretation}, $\astexpi{1}$ and $\astexpi{2}$ are principal values 
of provable solutions of $\chartof{\astexpi{1}}$ and $\chartof{\astexpi{2}}$.
These charts have the same bisimulation collapse $\achart$.
By ({\bf C},~Thm.~\ref{thm:LEEshaped:collapse}), $\achart$ is again a \LLEEchart.
Use ({\bf E},~Prop.~\ref{prop:extracted:fun:is:solution}) to
build a provable solution $\sasol$ of $\achart$; let its principal value be $\astexp$. 
Apply ({\bf P},~Prop.~\ref{prop:transf:sol:via:funbisim}) to
transfer $\sasol$ backwards over the functional bisimulations 
to obtain provable solutions $\sasoli{1}$ and $\sasoli{2}$ of $\chartof{\astexpi{1}}$ and $\chartof{\astexpi{2}}$, respectively.
By construction, $\sasoli{1}$ and $\sasoli{2}$ have the same principal value $\astexp$ as $\sasol$.
Finally, by using ({\bf S}, Prop.~\ref{prop:extrsol:vs:solution}),
$\astexpi{1}$ and $\astexpi{2}$ are both provably equal to $\astexp$.
Hence, $\astexpi{1}\BBPeq e\BBPeq \astexpi{2}$.

In his completeness proof for regular expressions in formal language theory, 
Salomaa \cite{salo:1966} argued `upwards' from two equivalent regular expressions to a larger regular expression 
that can be homomorphically collapsed onto both of them.
In contrast, our proof approach forces us `downwards' to the bisimulation collapse, 
because in the opposite direction the property of being a \LLEEchart\ may be lost.

\begin{example}\label{ex:salomaa}
  The picture below highlights why 
                                   we cannot adopt Salomaa's proof strategy 
  of linking two \languageequivalent\ regular expressions
  via the product of the DFAs they represent.
  The bisimilar LLEE-charts~$\acharti{1}$ and $\acharti{2}$
  are interpretations of 
  $\stexpbit{(\stexpsum{\stexpprod{a}{(\stexpsum{a}{b})}}{b})}{\stexpzero}$ 
    and 
   $\stexpbit{(\stexpsum{\stexpprod{b}{(\stexpsum{a}{b})}}{a})}{\stexpzero}$, respectively
  (the indicated labelings~$\acharthati{1}$ and $\acharthati{2}$ are \LLEEwitnesses). 
  But their product $\acharti{12}$ is a not a \LLEEchart; it is of the form  of
  one the not expressible charts from Ex.~\ref{ex:not:expressible}.
  Yet their common bisimulation collapse $\acharti{0}$,
  the chart interpretation of $\stexpbit{(\stexpsum{a}{b})}{\stexpzero}$, is a \LLEEchart\ with \LLEEwitness~$\acharthati{0}$.%
  \vspace{-2.5ex}
  \begin{center}
  %
\begin{center}
\begin{tikzpicture}[scale=0.6,every node/.style={transform shape}]


\matrix[matrix of nodes,matrix anchor=center,row sep=1cm,every node/.style={draw,very thick,circle,minimum width=2.5pt,fill,inner sep=0pt,outer sep=2pt}](C_0) at (0,0) {
  \node(v_0_0){};
  \\
};  
\draw[->,out=240,in=180,distance=1.5cm,very thick
                                      ] (v_0_0) to node[below,pos=0.2,scale=1.25]{$\loopnsteplab{1}$} 
                                                   node[left,pos=0.5,scale=1.25]{$\aact$}  (v_0_0);                                                  
\draw[->,out=-60,in=0,distance=1.5cm,very thick
                                    ]   (v_0_0) to node[below,pos=0.2,scale=1.25]{$\loopnsteplab{1}$} 
                                                   node[right,pos=0.5,scale=1.25]{$\bact$} (v_0_0);

\path (v_0_0) ++ (-3.75cm,-1.75cm) node(anchor_1){};
\path (v_0_0) ++ (3.75cm,-1.5cm)  node(anchor_2){};
\path (v_0_0) ++ (0cm,-3.9cm)    node(anchor_3){};

\matrix[matrix of nodes,anchor=north, 
        row sep=1cm,column sep=1cm,every node/.style={draw,very thick,circle,minimum width=2.5pt,fill,inner sep=0pt,outer sep=2pt}] (C_1) at (anchor_1) {
                    &  \node(v_0_1){};
  \\
  \node(v_1_1){};   &  \node[draw=none,fill=none](dummy_11){};  & \node[draw=none,fill=none]{};
  \\
}; 
\calcLength(v_0_1,dummy_11){mylen}; 
\draw[<-,very thick,>=latex,chocolate](v_0_1) -- ++ (90:{0.425*\mylen pt});        
\draw[->,very thick
        ] (v_0_1) to node[below,pos=0.45,yshift={-0.1*\mylen pt},scale=1.25]{$\loopnsteplab{1}$} 
                     node[left,pos=0.45,scale=1.25]{$\aact$} (v_1_1);  
\draw[->,very thick,
         out=-60,in=0,distance=1.75cm] (v_0_1) to 
                                                  node[right,pos=0.5,scale=1.25]{$\bact$} (v_0_1); 
\path(v_0_1) ++ ({0.25*\mylen pt},{0.25*\mylen pt}) node[scale=1.5]{$\averti{1}$};                                                 
\path(v_1_1) ++ ({-0.05*\mylen pt},{-0.25*\mylen pt}) node[scale=1.5]{$\averti{2}$};                  
\draw[->,out=210,in=180,distance=1.8cm]  (v_1_1) to node[left,pos=0.5,scale=1.25]{$\aact$}  (v_0_1);
\draw[->,out=-45,in=270,distance=1.35cm] (v_1_1) to node[below,pos=0.3,scale=1.25]{$\bact$} (v_0_1);

\path (v_0_1) ++ ({-0.7*\mylen pt},{0.35*\mylen pt}) node[scale=1.9]{$\acharti{1},\hspace*{1pt}\acharthati{1}$};

\draw[<-,very thick,>=latex,chocolate](v_0_0) -- ++ (90:{0.425*\mylen pt});  
  
\path (v_0_0) ++ ({0.7*\mylen pt},{0.45*\mylen pt}) node[scale=1.9]{$\acharti{0},\hspace*{1pt}\acharthati{0}$};

\matrix[matrix of nodes,anchor=north,
        row sep=1cm,column sep=1cm,every node/.style={draw,very thick,circle,minimum width=2.5pt,fill,inner sep=0pt,outer sep=2pt}] (C_2) at (anchor_2) {
                                & \node(v_0_2){};
  \\
  \node[draw=none,fill=none]{}; & \node[draw=none,fill=none](dummy_12){};  & \node(v_1_2){};   
  \\
}; 
\draw[<-,very thick,>=latex,chocolate](v_0_2) -- ++ (90:{0.425*\mylen pt});     
\draw[->,very thick
        ] (v_0_2) to node[below,pos=0.45,yshift={-0.1*\mylen pt},scale=1.25]{$\loopnsteplab{1}$} 
                     node[pos=0.45,right,xshift={0.1*\mylen pt},scale=1.25]{$\bact$} (v_1_2);
\path(v_0_2) ++ ({-0.25*\mylen pt},{0.25*\mylen pt}) node[scale=1.5]{$\bverti{1}$};                  
\draw[->,very thick,
         out=240,in=180,distance=1.75cm] (v_0_2) to node[below,pos=0.3,scale=1.25]{$\loopnsteplab{1}$} 
                                                    node[left,pos=0.5,scale=1.25]{$\aact$} (v_0_2);
\path(v_1_2) ++ ({0.05*\mylen pt},{-0.25*\mylen pt}) node[scale=1.5]{$\bverti{2}$};  
\draw[->,out=-30,in=0,distance=1.8cm]  (v_1_2) to node[right,pos=0.5,scale=1.25]{$\bact$} (v_0_2);
\draw[->,out=225,in=270,distance=1.35cm] (v_1_2) to node[below,pos=0.3,scale=1.25]{$\aact$} (v_0_2);

\path (v_0_2) ++ ({0.7*\mylen pt},{0.35*\mylen pt}) node[scale=1.9]{$\acharti{2},\hspace*{1pt}\acharthati{2}$};

\matrix[matrix of nodes,anchor=north,
        row sep=1cm,column sep=0.5cm,every node/.style={draw,very thick,circle,minimum width=2.5pt,fill,inner sep=0pt,outer sep=2pt}] (C_12) at (anchor_3) {
                    &                     & \node(v_0_12){};  &                     &
  \\
  \node(v_1_12){};  &                     &                   &                     &  \node(v_2_12){};
  \\
};  
\draw[<-,very thick,>=latex,chocolate](v_0_12) -- ++ (90:{0.425*\mylen pt});

\path(v_0_12) ++ ({-0.575*\mylen pt},{0.25*\mylen pt}) node[scale=1.5]{$\pair{\averti{1}}{\bverti{1}}$};                    
\draw[->] (v_0_12) to node[above,scale=1.25]{$\aact$} (v_1_12){};
\draw[->] (v_0_12) to node[above,scale=1.25]{$\bact$} (v_2_12){};

\path(v_1_12) ++ ({-0.1*\mylen pt},{-0.35*\mylen pt}) node[scale=1.5]{$\pair{\averti{2}}{\bverti{1}}$};  
\draw[->,bend right,distance=0.8cm] (v_1_12) to node[above,yshift=-0.325ex,scale=1.25]{${\bact}$} (v_2_12); 
\draw[->,out=210,in=180,distance=1.8cm,] (v_1_12) to node[left,pos=0.5,scale=1.25]{$\aact$} (v_0_12);

\path(v_2_12) ++ ({-0.1*\mylen pt},{-0.35*\mylen pt}) node[scale=1.5]{$\pair{\averti{1}}{\bverti{2}}$};  
\draw[->,bend right,distance=0.8cm,shorten <=4pt,shorten >= 5pt] (v_2_12) to node[above,yshift=-0.325ex,scale=1.25]{${\aact}$} (v_1_12); 
\draw[->,out=-30,in=0,distance=1.8cm] (v_2_12) to node[right,pos=0.5,scale=1.25]{$\bact$} (v_0_12);

\path (v_0_12) ++ ({0.45*\mylen pt},{0.35*\mylen pt}) node[scale=1.9]{$\acharti{12}$};

\node[scale=1.45] (label_2_0) at ($(C_2.center)!0.475!(C_0.center)$) {\reflectbox{\rotatebox{35}{\LARGE $\boldsymbol{\sfunbisim}$}}};    
\node[scale=1.45] (label_1_0) at ($(C_1.center)!0.475!(C_0.center)$) {\rotatebox{35}{\LARGE $\boldsymbol{\sfunbisim}$}};
  
  
\node[scale=1.45,yshift={-0.08*\mylen pt},xshift={-0.05*\mylen pt}] (label_12_1) at ($(C_12.center)!0.625!(C_1.center)$) {\reflectbox{\rotatebox{35}{\LARGE $\boldsymbol{\sfunbisim}$}}}; 
\node[scale=1.45,yshift={-0.08*\mylen pt},xshift={0.05*\mylen pt}] (label_12_2) at ($(C_12.center)!0.625!(C_2.center)$) {\rotatebox{35}{\LARGE $\boldsymbol{\sfunbisim}$}};  
  
\end{tikzpicture}

\end{center}
%
  \end{center}
  In view of $\acharti{1} \convfunbisim \acharti{12} \funbisim \acharti{2}$ this also shows
  that \LLEEchart{s} are not closed under converse functional bisimilarity~$\sconvfunbisim$.
\end{example}

\section{Extraction of star expressions from,
         and transferral between, LLEE-charts}
  \label{extraction:transferral}        

In this section we develop the results {\bf (E)}, {\bf (S)}, and {\bf (P)} as mentioned in Sect.~\ref{compl:proof}.
We start with the statement {\bf (P)}.

\begin{proposition}[requires \BBP-axioms (B1), (B2), (B3)]\label{prop:transf:sol:via:funbisim}%
  Let $\sphifun \funin \vertsi{1} \rightarrow \vertsi{2}$ be a functional bisimulation between charts $\acharti{1}$ and $\acharti{2}$.
  If $\sasoli{2} \funin \vertsi{2}\setminus\setexp{\tick} \to \StExpsover{\actions}$ is a provable solution of $\acharti{2}$,
  then $\scompfuns{\sasoli{2}}{\sphifun} \funin \vertsi{1}\setminus\setexp{\tick} \to \StExpsover{\actions}$ is a provable solution of $\acharti{1}$
  with the same principal value~as~$\sasoli{2}$. 
\end{proposition}

\begin{proof}[Proof (Idea)]
  The bisimulation clauses make it possible to 
  demonstrate the condition for $\scompfuns{\sasoli{2}}{\sphifun}$ to be a provable solution of $\acharti{1}$ at $\bvert$
  by using the condition for the provable solution $\sasoli{2}$ of $\acharti{2}$ at $\phifun{\bvert}$,
  together with the axioms $(\commstexpsum)$, $(\assocstexpsum)$, $(\idempotstexpsum)$.
\end{proof}

We now turn to proving results {\bf (E)} and {\bf (S)} from Sect.~\ref{compl:proof}.
We show that from every chart~$\achart$ with \LLEEwitness~$\acharthat$
a provable solution $\sextrsolof{\acharthat}$ of $\achart$ can be extracted.
Intuitively, the extraction process follows 
a run of the \loopelimination\ procedure on $\achart$, guided
by the \LLEEwitness~$\acharthat$.
All loop subcharts that are generated by the \loopentrytransitions\ from a vertex $\avert$ are removed
in a row.%
  \footnote{We
  repeatedly pick vertices $\avert$ in the remaining \LLEEwitness\
    with entry step level $\enl{\avert}$ (see in the text below) minimal.} 
Extraction synthesizes a star expression $\astexpi{1}$ 
whose behavior captures the eliminated loop subcharts of $\avert$ and their previously eliminated inner loop subcharts,
and that will later be part of an iteration expression $\stexpbit{\astexpi{1}}{\astexpi{2}}$ in the solution value at $\avert$. 
This idea motivates an inside-out extraction process that works with partial solutions,
and eventually builds up a provable solution of $\achart$. 

In particular, we inductively define
`relative extracted solutions' $\extrsoluntilof{\acharthat}{\bvert}{\avert}$ 
for vertices $\avert$ and $\bvert$ where $\bvert$ is in a loop subchart $\indsubchartinat{\acharthat}{\avert,\aLname}$ at $\avert$, for some $\aLname\in\natplus$,
that is, $\avert \descendsinlooplto{\aLname} \bvert$.
Hereby $\extrsoluntilof{\acharthat}{\bvert}{\avert}$ captures the part of the behavior in $\achart$ from $\bvert$ 
until 
      $\avert$ is reached. 
Then we define the from~$\acharthat$ `extracted solution' $\extrsolof{\acharthat}{\avert}$ at $\avert$
by using the relative solutions $\extrsoluntilof{\acharthat}{\bverti{j}}{\avert}$
for all targets $\bverti{j}$ of \loopentrytransitions\ from $\avert$ to define the iteration part $\astexpi{1}$
of the extracted solution $\extrsolof{\acharthat}{\avert} = \stexpbit{\astexpi{1}}{\astexpi{2}}$ at $\avert$. 
We start with a preparation.

Let $\acharthat$ be a \LLEEwitness,
and let $\avert$ be a vertex of $\acharthat$. 
By the \emph{entry step level $\enl{\avert}$ of $\avert$}
we mean the maximum loop level of a \loopentrytransition\ in $\acharthat$ that departs from $\avert$,
or $0$ if no \loopentrytransition\ departs from $\avert$.  
By the \emph{body step norm $\bosn{\avert}$ of $\avert$}
we mean the maximal length of a \bodytransition\ path in $\achart$ from $\avert$
(\welldefined\ by Lem.~\ref{lem:loop:relations}, \ref{it:bo:terminating}).

\begin{lemma}\label{lem:def:extrsol}
  For all vertices $\avert,\bvert$ in a chart $\achart$ with \LLEEwitness~$\acharthat$
  it holds (for the concepts as defined with respect to~$\acharthat$):
  \begin{enumerate}[label=(\roman{*})]
    \item{}\label{it:bosn:lem:def:extrsol}
      $\avert \redi{\bodylab}  \bvert \Rightarrow \bosn{\avert} > \bosn{\bvert}$,
    \item{}\label{it:enl:lem:def:extrsol}
      $\avert \descendsinloopto \bvert \Rightarrow \enl{\avert} > \enl{\bvert}$.
  \end{enumerate}
\end{lemma}

\begin{figure*}[t!]
\begin{center}
  \begin{equation*}
    \scalebox{0.92}{$
    \begin{aligned}[c]
      \extrsolof{\acharthat}{\averti{0}} 
        & \;\parbox[t]{\widthof{~~$\BBPeq$~~}}{~~$\defdby$}\,
      \stexpbit{\stexpzero}{(\stexpprod{\aact}{\extrsolof{\acharthat}{\averti{1}}})}
      \\ 
        & \;\parbox[t]{\widthof{~~$\BBPeq$~~}}{~~$\BBPeq$~~}\,
      \stexpprod{\aact}{\extrsolof{\acharthat}{\averti{1}}}
      \\ 
        & \;\parbox[t]{\widthof{~~$\BBPeq$~~}}{~~$\BBPeq$~~}\,
      \stexpprod{\aact}{\stexpbit{(\stexpsum{\stexpprod{\cact}{\aact}}
                                            {\stexpprod{\aact}{(\stexpsum{\bact}{\stexpprod{\bact}{\aact}})}})}
                                 {\stexpzero}}
      \\
      \extrsolof{\acharthat}{\averti{1}} 
        & \;\parbox[t]{\widthof{~~$\BBPeq$~~}}{~~$\defdby$}\,
      \stexpbit{(\stexpsum{\stexpprod{\cact}{\extrsoluntilof{\acharthat}{\averti{0}}{\averti{1}}}}
                          {\stexpprod{\aact}{\extrsoluntilof{\acharthat}{\averti{2}}{\averti{1}}})}}
               {\stexpzero}
      \\ 
        & \;\parbox[t]{\widthof{~~$\BBPeq$~~}}{~~$\BBPeq$~~}\,
      \stexpbit{(\stexpsum{\stexpprod{\cact}{\aact}}
                          {\stexpprod{\aact}{(\stexpsum{\bact}{\stexpprod{\bact}{\aact}})}})}
               {\stexpzero}
      \\
      \extrsoluntilof{\acharthat}{\averti{0}}{\averti{1}}
        & \;\parbox[t]{\widthof{~~$\BBPeq$~~}}{~~$\defdby$}\,
      \stexpbit{\stexpzero}{\aact}
      \\
        & \;\parbox[t]{\widthof{~~$\BBPeq$~~}}{~~$\BBPeq$~~}\,
      \aact
      \\
      \extrsoluntilof{\acharthat}{\averti{2}}{\averti{1}}
        & \;\parbox[t]{\widthof{~~$\BBPeq$~~}}{~~$\defdby$}\,
      \stexpbit{\stexpzero}{(\stexpsum{\bact}{\stexpprod{\bact}{\extrsoluntilof{\acharthat}{\averti{0}}{\averti{1}}}})}
      \\
        & \;\parbox[t]{\widthof{~~$\BBPeq$~~}}{~~$\BBPeq$~~}\,
      \stexpsum{\bact}{\stexpprod{\bact}{\aact}}
      \\
      \extrsolof{\acharthat}{\averti{2}}
        & \;\parbox[t]{\widthof{~~$\BBPeq$~~}}{~~$\defdby$}\,
      \stexpbit{\stexpzero}{(\stexpsum{\stexpprod{\bact}{\extrsolof{\acharthat}{\averti{1}}}}
                                      {\stexpprod{\bact}{\extrsolof{\acharthat}{\averti{0}}}})}
      \\
        & \;\parbox[t]{\widthof{~~$\BBPeq$~~}}{~~$\BBPeq$~~}\,
      \stexpsum{\stexpprod{\bact}{\extrsolof{\acharthat}{\averti{1}}}}
               {\stexpprod{\bact}{(\stexpprod{\aact}{\extrsolof{\acharthat}{\averti{1}}})}}
      \\
        & \;\parbox[t]{\widthof{~~$\BBPeq$~~}}{~~$\BBPeq$~~}\,
      \stexpprod{(\stexpsum{\bact}{\stexpprod{\bact}{\aact})}}
                {\extrsolof{\acharthat}{\averti{1}}}
      \\
        & \;\parbox[t]{\widthof{~~$\BBPeq$~~}}{~~$\BBPeq$~~}\,
      \stexpprod{(\stexpsum{\bact}{\stexpprod{\bact}{\aact})}}
                {(\stexpbit{(\stexpsum{\stexpprod{\cact}{\aact}}
                                      {\stexpprod{\aact}{(\stexpsum{\bact}{\stexpprod{\bact}{\aact}})}})}
                           {\stexpzero})}
    \end{aligned}$}
    \hspace*{-9ex}
    \begin{aligned}[c]
      \\[-4ex]
      \scalebox{1.1}{\begin{tikzpicture}[scale=1,every node/.style={transform shape}]

\hspace*{2mm}
\matrix[anchor=north,row sep=0.9cm,every node/.style={draw,very thick,circle,minimum width=2.5pt,fill,inner sep=0pt,outer sep=2pt}] {
  \node(v_0){};
  \\
  \node(v_1){};
  \\
  \node(v_2){};
  \\
};
\calcLength(v_0,v_1){mylen}
%
\draw[<-,very thick,>=latex,chocolate](v_0) -- ++ (90:{0.5*\mylen pt});
\path(v_0) ++ ({0.3*\mylen pt},{0.25*\mylen pt}) node{$\averti{0}$};
\draw[->](v_0) to node[right,xshift={-0.05*\mylen pt},pos=0.45]{\small $\aact$} (v_1); 
\path(v_1) ++ ({0.325*\mylen pt},0cm) node{$\averti{1}$};
\draw[->,thick](v_1) to node[right,xshift={-0.05*\mylen pt},pos=0.45]{$\loopnsteplab{1}$}
                        node[left,xshift={0.05*\mylen pt},pos=0.45]{\small $\aact$} (v_2);
\draw[->,thick,shorten <= 5pt](v_1) to[out=175,in=180,distance={0.75*\mylen pt}] 
         node
             [left]{$\loopnsteplab{1}$} 
         node[above,yshift={0.05*\mylen pt},pos=0.7]{\small $\cact$} (v_0);
\path(v_2) ++ (-0cm,{-0.275*\mylen pt}) node{$\averti{2}$};
\draw[->](v_2) to[out=180,in=185,distance={0.75*\mylen pt}]  
               node[below,yshift={0.0*\mylen pt},pos=0.2]{\small $\bact$} (v_1);
\draw[->](v_2) to[out=0,in=0,distance={1.3*\mylen pt}] 
               node[below,yshift={0.00*\mylen pt},pos=0.125]{\small $\bact$} (v_0);
\path (v_0) ++ ({0*\mylen pt},{0.95*\mylen pt}) node{\large $\achart,\,\acharthat$};
%
\end{tikzpicture}}
    \end{aligned}
    \hspace*{0ex}
    \scalebox{0.92}{$
    \begin{aligned}[c]
      \asol{\averti{0}} 
        & \;\parbox[t]{\widthof{~$\eqinsol{\BBP}$~}}{~$\eqinsol{\BBP}$~}\;
      \stexpprod{\aact}{\asol{\averti{1}}}
        \;\;\;\parbox[t]{\widthof{(use of `is provable solution')}}
                        {(${}^{\text{\scriptsize (sol)}}$ means 
                         \\[-0.5ex]\phantom{(}%
                         use of `is provable solution')}
      \\[0.75ex]
      \asol{\averti{1}}   
        & \;\parbox[t]{\widthof{~$\eqinsol{\BBP}$~}}{~$\eqinsol{\BBP}$~}\;
      \stexpsum{\stexpprod{\cact}{\asol{\averti{0}}}}
               {\stexpprod{\aact}{\asol{\averti{2}}}} 
      \\[-0.25ex] 
        & \;\parbox[t]{\widthof{~$\eqinsol{\BBP}$~}}{~$\eqinsol{\BBP}$~}\;
      \stexpsum{\stexpprod{\cact}{(\stexpprod{\aact}{\asol{\averti{1}}})}}
               {\stexpprod{\aact}{(\stexpsum{\stexpprod{\bact}{\asol{\averti{1}}}}
                                            {\stexpprod{\bact}{\asol{\averti{0}}}})}} 
      \\[-0.25ex]
        & \;\parbox[t]{\widthof{~$\eqinsol{\BBP}$~}}{~$\eqinsol{\BBP}$~}\;
      \stexpsum{\stexpprod{\cact}{(\stexpprod{\aact}{\asol{\averti{1}}})}}
               {\stexpprod{\aact}{(\stexpsum{\stexpprod{\bact}{\asol{\averti{1}}}}
                                            {\stexpprod{\bact}{(\stexpprod{\aact}{\asol{\averti{1}}})}})}}
      \\
        & \;\parbox[t]{\widthof{~$\eqinsol{\BBP}$~}}{~$\eqin{\BBP}$~}\;
      \stexpsum{\stexpprod{(\stexpsum{\stexpprod{\cact}{\aact}}
                                     {\stexpprod{\aact}  
                                                {(\stexpsum{\bact}{\stexpprod{\bact}{\aact}})}})}
                          {\asol{\averti{1}}}}
               {\stexpzero}           
      \\[-0.25ex]
        & \;\parbox[t]{\widthof{~$\eqinsol{\BBP}$~}}{~~~~~~~$\Downarrow$\text{~~applying $\RSPbit$}} 
      \\[-0.25ex]  
      \asol{\averti{1}} 
        & \;\parbox[t]{\widthof{~$\eqinsol{\BBP}$~}}{~$\seqin{\BBP}$~}\;
      \stexpbit{(\stexpsum{\stexpprod{\cact}{\aact}}
                          {\stexpprod{\aact}
                                     {(\stexpsum{\bact}{\stexpprod{\bact}{\aact}})}})}
               {\stexpzero} 
      \\
        & \;\parbox[t]{\widthof{~$\eqinsol{\BBP}$~}}{~$\eqin{\BBP}$~}\;
      \extrsolof{\acharthat}{\averti{1}}  \quad \text{(see in the derivation on the left)}
      \\[-0.25ex]
        & \;\parbox[t]{\widthof{~$\eqinsol{\BBP}$~}}{~~~~~~~$\Downarrow$\phantom{\text{~~applying $\RSPbit$}}} 
      \\[-0.25ex]
      \asol{\averti{0}}
        & \;\parbox[t]{\widthof{~$\eqinsol{\BBP}$~}}{~$\eqinsol{\BBP}$~}\;
      \stexpprod{\aact}{\asol{\averti{1}}}
        ~~\seqin{\BBP}~~
      \stexpprod{\aact}{\extrsolof{\acharthat}{\averti{1}}}
        ~~\seqinsol{\BBP}~~
      \extrsolof{\acharthat}{\averti{0}}
      \\[-0.25ex]
        & \;\parbox[t]{\widthof{~$\eqinsol{\BBP}$~}}{~~~~~~~$\Downarrow$\phantom{\text{~~applying $\RSPbit$}}} 
      \\[-0.25ex]
      \asol{\averti{2}}
        & \;\parbox[t]{\widthof{~$\eqinsol{\BBP}$~}}{~$\eqinsol{\BBP}$~}\;
      \stexpsum{\stexpprod{\bact}{\asol{\averti{1}}}}
               {\stexpprod{\bact}{\asol{\averti{0}}}} 
      \\[-0.25ex]
        & \;\parbox[t]{\widthof{~$\eqinsol{\BBP}$~}}{~$\eqin{\BBP}$~}\;
      \stexpsum{\stexpprod{\bact}{\extrsolof{\acharthat}{\averti{1}}}}
               {\stexpprod{\bact}{\extrsolof{\acharthat}{\averti{0}}}} ~~\eqinsol{\BBP}~~ \extrsolof{\acharthat}{\averti{2}}
    \end{aligned}$}
  \end{equation*}     
  \vspace*{-1.75ex}
  \caption{\label{fig:ex:extr::sol:vs:extrsol}%
           Left:
           the process of extracting the provable solution~$\sextrsolof{\acharthat}$ of a chart $\achart$ from an \protect\LLEEwitness~$\acharthat$ of $\achart$
           as in the middle.\vspace*{-0.25mm} 
           Right: 
           steps for showing that an arbitrary provable solution~$\sasol$ of $\achart$ 
           is \protect\provablyin{\BBP} equal to the extracted solution $\sextrsolof{\acharthat}$.%
           }
\end{center}
\end{figure*}  %

\begin{definition}\label{def:extrsol}
  Let $\acharthat$ be a \LLEEwitness\ of a chart~$\achart$.
  Then the \emph{relative extraction function of $\acharthat$} is defined inductively as:
  \begin{align*}
    &
    \sextrsoluntilof{\acharthat} 
      \funin 
        \descsetexp{\pair{\bvert}{\avert}}{\avert,\bvert\in\verts\setminus\setexp{\tick},\avert\descendsinloopto\bvert} 
          \to
        \StExpsover{\actions} \punc{,}
    \\[-0.75ex]
    &    
    \begin{aligned}[b]
      \extrsoluntilof{\acharthat}{\bvert}{\avert} 
      \;\defdby\; 
        \Bigl( &
                \Bigl(
                  \stexpsum{
                    \Bigl(
                      \sum_{i=1}^{m}
                        \aacti{i}
                    \Bigr)
                             }{         
                    \Bigl(
                      \sum_{j=1}^{n}
                            \stexpprod{\bacti{i}}
                                      {\extrsoluntilof{\acharthat}{\bverti{j}}{\bvert}}
                    \Bigr)
                              }
                \Bigr)^{\sstexpbit}\hspace*{-1pt} 
        \\[-1ex]     
        & \Bigl(
                  \stexpsum{         
                    \Bigl(
                      \sum_{i=1}^{p}
                        \cacti{i}
                    \Bigr) \:
                            }{\:
                    \Bigl(
                      \sum_{j=1}^{q}
                        \stexpprod{\dacti{j}}
                                  {\extrsoluntilof{\acharthat}{\cverti{j}}{\avert}}
                    \Bigr)   }
                \Bigr)
        \Bigr) \punc{,}          
    \end{aligned}
  \end{align*}
  provided that $\bvert$ has \loopentrytransitions\vspace*{-1mm}  
  $\{ \bvert \lti{\aacti{i}}{\loopnsteplab{\aLnamei{i}}} \bvert   \mid  i=1,\ldots,m \}
     \cup
   \{\bvert \lti{\bacti{j}}{\loopnsteplab{\bLnamei{j}}} \bverti{j}   \mid  j = 1,\ldots,n \land \bverti{j}\neq\bvert \}$
  and\vspace*{-1mm} body transitions
  $\{ \bvert \lti{\cacti{i}}{\bodylab} \avert   \mid   i=1,\ldots,p \}
     \cup
   \{ \bvert \lti{\dacti{j}}{\bodylab} \cverti{j}   \mid  j=1,\ldots,q\land \cverti{j}\neq v \}$.
  Hereby the induction proceeds on $\pair{\enl{\avert}}{\bosn{\bvert}}$ with the lexicographic order $\sltlex$ on $\nat\times\nat$:
  For $\extrsoluntilof{\acharthat}{\bverti{j}}{\bvert}$ 
  we have $\pair{\enl{\bvert}}{\bosnnf{\bverti{j}}} \ltlex \pair{\enl{\avert}}{\bosn{\bvert}}$
  due to $\enlnf{\bverti{j}} < \enlnf{\avert}$,
  which follows from $\avert \descendsinloopto \bvert$ by Lem.~\ref{lem:def:extrsol}, \ref{it:enl:lem:def:extrsol}.
  For $\extrsoluntilof{\acharthat}{\cverti{j}}{\avert}$ 
  we have $\pair{\enl{\avert}}{\bosnnf{\cverti{j}}} \ltlex \pair{\enl{\avert}}{\bosn{\bvert}}$
  due to $\bosnnf{\cverti{j}} < \bosn{\bvert}$,
  which follows from $\bvert \redi{\bodylab} \cverti{j}$ by Lem~\ref{lem:def:extrsol}, \ref{it:bosn:lem:def:extrsol}.
  
  The \emph{extraction function of $\acharthat$} is defined by:
  \begin{align*}
    &
    \sextrsolof{\acharthat} 
      \funin 
        \verts\setminus\setexp{\tick}
          \to
        \StExpsover{\actions} \punc{,}
    \\[-0.75ex]
    &    
    \begin{aligned}[b]
      \extrsolof{\acharthat}{\bvert} 
      \;\defdby\; 
        \Bigl( &
                \Bigl(
                  \stexpsum{
                    \Bigl(
                      \sum_{i=1}^{m}
                        \aacti{i}
                    \Bigr)
                             }{         
                    \Bigl(
                      \sum_{j=1}^{n}
                            \stexpprod{\bacti{j}}
                                      {\extrsoluntilof{\acharthat}{\bverti{j}}{\bvert}}
                    \Bigr)
                              }
                \Bigr)^{\sstexpbit}\hspace*{-1pt} 
        \\[-1ex]     
        & \Bigl(
                  \stexpsum{         
                    \Bigl(
                      \sum_{i=1}^{p}
                        \cacti{i}
                    \Bigr) \:
                            }{\:
                    \Bigl(
                      \sum_{j=1}^{q}
                        \stexpprod{\dacti{j}}
                                  {\extrsolof{\acharthat}{\cverti{j}}}
                    \Bigr)   }
                \Bigr)
        \Bigr) \punc{,}          
    \end{aligned}
  \end{align*}
  with induction on $\bosn{\bvert}$,
  provided that $\bvert$ has \loopentrytransitions\  
  $\{ \bvert \lti{\aacti{i}}{\loopnsteplab{\aLnamei{i}}} \bvert   \mid  i=1,\ldots,m \}
     \cup
   \{\bvert \lti{\bacti{j}}{\loopnsteplab{\bLnamei{j}}} \bverti{j}   \mid  j = 1,\ldots,n \land \bverti{j}\neq\bvert \}$
  and body transitions
  $\{ \bvert \lti{\cacti{i}}{\bodylab} \tick   \mid   i=1,\ldots,p \}
     \cup
   \{ \bvert \lti{\dacti{j}}{\bodylab} \cverti{j}   \mid  j=1,\ldots,q\land \cverti{j}\neq \tick \}$.
  For $\extrsolof{\acharthat}{\cverti{j}}$ the induction hypothesis holds  
  due to $\bosn{\cverti{j}} < \bosn{\bvert}$, which follows from $\bvert \redi{\bodylab} \cverti{j}$ 
  by Lem.~\ref{lem:def:extrsol}, \ref{it:bosn:lem:def:extrsol}. 
\end{definition}

\begin{lemma}[uses the \BBP-axioms (B1)--(B6), (BKS2), but not the rule $\RSPbit\,$]\label{lem:prop:extracted:fun:is:solution}
  In a chart $\achart$ with \LLEEwitness\ $\acharthat$, if $\,\avert \descendsinloopto \bvert\,$,
  then $\,\extrsolof{\acharthat}{\bvert} 
            \eqin{\BBP}
          \stexpprod{\extrsoluntilof{\acharthat}{\bvert}{\avert}}
                    {\extrsolof{\acharthat}{\avert}}\,$.
\end{lemma}

\begin{proposition}[uses the \BBP-axioms (B1)--(B6), (BKS1), (BKS2), but not the rule $\RSPbit\,$]\label{prop:extracted:fun:is:solution}
  For every \LLEEwitness~$\acharthat$ of a chart $\achart$,
  the extraction function $\sextrsolof{\acharthat}$ 
  is a provable solution~of~$\achart$.
\end{proposition} 

The proof of Lem.~\ref{lem:prop:extracted:fun:is:solution} 
proceeds by induction on $\bosn{\bvert}$; no induction is needed
for the proof of Prop.~\ref{prop:extracted:fun:is:solution} (cf. appendix).

\begin{example}\label{ex:prop:extracted:fun:is:solution}\mbox{}%
  Left in Fig.~\ref{fig:ex:extr::sol:vs:extrsol} we illustrate the extraction of a provable solution 
  for the \LLEEwitness~$\acharthat = \charthatof{\astexpi{0}}$ in Ex.~\ref{ex:entry-body-labeling}
  of the chart $\achart = \chartof{\astexpi{0}}$ in Ex.~\ref{ex:chart:interpretation}.
  In order to obtain the principal value $\extrsolof{\acharthat}{\averti{0}}$ of the extracted solution $\sextrsolof{\acharthat}$,
  its definition is expanded. It recurs on $\extrsolof{\acharthat}{\averti{1}}$,
  and then on $\extrsoluntilof{\acharthat}{\averti{0}}{\averti{1}}$ and $\extrsoluntilof{\acharthat}{\averti{2}}{\averti{1}}$.
  After computing those star expressions by using the definition of $\sextrsoluntilof{\acharthat}$, the principal value can be obtained
  by substitution. The star expressions $\extrsolof{\acharthat}{\averti{1}}$ and $\extrsolof{\acharthat}{\averti{2}}$ are obtained similarly.
  For readability we have simplified the arising terms on the way by using the equality $\stexpbit{\stexpzero}{\avar} \BBPeq \avar$
  (which follows by $(\commstexpsum)$, $(\neutralstexpsum)$, $(\stexpzerostexpprod)$, and (BKS1)).
\end{example}

\begin{lemma}[uses the \BBP-axioms (B1)--(B6), and the rule $\RSPbit\,$]\label{lem:prop:extrsol:vs:solution}
  If $\avert \descendsinloopto \bvert$,\vspace*{-0.5mm}
  then $\asol{\bvert} 
            \eqin{\BBP}
          \stexpprod{\extrsoluntilof{\acharthat}{\bvert}{\avert}}
                    {\asol{\avert}}$
  for every provable solution $\sasol$ of a chart~$\achart$ with \LLEEwitness~$\acharthat$.                   
\end{lemma}

\begin{proposition}[uses the \BBP-axioms (B1)--(B6), and the rule $\RSPbit\,$]\label{prop:extrsol:vs:solution}
  Let $\sasoli{1}$ and $\sasoli{2}$ be provable solutions of a \LLEEchart. 
  Then $\asoli{1}{\bvert} \BBPeq \asoli{2}{\bvert}$ for all vertices $\bvert\neq\tick$.
\end{proposition}

For the \emph{proof} of this proposition, see Fig.~\ref{fig:prf:prop:extrsol:vs:sol}.
The proof of Lem.~\ref{lem:prop:extrsol:vs:solution} (see in the appendix)
proceeds by the same induction measure as we used for the 
relative extraction function. 

\begin{figure*}
  \begin{minipage}{\textwidth}
    \begin{proof}[Proof (of Prop.~\ref{prop:extrsol:vs:solution})]
      Let $\acharthat$ be a \LLEEwitness\ of a chart~$\achart$. Let $\sasol$ be a provable solution of $\achart$.  
      We have to show that $\asol{\bvert} \BBPeq \extrsolof{\acharthat}{\bvert}$ for all $\bvert\neq\tick$.
      For this, let $\bvert \neq \tick$.   
      The derivation below \vspace*{-.4mm}is based on the set representation of transitions from $\bvert$ in $\acharthat$ 
      as formulated in the definition of $\extrsolof{\acharthat}{\bvert}$.
      The first derivation step uses that $\sasol$ is a provable solution of $\achart$ and axioms $(\commstexpsum)$, $(\assocstexpsum)$, and $(\idempotstexpsum)$,
      the second step uses Lem.~\ref{lem:prop:extrsol:vs:solution} in view of $\bvert \descendsinloopto \bverti{j}$ for $j=1,\ldots,n$, 
      and the third step uses axioms $(\distr)$, $(\assocstexpprod)$, and $(\neutralstexpsum)$.
      \begin{align*}
        \asol{\bvert}~
          & \;\parbox[t]{\widthof{$\eqin{\BBP}$}}{$\BBPeq$}\:
            \stexpsum{
              \Bigl(
                \stexpsum{
                  \Bigl(
                    \sum_{i=1}^{m}
                      \stexpprod{\aacti{i}}{\asol{\bvert}}
                  \Bigr)
                          }{         
                  \Bigl(
                    \sum_{j=1}^{n}
                      \stexpprod{\bacti{j}}
                                {\asol{\bverti{j}}}
                  \Bigr)
                            }   
              \Bigr)
                      }{
              \Bigr(           
                \stexpsum{          
                  \Bigl(
                    \sum_{i=1}^{p}
                      \cacti{i}
                  \Bigr) 
                          }{
                  \Bigl(
                    \sum_{j=1}^{q}
                      \stexpprod{\dacti{j}}{\asol{\cverti{j}}}
                  \Bigr)
                            }
              \Bigr)
                         }          
            \\
          & \;\parbox[t]{\widthof{$\eqin{\BBP}$}}{$\BBPeq$}\:
            \stexpsum{
              \Bigl(
                \stexpsum{
                  \Bigl(
                    \sum_{i=1}^{m}
                      \stexpprod{\aacti{i}}{\asol{\bvert}}
                  \Bigr)
                          }{         
                  \Bigl(
                    \sum_{j=1}^{n}
                      \stexpprod{\bacti{j}}
                                {\bigl(
                                   \stexpprod{\extrsoluntilof{\acharthat}{\bverti{j}}{\bvert}}
                                             {\asol{\bvert}}
                                 \bigr)}             
                  \Bigr)
                            }   
              \Bigr)
                      }{
              \Bigr(           
                \stexpsum{          
                  \Bigl(
                    \sum_{i=1}^{p}
                      \cacti{i}
                  \Bigr) 
                          }{
                  \Bigl(
                    \sum_{j=1}^{q}
                      \stexpprod{\dacti{j}}{\asol{\cverti{j}}}
                  \Bigr)
                            }
              \Bigr)
                         }          
            \\
          & \;\parbox[t]{\widthof{$\eqin{\BBP}$}}{$\BBPeq$}\:
            \stexpsum{
              \stexpprod{
                \Bigl(
                  \stexpsum{
                    \Bigl(
                      \sum_{i=1}^{m}
                        \aacti{i}
                    \Bigr)
                            }{         
                    \Bigl(
                      \sum_{j=1}^{n}
                          \bigl(
                            \stexpprod{\bacti{j}}
                                      {\extrsoluntilof{\acharthat}{\bverti{j}}{\bvert}}
                          \bigr)
                    \Bigr)
                              }   
                \Bigr)
                       }{\asol{\bvert}}
                      }{
              \Bigr(           
                \stexpsum{          
                  \Bigl(
                    \sum_{i=1}^{p}
                      \cacti{i}
                  \Bigr) 
                          }{
                  \Bigl(
                    \sum_{j=1}^{q}
                      \stexpprod{\dacti{j}}{\asol{\cverti{j}}}
                  \Bigr)
                            }
              \Bigr)
                         } 
      \end{align*}  
      In view of this derived provable equality for $\asol{\bvert}$,
      we can now apply the rule $\RSPbit$ in order to obtain:
      \begin{align*}
        \asol{\bvert}~
          & \;\parbox[t]{\widthof{$\eqin{\BBP}$}}{$\BBPeq$}\:
            \Bigl(
              \stexpsum{
                \Bigl(
                  \sum_{i=1}^{m}
                    \aacti{i}
                \Bigr)
                        }{         
                \Bigl(
                  \sum_{j=1}^{n}
                        \stexpprod{\bacti{j}}
                                  {\extrsoluntilof{\acharthat}{\bverti{j}}{\bvert}}
                \Bigr)
                          }
                \Bigr)^{\sstexpbit}\hspace*{-1pt} 
                \Bigl(
                  \stexpsum{         
                    \Bigl(
                      \sum_{i=1}^{p}
                        \cacti{i}
                    \Bigr) 
                            }{
                    \Bigl(
                      \sum_{j=1}^{q}
                        \stexpprod{\dacti{j}}
                                  {\extrsolof{\acharthat}{\cverti{j}}}
                    \Bigr)   }
                \Bigr)
          ~~\parbox[t]{\widthof{$\syntequal$}}{$\syntequal$}~~
            \extrsolof{\acharthat}{\bvert}
      \end{align*}
    In this last step we have used the definition of $\extrsolof{\acharthat}{\bvert}$.  
    \end{proof}
  \end{minipage}
  \caption{\label{fig:prf:prop:extrsol:vs:sol}%
           Proof of Prop.~\ref{prop:extrsol:vs:solution}.}
\end{figure*}  

\begin{example}\label{ex:prop:extrsol:vs:solution} \mbox{}
  In the right half of Fig.~\ref{fig:ex:extr::sol:vs:extrsol} we prove 
  that an arbitrary provable solution~$\sasol$ of \LLEEchart~$\achart = \chartof{\astexpi{0}}$ in Ex.~\ref{ex:chart:interpretation}
  with \LLEEwitness~$\acharthat = \charthatof{\astexpi{0}}$ in Ex.~\ref{ex:entry-body-labeling} 
  is provably equal to the extracted solution $\sextrsolof{\acharthat}$ of $\achart$. 
  Crucially, the defining conditions for $\sasol$ as a provable solution of $\achart$ are expanded along the loop at $\averti{1}$.
  The loop behavior obtained is the same as that which is used in the definition of $\extrsolof{\acharthat}{\averti{1}}$.
  By applying the fixed-point rule $\RSPbit$ 
  we can then deduce \provablein{\BBP} equality of $\asol{\averti{1}}$ and $\extrsolof{\acharthat}{\averti{1}}$.
  By using the solution conditions for $\sasol$ again, provable equality is then transferred to~$\averti{0}$~and~$\averti{1}$.
\end{example}

\section{Preservation of LLEE 
         under 
               collapse}%
  \label{collapse}         
\renewcommand{\ll}[1]{#1}%

In this section we establish the remaining result {\bf (C)} from Sect.~\ref{compl:proof} that is crucial for the completeness proof:
that the bisimulation collapse of a \LLEEchart\ is again a \LLEEchart. 

This result is achieved by a step-wise construction of a bisimulation collapse.
Pairs of bisimilar vertices $\bverti{1}$ and $\bverti{2}$ are collapsed one at a time,
whereby the incoming transitions of $\bverti{1}$ are redirected to $\bverti{2}$.
The crux is to take care, and to prove, that the resulting chart has again a \LLEEwitness.

\begin{definition}\label{def:connect-through}
  Let $\achart$ be a chart, with vertices $\bverti{1}$ and $\bverti{2}$.
  
  The \emph{connect-$\bverti{1}$-through-to-$\bverti{2}$ chart} $\connthroughin{\achart}{\bverti{1}}{\bverti{2}}$ 
  of $\achart$ 
  is obtained by redirecting all incoming transitions at $\bverti{1}$ over to $\bverti{2}$,
  and, if $\bverti{1}$ is the start vertex of $\achart$, making $\bverti{2}$ the new start vertex;
  in this way $\bverti{1}$ gets unreachable, and it is removed with other unreachable vertices
  to obtain a \startconnected\ chart.
  
  Let $\acharthat$ be an \entrybodylabeling\ of $\achart$.
  Then we define the \entrybodylabeling~$\connthroughin{\acharthat}{\bverti{1}}{\bverti{2}}$ of $\connthroughin{\achart}{\bverti{1}}{\bverti{2}}$ 
  as follows:\vspace{-0.5mm}
  every transition in $\connthroughin{\achart}{\bverti{1}}{\bverti{2}}$ that was already a transition~$\atrans$ in $\achart$
        inherits\vspace{-0.5mm} its marking label from $\atrans$ in $\acharthat$;
  and every transition in $\connthroughin{\achart}{\bverti{1}}{\bverti{2}}$ that arises as the redirection~$\atransi{\bverti{2}}$ to $\bverti{2}$ 
  of a transition $\atrans$ to $\bverti{1}$ in $\achart$
  such that $\atransi{\bverti{2}}$ does not coincide with a transition already in $\achart$ 
  inherits its marking label from $\atrans$ in $\acharthat$.  
  %
\end{definition}

\vspace*{-1.5ex}
\begin{lemma}\label{lem:connthroughchart:bisim}
  If $\bverti{1}\bisim \bverti{2}$ in $\achart$, then $\connthroughin{\achart}{\bverti{1}}{\bverti{2}}\bisim\achart$.
\end{lemma}

While the connect-through operation of bisimilar vertices in a chart thus results in a bisimilar chart,
its application to a \LLEEwitness\ (an \entrybodylabeling) does not need to 
yield a \LLEEwitness\ again: the property LEE may be lost.

 

\begin{example}\label{ex:trans-I}
  Consider the \LLEEwitness~$\acharthat$ in the middle below.
  The unspecified action labels are assumed to facilitate that $\bverti{1}$ and $\bverti{2}$ are bisimilar.
  Hence also $\bverthati{1}$ and $\bverthati{2}$ are bisimilar. 
  Bisimilarity is indicated by the broken lines.
  The \connectthroughchart{\bverti{1}}{\bverti{2}} on the left is not a LLEE-chart,\label{argument:counterex:transformation:I} 
  because it does not satisfy \LEE: after the loop subchart induced by the downwards transition from $\bverthati{2}$ is eliminated,
  and garbage collection is done, the remaining chart without the dotted transitions still has an infinite path;
  yet it does not contain another loop subchart, 
  because each infinite path can reach $\tick$ without returning to its source.
  An example of this is the red path from $\bverthati{1}$ via $\bverti{2}$ and $\bverthati{2}$ to $\tick$.
  In $\acharthat$, the bisimilar pair $\bverti{1},\bverti{2}$ progresses to the bisimilar pair $\bverthati{1},\bverthati{2}$. 
  The \connectthroughchart{\bverthati{1}}{\bverthati{2}} on the right 
  is a \LLEEchart, as witnessed by~the~\entrybodylabeling~$\connthroughin{\acharthat}{\bverthati{1}}{\bverthati{2}}$.
  \vspace*{-1.5ex}
  \begin{center}
    \scalebox{0.98}{
  %
%
%
\begin{tikzpicture}[scale=0.875]
\matrix[anchor=center,row sep=1cm,column sep=0.5cm,every node/.style={draw,very thick,circle,minimum width=2.5pt,fill,inner sep=0pt,outer sep=2pt}] at (0,0) {
  \node(0){};   & \node[draw=none,fill=none](root-anchor){};                    
                                       & \node(1){};
  \\
  \node(00){};  & \node[draw=none,fill=none](sink-anchor){};  
                                       & \node(10){};
  \\
  \node(000)[draw=none,fill=none]{}; 
                & \node[draw=none,fill=none](label-anchor){};  
                                       & \node(100){};
  \\
};
\calcLength(0,00){mylen};
\path (root-anchor) ++ (0cm,1cm) node[style={draw,very thick,circle,minimum width=2.5pt,fill,inner sep=0pt,outer sep=2pt}](root){};
\path (sink-anchor) ++ (0cm,0.25cm) node[style={draw,very thick,circle,minimum width=2.5pt,fill,inner sep=0pt,outer sep=2pt,red}](sink){};
\draw[thick,red] (sink) circle (0.12cm); 
%

%
%
\draw[<-,very thick,>=latex,chocolate](root) -- ++ (90:0.575cm);
\path (root) ++ (-0.15cm,1.25cm) node{\scalebox{1.45}{{$\connthroughin{\achart}{\bverti{1}}{\bverti{2}}$}}};
\draw[->] (root) to node[below,xshift=0.1cm]{} (0); 
\draw[->] (root) to node[below,xshift=-0.1cm]{} (1);
\path (0) ++ (-0.175cm,0.05cm) node[above]{$\bverthati{1}$};
\draw[->,red] (0) to node[left,pos=0.45,xshift=0.06cm]{
} (00);
\draw[->,shorten >= 2pt] (0) to (sink);
\draw[->,red] (00) to node[left,darkmagenta]{} (100);
\draw[->,distance=0.5cm,out=180,in=185] (00) to (0);
\path (1) ++ (0.15cm,0.05cm) node[above]{$\bverthati{2}$};
\draw[-{>[length=1mm,width=1.8mm]},thick,dotted] (1) to node[right,pos=0.45,xshift=-0.06cm]{
} (10);
\draw[->,shorten >= 2pt,red] (1) to (sink);
\path (sink) ++ (0cm,0.55cm) node{$\tick$};
\draw[-{>[length=1mm,width=1.8mm]},thick,dotted] (10) to (100);
\draw[-{>[length=1mm,width=1.8mm]},thick,dotted,distance=0.75cm,out=0,in=-5] (10) to (1);
\draw[->,distance=1.25cm,out=0,in=-5,red] (100) to (1);
\path (100) ++ (0cm,-0.3cm) node{$\bverti{2}$};

\draw[magenta,thick,densely dashed] (0) to (1);
\matrix[anchor=center,row sep=1cm,column sep=0.5cm,every node/.style={draw,very thick,circle,minimum width=2.5pt,fill,inner sep=0pt,outer sep=2pt}] at (4,0) {
  \node(C_0){};   & \node[draw=none,fill=none](C_root-anchor){};
                                       & \node(C_1){};
  \\
  \node(C_00){};  & \node[draw=none,fill=none](C_sink-anchor){};  
                                       & \node(C_10){};
  \\
  \node(C_000){}; & \node[draw=none,fill=none](C_label-anchor){};  
                                       & \node(C_100){};
  \\
};
\path (C_root-anchor) ++ (0cm,1cm) node[style={draw,very thick,circle,minimum width=2.5pt,fill,inner sep=0pt,outer sep=2pt}](C_root){};
\path (C_root) ++ (0cm,1.25cm) node{\scalebox{1.45}{{$\acharthat$}}};
\path (C_sink-anchor) ++ (0cm,0.25cm) node[style={draw,very thick,circle,minimum width=2.5pt,fill,inner sep=0pt,outer sep=2pt}](C_sink){};
\draw[thick] (C_sink) circle (0.12cm); 
\path (C_sink) ++ (0cm,0.55cm) node{$\tick$};
%
%
%
%
\draw[<-,very thick,>=latex,chocolate](C_root) -- ++ (90:0.575cm);
\draw[->] (C_root) to node[below,xshift=0.1cm]{} (C_0);
\draw[->] (C_root) to node[below,xshift=-0.1cm]{} (C_1);
\draw[->,thick] (C_0) to node[left,pos=0.45,xshift=0.1cm]{$\loopsteplabof{2}$} (C_00);
\draw[->,shorten >= 2pt] (C_0) to (C_sink);
\path (C_0) ++ (-0.09cm,0.05cm) node[above]{$\bverthati{1}$};
\draw[->] (C_00) to node[left]{} (C_000);
\draw[->,distance=0.75cm,out=180,in=185] (C_00) to (C_0);
\draw[->,distance=1.25cm,out=180,in=185] (C_000) to (C_0);
\path (C_000) ++ (0cm,-0.3cm) node{$\bverti{1}$};
\draw[->,thick] (C_1) to node[right,pos=0.45,xshift=-0.08cm]{$\loopsteplabof{1}$} (C_10);
\draw[->,shorten >= 2pt] (C_1) to (C_sink);
\path (C_1) ++ (0.15cm,0.05cm) node[above]{$\bverthati{2}$};
\draw[->] (C_10) to (C_100);
\draw[->,distance=0.75cm,out=0,in=-5] (C_10) to (C_1);
\draw[->,distance=1.25cm,out=0,in=-5] (C_100) to (C_1);
\path (C_100) ++ (0cm,-0.3cm) node{$\bverti{2}$};
\draw[magenta,thick,densely dashed] (C_0) to (C_1);
\draw[magenta,thick,densely dashed,bend right,distance=0.4cm,looseness=1] (C_000) to (C_100);

\matrix[anchor=center,row sep=1cm,column sep=0.5cm,every node/.style={draw,very thick,circle,minimum width=2.5pt,fill,inner sep=0pt,outer sep=2pt}] at (6.6,0) {
  \node[draw=none,fill=none](Cw1hatw2hat_0){};   & \node[draw=none,fill=none](Cw1hatw2hat_root-anchor){};                    
                                       & \node(Cw1hatw2hat_1){};
  \\
  \node[draw=none,fill=none](Cw1hatw2hat_00){};  & \node[draw=none,fill=none](Cw1hatw2hat_sink-anchor){};  
                                       & \node(Cw1hatw2hat_10){};
  \\
  \node[draw=none,fill=none](Cw1hatw2hat_000)[draw=none,fill=none]{}; 
                & \node[draw=none,fill=none](Cw1hatw2hat_label-anchor){};  
                                       & \node(Cw1hatw2hat_100){};
  \\
};
\path (Cw1hatw2hat_root-anchor) ++ (0cm,1cm) node[style={draw,very thick,circle,minimum width=2.5pt,fill,inner sep=0pt,outer sep=2pt}](Cw1hatw2hat_root){};
\path (Cw1hatw2hat_sink-anchor) ++ (0cm,0.25cm) node[style={draw,very thick,circle,minimum width=2.5pt,fill,inner sep=0pt,outer sep=2pt}](Cw1hatw2hat_sink){};
\draw[thick] (Cw1hatw2hat_sink) circle (0.12cm); 
\path (Cw1hatw2hat_sink) ++ (0cm,0.55cm) node{$\tick$};
%
%
%
\draw[<-,very thick,>=latex,chocolate](Cw1hatw2hat_root) -- ++ (90:0.575cm);
\draw[->] (Cw1hatw2hat_root) to node[below,xshift=-0.1cm]{} (Cw1hatw2hat_1);
\path (Cw1hatw2hat_root) ++ (0.65cm,1.15cm) node{\scalebox{1.45}{{$\connthroughin{\acharthat}{\bverthati{1}}{\bverthati{2}}$}}};
\draw[->,thick] (Cw1hatw2hat_1) to node[right,pos=0.45,xshift=-0.08cm]{$\loopsteplabof{1}$} (Cw1hatw2hat_10);
\draw[->,shorten >= 2pt] (Cw1hatw2hat_1) to (Cw1hatw2hat_sink);
\path (Cw1hatw2hat_1) ++ (0.15cm,0.05cm) node[above]{$\bverthati{2}$};
\draw[->] (Cw1hatw2hat_10) to (Cw1hatw2hat_100);
\draw[->,distance=0.75cm,out=0,in=-5] (Cw1hatw2hat_10) to (Cw1hatw2hat_1);
\draw[->,distance=1.25cm,out=0,in=-5] (Cw1hatw2hat_100) to (Cw1hatw2hat_1);
\path (Cw1hatw2hat_100) ++ (0cm,-0.3cm) node{$\bverti{2}$};

\draw[-implies,thick,double equal sign distance, bend right,distance={1.2*\mylen pt},
               shorten <= 0.8cm,shorten >= 0.8cm
               ] (C_root) to node[above,pos=0.5,yshift={0.075*\mylen pt}] {\scalebox{1.25}{$\connthroughin{\achart}{\bverti{1}}{\bverti{2}}\! \mapsfrom \achart$}}  
                                                                         (root) ;

\draw[-implies,thick,double equal sign distance, bend left,distance={0.8*\mylen pt},
               shorten <= 0.5cm,shorten >= 0.5cm
               ] (C_root) to node[above,pos=0.6,yshift={0.075*\mylen pt}]{\scalebox{1.25}{$(\text{\nf I})^{(\bverthati{1})}_{\bverthati{2}}$}} (Cw1hatw2hat_root) ;

\end{tikzpicture}\label{fig:counterex:ex:transformation:I}}
  \end{center}
  \vspace*{-1.5ex}
\end{example}

%

\noindent
This illustrates that bisimilar pairs of vertices must be selected carefully, to safeguard
that the connect-through construction preserves LLEE. 
The proposition below expresses that
a pair of \ul{distinct} bisimilar vertices can always be selected in one of three mutually exclusive categories.
Later, three LLEE-preserving transformations I, II, and III will be defined for each of these categories.

\begin{proposition}\label{prop:reduced:br}
  If a \LLEEchart~$\achart$ is not a bisimulation collapse,
  then it contains a pair of bisimilar vertices $w_1,w_2$ 
  that satisfy, for a \LLEEwitness\ of $\achart$, one of the conditions:
  \begin{enumerate}[label=(C\arabic*)$\:$]
    \item{}\label{cond:transf:I}
      $\lognot{(\bverti{2} \redrtc \bverti{1})} \logand (\descendsinloopto \bverti{1} \Rightarrow \text{$\bverti{2}$ is not normed}\,)$,
    \item{}\label{cond:transf:II}  
      $\bverti{2} \loopsbacktotc \bverti{1}$,
    \item{}\label{cond:transf:III}
      $\existsstzero{\avert\in\verts}
         \bigl(\,
           \bverti{1} \dloopsbackto \avert
             \logand 
           \bverti{2} \loopsbacktotc \avert
         \,\bigr)\logand\lognot(\bverti{2} \redrtci{\bodylab} \bverti{1})$.
  \end{enumerate}
\end{proposition}

Condition~\ref{cond:transf:I} requires that $\bverti{1}$ and $\bverti{2}$ are in different \sccs, as there is no path from $\bverti{2}$ to $\bverti{1}$.
The additional proviso in \ref{cond:transf:I} constrains the pair in such a way that if both are normed, then $\bverti{1}$ must be outside of all loops
(otherwise the connect-$w_1$-through-to-$w_2$ operation does not preserve \LLEEcharts, see Ex.~\ref{ex:trans-I});
its asymmetric formulation helps to avoid the assumption of bisimilarity in Prop.~\ref{prop:LEEshape:preserve:conds} below. 
The two other conditions concern the situation that $\bverti{1}$ and $\bverti{2}$ are in the same \scc.
While in \ref{cond:transf:II}$\!$ $\bverti{1}$ and $\bverti{2}$ are comparable (but different) by the \txtloopsbackto\ relation $\sloopsbacktortc$,
they are incomparable in \ref{cond:transf:III}$\!$.
In the situation that $w_1,w_2$ loop back to the same vertex $v$, but $w_1$ directly loops back to $v$,
\ref{cond:transf:III}$\!$ also demands that no body step path exists from $w_2$ to $w_1$
(otherwise the connect-$w_1$-through-to-$w_2$ construction does not preserve \LLEEcharts, 
 see an example in the appendix).


In the proof of Prop.~\ref{prop:reduced:br}
we progress, from a given pair of distinct bisimilar vertices,
repeatedly via transitions, at one side picking loop-back transitions,
over pairs of distinct bisimilar vertices, until one of the conditions \ref{cond:transf:I}, \ref{cond:transf:II}, \ref{cond:transf:III} is met.
%
We will use a subset of the \bodytransitions\ in a \LLEEwitness.
By a \emph{loop-back transition}, written as $\cvert \lbsred \avert$, 
we mean a transition $\cvert \redi{\bodylab} \avert$ that stays within an \scc,
that is, $\sccof{\cvert}=\sccof{\avert}$.
The \emph{loops-back-to norm} $\lbsminn{\cvert}$ of $\cvert$ 
is the maximal length of a $\slbsred$ path from $\cvert$
(which is \welldefined\ by Lem.~\ref{lem:loop:relations}, \ref{it:bo:terminating} and chart finiteness). 
Note that $\lbsminn{\cvert}=0$ if and only if $\cvert$ does not loop back (denoted by $\lognot(\cvert \loopsbackto)$).

\begin{proof}
[Proof of Prop.~\ref{prop:reduced:br}]\label{prf:prop:reduced:br:start}
  We pick distinct bisimilar vertices $u_1,u_2$. First we consider the case \mbox{$\sccof{\cverti{1}} \neq \sccof{\cverti{2}}$}. 
  Without loss of generality, suppose $\lognot{( \cverti{2} \redrtc \cverti{1} )}$.
  We progress to a pair of vertices where \ref{cond:transf:I} holds, using induction on $\lbsminn{\cverti{1}}$.
  In the base case, $\lbsminn{\cverti{1}} = 0$,
  it suffices to show that it is not possible that both $\sdescendsinloopto\,\cverti{1}$ holds and $\cverti{2}$ is normed,
  because then we can define $w_1=u_1$ and $w_2=u_2$, and are done.
  Therefore suppose, toward a contradiction, that $\sdescendsinloopto\,\cverti{1}$ holds and $\cverti{2}$ is normed.
  Then $\cverti{1}$ is normed, too, since $\cverti{1}$ and $\cverti{2}$ are bisimilar.
  Also $\lognot{(\cverti{1}\,\sloopsbackto)}$ follows from $\lbsminn{\cverti{1}} = 0$, which says that there are no \txtloopsbackto\ steps from $\cverti{1}$.
  So we get that $\sdescendsinloopto\,\cverti{1}$, $\lognot{(\cverti{1}\,\sloopsbackto)}$, and $\cverti{1}$ is normed.
  This contradicts Lemma~\ref{lem:loop:relations}, \ref{it:descendsinloopto:notloopsbackto:not:normed}.
  In the induction step, $\lbsminn{\cverti{1}} > 0$ implies
    $\cverti{1} \lbsred \cvertacci{1}$
      and
    $\lbsminn{\cvertacci{1}} < \lbsminn{\cverti{1}}$ for some $\cvertacci{1}$.
  Since $\cverti{1} \bisim \cverti{2}$,
  we have
    $\cverti{2} \red \cvertacci{2}$
      and
    $\cvertacci{1} \bisim \cvertacci{2}$ for some $\cvertacci{2}$.
  Since $\cverti{1} \lbsred \cvertacci{1}$,
  by definition, $\cverti{1}$ and $\cvertacci{1}$ are in the same scc. Hence
    $\cvertacci{1} \redrtc \cverti{1}$. This implies $\lognot{( \cvertacci{2} \redrtc \cvertacci{1} )}$, for else
    $\cverti{2}
      \red
    \cvertacci{2}
      \redrtc
    \cvertacci{1}
      \redrtc
    \cverti{1}$,
   which contradicts the assumption $\lognot{( \cverti{2} \redrtc \cverti{1} )}$.
  Since $\cvertacci{1} \bisim \cvertacci{2}$ and 
    $\lognot{( \cvertacci{2} \redrtc \cvertacci{1} )}$ 
        and
    $\lbsminn{\cvertacci{1}} < \lbsminn{\cverti{1}}$,
  by induction there exists a bisimilar pair $\bverti{1},\bverti{2}$ for which \ref{cond:transf:I}~holds.
  
  Now let $\sccof{\cverti{1}} = \sccof{\cverti{2}}$.
  Then by Lem.~\ref{lem:loop:relations},~\ref{it:bi:reachable:lLEEw}, $\cverti{1}\loopsbacktortc v$ and $\cverti{2}\loopsbacktortc v$ for some~$v$.
  By Lem.~\ref{lem:loop:relations},~\ref{it:least-upper-bound}
  we pick $v$ as the least upper bound of $u_1,u_2$ with regard to $\loopsbacktortc$.
  If $\cverti{1}=\avert$, then $\cverti{2}\loopsbacktotc \cverti{1}$, 
  so \ref{cond:transf:II} holds
  for $\bverti{1} = \cverti{1}$ and $\bverti{2} = \cverti{2}$.
  If $\cverti{2}=\avert$, then likewise \ref{cond:transf:II} holds
  for $\bverti{1} = \cverti{2}$ and $\bverti{2} = \cverti{1}$.
  Now let $\cverti{1},\cverti{2} \neq \avert$.
  Since $\avert$ is the least upper bound,
  $\cverti{1} \loopsbacktortc \averti{1} \dloopsbackto \avert
              \convdloopsbackto \averti{2} \convloopsbacktortc \cverti{2}$
  for distinct $\averti{1},\averti{2}\in\verts$.
  There cannot be a cycle of body transitions, so
  $\neg(\averti{2} \redrtci{\bodylab} \averti{1})$ or
  $\neg(\averti{1} \redrtci{\bodylab} \averti{2})$.
    By symmetry it suffices to consider $\neg(\averti{2} \redrtci{\bodylab} \averti{1})$.
    Summarizing,
    \mbox{%
    $ 
        \cverti{1} 
           \loopsbacktortc 
         \averti{1}
           \dloopsbackto 
         \avert
           \convdloopsbackto 
         \averti{2}
           \convloopsbacktortc 
         \cverti{2}$
         and
         $\lognot{( \averti{2} \redrtci{\bodylab} \averti{1} )} 
         $}.
    For this situation we use induction on $\lbsminn{\cverti{1}}$.
    If $\cverti{1}=\averti{1}$, then  $\cverti{1}\dloopsbackto \avert$;
    taking $\bverti{1}=\cverti{1}$ and $\bverti{2}=\cverti{2}$, \ref{cond:transf:III} holds.
    \vspace*{-.5mm}So we can assume $\cverti{1} \loopsbacktotc \averti{1} \dloopsbackto \avert$.
      Pick a transition  $\cverti{1} \lbsred \cverti{1}'$ with $\lbsminn{\cverti{1}'}<\lbsminn{\cverti{1}}$;
      by definition, $\sccof{\cverti{1}'}=\sccof{\cverti{1}}$.
      Since $\cverti{1}\bisim \cverti{2}$, there is a transition $\cverti{2}\ult \cverti{2}'$ with $\cverti{1}'\bisim \cverti{2}'$ for some $\cverti{2}'$.
      If $\sccof{\cverti{1}'}\neq\sccof{\cverti{2}'}$, 
      then as before we can find bisimilar $\bverti{1},\bverti{2}$ for which \ref{cond:transf:I} holds.
      Now let
$\sccof{\cverti{1}'}=\sccof{\cverti{2}'}$, so $\cverti{1},\cverti{2},\cverti{1}',\cverti{2}'$ are in the same scc.
          Since $\cverti{1}\loopsbacktotc \averti{1}$ and $\cverti{1}\ult \cverti{1}'$, either $\cverti{1}'=\averti{1}$ or $\averti{1} \descendsinlooptotc\cverti{1}'$. 
          Moreover, $\sccof{\cverti{1}'}=\sccof{\cverti{1}}=\sccof{\averti{1}}$, so by Lem.~\ref{lem:loop:relations},~\ref{it:descendsinloopto:scc:loopsbackto}, $\cverti{1}'\loopsbacktortc \averti{1}$. 
          Since $\cverti{2}\loopsbacktortc\averti{2}$, we can distinguish two cases 
          (for illustrations for each of the subcases, see the appendix).\vspace*{-.75mm}
      %
      \begin{description}\setlength{\itemsep}{.5ex}
          \item{\emph{Case 1:}}
          $\cverti{2}\loopsbacktotc \averti{2}$.
          Since $\cverti{2}\ult \cverti{2}'$, either $\cverti{2}'=\averti{2}$ or $\averti{2} \descendsinlooptotc\cverti{2}'$. 
          Moreover,  $\sccof{\cverti{2}'}=\sccof{\cverti{2}}=\sccof{\averti{2}}$, \vspace*{-.25mm}so by Lem.~\ref{lem:loop:relations},~\ref{it:descendsinloopto:scc:loopsbackto},
          $\cverti{2}'\loopsbacktortc \averti{2}$.
          Hence,
          $
             \cvertacci{1} 
                 \loopsbacktortc 
               \averti{1}
                 \dloopsbackto 
               \avert
                 \convdloopsbackto 
               \averti{2}
                 \convloopsbacktortc 
               \cvertacci{2}
               \logand
             \lognot{( \averti{2} \redrtci{\bodylab} \averti{1} )} 
             $,
         and $\lbsminn{\cverti{1}'}<\lbsminn{\cverti{1}}$. 
         \vspace*{.25mm}We apply the induction hypothesis to obtain 
         a bisimilar pair $w_1,w_2$ for which \ref{cond:transf:I}, \ref{cond:transf:II}, or \ref{cond:transf:III} holds.
         Below we illustrate both of the cases in which $\cverti{2} \red \cvertacci{2}$ 
         is a \loopentrytransition, or a \bodytransition.
          \begin{center}
            $
            \begin{aligned}[c]
             \scalebox{0.925}{\begin{tikzpicture}[scale=1,every node/.style={transform shape}]
%
\matrix[anchor=center,row sep=0.8cm,column sep=1.75cm,every node/.style={draw,thick,circle,minimum width=2.5pt,fill,inner sep=0pt,outer sep=2pt}] at (0,0) {
                   &  \node(v){};
  \\[0.25cm]
  \node(v_1){};    &  \node[draw=none,fill=none](h0){};
                                   &  \node(v_2){};  
  \\
  \node(v_11){};   &               &  \node(v_12){};
  \\[0.25cm]
  \node(v_n1){};   &               &  \node(v_n2){};
  \\
  \node(u_1){};    &  \node[draw=none,fill=none](h){}; 
                                   &  \node(u_2){};
  \\
};
\calcLength(v,h0){mylen};

\path (v) ++ (0pt,{0.25*\mylen pt}) node{$\avert$};
  \path (v) ++ ({-0.5*\mylen pt},{-0.4*\mylen pt}) node[draw,thick,circle,minimum width=2.5pt,fill,inner sep=0pt,outer sep=2pt] (v_01){};
    \draw[->,thick] (v) to (v_01);
    \draw[->>,bend right,distance={0.35*\mylen pt}] (v_01) to (v_1);
  \path (v) ++ ({0.5*\mylen pt},{-0.4*\mylen pt}) node[draw,thick,circle,minimum width=2.5pt,fill,inner sep=0pt,outer sep=2pt] (v_03){}; 
    \draw[->,thick] (v) to (v_03);
    \draw[->>,bend left,distance={0.35*\mylen pt}] (v_03) to (v_2);
   %
   %
%
\path (v_1) ++ ({-0.35*\mylen pt},{0.05*\mylen pt}) node{$\averti{1}$};
\draw[->>,out=170,in=220,distance={0.5*\mylen pt}] (v_11) to (v_1);

\path (v_1) ++ ({0.25*\mylen pt},{-0.4*\mylen pt}) node[draw,thick,circle,minimum width=2.5pt,fill,inner sep=0pt,outer sep=2pt](v_1_v_11){};
  \draw[->,thick] (v_1) to (v_1_v_11);
  \draw[->>,distance={0.25*\mylen pt},out=-30,in=20] (v_1_v_11) to (v_11);
\draw[->>,out=135,in=180,distance={1*\mylen pt}] (v_1) to (v);          

%
\draw[-,dotted,thick,shorten <={0.2*\mylen pt},shorten >={0.2*\mylen pt}] (v_11) to (v_n1);

\path (v_n1) ++ ({0.25*\mylen pt},{-0.4*\mylen pt}) node[draw,thick,circle,minimum width=2.5pt,fill,inner sep=0pt,outer sep=2pt](v_n1_u_1){};
  \draw[->,thick] (v_n1) to (v_n1_u_1);
  \draw[->>,distance={0.25*\mylen pt},out=-30,in=20] (v_n1_u_1) to (u_1);

\draw[->,thick] (v_11) to ($(v_11) + ({-0.25*\mylen pt},{-0.25*\mylen pt})$);
\draw[->,thick] (v_11) to ($(v_11) + ({0.25*\mylen pt},{-0.25*\mylen pt})$);

%
\draw[<<-,out=190,in=90,distance={0.25*\mylen pt}] (v_11) to ($(v_11) + ({-0.4*\mylen pt},{-0.3*\mylen pt})$);

\draw[->>,out=160,in=270,distance={0.25*\mylen pt}] (v_n1) to ($(v_n1) + ({-0.35*\mylen pt},{0.4*\mylen pt})$);

%
\path (v_2) ++ ({0.35*\mylen pt},{0.05*\mylen pt}) node{$\averti{2}$};
\draw[->>,out=10,in=-40,distance={0.5*\mylen pt}] (v_12) to (v_2);
\draw[->>,out=45,in=0,distance={1*\mylen pt}] (v_2) to (v);

\path (v_2) ++ ({-0.25*\mylen pt},{-0.4*\mylen pt}) node[draw,thick,circle,minimum width=2.5pt,fill,inner sep=0pt,outer sep=2pt](v_2_v_12){};
  \draw[->,thick] (v_2) to (v_2_v_12);
  \draw[->>,distance={0.25*\mylen pt},out=210,in=160] (v_2_v_12) to (v_12);

%
\draw[-,dotted,thick,shorten <={0.2*\mylen pt},shorten >={0.2*\mylen pt}] (v_12) to (v_n2);
\draw[->,thick] (v_12) to ($(v_12) + ({-0.25*\mylen pt},{-0.25*\mylen pt})$);
\draw[->,thick] (v_12) to ($(v_12) + ({0.25*\mylen pt},{-0.25*\mylen pt})$);

\draw[<<-,out=-10,in=90,distance={0.25*\mylen pt}] (v_12) to ($(v_12) + ({0.4*\mylen pt},{-0.3*\mylen pt})$); 

\draw[->>,shorten <= {0.1 *\mylen pt},shorten >={0.1 *\mylen pt}] (v_2) to node[pos=0.5,sloped]{$\small {/}$} 
                                                                          node[pos=0.15,yshift={-0.15 * \mylen pt}]{$\scriptstyle \bodylab$}(v_1);

\path (u_1) ++ ({0*\mylen pt},{-0.225*\mylen pt}) node[forestgreen]{$\cverti{1}$};
   \path (u_1) ++ ({-0.4*\mylen pt},{0.4*\mylen pt}) node[draw,thick,circle,minimum width=2.5pt,fill,inner sep=0pt,outer sep=2pt] (u'_1){};
     \path(u'_1) ++ ({-0.25*\mylen pt},0pt) node{$\colorred{\cvertacci{1}}$};

\draw[->,out=180,in=270,distance={0.25*\mylen pt},red] (u_1) to node[pos=0.55,below]{$\scriptstyle\slbs$} (u'_1); 
\draw[->>,out=90,in=180,distance={0.25*\mylen pt}] (u'_1) to (v_n1);

\path (u_2) ++ ({0.075*\mylen pt},{0.175*\mylen pt}) node[forestgreen]{$\cverti{2}$};
   \path (u_2) ++ ({0.4*\mylen pt},{0.35*\mylen pt}) node[draw,thick,circle,minimum width=2.5pt,fill,inner sep=0pt,outer sep=2pt] (u'_2_1){};
     \path(u'_2_1) ++ ({0.25*\mylen pt},0pt) node{$\colorred{\cvertacci{2}}$};
   \path (u_2) ++ ({0*\mylen pt},{-0.55*\mylen pt}) node[draw,thick,circle,minimum width=2.5pt,fill,inner sep=0pt,outer sep=2pt] (u'_2_2){};
     \path(u'_2_2) ++ ({0*\mylen pt},{-0.3*\mylen pt}) node{$\colorred{\cvertacci{2}}$};
     \draw[->>,out=-25,in=-40,distance={0.5*\mylen pt}] (u'_2_2) to (u_2);
     
\path (v_n2) ++ ({-0.25*\mylen pt},{-0.4*\mylen pt}) node[draw,thick,circle,minimum width=2.5pt,fill,inner sep=0pt,outer sep=2pt](v_n2_u_2){};
  \draw[->,thick] (v_n2) to (v_n2_u_2);
  \draw[->>,distance={0.25*\mylen pt},out=210,in=160] (v_n2_u_2) to (u_2);

\draw[->>,out=20,in=270,distance={0.25*\mylen pt}] (v_n2) to ($(v_n2) + ({0.35*\mylen pt},{0.4*\mylen pt})$);

\draw[->,thick,red] (u_2) to node[left,pos=0.3,xshift={0.05*\mylen pt}]{$\scriptstyle \loopnsteplab{\aLname} $}(u'_2_2);

\draw[->,out=0,in=270,distance={0.25*\mylen pt},red] (u_2) to node[below,pos=0.7]{$\scriptstyle \bodylab$} (u'_2_1); 
\draw[->>,out=90,in=0,distance={0.25*\mylen pt}] (u'_2_1) to (v_n2);

\draw[-,thick,magenta,densely dashed] 
  (u_1) to node[pos=0.5](mid){} 
           node[pos=0.65](left){} (u_2);

\draw[-,thick,magenta,densely dashed,out=25,in=145,distance={1.25*\mylen pt}] 
  (u'_1) to node[pos=0.5,above,sloped,black]{use ind.\ hyp.} 
            node[pos=0.562](left_1){}
                                                             (u'_2_1);
  \draw[-implies,double equal sign distance,thick] (left) to (left_1);

\draw[-,thick,magenta,densely dashed,out=-10,in=190,distance={1.5*\mylen pt}] 
  (u'_1) to node[below,pos=0.65,sloped,black]{use ind.\ hyp.} 
            node[pos=0.75](left-2){} (u'_2_2);
  \draw[-implies,double equal sign distance,thick] (left) to (left-2);   



\end{tikzpicture}}
            \end{aligned}
            $
          \end{center}
          
        \item{\emph{Case 2:}}
          $\cverti{2}=\averti{2}$. We distinguish two cases.
          
        \item{\emph{Case 2.1:}}
          $\cverti{2}\loopnstepto{\alpha}\cverti{2}'$.
          Then either $\cverti{2}'=\cverti{2}$ or $\cverti{2} \descendsinlooptotc\cverti{2}'$. 
          Moreover,  $\sccof{\cverti{2}'}=\sccof{\cverti{2}}$, so by Lem.~\ref{lem:loop:relations},~\ref{it:descendsinloopto:scc:loopsbackto},
          $\cverti{2}'\loopsbacktortc \cverti{2}$, and hence
          $\cverti{2}'\loopsbacktortc \averti{2}$. 
          \vspace*{-.25mm}Thus we have obtained
          $
             \cvertacci{1} 
                 \loopsbacktortc 
               \averti{1}
                 \dloopsbackto 
               \avert
                 \convdloopsbackto 
               \averti{2}
                 \convloopsbacktortc 
               \cvertacci{2}
               \logand
             \lognot{( \averti{2} \redrtci{\bodylab} \averti{1} )} 
             $.
         Due to $\lbsminn{\cverti{1}'}<\lbsminn{\cverti{1}}$,
         we can apply the induction hypothesis again.
          
        \item{\emph{Case 2.2:}}
          $\cverti{2}\redi{\bodylab}\cverti{2}'$. Then $\neg(\averti{2} \redrtci{\bodylab} \averti{1})$ together with
          $\averti{2}=\cverti{2}\redi{\bodylab}\cverti{2}'$ and
          $\cverti{1}' \redrtci{\bodylab} \averti{1}$
          (because $\cverti{1}'\loopsbacktortc \averti{1}$) imply $\cverti{1}'\neq \cverti{2}'$.
          We distinguish two cases.
          
        \item{\emph{Case 2.2.1:}}
          $\cverti{2}'=\avert$. Then $\cverti{1}'\loopsbacktortc \averti{1}\dloopsbackto v=\cverti{2}'$, i.e., $\cverti{1}'\loopsbacktotc \cverti{2}'$, 
          so we are done, because \ref{cond:transf:II} holds 
          for $\bverti{1} = \cverti{2}'$ and $\bverti{2} = \cverti{1}'$.
          
        \item{\emph{Case 2.2.2:}}
          $\cverti{2}'\neq \avert$. By Lem.~\ref{lem:loop:relations},~\ref{it:descendsinloopto:scc:loopsbackto}, $\cverti{2}'\loopsbacktotc \avert$.
          Hence, $\cverti{2}'\loopsbacktortc \avertacci{2}\dloopsbackto \avert$ for some $\avertacci{2}$.
          Since $\averti{2}=\cverti{2}\redi{\bodylab}\cverti{2}'\loopsbacktortc \avertacci{2}$ and
          $\neg(\averti{2} \redrtci{\bodylab} \averti{1})$, it follows that
          $\neg(\avertacci{2} \redrtci{\bodylab} \averti{1})$.
          So
          $ \cvertacci{1} 
              \loopsbacktortc 
            \averti{1}
              \dloopsbackto 
            \avert
              \convdloopsbackto 
            \avertacci{2}
              \convloopsbacktortc 
            \cvertacci{2}
               \logand
            \lognot{( \averti{2} \redrtci{\bodylab} \avertacci{1} )}$.
         %
         Due to $\lbsminn{\cverti{1}'}<\lbsminn{\cverti{1}}$,
         we can apply the induction hypothesis again.
  \label{prf:prop:reduced:br:end}
   \end{description}
   This exhaustive case analysis concludes the proof.
\end{proof}

Now we define, for \LLEEwitnesses~$\acharthat$ of a \LLEEchart~$\achart$,
and for bisimilar vertices $\bverti{1},\bverti{2}$ in $\achart$,
in each of the three cases \ref{cond:transf:I}, \ref{cond:transf:II}, or \ref{cond:transf:III} of Prop.~\ref{prop:reduced:br}
a transformation~of~$\acharthat$ into an \entrybodylabeling\ of the \connectthroughchart{\bverti{1}}{\bverti{2}}~$\connthroughin{\achart}{\bverti{1}}{\bverti{2}}$
that can be shown to be a \LLEEwitness\ again. 
%
%
%
\begin{figure*}[tb]
  \input{figs/ex-transformations.tex}
\end{figure*}
We number the \emph{transformations} for \ref{cond:transf:I}, \ref{cond:transf:II}, and \ref{cond:transf:III} as \emph{I}, \emph{II}, and \emph{III}, respectively. 
Each transformation makes use of the \emph{connect-through construction
for \entrybodylabelings} as defined in Def.~\ref{def:connect-through}. 
Additionally, in each transformation an adaptation of labels of transitions is performed, 
to avoid violations of LLEE-witness properties. 
In transformations I and III the adaptation is performed \ul{before} connecting $w_1$ through to $w_2$,
and is needed to guarantee that layeredness is preserved; 
in transformation II it is performed right \ul{after} eliminating~$w_1$, and avoids the creation of body step cycles. 
The \emph{level adaptations} for the three~transformations~are:
\begin{enumerate}[label={\textit{L}$_{\text{\Roman{*}}}$},leftmargin=*,align=right,labelsep=1ex,itemsep=0.5ex]
  \item{}\label{labels:I}
   Let $m=\max\{\,\ll{\bLname}\,:\,\mbox{there is a path }\bverti{2}\,\scomprewrels{\sredrtc}{\sredi{\loopnsteplab{\bLname}}}\mbox{ in }\acharthat\,\}$. In loop-entry transitions $\cvert\;\sredi{\loopnsteplab{\aLname}}\;\avert$ for which there is a path $\avert\;\sredrtc\;\bverti{1}$ in $\achart$, replace $\aLname$ by an $\aLnameacc$ with $\ll{\aLnameacc}=\ll{\aLname}+m$.
        This increases the labels of \loopentrytransitions\ that descend to $\bverti{1}$ in $\acharthat$ to a higher level
        than the loop labels reachable from~$\bverti{2}$.\vspace{1mm}
  \item{}\label{labels:II}
  Since $\bverti{2} \loopsbacktotc \bverti{1}$, \vspace*{-.5mm}there exists a $\bverthati{2}$ with $\bverti{2} \loopsbacktortc \bverthati{2} \dloopsbackto \bverti{1}$.
  Let $\cLname$ be the maximum loop level among the loop-entries at $\bverti{1}$ in $\acharthat$.
  (Note that since $\bverti{2} \loopsbacktotc \bverti{1}$, there is at least one such transition.)
  Turn the body transitions from $\bverthati{2}$
  into loop-entry transitions with loop label $\cLname$.\vspace{1mm}
  \item{}\label{labels:III}
   Let $\cLname$ be a loop label of maximum level among the loop-entry transitions at $\avert$ in $\acharthat$. (Note that since $\bverti{1} \loopsbackto \avert$, there is at least one such transition.) Turn the loop labels of the loop-entry transitions from $\avert$ into $\cLname$.
\end{enumerate}
Each of these transformations ends with a \emph{clean-up step}: 
if the \loopentry\ transitions from a vertex with the same loop label 
no longer induce an infinite path (due to the removal of $\bverti{1}$),
then they are changed into body transitions.

\begin{example}\label{ex:transformations}%
  The \LLEEwitness\ on the left in Fig.~\ref{fig:ex:transformations} is 
  reduced in three transformation steps 
  to a \LLEEwitness\ of the chart $\chartof{\astexpi{0}}$ in Ex.~\ref{ex:chart:interpretation}. 
  Broken lines are between bisimilar vertices.
  In step one, a transformation I, the start state $v_0$ is connected through to the bisimilar vertex $v_0''$, whereby $v_0''$ becomes the start vertex;
  note that there is no path from $v_0''$ to $v_0$, and no vertex descends into a loop to $v_0$.
  In step two, a transformation II, $v_1$ is connected through to the bisimilar vertex $v_1'$; note that $v_1' \loopsbacktotc v_1$.
  In step three, a transformation III, the start vertex $v_0''$ is connected through to the bisimilar vertex $v_0'''$, whereby $v_0'''$ becomes the start vertex; 
  note that $v_0''\dloopsbackto v_2$ and $v_0'''\loopsbacktotc v_2$ and there is no body step path from $v_0'''$ to $v_0''$.
  By the loop level adaptation \ref{labels:III}, all loop entries from $\averti{2}$ get level~3.
  The final step is an isomorphic deformation. 
  Only the left and right charts depict actions. 
\end{example}

The following examples provide more illustrations of the transformations~II and III.
Similarly as Ex.~\ref{ex:trans-I} does so for transformation~I and \ref{cond:transf:I},
they also show that the conditions \ref{cond:transf:II} and \ref{cond:transf:III}  
mark rather sharp borders between whether, 
on a given \LLEEwitness, a connect-through operation is possible while preserving LLEE, or not.

\vspace*{-0.5ex}
\begin{example}\label{ex:trans-II}
  For the \LLEEwitness~$\acharthat$ below in the middle,\vspace*{-0.5mm}
  the chart $\connthroughin{\achart}{\bverti{2}}{\bverti{1}}$ on the left has no \LLEEwitness.
  
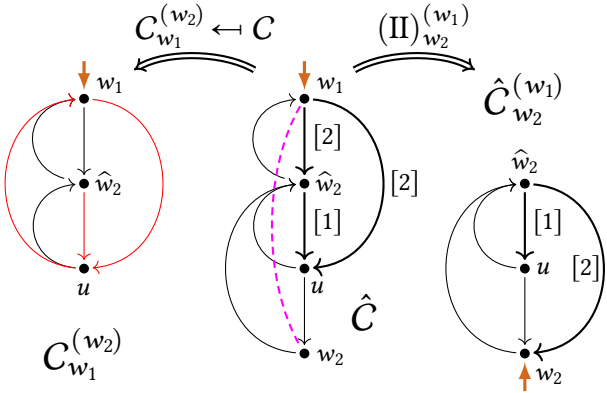
\vspace*{-1.5ex}
  \begin{center}
%
\begin{tikzpicture}
\matrix[anchor=center,row sep=1cm,column sep=1cm,every node/.style={draw,very thick,circle,minimum width=2.5pt,fill,inner sep=0pt,outer sep=2pt}] at (0,0) {
  \node(w1_Cw2w1){};                          &[0.8cm] & \node(w1_C){};     &[0.8cm]  & \node(label_Cw1w2)[draw=none,fill=none]{};                                             
  \\
  \node(w2hat_Cw2w1){};                       &        & \node(w2hat_C){};  &         & \node(w2hat_Cw1w2){};     
  \\
  \node(u_Cw2w1){};                           &        & \node(u_C){};      &         & \node(u_Cw1w2){};    
  \\
  \node(label_Cw2w1)[draw=none,fill=none]{};  &        & \node(w2_C){};     &         & \node(w2_Cw1w2){};
  \\
};

\draw[-implies,thick,double equal sign distance,bend right,distance={1*\mylen pt},
               shorten <= 0.65cm,shorten >= 0.65cm
               ] (w1_C) to node[above,pos=0.4425,yshift={0.05*\mylen pt}]{\scalebox{1.25}{$\connthroughin{\achart}{\bverti{2}}{\bverti{1}}\! \mapsfrom \achart$}} (w1_Cw2w1) ;

\draw[-implies,thick,double equal sign distance,bend left,distance={1*\mylen pt},
               shorten <= 0.65cm,shorten >= 0.65cm
               ] (w1_C) to node[above,pos=0.55,yshift={0.05*\mylen pt}]{\scalebox{1.25}{$(\text{\nf II})^{(\bverti{1})}_{\bverti{2}}$}} (label_Cw1w2);

\draw[<-,very thick,>=latex,chocolate](w1_Cw2w1) -- ++ (90:0.5cm);
\path (w1_Cw2w1) ++ (0.35cm,0.175cm) node{$\bverti{1}$};
\draw[->](w1_Cw2w1) to node[right,xshift=-0.025cm,yshift=0.05cm]{
} (w2hat_Cw2w1);
\draw[->,distance=1.25cm,out=-5,in=5,color=red] (w1_Cw2w1) to node[above,xshift=-0cm,yshift=0.15cm,
color=red,
pos=0.275]{
} (u_Cw2w1);
\path (w2hat_Cw2w1) ++ (0.325cm,0cm) node{$\bverthati{2}$};
\draw[->,color=red
](w2hat_Cw2w1) to node[right,xshift=-0.025cm,yshift=0.05cm]{
} (u_Cw2w1);
\draw[->,distance=0.75cm,out=165,in=185,shorten <= 4.5pt] (w2hat_Cw2w1) to (w1_Cw2w1);
\path (u_Cw2w1) ++ (0cm,-0.275cm) node{$\cvert$};
\draw[->,distance=0.75cm,out=175,in=185] (u_Cw2w1) to (w2hat_Cw2w1);
\draw[->,distance=1.25cm,out=175,in=185
,color=red
] (u_Cw2w1) to (w1_Cw2w1);
\path(label_Cw2w1) ++ (0cm,0cm) node{\scalebox{1.45}{$\connthroughin{\achart}{\bverti{2}}{\bverti{1}}$}};

\draw[<-,very thick,>=latex,chocolate](w1_C) -- ++ (90:0.5cm);
\path (w1_C) ++ (0.35cm,0.175cm) node{$\bverti{1}$};
\draw[->,thick](w1_C) to node[right,xshift=-0.025cm,yshift=0.05cm]{$\loopnsteplab{2}$} (w2hat_C);
\draw[->,thick,distance=1.25cm,out=-5,in=5] (w1_C) to node[right,xshift=-0.05cm,yshift=0cm,pos=0.5]{$\loopnsteplab{2}$} (u_C);
\path (w2hat_C) ++ (0.325cm,0cm) node{$\bverthati{2}$};
\draw[->,thick](w2hat_C) to node[right,xshift=-0.025cm,yshift=0.05cm]{$\loopnsteplab{1}$} (u_C);
\draw[->,distance=0.75cm,out=165,in=185,shorten <= 4.5pt] (w2hat_C) to (w1_C);
\path (u_C) ++ (0.175cm,-0.25cm) node{$\cvert$};
\draw[->](u_C) to (w2_C);
\draw[->,distance=0.75cm,out=175,in=185] (u_C) to (w2hat_C);
\path (w2_C) ++ (0.35cm,-0.05cm) node{$\bverti{2}$};
\draw[->,distance=1.25cm,out=175,in=185] (w2_C) to (w2hat_C);
\path (w2_C) ++ (0.8cm,0.55cm) node(label_C){\scalebox{1.45}{$\acharthat$}};

\draw[magenta,thick,densely dashed,bend right,distance=1cm,looseness=1] (w1_C) to (w2_C);

%
%

\path (w2hat_Cw1w2) ++ (0cm,0.3cm) node{$\bverthati{2}$};
\draw[->,thick](w2hat_Cw1w2) to node[right,xshift=-0.025cm,yshift=0.05cm]{$\loopnsteplab{1}$} (u_Cw1w2);
\draw[->,thick,distance=1.25cm,out=-5,in=5](w2hat_Cw1w2) to node[left,xshift=0.1cm,yshift=0cm,pos=0.5]{$\loopnsteplab{2}$} (w2_Cw1w2);
\path (u_Cw1w2) ++ (0.25cm,-0cm) node{$\cvert$};
\draw[->] (u_Cw1w2) to (w2_Cw1w2);
\draw[->,distance=0.75cm,out=175,in=185] (u_Cw1w2) to (w2hat_Cw1w2);
\draw[<-,very thick,>=latex,chocolate](w2_Cw1w2) -- ++ (270:0.5cm);
\path (w2_Cw1w2) ++ (0.325cm,-0.3cm) node{$\bverti{2}$};
\draw[->,distance=1.25cm,out=175,in=185] (w2_Cw1w2) to (w2hat_Cw1w2);
\path(label_Cw1w2) ++ (0cm,0cm) node{\scalebox{1.45}{$\connthroughin{\acharthat}{\bverti{1}}{\bverti{2}}$}};

\end{tikzpicture}
  %
  %
  %
  \end{center}
  \vspace*{0ex}
  It does not satisfy LEE:
  it has no loop subchart, since from each of its three vertices an infinite path starts that does not return to this vertex; 
  from $\bverthati{2}$ this path, drawn in red, cycles between $\cvert$ and $\bverti{1}$.
  Transformation II applied to the pair $\bverti{1},\bverti{2}$ (instead of $\bverti{2},\bverti{1}$) in $\acharthat$ yields
  the \entrybodylabeling\ $\connthroughin{\acharthat}{\bverti{1}}{\bverti{2}}$ 
  where $\bverthati{2}\,\sredi{\bodylab}\,\bverti{2}$ is turned into $\bverthati{2}\,\sredi{\loopnsteplab{2}}\,\bverti{2}$. 
  As the pair $\bverti{1},\bverti{2}$ satisfies \ref{cond:transf:II}, the proof of Prop.\ \ref{prop:LEEshape:preserve:conds} 
  ensures that this labeling, drawn on the right, is a \LLEEwitness.
\end{example}
\vspace*{-0.5ex}

\begin{example}\label{ex:trans-III}
  In the \LLEEwitness\ $\acharthat$ below in the middle,
  $w_1,w_2\loopsbacktotc v$ and there is no body step path from $w_2$ to $w_1$, 
  but \ref{cond:transf:III} does not hold for the pair $w_1,w_2$ due to $\lognot(w_1\dloopsbackto v)$.\vspace*{-1mm}
  The chart $\connthroughin{\achart}{\bverti{1}}{\bverti{2}}$ on the left has no LLEE-witness.
  It does not satisfy LEE:
    the downwards \loopentry\ transition from $\bverthati{2}$ can be eliminated,
    and then two more arising \loopentry\ transitions from $\avert$;
  the remaining chart of solid arrows has no further loop subchart,
  because from each of its vertices an infinite path starts that does not return to this vertex.
  
  In $\acharthat$, loop-entry transitions from $v$ have the same loop label, so the preprocessing step of transformation III is void.
  The bisimilar pair $w_1,w_2$ progresses to the bisimilar pair $\bverthati{1},\bverthati{2}$ in $\acharthat$, for which \ref{cond:transf:III} holds because $\bverthati{1} \dloopsbackto \avert
  \convloopsbackto \bverthati{2}$ and $\lognot(\bverthati{2}\,\sredrtci{\bodylab}\,\bverthati{1})$.
  \vspace*{-.5mm}Transformation III applied to this pair yields
  the labeling $\protect\connthroughin{\acharthat}{\bverthati{1}}{\bverthati{2}}$ on the right. In the proof of
  Prop.\ \ref{prop:LEEshape:preserve:conds} it is argued that this is guaranteed to be a \LLEEwitness.
  The remaining two bisimilar pairs can be eliminated by one or by two further applications of transformation III. 
  
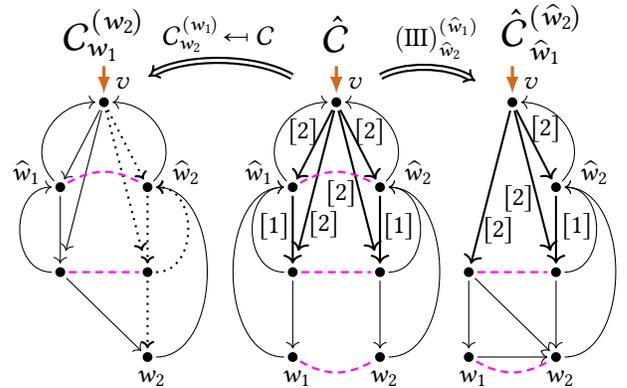
\vspace*{0ex}
  \begin{center}
  %
%
\begin{tikzpicture}
\matrix[anchor=center,row sep=1cm,column sep=0.45cm,every node/.style={draw,very thick,circle,minimum width=2.5pt,fill,inner sep=0pt,outer sep=2pt}] at (0,0) {
                        &  \node(v_Cw1w2){}; &                       &[0.9cm] &                   & \node(v_C){}; &                   &[0.15cm] &                           & \node(v_Chatw1hatw2){}; &
  \\
  \node(hatw1_Cw1w2){}; &                    & \node(hatw2_Cw1w2){}; &        & \node(hatw1_C){}; &               & \node(hatw2_C){}; &       &                           &                         & \node(hatw2_Chatw1hatw2){};
  \\
  \node(u1_Cw1w2){};    &                    & \node(u2_Cw1w2){};    &        & \node(u1_C){};    &               & \node(u2_C){};    &       & \node(u1_Chatw1hatw2){};  &                         & \node(u2_Chatw1hatw2){};
  \\
  \node[draw=none,fill=none]{};
                        &                    & \node(w2_Cw1w2){};    &        & \node(w1_C){};    &               & \node(w2_C){};    &       & \node(w1_Chatw1hatw2){};  &                         & \node(w2_Chatw1hatw2){};
  \\
};
\calcLength(hatw1_C,u1_C){mylen}; 
%
\draw[<-,very thick,>=latex,chocolate](v_Cw1w2) -- ++ (90:{0.425*\mylen pt}); 
\path (v_Cw1w2) ++ ({0.225*\mylen pt},{0.225*\mylen pt}) node{$\avert$};
\path (v_Cw1w2) ++ ({0*\mylen pt},{0.8*\mylen pt}) node{\scalebox{1.45}{$\connthroughin{\achart}{\bverti{2}}{\bverti{1}}$}};
\path (hatw1_Cw1w2) ++ ({-0.4*\mylen pt},{0.175*\mylen pt}) node{$\bverthati{1}$};
\draw[->] (v_Cw1w2) to node[above,pos=0.7,xshift={-0.05*\mylen pt},yshift={0.05*\mylen pt}
]{
} (hatw1_Cw1w2);
\draw[->,shorten >= 5pt
] (v_Cw1w2) to node[right,pos=0.58,xshift={-0.025*\mylen pt},yshift={0.00*\mylen pt}]{
} (u1_Cw1w2);
\draw[-{>[length=1mm,width=1.8mm]},thick,dotted] (v_Cw1w2) to node[above,pos=0.7,xshift={0.05*\mylen pt},yshift={0.05*\mylen pt}]{
} (hatw2_Cw1w2);
\draw[-{>[length=1mm,width=1.8mm]},thick,dotted,shorten >= 5pt] (v_Cw1w2) to node[left,pos=0.58,xshift={0.025*\mylen pt},yshift={0.00*\mylen pt}]{
} (u2_Cw1w2);
\draw[->
] (hatw1_Cw1w2) to node[left,pos=0.45,xshift={0.05*\mylen pt}]{
} (u1_Cw1w2);
\draw[->,distance={0.5*\mylen pt},out=135,in=180] (hatw1_Cw1w2) to (v_Cw1w2);
\draw[->] (u1_Cw1w2) to (w2_Cw1w2);
\draw[->,distance={0.5*\mylen pt},out=180,in=180
] (u1_Cw1w2) to (hatw1_Cw1w2);
\path (hatw2_Cw1w2) ++ ({0.45*\mylen pt},{0.175*\mylen pt}) node{$\bverthati{2}$};
\draw[-{>[length=1mm,width=1.8mm]},thick,dotted] (hatw2_Cw1w2) to node[right,pos=0.45,xshift={-0.075*\mylen pt}]{
} (u2_Cw1w2);
\draw[->,distance={0.5*\mylen pt},out=45,in=0
] (hatw2_Cw1w2) to (v_Cw1w2);
\draw[-{>[length=1mm,width=1.8mm]},thick,dotted] (u2_Cw1w2) to (w2_Cw1w2);
\draw[-{>[length=1mm,width=1.8mm]},thick,dotted,distance={0.5*\mylen pt},out=0,in=0] (u2_Cw1w2) to (hatw2_Cw1w2);
\path (w2_Cw1w2) ++ ({0.05*\mylen pt},{-0.25*\mylen pt}) node{$\bverti{2}$};
\draw[->,distance={0.8*\mylen pt},out=0,in=0
] (w2_Cw1w2) to (hatw2_Cw1w2);

\draw[magenta,thick,densely dashed,bend left,distance={0.35*\mylen pt},looseness=1] (hatw1_Cw1w2) to (hatw2_Cw1w2);
\draw[magenta,thick,densely dashed] (u1_Cw1w2) to (u2_Cw1w2);

\draw[<-,very thick,>=latex,chocolate](v_C) -- ++ (90:{0.425*\mylen pt});   
\path (v_C) ++ ({0.225*\mylen pt},{0.225*\mylen pt}) node{$\avert$};
\path (v_C) ++ ({0*\mylen pt},{0.8*\mylen pt}) node{\scalebox{1.45}{${\acharthat}$}};
\path (hatw1_C) ++ ({-0.4*\mylen pt},{0.175*\mylen pt}) node{$\bverthati{1}$};
\draw[->,thick] (v_C) to node[above,pos=0.7,xshift={-0.05*\mylen pt},yshift={0.05*\mylen pt}]{$\loopnsteplab{2}$} (hatw1_C);
\draw[->,thick,shorten >= 5pt] (v_C) to node[right,pos=0.74,xshift={-0.075*\mylen pt},yshift={0.00*\mylen pt}]{$\loopnsteplab{2}$} (u1_C);
\draw[->,thick] (v_C) to node[above,pos=0.7,xshift={0.05*\mylen pt},yshift={0.05*\mylen pt}]{$\loopnsteplab{2}$}  (hatw2_C);
\draw[->,thick,shorten >= 5pt] (v_C) to node[left,pos=0.54,xshift={0.075*\mylen pt},yshift={0.00*\mylen pt}]{$\loopnsteplab{2}$} (u2_C);
\draw[->,thick] (hatw1_C) to node[left,pos=0.45,xshift={0.075*\mylen pt}]{$\loopnsteplab{1}$} (u1_C);
\draw[->,distance={0.5*\mylen pt},out=135,in=180] (hatw1_C) to (v_C);
\draw[->] (u1_C) to (w1_C);
\draw[->,distance={0.5*\mylen pt},out=180,in=180] (u1_C) to (hatw1_C);
\path (w1_C) ++ ({0.05*\mylen pt},{-0.25*\mylen pt}) node{$\bverti{1}$};
\draw[->,distance={0.8*\mylen pt},out=180,in=180] (w1_C) to (hatw1_C);

\path (hatw2_C) ++ ({0.45*\mylen pt},{0.175*\mylen pt}) node{$\bverthati{2}$};
\draw[->,thick] (hatw2_C) to node[right,pos=0.45,xshift={-0.075*\mylen pt}]{$\loopnsteplab{1}$} (u2_C);
\draw[->,distance={0.5*\mylen pt},out=45,in=0] (hatw2_C) to (v_C);
\draw[->] (u2_C) to (w2_C);
\draw[->,distance={0.5*\mylen pt},out=0,in=0] (u2_C) to (hatw2_C);
\path (w2_C) ++ ({0.05*\mylen pt},{-0.25*\mylen pt}) node{$\bverti{2}$};
\draw[->,distance={0.8*\mylen pt},out=0,in=0] (w2_C) to (hatw2_C);

\draw[magenta,thick,densely dashed,bend left,distance={0.35*\mylen pt},looseness=1] (hatw1_C) to (hatw2_C);
\draw[magenta,thick,densely dashed] (u1_C) to (u2_C);
\draw[magenta,thick,densely dashed,bend right,distance={0.35*\mylen pt},looseness=1] (w1_C) to (w2_C);

\draw[<-,very thick,>=latex,chocolate](v_Chatw1hatw2) -- ++ (90:{0.425*\mylen pt}); 
\path (v_Chatw1hatw2) ++ ({0.225*\mylen pt},{0.225*\mylen pt}) node{$\avert$};
\path (v_Chatw1hatw2) ++ ({0.35*\mylen pt},{0.8*\mylen pt}) node{\scalebox{1.45}{$\connthroughin{\acharthat}{\bverthati{2}}{\bverthati{1}}$}};
\draw[->,thick] (v_Chatw1hatw2) to node[above,pos=0.7,xshift={0.05*\mylen pt},yshift={0.075*\mylen pt}]{$\loopnsteplab{2}$} (hatw2_Chatw1hatw2);
\draw[->,thick] (v_Chatw1hatw2) to node[right,pos=0.78,xshift={-0.075*\mylen pt},yshift={0.00*\mylen pt}]{$\loopnsteplab{2}$} (u1_Chatw1hatw2);
\draw[->,thick,shorten >= 5pt] (v_Chatw1hatw2) to node[left,pos=0.575,xshift={0.075*\mylen pt},yshift={0.00*\mylen pt}]{$\loopnsteplab{2}$} (u2_Chatw1hatw2);
\draw[->] (u1_Chatw1hatw2) to (w1_Chatw1hatw2);
\draw[->] (u1_Chatw1hatw2) to (w2_Chatw1hatw2);
\path (w1_Chatw1hatw2) ++ ({0.05*\mylen pt},{-0.25*\mylen pt}) node{$\bverti{1}$};
\draw[->] (w1_Chatw1hatw2) to (w2_Chatw1hatw2);
\path (hatw2_Chatw1hatw2) ++ ({0.45*\mylen pt},{0.175*\mylen pt}) node{$\bverthati{2}$};
\draw[->,thick] (hatw2_Chatw1hatw2) to node[right,pos=0.45,xshift={-0.05*\mylen pt}]{$\loopnsteplab{1}$} (u2_Chatw1hatw2);
\draw[->,distance={0.5*\mylen pt},out=45,in=0] (hatw2_Chatw1hatw2) to (v_Chatw1hatw2);
\draw[->] (u2_Chatw1hatw2) to (w2_Chatw1hatw2);
\draw[->,distance={0.5*\mylen pt},out=0,in=0] (u2_Chatw1hatw2) to (hatw2_Chatw1hatw2);
\path (w2_Chatw1hatw2) ++ ({0.05*\mylen pt},{-0.25*\mylen pt}) node{$\bverti{2}$};
\draw[->,distance={0.8*\mylen pt},out=0,in=0] (w2_Chatw1hatw2) to (hatw2_Chatw1hatw2);

\draw[magenta,thick,densely dashed] (u1_Chatw1hatw2) to (u2_Chatw1hatw2);
\draw[magenta,thick,densely dashed,bend right,distance={0.35*\mylen pt},looseness=1] (w1_Chatw1hatw2) to (w2_Chatw1hatw2);

\draw[-implies,thick,double equal sign distance, bend right,distance={1*\mylen pt},
               shorten <= {0.5*\mylen pt},shorten >= {0.5*\mylen pt},
               ] (v_C) to node[above,pos=0.505,yshift={0.05*\mylen pt}] {\scalebox{1}{$\connthroughin{\achart}{\bverti{1}}{\bverti{2}}\! \mapsfrom \achart$}} 
                                                                         (v_Cw1w2) ;

\draw[-implies,thick,double equal sign distance, bend left,distance={0.8*\mylen pt},
               shorten <= {0.5*\mylen pt},shorten >= {0.4*\mylen pt},
               ] (v_C) to node[above,pos=0.58,yshift={0.05*\mylen pt}]{\scalebox{1}{$(\text{\nf III})^{(\bverthati{1})}_{\bverthati{2}}$}} (v_Chatw1hatw2) ;
\end{tikzpicture}\label{fig:ex:progr:transf:III}%
  %
  %
  %

  \end{center}
  \vspace*{0ex}
\end{example}

\begin{proposition}\label{prop:LEEshape:preserve:conds}
  Let $\achart$ be a LLEE-chart.
  If a pair $\pair{\bverti{1}}{\bverti{2}}$ of vertices 
  satisfies \ref{cond:transf:I}$\!$, \ref{cond:transf:II}$\!$, or \ref{cond:transf:III}$\!$\vspace*{-0.5mm}
  with respect to a \LLEEwitness\ of $\achart$,\vspace*{-0.5mm} 
  then\/ $\connthroughin{\achart}{\bverti{1}}{\bverti{2}}$ 
  is a \mbox{\LLEEchart}.
\end{proposition}

\begin{proof}
  Let $\acharthat$ be a \LLEEwitness. 
  For vertices $\bverti{1}$, $\bverti{2}$ such that \ref{cond:transf:I}, \ref{cond:transf:II}, or \ref{cond:transf:III} holds,
  transformation~I, II, or III, respectively, produces an \entrybodylabeling~$\connthroughin{\acharthat}{\bverti{1}}{\bverti{2}}$.
  We prove for transformation~I that this is a \LLEEwitness,
  and refer to the appendix with regard to transformations II, and III. 
  
  We first argue\label{alleviation:prf:prop:LEEshape:preserve:conds} 
  it suffices to show that each of the transformations
  produces, before the final clean-up step, a labeling that satisfies the \LLEEwitness\ conditions, 
  except possible violations of loop property~\ref{loop:1} in \ref{LLEEw:2}\ref{LLEEw:2a}.
  Such violations
  can be
  removed from a \looplabeling\  while preserving the other \LLEEwitness\ conditions.
  To show this, suppose \ref{loop:1} is violated in some $\indsubchartinat{\acharthat}{\cvert,\aLname}$.
  Then $\cvert \,\sredi{\loopnsteplab{\aLname}}$ but $\lognot{(\cvert \comprewrels{\sredi{\loopnsteplab{\aLname}}}{\sredrtci{\bodylab}} \cvert)}$.
  Let $\acharthati{1}$ be the result of removing this violation by changing the $\aLname$\nb-\loopentry\ transitions from $\cvert$ into body transitions.
  No new violation of \ref{loop:1} is introduced~in~$\acharthati{1}$.
  \ref{LLEEw:1} and \ref{LLEEw:2}\ref{LLEEw:2a}, \ref{loop:2}, are preserved in $\acharthati{1}$
  because an introduced infinite body step path in $\acharthati{1}$ would be a body step cycle that stems from  
  a path $\cvert \redi{\loopnsteplab{\aLname}} \cvertacc \redrtci{\bodylab} \cvert$ in $\acharthat$. 
  \ref{LLEEw:2}\ref{LLEEw:2b} 
                              might only be violated by a path
  $ \bvert \comprewrels{\sredtavoidsvi{\bvert}{\loopnsteplab{\bLname}}}{\sredtavoidsvrtci{\bvert}{\bodylab}} \cvert \redtavoidsvi{\bvert,\cvert}{\bodylab} \cvertacc
           \comprewrels{\sredtavoidsvrtci{\bvert,\cvert}{\bodylab}}{\sredi{\loopnsteplab{\cLname}}}$ with $\ll{\bLname}\leq\ll{\cLname}$\vspace*{-1mm} 
  in $\acharthati{1}$ where $\cvert \redi{\bodylab} \cvertacc$ stems from
  $\cvert \redi{\loopnsteplab{\aLname}} \cvertacc$ in $\acharthat$;
  then \mbox{$\ll{\bLname} > \ll{\aLname} > \ll{\cLname}$} by layeredness of $\acharthat$; so  \ref{LLEEw:2}\ref{LLEEw:2b} is preserved.
  Analogously we find that also \ref{LLEEw:2}\ref{LLEEw:2a}, \ref{loop:3} is preserved, 
  because $\tick$ is never in $\indsubchartinat{\acharthat}{\cvert,\aLname}$.
  
  To show the correctness of transformation~I, consider vertices~$\bverti{1}$ and $\bverti{2}$ 
                                                                                               with \ref{cond:transf:I}. 
  We show that the result $\connthroughin{\acharthat}{\bverti{1}}{\bverti{2}}$ of transformation~I before the clean-up step
  satisfies the LLEE-witness properties, except for possible violations~of~\ref{loop:1}.
  
  To verify \ref{LLEEw:1} and part \ref{loop:2} of \ref{LLEEw:2}\ref{LLEEw:2a},
  it suffices to show that
  \vspace*{-0.35mm}
  $\connthroughin{\acharthat}{\bverti{1}}{\bverti{2}}$ does not contain body step cycles. 
  The original \looplabeling~$\acharthat$ is a \LLEEwitness, so it does not contain body step cycles.
  Since the level adaptation step does not turn \loopentry\ steps into body steps,
  body step cycles could only arise in the step connecting $\bverti{1}$ through to $\bverti{2}$.
  Suppose such a body step cycle arises. 
  Then there must be a transition $\cvert\redi{\bodylab}\bverti{1}$ in $\acharthat$
  (which is redirected to $\bverti{2}$ in $\connthroughin{\acharthat}{\bverti{1}}{\bverti{2}}$)
  and a path $\bverti{2}\redrtci{\bodylab}\cvert$ in $\acharthat$.
  But then $\bverti{2}\redrtci{\bodylab}\cvert\redi{\bodylab}\bverti{1}$ in $\achart$,
  which contradicts \ref{cond:transf:I} that there is no path from $\bverti{2}$ to $\bverti{1}$.
  Hence \ref{LLEEw:1} and part \ref{loop:2} of \ref{LLEEw:2}\ref{LLEEw:2a} hold for~$\connthroughin{\acharthat}{\bverti{1}}{\bverti{2}}$.
  
  Now we verify part~\ref{loop:3} of \ref{LLEEw:2}\ref{LLEEw:2a} in $\connthroughin{\acharthat}{\bverti{1}}{\bverti{2}}$.
  Consider a path 
  $\cvert \,\scomprewrels{\sredtavoidsvi{\cvert}{\loopnsteplab{\aLname}}}{\sredtavoidsvrtci{\cvert}{\bodylab}}\,~\bverti{1}$
  in $\acharthat$. 
  Then $\cvert\neq\bverti{1}$, and $\cvert\descendsinloopto\bverti{1}$. 
  It suffices to show that then $\lognot(\bverti{2}\redtc\surd)$ in $\achart$.
  But this is guaranteed, because otherwise
  $\bverti{2}$ were normed, and due to $\cvert\descendsinloopto\bverti{1}$
  we would have a contradiction with condition~\ref{cond:transf:I}.

  Finally we show that \ref{LLEEw:2}\ref{LLEEw:2b} is preserved in $\connthroughin{\acharthat}{\bverti{1}}{\bverti{2}}$
  by both the level adaptation and the connect-through step.
  First, since in the level adaptation step all adapted loop labels are increased with the same value $m$, 
  a violation of \ref{LLEEw:2}\ref{LLEEw:2b} would arise 
  by a path $\cvert \comprewrels{\sredi{\loopnsteplab{\aLname}}}{\comprewrels{\sredrtci{\bodylab}}{\sredi{\loopnsteplab{\bLname}}}} \avert$ in $\acharthat$ 
  where loop label $\bLname$ is increased while $\aLname$ is not. But such a path cannot exist. 
  Since $\bLname$ is increased, there is a path $\avert \redrtc \bverti{1}$ in $\achart$. 
  But then there is a path $\cvert \comprewrels{\sredi{\loopnsteplab{\aLname}}}{\sredtc} \avert \redrtc \bverti{1}$ in $\acharthat$, 
  which implies that also $\aLname$ is increased in the level adaptation step.
  Second, a violation of \ref{LLEEw:2}\ref{LLEEw:2b} in the connect-through step
  \vspace*{-.5mm}would arise from paths $\cvert \comprewrels{\redi{\loopnsteplab{\aLname}}}{\sredrtci{\bodylab}} \bverti{1}$ and 
                         $\bverti{2} \:\scomprewrels{\sredrtci{\bodylab}}{\sredi{\loopnsteplab{\bLname}}}\,$ in $\acharthatacc$ with $\ll{\aLname}\leq\ll{\bLname}$. 
  However, in view of the path $\cvert \comprewrels{\sredi{\loopnsteplab{\aLname}}}{\sredrtc} \bverti{1}$, 
  the loop label $\aLname$ was increased with $m$ in the level adaptation step . 
  On the other hand, in view of \ref{cond:transf:I} that there is no path from $\bverti{2}$ to $\bverti{1}$ in $\achart$, 
  $\bverti{1}$ is unreachable at the end of the path $\bverti{2} \,\scomprewrels{\sredrtc}{\sredi{\loopnsteplab{\bLname}}}$. 
  \vspace*{-.5mm}Hence this loop label $\bLname$ was not increased in the level adaptation step.
  So it is guaranteed that for such a pair of paths in $\connthroughin{\acharthat}{\bverti{1}}{\bverti{2}}$ always $\ll{\aLname}>\ll{\bLname}$.
  
  We conclude that the result of transformation~I is again a \LLEEwitness.
\end{proof}

\begin{theorem}\label{thm:LEEshaped:collapse}
  The bisimulation collapse of a \LLEEchart\ is again a \LLEEchart.
\end{theorem}

\begin{proof}
  Given a \LLEEchart~$\achart$, repeat the following step:
  based on a \LLEEwitness\ pick, by Prop.~\ref{prop:reduced:br}, 
  bisimilar vertices $\bverti{1}$ and $\bverti{2}$ with \ref{cond:transf:I}, \ref{cond:transf:II}, or \ref{cond:transf:III},
  and then connect $\bverti{1}$ through to $\bverti{2}$, obtaining by Prop.~\ref{prop:LEEshape:preserve:conds} a \LLEEchart\
  bisimilar to $\achart$, due to Lem.~\ref{lem:connthroughchart:bisim}.  
  Hence the bisimulation collapse of $\achart$, which is reached eventually, is a \LLEEchart.
\end{proof}

We mention that by using a refinement of the interpretation TSS
(that avoids creating concatenations~$\stexpprod{\astexpi{1}}{\astexpi{2}}$ where $\astexpi{1}$ is not normed,
 in favor of using just $\astexpi{1}$)
and a refinement of the extraction procedure (that ensures an eager use of the right distributive law~(B4) of $\sstexpprod$ over $\sstexpsum$)
this theorem can be strengthened: the bisimulation collapse of a \LLEEchart\ is the chart interpretation of some star expression.
which then is a \LLEEchart\ by Prop.~\ref{prop:id:is:sol:chart:interpretation}.  
This can be proved by showing that, on collapsed \LLEEcharts, (refined) chart interpretation is the converse of (refined) solution extraction.


\begin{corollary}\label{cor:expressible}
  If a chart is expressible by a star expression modulo bisimilarity, then its collapse is a \LLEEchart.  
\end{corollary}

The converse statement holds as well.
But this corollary
does not hold for star expressions with $\stexpone$ and unary star.
For example, with respect to the TSS for the process interpretation of star expressions from this class, see e.g.\ \cite{baet:corr:grab:2007},
the expression $\astexpi{1} \defdby \stexpprod{(\stexpprod{(\stexpprod{(\stexpone}{\stexpit{\aact})}}{(\stexpprod{\bact}{\stexpit{\cact}})})}{\astexp}$
with $\astexp \defdby \stexpit{(\stexpprod{\stexpit{\aact}}{(\stexpprod{\bact}{\stexpit{\cact}})})}$ 
has the following interpretation,
  where $\astexpi{2} \defdby \stexpprod{(\stexpprod{\stexpone}{\stexpit{\cact}})}{\astexp}\,$:\vspace{-3ex}         
\begin{center}
  \begin{tikzpicture}

\matrix[anchor=north,row sep=1cm,column sep=1cm,every node/.style={draw,very thick,circle,minimum width=2.5pt,fill,inner sep=0pt,outer sep=2pt}] at (6.3,0) {
  \node(e-1){};  &  \node[draw=none,fill=none](dummy){};  
                                   &[-0.3cm] \node(e-2){}; 
  \\
};
\calcLength(e-1,dummy){mylen}

\draw[<-,very thick,>=latex,chocolate,shorten <=2pt](e-1) -- ++ (90:{0.5*\mylen pt});
\draw[thick] (e-1) circle (0.12cm);
\path(e-1) ++ ({-0.35*\mylen pt},{-0.035*\mylen pt}) node{$\astexpi{1}$}; 
\draw[->,shorten <=2pt,shorten >=2pt,out=220,in=140,distance={1.25*\mylen pt}] (e-1) to node[left]{$\aact$} (e-1);
\draw[->,shorten <=2pt,shorten >=2pt,out=-20,in=200,distance={0.6*\mylen pt}] (e-1) to node[below]{$\bact$} (e-2);

\draw[thick] (e-2) circle (0.12cm);
\path(e-2) ++ ({-0.125*\mylen pt},{0.3*\mylen pt}) node{$\astexpi{2}$};  
\draw[->,shorten <=2pt,shorten >=2pt,out=-20,in=60,distance={1.25*\mylen pt}] (e-2) to node[right,xshift={-0.025*\mylen pt}]{$\bact$} (e-2);
\draw[->,shorten <=2pt,shorten >=2pt,out=-60,in=20,distance={1.25*\mylen pt}] (e-2) to node[right,xshift={-0.05*\mylen pt}]{$\cact$} (e-2);
\draw[->,shorten <=2pt,shorten >=2pt,out=160,in=20,distance={0.6*\mylen pt}] (e-2) to node[above]{$\aact$} (e-1);

\end{tikzpicture}
  
\end{center}\vspace{-4ex}  
This is a chart in the extended sense in which immediate termination is permitted
at arbitrary vertices. It is a bisimulation collapse that does not satisfy \LEE,
taking into account that in the definition of `loop' for charts in the extended sense
\ref{loop:3} needs to be changed to exclude immediate termination for vertices
in a loop chart other than~the~start~vertex.


%

\section{The completeness result, and conclusion}%
  \label{conclusion}         

That bisimulation collapse preserves LLEE was the last building block in the proof of the desired completeness result.

\begin{theorem}\label{thm:BBP:sound:complete}
  The proof system \BBP\ is complete with respect to
  the bisimulation semantics of star expressions, that is, 
  with respect to bisimilarity of charts that interpret
  star expressions without $\stexpone$ and with binary Kleene star $\stexpbit{}{}$. 
\end{theorem}   

\begin{proof}
  The proof steps were already explained in Sect.~\ref{compl:proof}.
\end{proof}

\begin{example}\label{ex:completeness1}
  The bisimilar LLEE-charts~$\acharti{1}$ and $\acharti{2}$ 
  in Ex.~\ref{ex:salomaa} have
  $\stexpbit{(\stexpsum{\stexpprod{a}{(\stexpsum{a}{b})}}{b})}{\stexpzero}$ and $\stexpbit{(\stexpsum{\stexpprod{b}{(\stexpsum{a}{b})}}{a})}{\stexpzero}$
  as their principal solutions. 
  Their bisimulation collapse $\acharti{0}$ has principal solution $\stexpbit{(\stexpsum{a}{b})}{\stexpzero}$.
  Then 
  $\stexpbit{(\stexpsum{\stexpprod{a}{(\stexpsum{a}{b})}}{b})}{\stexpzero} \BBPeq$ 
  $\stexpbit{(\stexpsum{a}{b})}{\stexpzero}\BBPeq\stexpbit{(\stexpsum{\stexpprod{b}{(\stexpsum{a}{b})}}{a})}{\stexpzero}$
  by Prop.~\ref{prop:transf:sol:via:funbisim}, Prop.~\ref{prop:extrsol:vs:solution}.
\end{example}

\begin{example}\label{ex:completeness2}
  Revisiting the star expressions $\astexpi{1},\astexpi{2}$ in Ex.~\ref{ex:chart:interpretation}
  with bisimilar chart interpretations $\chartof{\astexpi{1}}$ and $\chartof{\astexpi{2}}$,
  we can apply our proof in order to show that $\astexpi{1} \BBPeq \astexpi{2}$. 
  $\chartof{\astexpi{1}}$ and $\chartof{\astexpi{2}}$ have provable solutions with principal values $e_1$ and $e_2$
  by Prop.~\ref{prop:id:is:sol:chart:interpretation}. 
  As $\chartof{\astexpi{1}}$ and $\chartof{\astexpi{2}}$ are \LLEEcharts\ by Prop.~\ref{prop:lbl:chart:translation:is:LLEEw}
  with \LLEEwitnesses~$\charthatof{\astexpi{1}}$ and $\charthatof{\astexpi{2}}$,  
  their bisimulation collapse~$\achart$ is a \LLEEchart\ by Thm.~\ref{thm:LEEshaped:collapse}. 
  We take here the more familiar~$\acharthat$,
  but could also take the one obtained in Fig.~\ref{fig:ex:transformations}.
  We saw in Fig.~\ref{fig:ex:extr::sol:vs:extrsol} that $\acharthat$ has a provable solution with principal value
      $\extrsolof{\acharthat}{\averti{0}}=
          \stexpprod{\aact}{(\stexpbit{(\stexpsum{\stexpprod{\cact}{\aact}}
  {\stexpprod{\aact}{(\stexpsum{\bact}{\stexpprod{\bact}{\aact}})}})}{\stexpzero})}$.
  Then by Prop.~\ref{prop:transf:sol:via:funbisim} and Prop.~\ref{prop:extrsol:vs:solution} it follows that
  \mbox{$e_1\BBPeq\extrsolof{\acharthat}{\averti{0}}\BBPeq e_2$}.\vspace{-2mm}
  %
  \begin{center}
    \scalebox{0.96}{
  %
%
\hspace*{-1em}
\begin{tikzpicture}[scale=1,every node/.style={transform shape}]
%
\matrix[anchor=north,row sep=0.9cm,every node/.style={draw,very thick,circle,minimum width=2.5pt,fill,inner sep=0pt,outer sep=2pt}] at (0,-0.5) {
  \node(v_e1_0){};
  \\
  \node(v_e1_1){};
  \\
  \node(v_e1_2){};
  \\
  \node(v_e1_0'){};
  \\
};
\calcLength(v_e1_0,v_e1_1){mylen}
\draw[<-,very thick,>=latex,chocolate](v_e1_0) -- ++ (90:{0.45*\mylen pt});
%
\draw[->,very thick] (v_e1_0) to node[right,pos=0.45,xshift={-0.05*\mylen pt}]{$\loopnsteplab{2}$} 
                                 node[left,pos=0.45]{\small $\aact$} (v_e1_1);
%
\draw[->,very thick] (v_e1_1) to node[right,pos=0.45,xshift={-0.05*\mylen pt}]{$\loopnsteplab{1}$} 
                                 node[left,pos=0.45]{\small $\aact$} (v_e1_2);
\draw[->,shorten <= 4.5pt] (v_e1_1) to[out=170,in=180,distance={0.75*\mylen pt}] node[left,pos=0.5]{\small $\cact$} (v_e1_0);
%
\draw[->] (v_e1_2) to node[right,pos=0.5]{\small $\bact$} (v_e1_0');
\draw[->] (v_e1_2) to[out=180,in=190,distance={0.75*\mylen pt}] node[left,pos=0.5]{\small $\bact$} (v_e1_1);
%
\draw[->] (v_e1_0') to[out=0,in=0,distance={1.5*\mylen pt}] node[right,pos=0.5]{\small $\aact$} (v_e1_1); 
%
\path (v_e1_0') ++ ({0*\mylen pt},{-0.6*\mylen pt}) node{\large $\chartof{\astexpi{1}},\, \charthighhatof{\astexpi{1}}$};

\matrix[anchor=north,row sep=0.9cm,every node/.style={draw,very thick,circle,minimum width=2.5pt,fill,inner sep=0pt,outer sep=2pt}] at (2.9,0) {
  \node[draw=none,fill=none](v_1-dummy){};
  \\
  \node(v_0){};
  \\
  \node(v_1){};
  \\
  \node(v_2){};
  \\
};
\calcLength(v_0,v_1){mylen}
\draw[draw=none,<-,very thick,>=latex,chocolate](v_1-dummy) -- ++ (90:{0.45*\mylen pt});
\draw[<-,very thick,>=latex,chocolate](v_0) -- ++ (90:{0.5*\mylen pt});
\draw[->](v_0) to node[right,xshift={-0.05*\mylen pt},pos=0.45]{\small $\aact$} (v_1); 
%
\draw[->,very thick](v_1) to node[right,xshift={-0.05*\mylen pt},pos=0.45]{$\loopnsteplab{1}$}
                             node[left,xshift={0.05*\mylen pt},pos=0.45]{\small $\aact$} (v_2);
\draw[->,very thick,shorten <= 5pt](v_1) to[out=175,in=180,distance={0.75*\mylen pt}] 
         node[left,pos=0.5,xshift={0.05*\mylen pt}]{$\loopnsteplab{1}$} 
         node[right,xshift={-0.065*\mylen pt},pos=0.35]{\small $\cact$} (v_0);
%
\draw[->](v_2) to[out=180,in=185,distance={0.75*\mylen pt}]  
               node[below,yshift={0.0*\mylen pt},pos=0.2]{\small $\bact$} (v_1);
\draw[->](v_2) to[out=0,in=0,distance={1.3*\mylen pt}] 
               node[below,yshift={0.00*\mylen pt},pos=0.125]{\small $\bact$} (v_0);

\draw[-,magenta,thick,densely dashed] (v_e1_0) to (v_0);
\draw[-,magenta,thick,densely dashed] (v_e1_0') to (v_0);
\draw[-,magenta,thick,densely dashed] (v_e1_1) to (v_1);
\draw[-,magenta,thick,densely dashed] (v_e1_2) to (v_2);

\path (v_2) ++ ({0*\mylen pt},{-0.8*\mylen pt}) node{\large $\achart,\,\acharthat$};

%
\matrix[anchor=north,row sep=0.9cm,every node/.style={draw,very thick,circle,minimum width=2.5pt,fill,inner sep=0pt,outer sep=2pt}] at (5.85,0.1) {
  \node(v_e2_0){};
  \\
  \node(v_e2_1){};
  \\
  \node(v_e2_2){};
  \\
  \node(v_e2_0''){};
  \\
  \node(v_e2_1'){};
  \\
};
\calcLength(v_e2_0,v_e2_1){mylen}
  \draw[draw=none] (v_e2_2) arc (270:205:{\mylen pt}) node[style={draw,very thick,circle,minimum width=2.5pt,fill,inner sep=0pt,outer sep=2pt}](v_e2_0'){};
  \draw[draw=none] (v_e2_0'') arc (90:205:{\mylen pt}) node[style={draw,very thick,circle,minimum width=2.5pt,fill,inner sep=0pt,outer sep=2pt}](v_e2_0'''){};
\draw[<-,very thick,>=latex,chocolate](v_e2_0) -- ++ (90:{0.45*\mylen pt});
\draw[->](v_e2_0) to node[right,xshift={-0.05*\mylen pt},pos=0.4]{\small $\aact$} (v_e2_1);
%
%
\draw[->,very thick](v_e2_1) to node[right,pos=0.45,xshift={-0.075*\mylen}]{$\loopnsteplab{3}$}
                                node[left,pos=0.45,xshift={0.065*\mylen}]{\small $\aact$} (v_e2_2);
\draw[->,very thick](v_e2_1) to node[below,pos=0.45]{$\loopnsteplab{3}$} 
                                node[above,pos=0.7,yshift={-0.025*\mylen pt}]{\small $\cact$} (v_e2_0');
%
\draw[->] (v_e2_0') to[out=115,in=150,distance={0.7*\mylen pt}] node[above,pos=0.75]{\small $\aact$} (v_e2_1);
%
\draw[->,very thick](v_e2_2) to node[right,pos=0.45,xshift={-0.075*\mylen}]{$\loopnsteplab{2}$} 
                                node[left,pos=0.425,xshift={0.05*\mylen}]{\small $\bact$} (v_e2_0'');
\draw[->,shorten <= 5pt] (v_e2_2) to[out=10,in=0,distance=0.7cm] 
                                node[right,pos=0.5,xshift={-0.025*\mylen}]{\small $\bact$} (v_e2_1);
%
\draw[->](v_e2_0'') to node[right,xshift={-0.025*\mylen pt},pos=0.4]{\small $\aact$} (v_e2_1');
%
\draw[->,very thick](v_e2_1') to node[below,pos=0.45]{$\loopnsteplab{1}$} 
                                 node[above,pos=0.7,yshift={-0.05*\mylen pt}]{\small $\cact$} (v_e2_0''');
\draw[->](v_e2_1') to[out=0,in=0,distance=1.4cm] node[right,pos=0.6,xshift={-0.025*\mylen}]{\small $\aact$} (v_e2_2);
%
\draw[->] (v_e2_0''') to[out=115,in=150,distance={0.7*\mylen pt}] node[above,pos=0.4]{\small $\aact$} (v_e2_1');

\path (v_e2_1') ++ ({1*\mylen pt},{-0.3*\mylen pt}) node{\large $\chartof{\astexpi{2}},\,\charthighhatof{\astexpi{2}}$};

\draw[-,magenta,thick,densely dashed] (v_e2_0) to (v_0);
\draw[-,magenta,thick,densely dashed] (v_e2_0') to (v_0);
\draw[-,magenta,thick,densely dashed] (v_e2_0'') to (v_0);
\draw[-,magenta,thick,densely dashed] (v_e2_0''') to (v_0);
\draw[-,magenta,thick,densely dashed,out=160,bend right,distance={0.75*\mylen pt}] (v_e2_1) to (v_1);
\draw[-,magenta,thick,densely dashed] (v_e2_1') to (v_1);
\draw[-,magenta,thick,densely dashed] (v_e2_2) to (v_2);

\end{tikzpicture} 
%
  %
}
  \end{center}
\end{example}

\vspace{-1mm}
We have shown that Milner's axiomatization, tailored to star expressions without~1 and with ${}^{\sstexpbit}$, 
is complete in bisimulation semantics. 
At the core of our proof is the graph structure property LLEE, 
which characterizes the process graphs that can be expressed by star expressions without 1 and with ${}^{\sstexpbit}$
as charts whose bisimulation collapse is a \LLEEchart. 

Completeness of \BBP\ covers completeness of the theory $\thplus{\BPAzeropl}{\RSPpl}$ of perpetual loop iteration $\stexppl{(\cdot)}$ \cite{fokk:1997:pl:ICALP}
in the sense that the latter result can be shown by our means, or by a faithful interpretation $\stexppl{\astexp} \mapsto \stexpbit{\astexp}{\stexpzero}$
of $\thplus{\BPAzeropl}{\RSPpl}$ in \BBP. 

Completeness of \BBP\ can be extended, also by means of a faithful interpretation,
to cover star expressions with 0, 1, and ${}^{\sstexpit}$, 
but with a syntactic restriction on terms directly under a~${}^{\sstexpit}$:
that they can be rewritten to star expressions with only 'harmless' occurrences of 1. 
This is analogous to the situation that the completeness result from \cite{fokk:zant:1994,fokk:1996:kleene:star:AMAST} 
for star expressions without 0 and 1 and with ${}^{\sstexpbit}$ was extended in \cite{corr:nico:labe:2002} to a setting with 1 (but not 0) 
and ${}^{\sstexpit}$, where a generalized version of the non-empty-word property is disallowed for terms directly under a ${}^{\sstexpit}$. 
With the interpretation approach, also the result in \cite{corr:nico:labe:2002} can be obtained 
from the one in \cite{fokk:zant:1994,fokk:1996:kleene:star:AMAST}.



The main future goal is to solve Milner's problem entirely by extending our result to the full class of star expressions.

\begin{acks}
  We thank Alban Ponse for his suggestion to consider completeness of Milner's axiomatization for the fragment without 1,
  and Luca Aceto for encouragement, comments on the exposition, 
  and facilitating a visit of the second author to GSSI, from
  which this paper developed.
  Also, we thank the reviewers for detailed remarks
  and suggestions on how to improve the positioning of our completeness result.
\end{acks}

\bibliography{BBP-complete.bib}



\newpage\onecolumn
\appendix%
\section{Appendix: supplements, more proof details, and omitted proofs}%
  \label{appendix}%

\subsection{Proofs in Section~\ref{prelims}: Preliminaries}

\begin{repeatedprop}[= Proposition~\ref{prop:id:is:sol:chart:interpretation}, uses \BBP-axioms (B1)--(B7), (BKS1)]
  For every $\astexp\in\StExpsover{\actions}$,
  the identity function $\sidfunon{\vertsof{\astexp}} \funin \vertsof{\astexp} \to \vertsof{\astexp}\subseteq\StExpsover{\actions}$, $\astexpacc \mapsto \astexpacc$,
  is a provable solution of the chart interpretation $\chartof{\astexp}$ of $\astexp$. 
\end{repeatedprop}

In the proof of this proposition we will use the following definition concerning `action derivatives',
and the subsequent lemma. That statement can be viewed as the `fundamental theorem of differential calculus for star expressions'
which says that every star expressions can be reassembled by a form of `integration' from its action derivatives. 
In this context `differentiation' follows the definition of action derivatives in Definition~\ref{def:chart:interpretation}
(corresponding to Antimirov's concept of `partial derivative' in \cite{anti:1996}),
and `integration' means sum formation over products of pairs $\pair{\aact}{\atickstexp}$
for actions $\aact$ and \aderivatives{\aact} $\atickstexp$.

\begin{definition}
  For star expressions $\astexp\in\StExpsover{\actions}$ we define the set $\actderivs{\astexp}$ of \emph{action derivatives} of $\astexp$ as follows:
  \begin{equation*}
    \actderivs{\astexp}
      \defdby
        \descsetexpbig{ \pair{\aact}{\atickstexp} }
                      { \aact\in\actions,\, \atickstexp\in\tickStExpsover{\actions},\, \astexp \lt{\aact} \atickstexp } \punc{.}
  \end{equation*}
\end{definition}

\begin{lemma}\label{lem:ft}
  Every $\astexp\in\StExpsover{\actions}$ can be provably reassembled from its action derivatives as: 
  \begin{align}
    \astexp
      \:\BBPeq\: &
    \stexpsum{  
      \Bigl(
        \sum_{i=1}^{m}
          \aacti{i} 
      \Bigr)     
              }{
      \Big(          
        \sum_{j=1}^{n}
          \stexpprod{\bacti{j}}{\astexpacci{j}}
      \Bigr)  
                } \punc{,} 
      \label{eq:1:lem:ft}          
      \\[-0.75ex]
      & \;          
      \text{provided that }
      \actderivs{\astexp}
        \: =\:
      \setexpbig{ \pair{\aacti{1}}{\tick}, \ldots, \pair{\aacti{m}}{\tick},
                  \pair{\bacti{1}}{\astexpacci{1}}, \ldots, \pair{\bacti{n}}{\astexpacci{n}} } \punc{.}
      \label{eq:2:lem:ft}            
  \end{align}
\end{lemma}

\begin{proof}
  We start by noting that we need to show \eqref{eq:1:lem:ft}, for all $\astexp\in\StExpsover{\actions}$, only for one list representation of $\actderivs{\astexp}$ of the form \eqref{eq:2:lem:ft}.
  This is because then \eqref{eq:1:lem:ft} follows also for all other list representations of $\actderivs{\astexp}$ the form \eqref{eq:2:lem:ft}.
  Indeed, the axioms ($\commstexpsum$), ($\assocstexpsum$), and ($\idempotstexpsum$) of \BBP\ 
    (the \ACI\nb-axioms for associativity, commutativity, and idempotency of $\sstexpsum$)
  can be used to permute and duplicate summands as well as to remove duplicates of summands in sums \eqref{eq:1:lem:ft}
  according to permutations, duplications, and removal of duplicates in list representations of $\actderivs{\astexp}$ of the form \eqref{eq:2:lem:ft}.  
  
  We proceed by induction on the structure of star expressions in $\StExpsover{\actions}$.
  For performing the induction step, we distinguish 
  the five cases of productions in the grammar in Definition~\ref{def:StExps}.
  \begin{description}\setlength{\itemsep}{1.5ex}\vspace{0.5ex}
    \item{\emph{Case~1:}} \mbox{}
      $\astexp \syntequal \stexpzero$. 
      \vspace*{0.75ex}
      
      Then $\astexp$ does not enable any transitions, and hence $\actderivs{\astexp} = \emptyset$.
      We find the provable equality: 
      \begin{align*}
        \astexp
          \:\syntequal\:
        \stexpzero
          \BBPeq
        \stexpsum{\stexpzero}{\stexpzero}
            \qquad \text{(by axiom ($\commstexpsum$) of \BBP)} \punc{.}   
      \end{align*}
      This is of the form as in \eqref{eq:1:lem:ft} with $m = n = 0$
      when we construe $\actderivs{\astexp} = \emptyset$ as a list representation of the form \eqref{eq:2:lem:ft}.
    \item{\emph{Case~2:}} \mbox{}
      $\astexp \syntequal \aact$ for some $\aact\in\actions$. 
      \vspace*{0.75ex}  
    
      Then according to the TSS in Definition~\ref{def:chart:interpretation}
      the expression $\astexp$ enables precisely one transition, an \transitionact{\aact} to $\tick$.
      Hence the set of \actionderivatives\ of $\astexp$ consists only of one element:
      \begin{equation}\label{eq:2:prf:lem:ft}
        \actderivs{\astexp} = \setexp{ \pair{\aact}{\tick} } \punc{.}
      \end{equation}
      
      We find the provable equality: 
      \begin{align*}
        \astexp
          \BBPeq
        \stexpsum{\aact}{\stexpzero}
          \qquad \text{(by axiom ($\neutralstexpsum$) of \BBP)} \punc{.}  
      \end{align*}
      The right-hand side is of the form \eqref{eq:1:lem:ft} with $m = 1$, $\aacti{1} = \aact$ and $n = 0$
      in relation to \eqref{eq:2:prf:lem:ft}
      when we construe $\actderivs{\astexp}$ as a list representation of the form \eqref{eq:2:lem:ft}.
    \item{\emph{Case~3:}} \mbox{}
      $\astexp \syntequal \stexpsum{\astexpi{1}}{\astexpi{2}}$. 
      \vspace*{0.75ex}
      
      Since every star expression has only finitely many derivatives, each of which is either $\tick$ or a star expression,
      we may assume that the sets of \actionderivatives\ of the constituent expressions $\astexpi{1}$ and $\astexpi{2}$ of $\stexpsum{\astexpi{1}}{\astexpi{2}}$ 
      have list representations:
      \begin{equation}\label{eq:3:prf:lem:ft}
        \begin{aligned}
          \actderivs{\astexpi{1}}
            \: & = \:
          \setexpbig{ \pair{\aacti{11}}{\tick}, \ldots, \pair{\aacti{m_{1}1}}{\tick},
                      \pair{\bacti{11}}{\astexpacci{11}}, \ldots, \pair{\bacti{n_{1}1}}{\astexpacci{n_{1}1}} } \punc{,}
          \\
          \actderivs{\astexpi{2}}
            \: & = \:
          \setexpbig{ \pair{\aacti{12}}{\tick}, \ldots, \pair{\aacti{{m_{2}2}}}{\tick},
                      \pair{\bacti{12}}{\astexpacci{12}}, \ldots, \pair{\bacti{n_{2}2}}{\astexpacci{n_{2}2}} } \punc{.}    
        \end{aligned}                       
      \end{equation}
      Then it follows from the form of the TSS rules in Definition~\ref{def:chart:interpretation} concerning sums of star expressions
      that the sets of \actionderivatives\ of $\stexpsum{\astexpi{1}}{\astexpi{2}}$ is the union
      of the sets of \actionderivatives\ of $\astexpi{1}$, and of $\astexpi{2}$. 
      By permuting the \actionderivatives\ with tick to the front, this union has the list representation:
      \begin{equation}\label{eq:4:prf:lem:ft}
        \begin{aligned}
          &
          \actderivs{\stexpsum{\astexpi{1}}{\astexpi{2}}}
          \: = \:
          \begin{aligned}[t]
            \bigl\{ & \pair{\aacti{11}}{\tick}, \ldots, \pair{\aacti{m_{1}1}}{\tick},
                      \pair{\aacti{12}}{\tick}, \ldots, \pair{\aacti{{m_{2}2}}}{\tick},
                      \\[-0.25ex]
                    & \pair{\bacti{11}}{\astexpacci{11}}, \ldots, \pair{\bacti{n_{1}1}}{\astexpacci{n_{1}1}},
                      \pair{\bacti{12}}{\astexpacci{12}}, \ldots, \pair{\bacti{n_{2}2}}{\astexpacci{n_{2}2}}
            \bigr\} \punc{.}
          \end{aligned}
        \end{aligned}
      \end{equation}
      Now we can argue as follows to reassemble $\stexpsum{\astexpi{1}}{\astexpi{2}}$
      from its \actionderivatives: 
      \begin{align*}
        \astexp
          \:\syntequal\:
        \stexpsum{\astexpi{1}}{\astexpi{2}} \:
        & \;\parbox[t]{\widthof{$\eqin{\BBP}$}}{$\BBPeq$}\:
        \stexpsum{
          \Bigl(
            \stexpsum{  
              \Bigl(
                \sum_{i=1}^{m_{1}}
                  \aacti{i1} 
              \Bigr)     
                      }{
              \Big(          
                \sum_{j=1}^{n_{1}}
                  \stexpprod{\bacti{j1}}{\astexpacci{j1}}
              \Bigr) 
                       } 
          \Bigr)  
                  }{
          \Bigl(
            \stexpsum{  
              \Bigl(
                \sum_{i=1}^{m_{2}}
                  \aacti{i2} 
              \Bigr)     
                      }{
              \Big(          
                \sum_{j=1}^{n_{2}}
                  \stexpprod{\bacti{j2}}{\astexpacci{j2}}
              \Bigr) 
                       } 
          \Bigr)  
                 }     
          \displaybreak[0]\\
          & \;\parbox[t]{\widthof{$\eqin{\BBP}$}}{\hspace*{\fill}}\: \text{(by the induction hypothesis, using representation~\eqref{eq:4:prf:lem:ft})} 
        \displaybreak[0]\\
        & \;\parbox[t]{\widthof{$\eqin{\BBP}$}}{$\BBPeq$}\:
        \stexpsum{
          \Bigl(
            \stexpsum{
              \Bigl( 
                \stexpsum{
                  \Bigl(
                    \sum_{i=1}^{m_{1}}
                      \aacti{i1} 
                  \Bigr)     
                          }{
                  \Bigl(
                    \sum_{i=1}^{m_{2}}
                      \aacti{i2} 
                  \Bigr)
                          }
              \Bigr)  
                      }{  
                  \Big(          
                    \sum_{j=1}^{n_{1}}
                      \stexpprod{\bacti{j1}}{\astexpacci{j1}}
                  \Bigr) 
                        }
              \Bigr)
                   }{           
            \Big(          
              \sum_{j=1}^{n_{2}}
                \stexpprod{\bacti{j2}}{\astexpacci{j2}}
            \Bigr) 
                     } \punc{.}
          \displaybreak[0]\\
          & \;\parbox[t]{\widthof{$\eqin{\BBP}$}}{\hspace*{\fill}}\: \text{(by axioms ($\assocstexpsum$) and ($\commstexpsum$))} 
      \end{align*}
      Since \ACI\ is a subsystem of \BBP, this chain of provably equalities is one in \BBP.
      It demonstrates,
      together with applications of the axiom~($\assocstexpsum$) that are needed to bring each of the subexpressions of the two outermost summands
      into a form with association of summation subterms to the left,
      that $\astexp$ satisfies \eqref{eq:1:lem:ft} 
      when we construe $\actderivs{\astexp}$ in \eqref{eq:4:prf:lem:ft} as a list representation of the form \eqref{eq:2:lem:ft}
      with $m = m_{1} + m_{2}$ and $n = n_{1} + n_{2}$.
    \item{\emph{Case~4:}} \mbox{}
      $\astexp \syntequal \stexpprod{\astexpi{1}}{\astexpi{2}}$. 
      \vspace*{0.75ex}
      
      As argued in the previous case, we may assume
      that the \actionderivatives\ of $\astexpi{1}$ are of the form: 
      \begin{equation}\label{eq:5:prf:lem:ft}
        \actderivs{\astexpi{1}}
          \: = \:
        \setexpbig{ \pair{\aacti{11}}{\tick}, \ldots, \pair{\aacti{m_{1}1}}{\tick},
                    \pair{\bacti{11}}{\astexpacci{11}}, \ldots, \pair{\bacti{n_{1}1}}{\astexpacci{n_{1}1}} } \punc{.} 
      \end{equation}
      Then it follows from the forms of the two rules in the TSS in Definition~\ref{def:chart:interpretation}
      concerning transitions from expressions with concatenation as their outermost symbol 
      that the set of \actionderivatives\ of $\stexpprod{\astexpi{1}}{\astexpi{2}}$ has a list representation of the form: 
      \begin{equation}\label{eq:6:prf:lem:ft}
        \actderivs{\stexpprod{\astexpi{1}}{\astexpi{2}}}
          \;\, = \:
            \bigl\{ \pair{\aacti{11}}{\astexpi{2}}, \ldots, \pair{\aacti{m_{1}1}}{\astexpi{2}},
                    \pair{\bacti{11}}{\stexpprod{\astexpacci{11}}{\astexpi{2}}}, \ldots, \pair{\bacti{n_{1}1}}{\stexpprod{\astexpacci{n_{1}1}}{\astexpi{2}}} 
            \bigr\} \punc{.}
      \end{equation}
      \begin{description}\setlength{\itemsep}{1.5ex}\vspace{0.5ex}
        \item{\emph{Case 4.1:}} \mbox{}
          $m_1,n_1 > 0$.
          \vspace*{0.5ex}
      
          Then we can reassemble $\stexpprod{\astexpi{1}}{\astexpi{2}}$ as follows:
          \begin{align*}
            \astexp
              \:\syntequal\:
            \stexpprod{\astexpi{1}}{\astexpi{2}} \:
            & \;\parbox[t]{\widthof{$\eqin{\BBP}$}}{$\BBPeq$}\:
            \stexpprod{
              \Bigl(
                \stexpsum{  
                  \Bigl(
                    \sum_{i=1}^{m_{1}}
                      \aacti{i1} 
                  \Bigr)     
                          }{
                  \Big(          
                    \sum_{j=1}^{n_{1}}
                      \stexpprod{\bacti{j1}}{\astexpacci{j1}}
                  \Bigr)  
              \Bigr)  
                          }
                       }{\astexpi{2}}
              & \parbox{\widthof{(by the induction hypothesis,}}
                       {(by the induction hypothesis,
                        \\\phantom{(}%
                        using representation~\eqref{eq:5:prf:lem:ft})}           
            \displaybreak[0]\\
            & \;\parbox[t]{\widthof{$\eqin{\BBP}$}}{$\BBPeq$}\:
              \stexpsum{ 
                \Bigl(
                  \sum_{i=1}^{m_{1}}
                    \stexpprod{\aacti{i1}}{\astexpi{2}} 
                \Bigr)  
                     }{
                \Big(          
                  \sum_{j=1}^{n_{1}}
                    \stexpprod{(\stexpprod{\bacti{j1}}{\astexpacci{j1}})}{\astexpi{2}}
                \Bigr)  
                         }
              & \text{(by axiom ($\distr$))}          
            \displaybreak[0]\\
            & \;\parbox[t]{\widthof{$\eqin{\BBP}$}}{$\BBPeq$}\:
              \stexpsum{ 
                \Bigl(
                  \sum_{i=1}^{m_{1}}
                    \stexpprod{\aacti{i1}}{\astexpi{2}} 
                \Bigr)  
                     }{
                \Big(          
                  \sum_{j=1}^{n_{1}}
                    \stexpprod{\bacti{j1}}{(\stexpprod{\astexpacci{j1}}{\astexpi{2}})}
                \Bigr)  
                         }
              & \text{(by axiom  ($\assocstexpprod$))}
            \displaybreak[0]\\
            & \;\parbox[t]{\widthof{$\eqin{\BBP}$}}{$\BBPeq$}\:
              \stexpsum{\stexpzero}
                       {\Bigl(
                          \stexpsum{ 
                            \Bigl(
                              \sum_{i=1}^{m_{1}}
                                \stexpprod{\aacti{i1}}{\astexpi{2}} 
                            \Bigr)  
                                 }{
                            \Big(          
                              \sum_{j=1}^{n_{1}}
                                \stexpprod{\bacti{j1}}{(\stexpprod{\astexpacci{j1}}{\astexpi{2}})}
                            \Bigr)  
                                     }
                        \Bigr)}
              & \text{(by axiom ($\neutralstexpsum$))}
          \end{align*}
          This chain of provable equalities demonstrates,
          together with applications of the axiom~(B2) that are needed to bring each of the subexpressions of the right outermost summands
          into a form with association of summation subterms to the left,
          that $\astexp$ satisfies \eqref{eq:1:lem:ft} 
          when we construe $\actderivs{\astexp}$ in \eqref{eq:6:prf:lem:ft} as a list representation \eqref{eq:2:lem:ft}
          with $m = 0$ and $n = m_{1} + n_{1}$.
        \item{\emph{Case 4.2:}} \mbox{}
          $m_1 > 0$, $n_1 = 0$.
          \vspace*{0.5ex}
      
          Then we can reassemble $\stexpprod{\astexpi{1}}{\astexpi{2}}$ as follows:
          \begin{align*}
            \astexp
              \:\syntequal\:
            \stexpprod{\astexpi{1}}{\astexpi{2}} \:
            & \;\parbox[t]{\widthof{$\eqin{\BBP}$}}{$\BBPeq$}\:
            \stexpprod{
              \Bigl(
                \stexpsum{  
                  \Bigl(
                    \sum_{i=1}^{m_{1}}
                      \aacti{i1} 
                  \Bigr)     
                          }{
                  \Big(          
                    \sum_{j=1}^{n_{1}}
                      \stexpprod{\bacti{j1}}{\astexpacci{j1}}
                  \Bigr)
                          }  
              \Bigr)  
                       }{\astexpi{2}}
              & \parbox{\widthof{(by the induction hypothesis,}}
                       {(by the induction hypothesis,
                        \\\phantom{(}%
                        using representation~\eqref{eq:5:prf:lem:ft})}
          \displaybreak[0]\\
            & \;\parbox[t]{\widthof{$\eqin{\BBP}$}}{$\BBPeq$}\:
            \stexpprod{
              \Bigl(
                \stexpsum{  
                  \Bigl(
                    \sum_{i=1}^{m_{1}}
                      \aacti{i1} 
                  \Bigr)     
                          }{\stexpzero}
              \Bigr)  
                       }{\astexpi{2}}
              & \text{(since $n_1 = 0$))}                         
          \displaybreak[0]\\
            & \;\parbox[t]{\widthof{$\eqin{\BBP}$}}{$\BBPeq$}\:
            \stexpprod{
                \stexpsum{  
                  \Bigl(
                    \sum_{i=1}^{m_{1}}
                      \stexpprod{\aacti{i1}}{\astexpi{2}} 
                  \Bigr)    
                          }{\stexpzero}
                       }{\astexpi{2}}
              & \text{(by axiom $\distr$)}
          \displaybreak[0]\\  
            & \;\parbox[t]{\widthof{$\eqin{\BBP}$}}{$\BBPeq$}\:
              \stexpsum{ 
                \Bigl(
                  \sum_{i=1}^{m_{1}}
                    \stexpprod{\aacti{i1}}{\astexpi{2}} 
                \Bigr)  
                     }{\stexpzero}
              & \text{(by axiom ($\stexpzerostexpprod$))}
            \displaybreak[0]\\
            & \;\parbox[t]{\widthof{$\eqin{\BBP}$}}{$\BBPeq$}\:
              \stexpsum{\stexpzero}
                       {\Bigl(
                          \sum_{i=1}^{m_{1}}
                            \stexpprod{\aacti{i1}}{\astexpi{2}} 
                        \Bigr)}
              & \text{(by axioms ($\commstexpsum$) and ($\neutralstexpsum$))}
          \end{align*}
          This chain of provable equalities demonstrates that $\astexp$ satisfies \eqref{eq:1:lem:ft} 
          when we construe $\actderivs{\astexp}$ in \eqref{eq:6:prf:lem:ft}, recalling that $n_1 = 0$, 
          as a list representation \eqref{eq:2:lem:ft} with $m = 0$ and $n = m_{1}$.
        \item{\emph{Case 4.3:}} \mbox{}
          $m_1 = 0$, $n_1 > 0$.
          \vspace*{0.5ex}
      
          Then we can reassemble $\stexpprod{\astexpi{1}}{\astexpi{2}}$ as follows:
          \begin{align*}
            \astexp
              \:\syntequal\:
            \stexpprod{\astexpi{1}}{\astexpi{2}} \:
            & \;\parbox[t]{\widthof{$\eqin{\BBP}$}}{$\BBPeq$}\:
            \stexpprod{
              \Bigl(
                \stexpsum{  
                  \Bigl(
                    \sum_{i=1}^{m_{1}}
                      \aacti{i1} 
                  \Bigr)     
                          }{
                  \Big(          
                    \sum_{j=1}^{n_{1}}
                      \stexpprod{\bacti{j1}}{\astexpacci{j1}}
                  \Bigr)  
              \Bigr)  
                          }
                       }{\astexpi{2}}
              & \parbox{\widthof{(by the induction hypothesis,}}
                       {(by the induction hypothesis,
                        \\\phantom{(}%
                        using representation~\eqref{eq:5:prf:lem:ft})}           
            \displaybreak[0]\\
            & \;\parbox[t]{\widthof{$\eqin{\BBP}$}}{$\BBPeq$}\:
            \stexpprod{
              \Bigl(
                \stexpsum{\stexpzero}{
                  \Big(          
                    \sum_{j=1}^{n_{1}}
                      \stexpprod{\bacti{j1}}{\astexpacci{j1}}
                  \Bigr)  
              \Bigr)  
                          }
                       }{\astexpi{2}}
              & \text{(since $m_1 = 0$)}
            \displaybreak[0]\\           
            & \;\parbox[t]{\widthof{$\eqin{\BBP}$}}{$\BBPeq$}\:
              \stexpsum{ 
                \stexpprod{\stexpzero}{\astexpi{2}} 
                     }{
                \Big(          
                  \sum_{j=1}^{n_{1}}
                    \stexpprod{(\stexpprod{\bacti{j1}}{\astexpacci{j1}})}{\astexpi{2}}
                \Bigr)  
                         }
              & \text{(by axiom ($\distr$))}
            \displaybreak[0]\\           
            & \;\parbox[t]{\widthof{$\eqin{\BBP}$}}{$\BBPeq$}\:
              \stexpsum{ 
                \stexpzero 
                     }{
                \Big(          
                  \sum_{j=1}^{n_{1}}
                    \stexpprod{(\stexpprod{\bacti{j1}}{\astexpacci{j1}})}{\astexpi{2}}
                \Bigr)  
                         }
              & \text{(by axiom ($\stexpzerostexpprod$))}
            \displaybreak[0]\\           
            & \;\parbox[t]{\widthof{$\eqin{\BBP}$}}{$\BBPeq$}\:
              \stexpsum{ 
                \stexpzero 
                     }{
                \Big(          
                  \sum_{j=1}^{n_{1}}
                    \stexpprod{\bacti{j1}}{(\stexpprod{\astexpacci{j1}}{\astexpi{2}})}
                \Bigr)  
                         }
              & \text{(by axiom ($\assocstexpprod$))}
          \end{align*}
          This chain of provable equalities demonstrates that $\astexp$ satisfies \eqref{eq:1:lem:ft} 
          when we construe $\actderivs{\astexp}$ in \eqref{eq:6:prf:lem:ft}, recalling that $m_1 = 0$, 
          as a list representation \eqref{eq:2:lem:ft} with $m = 0$ and $n = n_{1}$.
        \item{\emph{Case 4.4:}} \mbox{}
          $m_1 = n_1 = 0$.
          \vspace*{0.5ex}
      
          Then we can reassemble $\stexpprod{\astexpi{1}}{\astexpi{2}}$ as follows:
          \begin{align*}
            \astexp
              \:\syntequal\:
            \stexpprod{\astexpi{1}}{\astexpi{2}} \:
            & \;\parbox[t]{\widthof{$\eqin{\BBP}$}}{$\BBPeq$}\:
            \stexpprod{
              \Bigl(
                \stexpsum{  
                  \Bigl(
                    \sum_{i=1}^{m_{1}}
                      \aacti{i1} 
                  \Bigr)     
                          }{
                  \Big(          
                    \sum_{j=1}^{n_{1}}
                      \stexpprod{\bacti{j1}}{\astexpacci{j1}}
                  \Bigr)
                          }  
              \Bigr)  
                       }{\astexpi{2}}
              & \parbox{\widthof{(by the induction hypothesis,}}
                       {(by the induction hypothesis,
                        \\\phantom{(}%
                        using representation~\eqref{eq:5:prf:lem:ft})}
          \displaybreak[0]\\
            & \;\parbox[t]{\widthof{$\eqin{\BBP}$}}{$\BBPeq$}\:
            \stexpprod{
              \bigl(
                \stexpsum{\stexpzero}  
                         {\stexpzero}
              \bigr)  
                       }{\astexpi{2}}
              & \text{(since $m_1 = n_1 = 0$)}
          \displaybreak[0]\\
            & \;\parbox[t]{\widthof{$\eqin{\BBP}$}}{$\BBPeq$}\:
            \stexpprod{\stexpzero}
                      {\astexpi{2}}
              & \text{(by axiom ($\neutralstexpsum$))} 
          \displaybreak[0]\\
            & \;\parbox[t]{\widthof{$\eqin{\BBP}$}}{$\BBPeq$}\:
            \stexpzero
              & \text{(by axiom ($\stexpzerostexpprod$))} 
          \displaybreak[0]\\
            & \;\parbox[t]{\widthof{$\eqin{\BBP}$}}{$\BBPeq$}\:
            \stexpsum{\stexpzero}{\stexpzero}
              & \text{(by axiom ($\neutralstexpsum$))}
          \end{align*}
          This chain of provable equalities demonstrates
          that $\astexp$ satisfies \eqref{eq:1:lem:ft} 
          when we construe $\actderivs{\astexp}$ in \eqref{eq:6:prf:lem:ft},  recalling that $m_1 = n_1 = 0$, 
          as a list representation \eqref{eq:2:lem:ft} with $m = 0$ and $n = 0$.
  
     \end{description}

   \item{\emph{Case~5:}} \mbox{}
      $\astexp \syntequal \stexpbit{\astexpi{1}}{\astexpi{2}}$. 
      \vspace*{0.75ex}
  
      As in Case~3 we may assume that the sets of \actionderivatives\ 
      of the constituent expressions $\astexpi{1}$ and $\astexpi{2}$ of $\stexpsum{\astexpi{1}}{\astexpi{2}}$ 
      have list representations of the form \eqref{eq:3:prf:lem:ft}.
      Then it follows from the forms of the three rules in Definition~\ref{def:chart:interpretation}
      concerning transitions from expressions with binary iteration as their outermost symbol,
      that the set of \actionderivatives\ of $\stexpbit{\astexpi{1}}{\astexpi{2}}$ has a list representation of the form:
      \begin{equation*}
        \actderivs{\stexpbit{\astexpi{1}}{\astexpi{2}}}
          \;\, = \:
          \begin{aligned}[t]
            \bigl\{ & \pair{\aacti{11}}{\stexpbit{\astexpi{1}}{\astexpi{2}}}, \ldots, \pair{\aacti{m_{1}1}}{\stexpbit{\astexpi{1}}{\astexpi{2}}},
                    \\[-0.25ex]
                    & \pair{\bacti{11}}{\stexpprod{\astexpacci{11}}{(\stexpbit{\astexpi{1}}{\astexpi{2}})}}, 
                        \ldots, \pair{\bacti{n_{1}1}}{\stexpprod{\astexpacci{n_{1}1}}{(\stexpbit{\astexpi{1}}{\astexpi{2}})}},
                    \\[-0.25ex]
                    & \pair{\aacti{12}}{\tick}, \ldots, \pair{\aacti{m_{2}2}}{\tick},
                      \pair{\bacti{12}}{\astexpacci{12}}, \ldots, \pair{\bacti{n_{2}2}}{\astexpacci{n_{2}2}} 
            \bigr\} \punc{.}
          \end{aligned}
      \end{equation*}
      By permuting the \actionderivatives\ with tick to the front, this representation can be changed into: 
      \begin{equation}\label{eq:7:prf:lem:ft}
        \left.\;
          \begin{aligned}
            \actderivs{\stexpbit{\astexpi{1}}{\astexpi{2}}}
              \;\, = \:
              \begin{aligned}[t]
                \bigl\{ & \pair{\aacti{12}}{\tick}, \ldots, \pair{\aacti{m_{2}2}}{\tick},
                        \\[-0.25ex]
                        & \pair{\aacti{11}}{\stexpbit{\astexpi{1}}{\astexpi{2}}}, \ldots, \pair{\aacti{m_{1}1}}{\stexpbit{\astexpi{1}}{\astexpi{2}}},
                        \\[-0.25ex]
                        & \pair{\bacti{11}}{\stexpprod{\astexpacci{11}}{(\stexpbit{\astexpi{1}}{\astexpi{2}})}}, 
                            \ldots, \pair{\bacti{n_{1}1}}{\stexpprod{\astexpacci{n_{1}1}}{(\stexpbit{\astexpi{1}}{\astexpi{2}})}},
                        \\[-0.25ex]
                        & \pair{\bacti{12}}{\astexpacci{12}}, \ldots, \pair{\bacti{n_{2}2}}{\astexpacci{n_{2}2}} 
                \bigr\} \punc{.}
              \end{aligned}
          \end{aligned}
        \;\;\right\}
      \end{equation}
      Now we argue as follows in order to reassemble $\stexpbit{\astexpi{1}}{\astexpi{2}}$ 
      from its \actionderivatives\ in $\actderivs{\astexp}\,$:
      \begin{alignat*}{2}
        \astexp
        & \;\parbox[t]{\widthof{$\eqin{\BBP}$}}{$\syntequal$}\:
          \stexpbit{\astexpi{1}}{\astexpi{2}}
          & & \hspace*{-10.5ex} \text{(assumption in this case)}
        \\
        & \;\parbox[t]{\widthof{$\eqin{\BBP}$}}{$\BBPeq$}\:
          \stexpsum{
            \stexpprod{\astexpi{1}}{(\stexpbit{\astexpi{1}}{\astexpi{2}})}
                    }{\astexpi{2}}
          & & \hspace*{-10.5ex} \text{(by axiom (BKS1))}
        \\
        & \;\parbox[t]{\widthof{$\eqin{\BBP}$}}{$\BBPeq$}\:
          \stexpsum{
            \stexpprod{
                \stexpsum{
                  \Bigl(
                    \sum_{i=1}^{m_{1}}
                      \aacti{i1} 
                  \Bigr)     
                          }{
                  \Big(          
                    \sum_{j=1}^{n_{1}}
                      \stexpprod{\bacti{j1}}{\astexpacci{j1}}
                  \Bigr) 
                          }  
                       }{(\stexpbit{\astexpi{1}}{\astexpi{2}})}
              \Bigr)   
                    }{
              \Bigl(
                \stexpsum{
                  \Bigl(
                    \sum_{i=1}^{m_{2}}
                      \aacti{i2} 
                  \Bigr)     
                          }{
                  \Big(          
                    \sum_{j=1}^{n_{2}}
                      \stexpprod{\bacti{j2}}{\astexpacci{j2}}
                  \Bigr) 
                          }
              \Bigr)        
                      }
          & & \hspace*{-10.5ex}
              \parbox{\widthof{(by the induction hypothesis,}}
                     {(by the induction hypothesis,
                      \\\phantom{(}%
                      using representation \eqref{eq:3:prf:lem:ft})} 
        \displaybreak[0]\\
        & \;\parbox[t]{\widthof{$\eqin{\BBP}$}}{$\BBPeq$}\:
          \stexpsum{
              \Bigl(
                \stexpsum{
                  \Bigl(
                    \sum_{i=1}^{m_{1}}
                      \stexpprod{\aacti{i1}}{(\stexpbit{\astexpi{1}}{\astexpi{2}})} 
                  \Bigr)     
                          }{
                  \Big(          
                    \sum_{j=1}^{n_{1}}
                      \stexpprod{(\stexpprod{\bacti{j1}}{\astexpacci{j1}})}
                                {(\stexpbit{\astexpi{1}}{\astexpi{2}})}
                  \Bigr) 
                          }
              \Bigr)
                    }{
              \Bigl(
                \stexpsum{
                  \Bigl(
                    \sum_{i=1}^{m_{2}}
                      \aacti{i2} 
                  \Bigr)     
                          }{
                  \Big(          
                    \sum_{j=1}^{n_{2}}
                      \stexpprod{\bacti{j2}}{\astexpacci{j2}}
                  \Bigr) 
                          }
              \Bigr)        
                      }
        \\              
          & & & \hspace*{-10.5ex} \text{(by axiom ($\distr$))} 
        \displaybreak[0]\\[0.75ex]
        & \;\parbox[t]{\widthof{$\eqin{\BBP}$}}{$\BBPeq$}\:
          \stexpsum{
            \Bigl(
              \stexpsum{
                \Bigl(
                  \sum_{i=1}^{m_{1}}
                    \stexpprod{\aacti{i1}}{(\stexpbit{\astexpi{1}}{\astexpi{2}})} 
                \Bigr)     
                        }{
                \Big(          
                  \sum_{j=1}^{n_{1}}
                    \stexpprod{\bacti{j1}}
                              {(\stexpprod{\astexpacci{j1}}
                                          {(\stexpbit{\astexpi{1}}{\astexpi{2}})})}
                \Bigr) 
                        }
            \Bigr)
                  }{
            \Bigl(
              \stexpsum{
                \Bigl(
                  \sum_{i=1}^{m_{2}}
                    \aacti{i2} 
                \Bigr)     
                        }{
                \Big(          
                  \sum_{j=1}^{n_{2}}
                    \stexpprod{\bacti{j2}}{\astexpacci{j2}}
                \Bigr) 
                        }
            \Bigr)        
                    }
          \\          
          & & & \hspace*{-10.5ex} \text{(by axiom ($\assocstexpprod$))} 
        \displaybreak[0]\\[0.75ex]
        & \;\parbox[t]{\widthof{$\eqin{\BBP}$}}{$\ACIeq$}\:
          \stexpsum{
            \Bigl(
              \sum_{i=1}^{m_{2}}
                \aacti{i2} 
            \Bigr) 
                    }{
            \Bigr(           
              \stexpsum{            
                \Bigl(
                  \sum_{i=1}^{m_{1}}
                    \stexpprod{\aacti{i1}}{(\stexpbit{\astexpi{1}}{\astexpi{2}})} 
                \Bigr) 
                       }{
                \Bigl(             
                  \stexpsum{
                    \Bigl(   
                      \sum_{j=1}^{n_{1}}
                        \stexpprod{\bacti{j1}}
                                  {(\stexpprod{\astexpacci{j1}}
                                              {(\stexpbit{\astexpi{1}}{\astexpi{2}})})}
                    \Bigr)     
                            }{
                    \Big(          
                      \sum_{j=1}^{n_{2}}
                        \stexpprod{\bacti{j2}}{\astexpacci{j2}}
                    \Bigr) 
                              }
                \Bigr)                             
                         }
             \Bigr)                    
                    }
          \\          
          & & & \hspace*{-10.5ex} \text{(by axioms ($\assocstexpsum$) and ($\commstexpsum$))} 
      \end{alignat*}
      This chain of provably equalities demonstrates,
      together with applications of the axiom~($\assocstexpsum$) that are needed to bring each of the subexpressions of the right outermost summand
      into a form with association of summation subterms to the left,
      that $\astexp$ satisfies \eqref{eq:1:lem:ft} 
      when we construe $\actderivs{\astexp}$ in \eqref{eq:7:prf:lem:ft} as a list representation of the form \eqref{eq:2:lem:ft}
      with $m = m_{2}$ and $n = m_{1} + n_{1} + n_{2}$.
  \end{description}  
  In each of these five possible cases concerning the outermost structure of $\astexp$
  we have successfully performed the induction step. 
  In this way we have proved the statement of the lemma.
\end{proof}

\begin{proof}[Proof of Proposition~\ref{prop:id:is:sol:chart:interpretation}]
  Let $\chartof{\astexp} = \tuple{\vertsof{\astexp},\tick,\astexp,\actions,\transsof{\astexp}}$ 
  be the chart interpretation of a star expression $\astexp\in\StExpsover{\actions}$.
 
  Let $\bstexp\in\vertsof{\astexp}\subseteq\StExpsover{\actions}$ be a vertex of $\chartof{\astexp}$. 
  By Lemma~\ref{lem:ft}
  every star expression in $\StExpsover{\actions}$ can be reassembled
  as the \provablein{\BBP} sum over products of 
  over its action derivatives $\pair{\aact}{\atickstexp}$, 
  that is, over all actions $\aact\in\actions$ and \aderivatives{\aact} $\atickstexp$ of $\astexp$. 
  In particular, 
  \eqref{eq:1:lem:ft} guarantees that $\idfunon{\vertsof{\astexp}}{\bstexp} = \bstexp$
  satisfies the condition for $\sidfunon{\vertsof{\astexp}}$ 
  to be a provable solution at the vertex $\bstexp$ of $\chartof{\astexp}$,
  relative to a representation \eqref{eq:2:lem:ft} of the action derivatives of $\bstexp$
  which corresponds to a representation as assumed in Definition~\ref{def:provable-solution}.
  Since $\bstexp\in\vertsof{\astexp}$ was arbitrary in this argument,
  it follows that $\sidfunon{\vertsof{\astexp}}$ is a provable solution of $\chartof{\astexp}$. 
\end{proof}

\subsection{Proofs in Section~\ref{LLEE}: 
            Layered loop existence and elimination}
\renewcommand{\ll}[1]{#1}

\begin{repeatedprop}[= Proposition~\ref{prop:lbl:chart:translation:is:LLEEw}]
  For every $\astexp\in\StExpsover{\actions}$,
  the \entrybodylabeling\/ $\charthatof{\astexp}$ of $\chartof{\astexp}$ is a \LLEEwitness\ of $\chartof{\astexp}$.  
\end{repeatedprop}

\begin{proof}
  To verify \ref{LLEEw:1}
  it suffices to show that there are no infinite body step paths from any star expression $\astexp$
  (this is also a preparation for \ref{LLEEw:2}\ref{LLEEw:2a}, part~\ref{loop:2}).
  We prove, by induction on the syntactic structure of $\astexp$, the stronger statement that if 
  $e \redtc f$, then there does not exist an infinite body step path from $f$.
  The base cases, in which $e$ is of the form $a$ or 0, are trivial.
  Suppose $e\syntequal \stexpsum{\astexpi{1}}{\astexpi{2}}$. Then $e_i\,\sredtc\,f$ for some $i\in\{1,2\}$. So by induction, $f$ does not exhibit an infinite body step path.
  Suppose $e\syntequal \stexpprod{\astexpi{1}}{\astexpi{2}}$. Then $e\,\sredtc\,f$ means either $\astexpi{1}\,\sredtc\,f_1$ and $f\syntequal\stexpprod{f_1}{\astexpi{2}}$, or $\astexpi{2}\,\sredrtc\,f$.
  In the first case, by induction, $f_1$ and $\astexpi{2}$ do not exhibit infinite body step paths.
  This induces that $\stexpprod{f_1}{\astexpi{2}}$ does not exhibit an infinite body step path.
  In the second case, by induction, $f$ does not exhibit an infinite body step path.
  Suppose $e\syntequal \stexpbit{\astexpi{1}}{\astexpi{2}}$. Then $e\,\sredtc\,f$ means (A) $f\syntequal\stexpbit{\astexpi{1}}{\astexpi{2}}$, or (B) $\astexpi{1}\,\sredtc\,f_1$ and $f\syntequal\stexpprod{f_1}{(\stexpbit{\astexpi{1}}{\astexpi{2}})}$, or (C) $\astexpi{2}\,\sredtc\,f$.
  In case (A), each body step path from $f$ starts with either $f\lt{}_{\bodylab}\stexpprod{\astexpacci{1}}{(\stexpbit{\astexpi{1}}{\astexpi{2}})}$ where $\astexpi{1}\lt{}\astexpacci{1}$ and $\astexpacci{1}$ is not normed, or $f\lt{}_{\bodylab}\astexpacci{2}$ where $\astexpi{2}\lt{}\astexpacci{2}$.
  In the first case, by induction, $\astexpacci{1}$ does not exhibit an infinite body step path, so since $\astexpacci{1}$ is not normed, $\stexpprod{\astexpacci{1}}{(\stexpbit{\astexpi{1}}{\astexpi{2}})}$ does not exhibit an infinite body step path.
  In the second case, by induction, $e_2'$ does not exhibit an infinite body step path.
  In case (B), since by induction $f_1$ and by case (A) $\stexpbit{\astexpi{1}}{e_2}$ do not exhibit infinite body step paths, $\stexpprod{f_1}{(\stexpbit{\astexpi{1}}{e_2})}$ does not exhibit an infinite body step path.
  In case (C), by induction, $f$ does not exhibit an infinite body step path.

  We verify \ref{LLEEw:2}.
  From the TSS-rules in Definition~\ref{def:chart:interpretation} it follows that if $\astexp$ has a \loopentrytransition, then 
  $\astexp \syntequal ((\cdots((\stexpbit{\astexpi{1}}{\astexpi{2}})\cdot\bstexpi{1})\cdots )\cdot\bstexpi{n})$
  for some $n\geq 0$ and $\astexpi{1}$ normed.
  Let $\acharthat$ denote the \entrybodylabeling\ defined by the TSS-rules in Definition~\ref{def:chart:interpretation}
  on the `free' (= start-vertex free) chart of all star expressions in $\StExpsover{\actions}$. 
  We prove \ref{LLEEw:2} for a subchart $\indsubchartinat{\acharthat}{\astexp,\aLname}$ of $\acharthat$.  
  We first consider the case $n=0$, and then generalize it. 

  Let $\astexp \syntequal \stexpbit{\astexpi{1}}{\astexpi{2}}$ with $\astexpi{1}$ normed, and 
  $\aLname = \bsth{\astexpi{1}}+1$.
  Either $\astexp \redi{\loopnsteplab{\aLname}} \astexp$
  or $\astexp \redi{\loopnsteplab{\aLname}} \stexpprod{\astexpacci{1}}{\astexp}$
  for some normed $\astexpacci{1}$ with $\astexpi{1} \red \astexpacci{1}$. 
  In the first case \ref{loop:1} is clearly satisfied; we focus on the second case.
  It can be argued, by induction on syntactic structure, that every normed star expression has a body step path to $\tick$.
  Then so does $\astexpacci{1}$.
  This means $\stexpprod{\astexpacci{1}}{\astexp}$ has a body step path to $\astexp$.
  Hence \ref{loop:1} holds.
  For the remainder of \ref{LLEEw:2} it suffices to consider 
  \loopentry\ transitions $\astexp \redi{\loopnsteplab{\aLname}} \stexpprod{\astexpdacci{1}}{\astexp}$
  where $\astexpi{1} \red \astexpdacci{1}$.
  Since we showed above there are no body step cycles, every body step path from $\astexpdacci{1}$ eventually leads to deadlock or $\tick$; 
  in the first case the corresponding body step path of $\stexpprod{\astexpdacci{1}}{\astexp}$ also deadlocks, 
  and in the second case it returns to $\astexp$. Hence \ref{loop:2} holds.
  Since $\stexpprod{\astexpdacci{1}}{\astexp}$ cannot reach $\surd$ without returning to $\astexp$, \ref{loop:3} holds.
  It can be shown, by induction on derivation depth, that $f\lt{}f'$ implies $\bsth{f}\geq\bsth{f'}$, 
  and clearly $f\lt{}_{\loopnsteplab{\bLname}}$ implies $\ll{\bLname}\leq\bsth{f}$.
  So if $e_1'' \:\scomprewrels{\sredrtc}{\sredi{\loopnsteplab{\bLname}}}$, then $\ll{\bLname}\leq\bsth{e_1''}\leq\bsth{e_1}$.
  Hence, if 
  $\stexpprod{\astexpdacci{1}}{\astexp}
    \:\scomprewrels{\sredtavoidsvrtci{\stexpbit{\astexpi{1}}{\astexpi{2}}}{\bodylab}}{\sredi{\loopnsteplab{\bLname}}}$,
  then $\ll{\bLname}<\bsth{e_1}+1=\ll{\aLname}$. So $\ref{LLEEw:2}\ref{LLEEw:2b}$ holds.

  Now consider $\astexp \syntequal ((\cdots((\stexpbit{\astexpi{1}}{\astexpi{2}})\cdot\bstexpi{1})\cdots )\cdot\bstexpi{n})$ for $n>0$,
  with $\astexpi{1}$ normed.
  Again $\aLname = \bsth{\astexpi{1}}+1$.
  The subchart $\indsubchartinat{\acharthat}{\astexp,\aLname}$ basically coincides with $\indsubchartinat{\acharthat}{\stexpbit{\astexpi{1}}{\astexpi{2}},\aLname}$, except that the star expressions in the first chart are post-fixed with $\bstexpi{1}$, \ldots, $\bstexpi{n}$;
  its transitions are derived by $n$ additional applications 
  of the first rule for concatenation in Definition~\ref{def:chart:interpretation}, to affix these expressions.
  This chart isomorphism between $\indsubchartinat{\acharthat}{\stexpbit{\astexpi{1}}{\astexpi{2}},\aLname}$
  and $\indsubchartinat{\acharthat}{\astexp,\aLname}$ 
  preserves action labels as well as the loop-labeling, because 
  the first rule for concatenation preserves these labels.
  We showed that $\indsubchartinat{\acharthat}{\stexpbit{\astexpi{1}}{\astexpi{2}},\aLname}$ satisfies \ref{LLEEw:2}, 
  so the same holds for $\indsubchartinat{\acharthat}{\astexp,\aLname}$.
\end{proof}



\bigskip

We now turn to the proof of Lemma~\ref{lem:loop:relations}, which expresses properties of the body transition relation $\sredi{\bodylab}$,
the descends-in-loop-to relation $\sdescendsinloopto$, 
the loops-back-to relation $\sloopsbackto$, and the directly-loops-back-to relation $\sdloopsbackto$. 

\begin{repeatedlem}[= Lemma~\ref{lem:loop:relations}]
  The relations $\sredi{\bodylab}$, $\sdescendsinloopto$, $\sloopsbackto$, $\sdloopsbackto$\vspace*{-0.5mm}
  as defined by a \LLEEwitness~$\acharthat$ on a chart~$\achart$ satisfy the following properties:
  \begin{enumerate}[label=(\roman{*})]
    \item{}
      $\acharthat$ does not have infinite $\sredi{\bodylab}$ paths (so no $\sredi{\bodylab}$ cycles).
    \item{}
      If $\sccof{\cvert} = \sccof{\avert}$, then $\cvert \descendsinlooptortc \avert$ implies $\avert \loopsbacktortc \cvert$.
    \item{}
      If $\avert \descendsinloopto \bvert$ and $\lognot{(\bvert\,\sloopsbackto)}$, then $\bvert$ is not normed.
    \item{}
      $\sccof{\cvert}=\sccof{\avert}$ if and only if  $\cvert\loopsbacktortc \bvert$ and $\avert\loopsbacktortc \bvert$ for some vertex~$\bvert$.
    \item{}
      $\sloopsbacktortc$ is a partial order that has the least-upper-bound property: 
      if a non\-empty set of vertices has an upper bound with respect to $\sloopsbacktortc$, then it has a least upper bound.
    \item{}
      $\sloopsbackto$ is a total order on $\sloopsbackto$-successor vertices: 
      if $\bvert \loopsbackto \averti{1}$ and $\bvert \loopsbackto \averti{2}$, 
      then $\averti{1} \loopsbackto \averti{2}$ or $\averti{1} = \averti{2}$ or $\averti{2} \loopsbackto \averti{1}$.
    \item{}
      If $\averti{1}\dloopsbackto \cvert$ and $\averti{2}\dloopsbackto \cvert$ for distinct $\averti{1},\averti{2}$, 
      then there is no vertex $\bvert$ such that both $\bvert\loopsbacktortc \averti{1}$ and $w\loopsbacktortc \averti{2}$.
  \end{enumerate}
\end{repeatedlem}

We split the proof into the arguments for the parts~\ref{it:bo:terminating}--\ref{it:direct-subordinates}, respectively.
In doing so we repeat these statements as individual lemmas, and add a few more on the way.

\begin{lemma}\label{lem:layeredness}
  In a chart with a \LLEEwitness, if $\avert \comprewrels{\sdescendsinlooplto{\aLname}}{\scomprewrels{\sdescendsinlooptortc}{\sredi{\loopnsteplab{\bLname}}}}$,
  then $\ll{\aLname} > \ll{\bLname}$. 
\end{lemma}

\begin{proof}
  By induction on the number $n$ of $\sdescendsinloopto\,$\nb-steps in a path 
  $\avert \comprewrels{\sdescendsinlooplto{\aLname}}{\scomprewrels{\sdescendsinloopto^{n}}{\sredi{\loopnsteplab{\bLname}}}}\,$. 
  If $n = 0$, then from 
  $\avert \comprewrels{\sdescendsinlooplto{\aLname}}{\sredi{\loopnsteplab{\bLname}}}$ 
  we get $\ll{\aLname} > \ll{\bLname}$ by means of the \LLEEwitness\ condition \ref{LLEEw:2}\ref{LLEEw:2b}.
  If $n > 0$, 
  then the path 
  $\avert \comprewrels{\sdescendsinlooplto{\aLname}}{\scomprewrels{\sdescendsinloopto^{n}}{\sredi{\loopnsteplab{\bLname}}}}$
  is of the form 
  $\avert \comprewrels{\sdescendsinlooplto{\aLname}}{\scomprewrels{\scomprewrels{\sdescendsinloopto^{n-1}}{\sdescendsinlooplto{\cLname}}}{\sredi{\loopnsteplab{\bLname}}}}$
  for some loop name $\cLname$. 
  This path contains an initial segment
  $\avert \:\scomprewrels{\sdescendsinlooplto{\aLname}}{\scomprewrels{\sdescendsinloopto^{n-1}}{\sredi{\loopnsteplab{\cLname}}}}$.
  Then $\ll{\aLname} > \ll{\cLname}$ follows by the induction hypothesis.
  From the part $\scomprewrels{\sdescendsinlooplto{\cLname}}{\sredi{\loopnsteplab{\bLname}}}$ of this path 
  we get $\ll{\cLname} > \ll{\bLname}$ by \LLEEwitness\ condition \ref{LLEEw:2}\ref{LLEEw:2b}.
  So we conclude that $\ll{\aLname} > \ll{\bLname}$ holds.
\end{proof}

\begin{lemma}\label{lem:descendsinlooptotc:tick}
  In a chart with a \LLEEwitness, if $\avert \descendsinlooptotc \bvert$, then $\bvert \neq \tick$. 
\end{lemma}

\begin{proof}
  Let $\acharthat$ be a \LLEEwitness\ of a chart~$\achart$. 
  It suffices to show that $\sdescendsinloopto\, \bvert$ implies $\bvert \neq \tick$.
  For this, we let $\avert$ and $\bvert$ be vertices such that $\avert \descendsinloopto \bvert$.
  Then we can pick $\aLname\in\natplus$ such that $\avert \descendsinlooplto{\aLname} \bvert$.
  Since this means $\avert \comprewrels{\sredtavoidsvi{\avert}{\loopnsteplab{\aLname}}}{\sredtavoidsvrtci{\avert}{\bodylab}} \bvert$,
  it follows that $\bvert\in\indsubchartinat{\acharthat}{\avert,\aLname}$.
  Now since $\indsubchartinat{\acharthat}{\avert,\aLname}$ is a loop chart by condition~\ref{LLEEw:2}\ref{LLEEw:2a} for the \LLEEwitness~$\acharthat$,
  it follows that $\bvert \neq \tick$.
\end{proof}

\begin{lemma}\label{lem:to:bodystep:descendsinloopto:path}
  In a chart with a \LLEEwitness\ (assumed to be \startconnected, see Definition~\ref{def:charts}), 
  every vertex is reachable by an acylic $\scomprewrels{\sredrtci{\bodylab}}{\sdescendsinlooptortc}$ path from the start vertex~$\start$,
  that is, $\start \comprewrels{\sredrtci{\bodylab}}{\sdescendsinlooptortc} \bvert$ holds for all vertices $\bvert$.
\end{lemma}

\begin{proof}
  Let $\apath$ be a path from $\start$ to $\bvert$.  
  By removing cycles from $\apath$ we obtain 
  an acyclic path $\apathacc$ from $\start$ to $\bvert$
  that consists of a sequence of \loopentry\ and body transitions.
  Hence $\apathacc$ is of the form
  $\start
     \redrtci{\bodylab}
   \bvert$
  or  
  $\start
     \redrtci{\bodylab}
   \cverti{0}
     \comprewrels{\sredtavoidsvi{\avert}{\loopnsteplab{\cLnamei{0}}}}{\sredtavoidsvrtci{\cverti{0}}{\bodylab}} 
   \cverti{1}
     \comprewrels{\sredtavoidsvi{\cverti{1}}{\loopnsteplab{\cLnamei{1}}}}{\sredtavoidsvrtci{\cverti{1}}{\bodylab}}
       \:\cdots\:
     \comprewrels{\sredtavoidsvi{\cverti{1}}{\loopnsteplab{\cLnamei{1}}}}{\sredtavoidsvrtci{\cverti{n-1}}{\bodylab}}
   \cverti{n}
     \syntequal
   \bvert$  
  for some $n\in\nat$, 
  and $\cLnamei{0},\ldots,\cLnamei{n}\in\natplus$,
  where the target-avoidance parts are due to acyclicity of $\apathacc$.
  Hence $\apathacc$ is of the form 
  $\start
     \redrtci{\bodylab}
   \cverti{0} 
     \comprewrels{\sdescendsinlooplto{\cLnamei{0}}}
                 {\sdescendsinlooplto{\cLnamei{1}}{\:\cdots\,\scomprewrels{\sdescendsinlooplto{\cLnamei{n-2}}}
                                                                          {\sdescendsinlooplto{\cLnamei{n-1}}}}}
   \bvert$,
  for some $n\in\nat$, and $\cLnamei{0},\ldots,\cLnamei{n}\in\natplus$,
  and therefore of the form $\start \comprewrels{\redrtci{\bodylab}}{\descendsinlooptortc} \bvert$.
\end{proof}

\begin{lemma}\label{lem:to:descendsinloopto:path}
  In a chart with a \LLEEwitness, for every path 
  \vspace*{-2mm}%
  $\avert \comprewrels{\sredtavoidsvi{\avert}{\loopnsteplab{\aLname}}}{\redtavoidsvrtc{\avert}} \bvert$
  there is an acyclic path 
                           $\avert \comprewrels{\sdescendsinlooplto{\aLname}}{\sdescendsinlooptortc} \bvert$.
\end{lemma}

\begin{proof}
  Let $\apath$ be a path from $\avert$ to $\bvert$ that starts with a \loopentry\ step with loop name $\aLname$ such that
  all targets of transitions in $\apath$ avoid $\avert$.  
  By removing cycles we obtain 
  an acyclic path $\apathacc$ from $\avert$ to $\bvert$ that starts with an $\alpha$\nb-\loopentry\ step whose target is not $\avert$.
  We can write $\apathacc$ as a sequence of \loopentry\ and body steps of the form
  $\avert 
     \comprewrels{\sredtavoidsvi{\avert}{\loopnsteplab{\aLname}}}{\sredtavoidsvrtci{\avert}{\bodylab}} 
   \cverti{1}
     \comprewrels{\sredtavoidsvi{\cverti{0}}{\loopnsteplab{\cLnamei{0}}}}{\sredtavoidsvrtci{\cverti{0}}{\bodylab}}
       \:\cdots\:
   \cverti{n-2}    
     \comprewrels{\sredtavoidsvi{\cverti{n-2}}{\loopnsteplab{\cLnamei{n-2}}}}{\sredtavoidsvrtci{\cverti{n-2}}{\bodylab}}
   \cverti{n-1}
     \comprewrels{\sredtavoidsvi{\cverti{n-1}}{\loopnsteplab{\cLnamei{n-1}}}}{\sredtavoidsvrtci{\cverti{n-1}}{\bodylab}}
   \bvert$  
  for some $n\geq 1$, where the target-avoidance parts are due to acyclicity of $\apathacc$.
  Hence $\apathacc$ is of the form 
  $\avert 
    \comprewrels{\sdescendsinlooplto{\aLname}}
                {\sdescendsinlooplto{\cLnamei{1}}{\:\cdots\,\scomprewrels{\sdescendsinlooplto{\cLnamei{n-2}}}
                                                                         {\sdescendsinlooplto{\cLnamei{n-1}}}}}
   \bvert$,
  and therefore of the form $\avert \comprewrels{\sdescendsinlooplto{\aLname}}{\sdescendsinlooptortc} \bvert$.
\end{proof}

The following lemma was also used implicitly in the proof of Lem.~\ref{lem:loop:relations},~\ref{it:least-upper-bound}.

\begin{lemma}\label{lem:layeredness:ext}
  In a chart with a \LLEEwitness, 
  if
  $\avert \comprewrels{\sredtavoidsvi{\avert}{\loopnsteplab{\aLname}}}{\scomprewrels{\sredtavoidsvrtc{\avert}}{\sredi{\loopnsteplab{\bLname}}}}$,
  then $\ll{\aLname} > \ll{\bLname}$.
\end{lemma}

\begin{proof}
  This is a direct consequence of Lem.~\ref{lem:to:descendsinloopto:path} and Lem.~\ref{lem:layeredness}.
\end{proof}

\begin{lemma}\label{lem:loopsbackto:channel}
  In a chart with a \LLEEwitness, if $\cvert \loopsbacktortc \avert \loopsbacktortc \bvert$, then each path $\cvert \redrtci{\bodylab} \bvert$ visits~$\avert$.
\end{lemma}

\begin{proof}
Let $\avert\neq\cvert,\bvert$, as else the lemma trivially holds. Since $\cvert \loopsbacktotc \avert \loopsbacktotc \bvert$, there is a path $\bvert\,\sredtavoidsvi{\bvert}{\loopnsteplab{\aLname}}\cdot\sredtavoidsvrtc{\bvert}\,\avert\,\sredtavoidsvi{\avert}{\loopnsteplab{\bLname}}\cdot\sredtavoidsvrtc{\avert}\,\cvert$.
  \vspace*{-.5mm}By layeredness, $\ll{\aLname} > \ll{\bLname}$.
  A path $\cvert\,\sredtavoidsvrtci{\avert}{\bodylab}\,\bvert$
  would yield $\avert\,\sredtavoidsvi{\avert}{\loopnsteplab{\bLname}}\cdot\sredtavoidsvrtc{\avert}\,\cvert\,\sredtavoidsvrtci{\avert}{\bodylab}\, \bvert\,\sredi{\loopnsteplab{\aLname}}$.
  \vspace*{-.5mm}Then layeredness would require $\ll{\bLname} > \ll{\aLname}$, which cannot be the case.
\end{proof}

\begin{repeatedlem}[= Lemma~\ref{lem:loop:relations},~\ref{it:bo:terminating}]
  In a chart with a \LLEEwitness,
    there are no infinite $\sredi{\bodylab}$ paths (so no $\sredi{\bodylab}$ cycles).
\end{repeatedlem}

\begin{proof}
  Let $\achart$ be a chart with \LLEEwitness~$\acharthat$, and with start vertex $\start$. 
  Due to Lemma~\ref{lem:to:bodystep:descendsinloopto:path} every vertex of $\avert$
  is reachable by a $\comprewrels{\sredrtci{\bodylab}}{\sdescendsinlooptortc}$ path from $\start$. 
  In order to show that there are no infinite $\sredi{\bodylab}$ paths in $\acharthat$
  it therefore suffices to show
  that if $\start \comprewrels{\sredrtci{\bodylab}}{\sdescendsinloopto^{n}} \avert$, then there is no infinite $\sredi{\bodylab}$ path from $\avert$. 
  
  For the base case, $n=0$, let $\bvert$ be such that $\start \redrtci{\bodylab} \bvert$.
  Now suppose that there is an infinite $\sredrtci{\bodylab}$ path from $\bvert$ in $\acharthat$.
  Then due to $\start \redrtci{\bodylab} \bvert$ it follows that there is also an infinite $\sredrtci{\bodylab}$ path from $\start$ in $\acharthat$.
  This, however, contradicts with the condition~\ref{LLEEw:1} that the \LLEEwitness~$\acharthat$ must satisfy.
  We conclude that there is no infinite $\sredrtci{\bodylab}$ path from $\bvert$ in $\acharthat$.\vspace{1.5mm}
  
  For performing the induction step from $n$ to $n+1$, let $\bvert$ be such that $\start \comprewrels{\sredrtci{\bodylab}}{\sdescendsinloopto^{n+1}} \bvert$.
  Then we can pick $\bverti{0}$ with $\start \comprewrels{\sredrtci{\bodylab}}{\sdescendsinloopto^{n}} \bverti{0} \descendsinloopto \bvert$.
  It follows that $\bverti{0} \comprewrels{\redtavoidsvi{\bverti{0}}{\loopnsteplab{\aLname}}}{\redtavoidsvrtci{\bverti{0}}{\bodylab}} \bvert$
  for some $\aLname\in\natplus$, which we pick accordingly.  
  Now suppose that there is an infinite $\sredrtci{\bodylab}$ path $\apath$ from $\bvert$ in $\acharthat$.
  Then it cannot be the case that $\apath$ avoids $\bverti{0}$ forever, because otherwise it would give rise to an infinite path
  $\bverti{0} 
     \comprewrels{\redtavoidsvi{\bverti{0}}{\loopnsteplab{\aLname}}}{\redtavoidsvi{\bverti{0}}{\bodylab}} 
   \bvert
     \redtavoidsvi{\bverti{0}}{\bodylab}
   \bverti{1}
     \redtavoidsvi{\bverti{0}}{\bodylab}
   \bverti{2}
     \redtavoidsvi{\bverti{0}}{\bodylab}
       \cdots$,
   which is not possible since the condition \ref{LLEEw:2}\ref{LLEEw:2a} for the \LLEEwitness~$\achart$
   implies that $\indsubchartinat{\acharthat}{\cverti{0},\aLname}$ is a loop chart.      
   Therefore it follows that $\apath$ must visit $\averti{0}$.
   But then $\apath$ also gives rise to an infinite $\sredrtci{\bodylab}$ path from $\bverti{0}$. 
   This, however, contradicts the the statement that the induction hypothesis guarantees for $\bverti{0}$ 
   due to $\start \comprewrels{\sredrtci{\bodylab}}{\sdescendsinloopto^{n}} \bverti{0}$,
   namely that there is no infinite $\sredrtci{\bodylab}$ path from $\bverti{0}$. 
   We have reached a contradiction. 
   Therefore we can conclude that there is no infinite $\sredrtci{\bodylab}$ path $\apath$ from $\bvert$ in $\acharthat$.
   In this way we have successfully performed the induction step.
\end{proof}

%

\begin{repeatedlem}[= Lemma~\ref{lem:loop:relations},~\ref{it:descendsinloopto:scc:loopsbackto}]
  In a chart with a \LLEEwitness, 
    if $\sccof{\cvert} = \sccof{\avert}$, then $\cvert \descendsinlooptortc \avert$ implies $\avert \loopsbacktortc \cvert$.
\end{repeatedlem}

\begin{proof}
 We prove that $u \descendsinloopto^n v$ implies $v \loopsbackto^n u$ for all $n\geq 0$, by induction on $n$. 
 The base case $n=0$ is trivial, as then $u=v$. If $n>0$, $u \descendsinloopto^{n-1} u'\descendsinloopto v$ for some $u'$. 
 Clearly $\sccof{u}=\sccof{u'}=\sccof{v}$. 
 By induction, $u' \loopsbackto^{n-1} u$. 
 Since $u' \descendsinloopto v$, there is an acyclic path $u'\, \scomprewrels{\sredi{\loopnsteplab{\aLname}}}{\sredrtci{\bodylab}}\,v$. 
 And since $\sccof{u'}=\sccof{v}$, 
 there is an acyclic path 
 $v \,\sredrtci{\bodylab}
        \cdot
      \sloopnstepto{\bLname_1}
        \cdot
      \sredrtci{\bodylab}
        \cdot
          \,\cdots\,
        \cdot
      \sloopnstepto{\bLname_k}
        \cdot
      \sredrtci{\bodylab} \, u'$. 
 By \ref{LLEEw:2}\ref{LLEEw:2b}, $\ll{\aLname} > \ll{\bLname_1} > \cdots > \ll{\bLname_k} > \ll{\aLname}$. 
 This means $k=0$, so $v \redrtci{\bodylab}u'$. This implies $v \loopsbackto u'$ and hence $v\loopsbackto^n u$.
\end{proof}

\begin{repeatedlem}[= Lemma~\ref{lem:loop:relations}, \ref{it:descendsinloopto:notloopsbackto:not:normed}]
  If, in a chart with a \LLEEwitness, $\sdescendsinloopto\,\bvert$ and $\lognot{(\bvert\,\sloopsbackto)}$, then $\bvert$ is not normed.
\end{repeatedlem}

\begin{proof}
  We argue indirectly by showing that the negation of the implication in the statement of the lemma leads to a contradiction.
  For this, suppose that $\avert\descendsinloopto\,\bvert$ and $\lognot{(\bvert\,\sloopsbackto)}$ hold for some vertices $\avert$ and $\bvert$,
  and that additionally $\bvert$ is normed. 
  From $\avert\descendsinloopto\,\bvert$ and $\lognot{(\bvert\,\sloopsbackto)}$ we obtain by Lem.~\ref{lem:loop:relations},~\ref{it:descendsinloopto:scc:loopsbackto}
  that $\bvert\notin\sccof{\avert}$. Since $\avert \descendsinloopto \bvert$ entails $\avert \redrtc \bvert$ this entails $\lognot{(\bvert \redrtc \avert)}$. 
  Now since that $\bvert$ is normed means $\bvert \redrtc \tick$,
  we obtain $\avert \descendsinlooptortc \bvert \redtavoidsvrtc{\avert} \tick$,
  which means $\avert \comprewrels{\sredtavoidsvi{\avert}{\loopnsteplab{\aLname}}}{\sredtavoidsvrtci{\avert}{\bodylab}} \bvert \redtavoidsvrtc{\avert} \tick$
  for some $\aLname\in\natplus$.
  Then it follows from Lemma~\ref{lem:to:descendsinloopto:path}
  that $\avert \descendsinlooptotc \tick$. 
  This, however, contradicts, Lemma~\ref{lem:descendsinlooptotc:tick}.
\end{proof}

\begin{repeatedlem}[= Lemma~\ref{lem:loop:relations},~\ref{it:bi:reachable:lLEEw}]
  In a chart with a \LLEEwitness, 
    $\sccof{\cvert}=\sccof{\avert}$ if and only if $\cvert\loopsbacktortc \bvert$ and $\avert\loopsbacktortc \bvert$ for some vertex~$\bvert$.
\end{repeatedlem}

\begin{proof}
  The direction from right to left of the lemma trivially holds; we focus on the direction from left to right.
  Let $\sccof{\cvert}=\sccof{\avert}$. The case $u=v$ is trivial. Let $u\neq v$.
  Then they are on a cycle, which,
  since there is no body step cycle, contains a loop-entry transition from some $w$.
  Without loss of generality, suppose $w\neq u$. 
  Then $w\descendsinlooptotc u$, so by Lemma~\ref{lem:loop:relations}, \ref{it:descendsinloopto:scc:loopsbackto}, $u\loopsbacktotc w$. 
  If $w=v$ we have $v\loopsbacktortc w$, and if $w\neq v$ we can argue in the same fashion that $v\loopsbacktotc w$.
\end{proof}

\begin{lemma}\label{lem:loopsbacktotc:irreflexive}
  In a chart with a \LLEEwitness, $\loopsbacktotc$ is irreflexive. 
\end{lemma}

\begin{proof} 
  Let $\acharthat$ be a \LLEEwitness\ of a \LLEEchart~$\achart$. 
  Suppose that $\bvert \loopsbacktotc \bvert$ holds for some vertex $\bvert$ of $\achart$ and $\acharthat$.
  Then it follows from the definition of $\sloopsbacktotc$ that there is a $\sredi{\bodylab}$ path of \nonzero\ length from $\bvert$ to $\bvert$ itself.
  But such a $\sredi{\bodylab}$ cycle in $\acharthat$ is not possible, as it would give rise to an infinite $\sredi{\bodylab}$ path in $\acharthat$,
  contradicting Lemma~\ref{lem:loop:relations}, \ref{it:bo:terminating}.
\end{proof}

\begin{lemma}\label{lem:loopsbacktortc:po}
  In a chart with a \LLEEwitness, $\loopsbacktortc$ is a partial order.
\end{lemma}

\begin{proof} 
  By definition, $\sloopsbackto$ is transitive--reflexive. Moreover, $\sloopsbackto$ is anti-symmetric,
  because $u\loopsbacktotc v$ and $v\loopsbacktotc u$ for $u \neq v$ would imply 
  $u\loopsbacktotc v$ and $v\loopsbacktotc u$, in contradiction with irreflexivity of $\sloopsbacktotc$, see Lemma~\ref{lem:loopsbacktotc:irreflexive}.   
\end{proof}

\begin{repeatedlem}[= Lemma~\ref{lem:loop:relations},~\ref{it:least-upper-bound}]
  In a chart with a \LLEEwitness,   
    $\sloopsbacktortc$ is a partial order that has the least-upper-bound property: 
    if a non\-empty set of vertices has an upper bound with respect to $\sloopsbacktortc$, then it has a least upper bound.
\end{repeatedlem}

\begin{proof}
  Let $\achart$ be a chart with a \LLEEwitness~$\achart$. Let the relation $\sloopsbackto$ be defined on $\achart$ according to $\acharthat$. 
  
  $\loopsbacktortc$ is a partial order by Lemma~\ref{lem:loopsbacktortc:po}.
  %
  %
  Since $\achart$ as a chart is finite, it suffices to show that for each vertex $v$ the set of vertices $x$ with $v\loopsbacktortc x$ 
  is totally ordered with regard to $\loopsbacktortc$. Let $v\loopsbacktotc u_1$ and $v\loopsbacktotc u_2$ with $u_1\neq u_2$. 
  There is a path 
   $u_1 
      \comprewrels{\sredtavoidsvi{u_1}{\loopnsteplab{\aLname}}}{\sredtavoidsvrtc{u_1}} 
    v
      \redtci{\bodylab} 
    u_2
      \comprewrels{\sredtavoidsvi{u_2}{\loopnsteplab{\bLname}}}{\sredtavoidsvrtc{u_2}}
    v
      \redtci{\bodylab}
    u_1$.
  Without loss of generality, suppose $\ll{\bLname}\geq\ll{\aLname}$.
  Then layeredness implies that each path $v \redtci{\bodylab} u_2$ must visit $u_1$, 
  so $v \redtavoidsvtci{u_2}{\bodylab} u_1 \redtci{\bodylab} u_2$.
  \vspace*{-.75mm}Hence there is a path 
  $u_2  
     \comprewrels{\sredtavoidsvi{u_2}{\loopnsteplab{\bLname}}}{\sredtavoidsvrtc{u_2}} 
   v
     \redtavoidsvtci{u_2}{\bodylab}
   u_1 
     \redtci{\bodylab}
   u_2$, which implies $u_1 \loopsbacktotc u_2$.
\end{proof}

\begin{repeatedlem}[= Lemma~\ref{lem:loop:relations},~\ref{it:direct-subordinates}]
  In a chart with a \LLEEwitness, 
    if $\averti{1}\dloopsbackto \cvert$ and $\averti{2}\dloopsbackto \cvert$ for distinct $\averti{1},\averti{2}$, 
    then there is no vertex $\bvert$ such that both $\bvert\loopsbacktortc \averti{1}$ and $w\loopsbacktortc \averti{2}$.
\end{repeatedlem}

\begin{proof}
  $\lognot(v_2\loopsbacktotc v_1)$ and $\lognot(v_1\loopsbacktotc v_2)$, for else the definition of $\dloopsbackto$
  would imply $u\loopsbacktortc v_1$ or $u\loopsbacktortc v_2$,
  and so $v_1\loopsbacktotc v_1$ or $v_2\loopsbacktotc v_2$,
  contradicting irreflexivity of $\sloopsbacktotc$, see Lemma~\ref{lem:loopsbacktotc:irreflexive}. 
  In the proof of Lemma~\ref{lem:loop:relations}, \ref{it:least-upper-bound}, we furthermore saw
  that for each $w$, $\{x\mid w\loopsbacktortc x\}$ is totally ordered with regard to $\loopsbacktortc$, which implies that any such sets cannot contain both $v_1$ and $v_2$. 
\end{proof}

\subsection{Proofs in Section~\ref{extraction:transferral}: 
            Extraction of star expressions from, 
            and transferral between, LLEE-charts}

\begin{repeatedprop}[= Proposition~\ref{prop:transf:sol:via:funbisim}, requires \BBP-axioms (B1), (B2), (B3)]%
  Let $\sphifun \funin \vertsi{1} \rightarrow \vertsi{2}$ be a functional bisimulation between charts $\acharti{1}$ and $\acharti{2}$.
  Let $\sasoli{2} \funin \vertsi{2}\setminus\setexp{\tick} \to \StExpsover{\actions}$ be a provable solution of $\acharti{2}$.
  Then $\scompfuns{\sasoli{2}}{\sphifun} \funin \vertsi{1}\setminus\setexp{\tick} \to \StExpsover{\actions}$ is a provable solution of $\acharti{1}$
  with the same principal value~as~$\sasoli{2}$. 
\end{repeatedprop}

\begin{proof}
  Let $\sasoli{2}$ be a provable solution of $\acharti{2}$. 
  Let $\avert\in\vertsi{1}\backslash\{\tick\}$. 
  Since $\sphifun$ is a functional bisimulation between $\acharti{1}$ and $\acharti{2}$, 
  the forth, back, and termination conditions for the graph of $\sphifun$ as a bisimulation hold
  for the pair $\pair{\avert}{\phifun{\avert}}$ of vertices. 
  This makes it possible to bring the sets of transitions $T_1(\avert)$ from $\avert$ in $\acharti{1}$, and  $T_2(\phifun{\avert})$ from $\phifun{\avert}$ in $\acharti{2}$
  into a 1--1~correspondence such that $\sphifun$ again relates their targets:
  \begin{align}
    T_1(\avert)
      & {} =
      \descsetexpbig{ \avert \lt{\aacti{i}} \tick }{ i = 1,\ldots,m }
        \cup
      \descsetexpbig{ \avert \lt{\bacti{j}} \avertacci{j 1} }{ j = 1,\ldots,n } \punc{,} 
        \label{eq:1:prf:prop:transf:sol:via:funbisim} 
    \displaybreak[0]\\[-0.5ex]
    T_2(\phifun{\avert})
      & {} =
      \descsetexpbig{ \phifun{\avert} \lt{\aacti{i}} \tick }{ i = 1,\ldots,m }
        \cup
      \descsetexpbig{ \phifun{\avert} \lt{\bacti{j}} \avertacci{j 2} }{ j = 1,\ldots,n } \punc{,}
        \label{eq:2:prf:prop:transf:sol:via:funbisim} 
    \displaybreak[0]\\
    \phifun{\avertacci{j 1}} 
      & {} =
    \avertacci{j 2} \punc{,} 
      \quad \text{for all $j\in\setexp{1,\ldots,n}$} \punc{,} 
        \label{eq:3:prf:prop:transf:sol:via:funbisim} 
  \end{align}
  with $n,m\in\nat$, and vertices $\avertacci{j 1}\in\vertsi{1}\setminus\setexp{\tick}$, and $\avertacci{j 2}\in\vertsi{2}\setminus\setexp{\tick}$, for $j\in\setexp{1,\ldots,n}$. 
  Note that the same transition may be listed multiple times in the set $T_2(\phifun{\avert})$.
  On this basis we can argue as follows.
  \[
  \begin{array}{rcl}
  \compfuns{(\sasoli{2}}{\sphifun)}{\avert}
  ~\syntequal~
  \asoli{2}{\phifun{\avert}} & ~\BBPeq~ &
  {\displaystyle\stexpsum{\Bigl(\sum_{i=1}^{m} \aacti{i}\Bigr)}
                 {\Bigl(\sum_{j=1}^{n} 
        \stexpprod{\bacti{j}}{\asoli{2}{\avertacci{j 2}}}\Bigr)}}\vspace{1mm}\\
  \multicolumn{3}{r}{\mbox{(since $\sasoli{2}$ is a provable solution of $\acharti{2}$,
  using \eqref{eq:2:prf:prop:transf:sol:via:funbisim} and axioms (\commstexpsum), (\assocstexpsum), (\idempotstexpsum))}}\vspace{2mm}\\
  &~\syntequal~&
   {\displaystyle\stexpsum{\Bigl(\sum_{i=1}^{m} \aacti{i}\Bigr)}
             {\Bigl(\sum_{j=1}^{n} 
  \stexpprod{\bacti{j}}{\compfuns{(\sasoli{2}}{\sphifun)}{\avertacci{j 1}}}\Bigr)}} \vspace{1mm}\\
 \multicolumn{3}{r}{\text{(using \eqref{eq:3:prf:prop:transf:sol:via:funbisim} and $\compfuns{(\sasoli{2}}{\sphifun)}{\avertacci{j 1}}\equiv\asoli{2}{\phifun{\avertacci{j 1}}}$)}}
  \end{array}
  \]
  This shows, in view of \eqref{eq:1:prf:prop:transf:sol:via:funbisim}, that $\scompfuns{\sasol}{\sphifun}$ 
  satisfies the condition for a provable solution at $\avert$. Now as $\avert\in\vertsi{1}\setminus\setexp{\tick}$ was arbitrary,
  $\scompfuns{\sasoli{2}}{\sphifun}$ (with domain $\vertsi{1}\setminus\setexp{\tick}$) is a provable solution of $\acharti{1}$.
  Since furthermore the functional bisimulation $\sphifun$ must relate the start vertices of $\acharti{1}$ and $\acharti{2}$,
  the principal value of $\scompfuns{\sasoli{2}}{\sphifun}$ coincides with that of $\sasoli{2}$. 
\end{proof}

\begin{repeatedlem}[= Lemma~\ref{lem:def:extrsol}]
  In a chart with a \LLEEwitness,
  for all vertices $\avert,\bvert$: 
  \begin{enumerate}[label=(\roman{*})]
    \item{}\label{it:bosn:repeatedlem:def:extrsol}
      $\avert \redi{\bodylab}  \bvert \Rightarrow \bosn{\avert} > \bosn{\bvert}$,
    \item{}\label{it:enl:repeatedlem:def:extrsol}
      $\avert \descendsinloopto \bvert \Rightarrow \enl{\avert} > \enl{\bvert}$.
  \end{enumerate}
\end{repeatedlem}

\begin{proof}
  For statement~\ref{it:bosn:repeatedlem:def:extrsol} we argue as follows.
  Recall that the body step norm $\bosn{\avert}$ in a \LLEEwitness~$\acharthat$ was defined 
  as the maximal length of a body step path from $\avert$ in $\acharthat$.
  This was \welldefined\ due to Lemma~\ref{lem:loop:relations}, \ref{it:bo:terminating}, and the finiteness of charts.
  Now suppose that $\avert \redi{\bodylab}  \bvert$.
  Then every body step path from $\bvert$ gives rise to a body step path 
  from $\avert$ that starts with the transition $\avert \redi{\bodylab}  \bvert$. 
  Hence a longest body step path from $\bvert$ of length $\bosn{\bvert}$ gives rise to a body step path from $\avert$
  of length $\bosn{\bvert} + 1$. It follows that $\bosn{\avert} \ge \bosn{\bvert} + 1 > \bosn{\bvert}$, and hence $\bosn{\avert} > \bosn{\bvert}$.
  
  \smallskip
  For showing statement~\ref{it:enl:lem:def:extrsol}, suppose that $\avert \descendsinloopto \bvert$.
  Then $\avert \descendsinlooplto{\aLname} \bvert$ holds for some $\aLname\in\natplus$.
  Then $\enl{\avert} \ge \aLname$. 
 %
  If there is no \loopentrytransition\ that departs from $\bvert$, then $\enl{\bvert} = 0$ holds,
  and hence we get $\enl{\avert} \ge \aLname > 0 = \enl{\bvert}$. 
  Otherwise we let $\bLname\in\natplus$ be the maximal index of a \loopentrytransition\ from $\bvert$. 
  Then $\avert \descendsinlooplto{\aLname} \bvert \,\sredi{\loopnsteplab{\bLname}}$. 
  By Lemma~\ref{lem:layeredness} it follows that $\aLname > \bLname$. 
  Consequently we find $\enl{\avert} \ge \aLname > \bLname = \enl{\bvert}$.
  In both cases we have shown $\enl{\avert} > \enl{\bvert}$.
\end{proof}

\begin{repeatedlem}[= Lemma~\ref{lem:prop:extracted:fun:is:solution}, uses the \BBP-axioms (B1)--(B6), (BKS2), but not the rule $\RSPbit\,$]
  For a \LLEEchart~$\achart$ with \LLEEwitness~$\acharthat$
  the following connection holds 
  between the extracted solution $\sextrsolof{\acharthat}$
  and the relative extracted solution $\sextrsoluntilof{\acharthat}$,
  for all vertices $\avert,\bvert$:  
  \begin{equation}\label{eq:1:prf:prop:extracted:fun:is:solution}
    \avert \descendsinloopto \bvert
      \;\;\Longrightarrow\;\;
        \extrsolof{\acharthat}{\bvert}
          \:\eqin{\BBP}\:
        \stexpprod{ \extrsoluntilof{\acharthat}{\bvert}{\avert} }
                  { \extrsolof{\acharthat}{\avert} } \punc{.}
  \end{equation}
  Note that if $\avert \descendsinloopto \bvert$, then $\avert \neq \tick$, and also $\bvert\neq\tick$, because $\bvert$ is in the body
  of a loop at $\avert$, and therefore cannot be $\tick$ (see Lem.\ \ref{lem:descendsinlooptotc:tick}).
\end{repeatedlem}

\begin{proof}
  In order to show \eqref{eq:1:prf:prop:extracted:fun:is:solution}
  we proceed by complete induction (without explicit treatment of the base case)
  on the length $\bosn{\bvert}$ of a longest body step path from $\bvert$.
  For performing the induction step, we consider arbitrary $\avert,\bvert\neq\tick$ with $\avert \descendsinloopto \bvert$.
  We assume a representation of the set $\fap{\hat{T}}{\bvert}$ of transitions from $\bvert$ in $\acharthat\,$:
  \begin{align}
      \fap{\hat{T}}{\bvert}
        = {} & 
        \descsetexpbig{ \bvert \lti{\aacti{i}}{\loopnsteplab{\aLnamei{i}}} \bvert }{ i = 1,\ldots,m }
          \cup
        \descsetexpbig{ \bvert \lti{\bacti{j}}{\loopnsteplab{\bLnamei{j}}} \bverti{j} }{ \, \bverti{j} \neq \bvert, \: j = 1,\ldots,n }  
        \notag\\[-0.75ex]
        & {} \; \cup 
        \descsetexpbig{ \bvert \lti{\cacti{i}}{\bodylab} \avert }{ i = 1,\ldots,p }
          \cup
        \descsetexpbig{ \bvert \lti{\dacti{j}}{\bodylab} \cverti{j} }{ \cverti{j} \neq \avert, \: j = 1,\ldots,q }
        \label{eq:2:prf:prop:extracted:fun:is:solution} 
  \end{align} 
  that partitions $\fap{\hat{T}}{\bvert}$ into \loopentry\ transitions to $\bvert$ and to other targets $\bverti{1},\ldots,\bverti{n}$, 
  and \bodytransitions\ to $\avert$ and to other targets $\cverti{1},\ldots,\cverti{q}$.
  Since $\bvert$ is contained in a loop at $\avert$, none of these targets can be $\tick$.
  In order to show provable equality at the right-hand side of \eqref{eq:1:prf:prop:extracted:fun:is:solution}, we argue as follows:
  \begin{align*}
    \extrsolof{\acharthat}{\bvert}
      & \;\parbox[t]{\widthof{$\eqin{\BBP}$}}{$\syntequal$}\:
      \Bigl(
          \stexpsum{
            \Bigl(
              \sum_{i=1}^{m}
                \aacti{i}
            \Bigr)
                    }{
            \Bigl(
              \sum_{j=1}^{n}
                \stexpprod{\bacti{j}}{\extrsoluntilof{\acharthat}{\bverti{j}}{\bvert}}
            \Bigr)
                    }
        \Bigr)^{\sstexpbit}
        \Bigl(
          \stexpsum{\stexpzero}
                   {\Bigl(
                      \stexpsum{
                        \Bigl(
                          \sum_{i=1}^{p}
                            \stexpprod{\cacti{i}}{\extrsolof{\acharthat}{\avert}}
                        \Bigr)
                                }{
                        \Bigl(
                          \sum_{j=1}^{q}
                            \stexpprod{\dacti{j}}{\extrsolof{\acharthat}{\cverti{j}}}
                        \Bigr)
                    \Bigr)}}
        \Bigr)
        \\[0.25ex]
        & \;\,\parbox[t]{\widthof{$\eqin{\milnersysmin}$\hspace*{3ex}}}{\mbox{}}\:
          \parbox{\widthof{\widthof{(using that none of the target vertices is the terminating sink~$\tick$,)}}}
                 {(by the definition of $\extrsolof{\acharthat}{\bvert}$, 
                  based on the representation \eqref{eq:2:prf:prop:extracted:fun:is:solution},\\\phantom{(}%
                  using that none of the target vertices is $\tick$)}
    \displaybreak[0]\\[0.5ex]
      & \;\parbox[t]{\widthof{$\eqin{\BBP}$}}{$\eqin{\BBP}$}\:
      \Bigl(
          \stexpsum{
            \Bigl(
              \sum_{i=1}^{m}
                \aacti{i}
            \Bigr)
                    }{
            \Bigl(
              \sum_{j=1}^{n}
                \stexpprod{\bacti{j}}{\extrsoluntilof{\acharthat}{\bverti{j}}{\bvert}}
            \Bigr)
                    }
        \Bigr)^{\sstexpbit}
        \Bigl(
          \stexpsum{
            \Bigl( 
              \sum_{i=1}^{p}
              \stexpprod{\cacti{i}}{\extrsolof{\acharthat}{\avert}}
            \Bigr)
                    }{
            \Bigl(
              \sum_{j=1}^{q}
                \stexpprod{\dacti{j}}{\extrsolof{\acharthat}{\cverti{j}}}
            \Bigr)
                    }
        \Bigr)
        \\
        & \;\,\parbox[t]{\widthof{$\eqin{\milnersysmin}$\hspace*{3ex}}}{\mbox{}}\:
          \text{(using axiom (\neutralstexpsum))}
    \displaybreak[0]\\
      & \;\parbox[t]{\widthof{$\eqin{\BBP}$}}{$\eqin{\BBP}$}\:
      \Bigl(
          \stexpsum{
            \Bigl(
              \sum_{i=1}^{m}
                \aacti{i}
            \Bigr)
                    }{
            \Bigl(
              \sum_{j=1}^{n}
                \stexpprod{\bacti{j}}{\extrsoluntilof{\acharthat}{\bverti{j}}{\bvert}}
            \Bigr)
                    }
        \Bigr)^{\sstexpbit}
        \Bigl(
          \stexpsum{
            \Bigl(
              \sum_{i=1}^{p}
                \stexpprod{\cacti{i}}{\extrsolof{\acharthat}{\avert}}
            \Bigr)}{
            \Bigl(  
              \sum_{j=1}^{q}
                \stexpprod{\dacti{j}}
                          {\bigl(
                            \stexpprod{ \extrsoluntilof{\acharthat}{\cverti{j}}{\avert} }
                                      { \extrsolof{\acharthat}{\avert}}
                           \bigr)}
            \Bigr)
                   }              
        \Bigr)   
        \\
        & \;\,\parbox[t]{\widthof{$\eqin{\milnersysmin}$\hspace*{3ex}}}{\mbox{}}\:
          \parbox[t]{\widthof{(by the induction hypothesis, using that $\avert\descendsinloopto\cverti{j}$ and $\bosn{\cverti{j}} < \bosn{\bvert}$)}}
                    {(by the induction hypothesis, using that 
             $\avert\descendsinloopto\cverti{j}$ and $\bosn{\cverti{j}} < \bosn{\bvert}$ \\\phantom{(}
                      because $\bvert \redi{\bodylab} \cverti{j}$
                      for $j=1,\ldots,q$, see \eqref{eq:2:prf:prop:extracted:fun:is:solution})}
    \displaybreak[0]\\[0.5ex]
      & \;\parbox[t]{\widthof{$\eqin{\BBP}$}}{$\eqin{\BBP}$}\:
      \Bigl(
          \stexpsum{
            \Bigl(
              \sum_{i=1}^{m}
                \aacti{i}
            \Bigr)
                    }{
            \Bigl(
              \sum_{j=1}^{n}
                \stexpprod{\bacti{j}}{\extrsoluntilof{\acharthat}{\bverti{j}}{\bvert}}
            \Bigr)
                    }
        \Bigr)^{\sstexpbit}
        \Bigl(
          \stexpprod{
            \Bigl(
              \stexpsum{
                \Bigl(
                  \sum_{i=1}^{p}
                    \cacti{i}
                \Bigr)}{
                \Bigl(  
                  \sum_{j=1}^{q}
                    \stexpprod{\dacti{j}}
                              {\extrsoluntilof{\acharthat}{\cverti{j}}{\avert}}
                \Bigr)
                        }
            \Bigr)
                     }{\extrsolof{\acharthat}{\avert}}             
        \Bigr)
        \\
        & \;\,\parbox[t]{\widthof{$\eqin{\milnersysmin}$\hspace*{3ex}}}{\mbox{}}\:
          \text{(using axioms (\assocstexpprod), (\distr))}
    \displaybreak[0]\\
      & \;\parbox[t]{\widthof{$\eqin{\BBP}$}}{$\eqin{\BBP}$}\:
      \stexpprod{
        \Bigl(\!
          \stexpsum{
            \Bigl(
              \sum_{i=1}^{m}
                \aacti{i}
            \Bigr)\!
                    }{\!
            \Bigl(
              \sum_{j=1}^{n}
                \stexpprod{\bacti{j}}{\extrsoluntilof{\acharthat}{\bverti{j}}{\bvert}}
            \Bigr)
                    }
        \!\Bigr)^{\sstexpbit}
        \Bigl(\!
          \stexpsum{
            \Bigl(
              \sum_{i=1}^{p}
                \cacti{i}
            \Bigr)\!
                    }{ \!
            \Bigl(
              \sum_{j=1}^{q}
                \stexpprod{\dacti{j}}{\extrsoluntilof{\acharthat}{\cverti{j}}{\avert}}
            \Bigr)
                    }
        \!\Bigr)\!
                }{\!\extrsolof{\acharthat}{\avert}}
        \\
        & \;\,\parbox[t]{\widthof{$\eqin{\milnersysmin}$\hspace*{3ex}}}{\mbox{}}\:
          \text{(using axiom (BKS2))}      
    \displaybreak[0]\\[0.5ex]
      & \;\parbox[t]{\widthof{$\eqin{\BBP}$}}{$\syntequal$}\:
      \stexpprod{ \extrsoluntilof{\acharthat}{\bvert}{\avert} }
                { \extrsolof{\acharthat}{\avert} }
      \\          
        & \;\,\parbox[t]{\widthof{$\eqin{\milnersysmin}$\hspace*{3ex}}}{\mbox{}}\:
          \parbox{300pt}{(by the definition of $\extrsoluntilof{\acharthat}{\bvert}{\avert}$, 
                          based on the representation \eqref{eq:2:prf:prop:extracted:fun:is:solution})}
  \end{align*}  
  This chain of provable equalities demonstrates
  \eqref{eq:1:prf:prop:extracted:fun:is:solution}.
\end{proof}

\begin{repeatedprop}[= Proposition~\ref{prop:extracted:fun:is:solution}, uses the \BBP-axioms (B1)--(B6), (BKS1), (BKS2), but not the rule $\RSPbit\,$]
  In a chart $\achart$ with a LLEE-witness $\acharthat$, 
  $\sextrsolof{\acharthat}$ is a provable solution of~$\achart$.
\end{repeatedprop}

\begin{proof}
  We prove that $\sextrsolof{\acharthat}$ is a provable solution of the chart $\achart$. 
  Let $\bvert\neq{\tick}$.
  We show that $\extrsolof{\acharthat}{\bvert}$ satisfies the defining equation of $\sextrsolof{\acharthat}$ to be a provable solution of $\achart$ at $\bvert$.
  
  We consider a representation of the set $\fap{\hat{T}}{\bvert}$ of transitions from $\bvert$ in~$\acharthat$ as follows: 
  \begin{align}
      \fap{\hat{T}}{\bvert}
        = {} & 
        \descsetexpbig{ \bvert \lti{\aacti{i}}{\loopnsteplab{\aLnamei{i}}} \bvert }{ i = 1,\ldots,m }
          \cup
        \descsetexpbig{ \bvert \lti{\bacti{j}}{\loopnsteplab{\bLnamei{j}}} \bverti{j} }{ \, \bverti{j} \neq \bvert, \: j = 1,\ldots,n }  
        \notag\\[-0.75ex]
        & {} \; \cup 
        \descsetexpbig{ \bvert \lti{\cacti{i}}{\bodylab} \tick }{ i = 1,\ldots,p }
          \cup
        \descsetexpbig{ \bvert \lti{\dacti{j}}{\bodylab} \cverti{j} }{ \cverti{j} \neq \tick, \: j = 1,\ldots,q }
        \label{eq:3:prf:prop:extracted:fun:is:solution} 
  \end{align} 
  that partitions $\fap{\hat{T}}{\bvert}$ into \loopentry\ transitions to $\bvert$ and to other targets $\bverti{1},\ldots,\bverti{n}$, 
  and \bodytransitions\ to $\tick$ and to other targets $\cverti{1},\ldots,\cverti{q}$.
  We argue as follows:
  \begin{align*}
    \extrsolof{\acharthat}{\bvert}~
      & {} \;\parbox[t]{\widthof{$\eqin{\BBP}$}}{$\syntequal$}\:
      \Bigl(
          \stexpsum{
            \Bigl(
              \sum_{i=1}^{m}
                \aacti{i}
            \Bigr)
                    }{
            \Bigl(
              \sum_{j=1}^{n}
                \stexpprod{\bacti{j}}{\extrsoluntilof{\acharthat}{\bverti{j}}{\bvert}}
            \Bigr)
                    }
        \Bigr)^{\sstexpbit}
        \Bigl(
          \stexpsum{
            \Bigl(
              \sum_{i=1}^{p}
                \cacti{i}
            \Bigr)
                    }{
            \Bigl(
              \sum_{j=1}^{q}
                \stexpprod{\dacti{j}}{\extrsolof{\acharthat}{\cverti{j}}}
            \Bigr)
                    }
        \Bigr)
        \\
        & \;\,\parbox[t]{\widthof{$\eqin{\milnersysmin}$\hspace*{3ex}}}{\mbox{}}\:
          \parbox{300pt}{(by the definition of $\sextrsolof{\acharthat}$,
                          in view of \eqref{eq:3:prf:prop:extracted:fun:is:solution})}
    \displaybreak[0]\\
      & \;\parbox[t]{\widthof{$\eqin{\BBP}$}}{$\eqin{\BBP}$}\:
      \stexpsum{
        \stexpprod{
          \Bigl(
              \stexpsum{
                \Bigl(
                  \sum_{i=1}^{m}
                    \aacti{i}
                \Bigr)
                        }{
                \Bigl(
                  \sum_{j=1}^{n}
                    \stexpprod{\bacti{j}}{\extrsoluntilof{\acharthat}{\bverti{j}}{\bvert}}
                \Bigr)
                        }
            \Bigr)
                   }{\extrsolof{\acharthat}{\bvert}}
                }{        
        \Bigl(
          \stexpsum{
            \Bigl(
              \sum_{i=1}^{p}
                \cacti{i}
            \Bigr)
                    }{
            \Bigl(
              \sum_{j=1}^{q}
                \stexpprod{\dacti{j}}{\extrsolof{\acharthat}{\cverti{j}}}
            \Bigr)
                    }
        \Bigr)
                   }
        \\
        & \;\,\parbox[t]{\widthof{$\eqin{\milnersysmin}$\hspace*{3ex}}}{\mbox{}}\:
          \text{(using axiom (BKS1) and the defining equality in the first step)}
    \displaybreak[0]\\
      & \;\parbox[t]{\widthof{$\eqin{\BBP}$}}{$\eqin{\BBP}$}\:
      \stexpsum{
        \Bigl(
          \stexpsum{
            \Bigl(
              \sum_{i=1}^{m}
                \stexpprod{\aacti{i}}{\extrsolof{\acharthat}{\bvert}}
            \Bigr)
                    }{         
            \Bigl(
              \sum_{j=1}^{n}
                \stexpprod{\bacti{j}}
                          {(\stexpprod{\extrsoluntilof{\acharthat}{\bverti{j}}{\bvert}}
                                      {\extrsolof{\acharthat}{\bvert}})}
            \Bigr)
                      }   
        \Bigr)
                  }{     
        \Bigl(
          \stexpsum{
            \Bigl(
              \sum_{i=1}^{p}
                \cacti{i}
            \Bigr)
                    }{
            \Bigl(
              \sum_{j=1}^{q}
                \stexpprod{\dacti{j}}{\extrsolof{\acharthat}{\cverti{j}}}
            \Bigr)
                    }
        \Bigr)
                   }
        \\
        & \;\,\parbox[t]{\widthof{$\eqin{\milnersysmin}$\hspace*{3ex}}}{\mbox{}}\:
          \text{(using axioms (\assocstexpprod), (\distr))}
    \displaybreak[0]\\
      & \;\parbox[t]{\widthof{$\eqin{\BBP}$}}{$\eqin{\BBP}$}\:
      \stexpsum{
        \Bigl(
          \stexpsum{
            \Bigl(
              \sum_{i=1}^{m}
                \stexpprod{\aacti{i}}{\extrsolof{\acharthat}{\bvert}}
            \Bigr)
                    }{         
            \Bigl(
              \sum_{j=1}^{n}
                \stexpprod{\bacti{j}}
                          {\extrsolof{\acharthat}{\bverti{j}}}
            \Bigr)
                      }   
        \Bigr)
                  }{     
        \Bigl(
          \stexpsum{
            \Bigl(
              \sum_{i=1}^{p}
                \cacti{i}
            \Bigr)
                    }{
            \Bigl(
              \sum_{j=1}^{q}
                \stexpprod{\dacti{j}}{\extrsolof{\acharthat}{\cverti{j}}}
            \Bigr)
                    }
        \Bigr)
                   }
        \\
        & \;\,\parbox[t]{\widthof{$\eqin{\milnersysmin}$\hspace*{3ex}}}{\mbox{}}\:
          \parbox{\widthof{(using \eqref{eq:1:prf:prop:extracted:fun:is:solution} of Lemma~\ref{lem:prop:extracted:fun:is:solution}, 
                            in view of $\bvert \descendsinloopto \bverti{i}$ for $j=1,\ldots,n$,}}
                 {(using \eqref{eq:1:prf:prop:extracted:fun:is:solution} of Lemma~\ref{lem:prop:extracted:fun:is:solution}, 
                   in view of $\bvert \descendsinloopto \bverti{j}$ for $j= 1,\ldots,n$)}
    \displaybreak[0]\\
      & \;\parbox[t]{\widthof{$\eqin{\BBP}$}}{$\eqin{\BBP}$}\:
      \stexpsum{
        \Bigl(
          \sum_{i=1}^{p}
            \cacti{i}
        \Bigr)  }{
        \Bigr(    
          \stexpsum{
            \Bigl(
              \stexpsum{
                \Bigl(
                  \sum_{i=1}^{m}
                    \stexpprod{\aacti{i}}{\extrsolof{\acharthat}{\bvert}}
                \Bigr)
                     }{         
                \Bigl(
                  \sum_{j=1}^{n}
                    \stexpprod{\bacti{j}}
                              {\extrsolof{\acharthat}{\bverti{j}}}
                \Bigr)
                          }   
            \Bigr)
                    }{ 
                \Bigl(
                  \sum_{j=1}^{q}
                    \stexpprod{\dacti{j}}{\extrsolof{\acharthat}{\cverti{j}}}
                \Bigr)
                       }
        \Bigr)
                   }
        \\
        & \;\,\parbox[t]{\widthof{$\eqin{\milnersysmin}$\hspace*{3ex}}}{\mbox{}}\:
          \text{(using axioms (\commstexpsum), (\assocstexpsum))}
  \end{align*}
  This chain of provable equalities demonstrates that
  $\extrsolof{\acharthat}{\bvert}$ is a provable solution of $\achart$ at $\bvert$,
  in view of \eqref{eq:3:prf:prop:extracted:fun:is:solution}.
  As $\bvert\neq\tick$ is arbitrary, $\sextrsolof{\acharthat}$ is indeed a provable solution of $\achart$.
\end{proof}

\begin{repeatedlem}[= Lemma~\ref{lem:prop:extrsol:vs:solution}, uses the \BBP-axioms (B1)--(B6), and the rule $\RSPbit\,$]
  For every provable solution $\sasol$ of a chart~$\achart$ with \LLEEwitness~$\acharthat$,
  the following connection holds with the relative extraction function $\sextrsoluntilof{\acharthat}$ holds,
  for all vertices $\avert,\bvert$:
  \begin{equation}\label{eq:repeatedlem:prop:extrsol:vs:solution}
    \avert \descendsinloopto \bvert
      \;\;\,\Longrightarrow\;\;\;
        \asol{\bvert} 
          \,\eqin{\BBP}\,
        \stexpprod{\extrsoluntilof{\acharthat}{\bvert}{\avert}}
                  {\asol{\avert}} 
  \end{equation}
  Note that if $\avert \descendsinloopto \bvert$, then $\avert \neq \tick$, and also $\bvert\neq\tick$, because $\bvert$ is in the body
   of a loop at $\avert$, and therefore cannot be $\tick$.
\end{repeatedlem}


\begin{proof}
  
  In order to prove \eqref{eq:repeatedlem:prop:extrsol:vs:solution}
  we proceed by complete induction
  on the same measure as used in the definition of the relative extraction function~$\sextrsoluntilof{\acharthat}$,
  namely, 
  induction on the maximal loop level of a loop at $\avert$, with a subinduction on $\bosn{\bvert}$.
  For performing the induction step, consider vertices $\avert$, $\bvert$ with $\avert \descendsinloopto \bvert$.
  As in the proof of Prop.~\ref{prop:extracted:fun:is:solution} 
  we assume the representation~\eqref{eq:2:prf:prop:extracted:fun:is:solution} of the set $\fap{\hat{T}}{\bvert}$ of transitions from $\bvert$ in $\acharthat$,
  which partitions $\fap{\hat{T}}{\bvert}$ into \loopentry\ transitions to $\bvert$ and to other targets $\bverti{1},\ldots,\bverti{n}$, 
  and \bodytransitions\ to $\avert$ and to other targets $\cverti{1},\ldots,\cverti{q}$.
  Since $\bvert$ is contained in a loop at $\avert$, none of these targets can be $\tick$.
  We now argue as follows:
  
  \begin{align*}
    \asol{\bvert}~
      & \;\parbox[t]{\widthof{$\eqin{\BBP}$}}{$\BBPeq$}\:
      \stexpsum{\stexpzero}
               {\Bigl(
                  \stexpsum{
                    \Bigl(
                      \sum_{i=1}^{m}
                        \stexpprod{\aacti{i}}{\asol{\bvert}}
                    \Bigr) 
                            }{
                    \Bigl(
                      \stexpsum{
                          \stexpsum{
                            \Bigl(
                              \sum_{j=1}^{n}
                                \stexpprod{\bacti{j}}
                                          {\asol{\bverti{j}}}
                            \Bigr)
                                    }{         
                            \Bigl(
                              \sum_{i=1}^{p}
                                \stexpprod{\cacti{i}}{\asol{\avert}}
                            \Bigr)  
                                      }
                                }{
                        \Bigl(
                          \sum_{j=1}^{q}
                            \stexpprod{\dacti{j}}{\asol{\cverti{j}}}
                        \Bigr)
                                  }
                     \Bigr)
                            }
                \Bigr)}        
            \\[0.25ex]
            & \;\,\parbox[t]{\widthof{$\eqin{\milnersysmin}$\hspace*{3ex}}}{\mbox{}}\:
          \parbox{\widthof{(as a representation of $\fap{\transshat}{\bvert}$ in Def.~\ref{def:provable-solution} without $\tick$)}}
                 {(since $\sasol$ is a provable solution of $\achart$ at $\bvert$,
                  using \eqref{eq:2:prf:prop:extracted:fun:is:solution})}
       \displaybreak[0]\\
          & \;\parbox[t]{\widthof{$\eqin{\BBP}$}}{$\BBPeq$}\:
            \stexpsum{
              \Bigl(
                \stexpsum{
                  \Bigl(
                    \sum_{i=1}^{m}
                      \stexpprod{\aacti{i}}{\asol{\bvert}}
                  \Bigr)
                          }{         
                  \Bigl(
                    \sum_{j=1}^{n}
                      \stexpprod{\bacti{j}}
                                {\asol{\bverti{j}}}
                  \Bigr)
                            }   
              \Bigr)
                      }{
              \Bigr(           
                \stexpsum{          
                  \Bigl(
                    \sum_{i=1}^{p}
                      \stexpprod{\cacti{i}}{\asol{\avert}}
                  \Bigr) 
                          }{
                  \Bigl(
                    \sum_{j=1}^{q}
                      \stexpprod{\dacti{j}}{\asol{\cverti{j}}}
                  \Bigr)
                            }
              \Bigr)
                         }
        \\
        & \;\,\parbox[t]{\widthof{$\eqin{\milnersysmin}$\hspace*{3ex}}}{\mbox{}}\:
          \text{(using axioms (\neutralstexpsum), (\assocstexpsum))}
    \displaybreak[0]\\
      & \;\parbox[t]{\widthof{$\eqin{\BBP}$}}{$\BBPeq$}\:
        \stexpsum{
          \Bigl(
            \stexpsum{
              \Bigl(
                \sum_{i=1}^{m}
                  \stexpprod{\aacti{i}}{\asol{\bvert}}
              \Bigr)
                      }{         
              \Bigl(
                \sum_{j=1}^{n}
                  \stexpprod{\bacti{j}}
                            {\bigl(
                               \stexpprod{\extrsoluntilof{\acharthat}{\bverti{j}}{\bvert}}
                                         {\asol{\bvert}}
                             \bigr)}
              \Bigr)
                        }   
          \Bigr)
                  }{
          \Bigr(           
            \stexpsum{          
              \Bigl(
                \sum_{i=1}^{p}
                  \stexpprod{\cacti{i}}{\asol{\avert}}
              \Bigr) 
                      }{
              \Bigl(
                \sum_{j=1}^{q}
                  \stexpprod{\dacti{j}}
                            {\bigl(
                               \stexpprod{\extrsoluntilof{\acharthat}{\cverti{j}}{\avert}}
                                         {\asol{\avert}}
                             \bigr)}
              \Bigr)
                        }
          \Bigr)
                     }          
        \\[0.25ex]
        & \;\,\parbox[t]{\widthof{$\eqin{\milnersysmin}$\hspace*{3ex}}}{\mbox{}}\:
          \parbox{\widthof{( $\avert \descendsinloopto \cverti{i}$ and $\bosn{\cverti{j}} < \bosn{\bvert}$ due to $\bvert \redi{\bodylab} \cverti{j}$ for $j\in\setexp{1,\ldots,q}$,
                            see \eqref{eq:2:prf:prop:extracted:fun:is:solution})}}
                 {(using the induction hypothesis, which is applicable because%
                  \\\phantom{(}%
                  the maximal loop level at $\bvert$ is smaller than that at $\avert$ due to $\avert \descendsinloopto \bvert$, and 
                  \\\phantom{(}%
                  $\avert \descendsinloopto \cverti{i}$ and $\bosn{\cverti{j}} < \bosn{\bvert}$ due to $\bvert \redi{\bodylab} \cverti{j}$ for $j=1,\ldots,q$,
                  see \eqref{eq:2:prf:prop:extracted:fun:is:solution})}
    \displaybreak[0]\\[0.5ex]
      & \;\parbox[t]{\widthof{$\eqin{\BBP}$}}{$\BBPeq$}\:
      \stexpsum{
        \stexpprod{
          \Bigl(
            \stexpsum{
              \Bigl(
                \sum_{i=1}^{m}
                  \aacti{i}
              \Bigr)
                       }{         
              \Bigl(
                \sum_{j=1}^{n}
                      \stexpprod{\bacti{j}}
                                {\extrsoluntilof{\acharthat}{\bverti{j}}{\bvert}}
              \Bigr)
                        }
          \Bigr) 
                    }{\asol{\bvert}}
                 }{ 
        \stexpprod{  
          \Bigl(
            \stexpsum{         
              \Bigl(
                \sum_{i=1}^{p}
                  \cacti{i}
              \Bigr) 
                      }{
              \Bigl(
                \sum_{j=1}^{q}
                  \stexpprod{\dacti{j}}
                            {\extrsoluntilof{\acharthat}{\cverti{j}}{\avert}}
              \Bigr)   }
          \Bigr)
                    }{\asol{\avert}}
                     }
        \\
        & \;\,\parbox[t]{\widthof{$\eqin{\milnersysmin}$\hspace*{3ex}}}{\mbox{}}\:
          \text{(using axioms (\assocstexpprod), (\distr))}
  \end{align*}
  This chain of provable equalities justifies: 
  \begin{align*}
    \asol{\bvert}
      & \;\parbox[t]{\widthof{$\eqin{\BBP}$}}{$\BBPeq$}\:
      \stexpsum{
        \stexpprod{
          \Bigl(
            \stexpsum{
              \Bigl(
                \sum_{i=1}^{m}
                  \aacti{i}
              \Bigr)
                       }{         
              \Bigl(
                \sum_{j=1}^{n}
                      \stexpprod{\bacti{j}}
                                {\extrsoluntilof{\acharthat}{\bverti{j}}{\bvert}}
              \Bigr)
                        }
          \Bigr) 
                    }{\asol{\bvert}}
                 }{ 
        \stexpprod{  
          \Bigl(
            \stexpsum{         
              \Bigl(
                \sum_{i=1}^{p}
                  \cacti{i}
              \Bigr) 
                      }{
              \Bigl(
                \sum_{j=1}^{q}
                  \stexpprod{\dacti{j}}
                            {\extrsoluntilof{\acharthat}{\cverti{j}}{\avert}}
              \Bigr)   }
          \Bigr)
                    }{\asol{\avert}}
                     }
  \end{align*}
  To this equality we can apply the rule $\RSPbit$:
  \begin{align*}
    \asol{\bvert}
      & \;\parbox[t]{\widthof{$\eqin{\BBP}$}}{$\BBPeq$}\:
      \stexpprod{
        \Bigl(
              \Bigl(
                \stexpsum{
                  \Bigl(
                    \sum_{i=1}^{m}
                      \aacti{i}
                  \Bigr)
                           }{         
                  \Bigl(
                    \sum_{j=1}^{n}
                          \stexpprod{\bacti{i}}
                                    {\extrsoluntilof{\acharthat}{\bverti{j}}{\bvert}}
                  \Bigr)
                            }
              \Bigr)^{\sstexpbit}\hspace*{-1pt} 
              \Bigl(
                \stexpsum{         
                  \Bigl(
                    \sum_{i=1}^{p}
                      \cacti{i}
                  \Bigr) 
                          }{
                  \Bigl(
                    \sum_{j=1}^{q}
                      \stexpprod{\dacti{j}}
                                {\extrsoluntilof{\acharthat}{\cverti{j}}{\avert}}
                  \Bigr)   }
              \Bigr)
      \Bigr)
             }{\asol{\avert}} 
        \\
        & \;\,\parbox[t]{\widthof{$\eqin{\milnersysmin}$\hspace*{3ex}}}{\mbox{}}\:
          \text{(by applying rule $\RSPbit$)}
    \displaybreak[0]\\[0.75ex]
      & \;\parbox[t]{\widthof{$\eqin{\BBP}$}}{$\syntequal$}\:
        \stexpprod{\extrsoluntilof{\acharthat}{\bvert}{\avert}}
                  {\asol{\avert}} \punc{,}
  \end{align*}
  The last step uses the definition of $\extrsoluntilof{\acharthat}{\bvert}{\avert}$,
  based on representation \eqref{eq:2:prf:prop:extracted:fun:is:solution} of $\fap{\transshat}{\bvert}$.
  In this way we have carried out the induction step.
  We conclude that \eqref{eq:repeatedlem:prop:extrsol:vs:solution} holds for all vertices $\avert$ and $\bvert$ of $\achart$.
\end{proof}

\subsection{Proofs in Section~\ref{collapse}: 
    Preservation of {\bf LLEE} under 
                               collapse}

\begin{repeatedlem}[= Lemma~\ref{lem:connthroughchart:bisim}]
  If $\bverti{1}\bisim \bverti{2}$ in $\achart$, then $\connthroughin{\achart}{\bverti{1}}{\bverti{2}}\bisim\achart$.
\end{repeatedlem}

\begin{proof}
Let $\achart=\tuple{\vertsi{1},\tick,\starti{1},\transsi{1}}$ and $\connthroughin{\achart}{\bverti{1}}{\bverti{2}}=\tuple{\vertsi{2},\tick,\starti{2},\transsi{2}}$.
Let  $\abisimi{1} \subseteq \vertsi{1}\times\vertsi{1}$ be the largest bisimulation relation on $\achart$.
In particular, $\pair{\bverti{1}}{\bverti{2}}\in\abisimi{1}$.
We argue that $\abisimi{2} = \abisimi{1}\cap(\vertsi{1}\times\vertsi{2})$ is a bisimulation relation between $\achart$ and $\connthroughin{\achart}{\bverti{1}}{\bverti{2}}$.
Take any $\pair{\cvert}{\avert}\in\abisimi{2}\subseteq\abisimi{1}$.
\begin{itemize} 
  \item
    ({\it forth}): Let  $\cvert \lt{\aact} \cvertacc\in\transsi{1}$.
    Then $\pair{\cvert}{\avert}\in\abisimi{1}$ implies there is a $\avert \lt{\aact} \avertacc\in\transsi{1}$ with
     $\pair{\cvertacc}{\avertacc}\in\abisimi{1}$.
    If $\avert \lt{\aact} \avertacc\in\transsi{2}$, then $\avertacc\in\vertsi{2}$, so $\pair{\cvertacc}{\avertacc}\in\abisimi{2}$ and we are done.
    If $\avert \lt{\aact} \avertacc\not\in\transsi{2}$, then $\avertacc=\bverti{1}$ and $\avert \lt{\aact} \bverti{2}\in\transsi{2}$.
    Since $\pair{\cvertacc}{\bverti{1}}\in\abisimi{1}$ and $\pair{\bverti{1}}{\bverti{2}}\in\abisimi{1}$, also $\pair{\cvertacc}{\bverti{2}}\in\abisimi{1}$.
    Since $\bverti{2}\in\vertsi{2}$, it follows that $\pair{\cvertacc}{\bverti{2}}\in\abisimi{2}$.
  \item
    ({\it back}): Let $\avert \lt{\aact} \avertacc\in\transsi{2}$.
    If $\avert \lt{\aact} \avertacc\in\transsi{1}$, then  $\pair{\cvert}{\avert}\in\abisimi{1}$ implies there is a $\cvert \lt{\aact} \cvertacc\in\transsi{1}$ with $\pair{\cvertacc}{\avertacc}\in\abisimi{1}$.
    Since $\avertacc\in\vertsi{2}$, also $\pair{\cvertacc}{\avertacc}\in\abisimi{2}$ and we are done.
    If $\avert \lt{\aact} \avertacc\not\in\transsi{1}$, then $\avertacc=\bverti{2}$ and $\avert \lt{\aact} \bverti{1}\in\transsi{1}$.
    So $\pair{\cvert}{\avert}\in\abisimi{1}$ implies there is a $\cvert \lt{\aact} \cvertacc\in\transsi{1}$ with $\pair{\cvertacc}{\bverti{1}}\in\abisimi{1}$.
    Since $\pair{\cvertacc}{\bverti{1}}\in\abisimi{1}$ and $\pair{\bverti{1}}{\bverti{2}}\in\abisimi{1}$, also $\pair{\cvertacc}{\bverti{2}}\in\abisimi{1}$.
    Since $\bverti{2}\in\vertsi{2}$, it follows that $\pair{\cvertacc}{\bverti{2}}\in\abisimi{2}$.
  \item
    ({\it termination}): Since $\abisimi{2}\subseteq\abisimi{1}$ clearly $\cvert=\tick$ if and only if $\avert=\tick$.
\end{itemize}
Finally, concerning ({\it start}):
  If $\starti{1}=\starti{2}$, then trivially $\pair{\starti{1}}{\starti{2}}\in\abisimi{2}$.
  If $\starti{1}\neq \starti{2}$, then $\starti{1}=\bverti{1}$ and $\starti{2}=\bverti{2}$.
  Since $\pair{\bverti{1}}{\bverti{2}}\in\abisimi{1}$ and $\bverti{2}\in\vertsi{2}$, we have $\pair{\bverti{1}}{\bverti{2}}\in\abisimi{2}$. 
\end{proof}

\begin{repeatedprop}[= Proposition~\ref{prop:reduced:br}]
  If a \LLEEchart~$\achart$ is not a bisimulation collapse,
  then it contains a pair of bisimilar vertices $w_1,w_2$ 
  that satisfy, for a \LLEEwitness\ of $\achart$, one of the following conditions:
  \begin{enumerate}[label=(C\arabic*)$\:$,itemsep=0.35ex]
    \item{}
      $\lognot{(\bverti{2} \redrtc \bverti{1})} \logand (\descendsinloopto \bverti{1} \,\Longrightarrow\, \text{$\bverti{2}$ is not normed}\,)$,
    \item{}
      $\bverti{2} \loopsbacktotc \bverti{1}$,
    \item{}
      $\existsstzero{\avert\in\verts}
         \bigl(\,
           \bverti{1} \dloopsbackto \avert
             \logand 
           \bverti{2} \loopsbacktotc \avert
         \,\bigr)\logand\lognot(\bverti{2} \redrtci{\bodylab} \bverti{1})$.
  \end{enumerate}
\end{repeatedprop}

\begin{proof}[More supplementary illustrations for the proof of Prop.~\ref{prop:reduced:br} on pages~\pageref{prf:prop:reduced:br:start}--\pageref{prf:prop:reduced:br:end}]
  The proof started from a pair $\cverti{1}$, $\cverti{2}$ of distinct bisimilar vertices. 
  In the case $\sccof{\cverti{1}} = \sccof{\cverti{2}}$, we had the following situation:
  \begin{equation}\label{eq:1:suppl:prf:prop:reduced:br}
     \cverti{1} 
       \loopsbacktortc 
     \averti{1}
       \dloopsbackto 
     \avert
       \convdloopsbackto 
     \averti{2}
       \convloopsbacktortc 
     \cverti{2}
       \logand
     \lognot{( \averti{2} \redrtci{\bodylab} \averti{1} )} \punc{.}
  \end{equation}
  For pairs of vertices $\cverti{1}$ and $\cverti{2}$ 
  such that \eqref{eq:1:suppl:prf:prop:reduced:br} holds, for some $\averti{1}$, $\averti{2}$, and $\avert$,
  we used induction on $\lbsminn{\cverti{1}}$ in order show that 
  $\cverti{1}$ and $\cverti{2}$ progress, via pairs of distinct bisimilar vertices, to
  bisimilar vertices $\bverti{1}$ and $\bverti{2}$ such that one of the conditions~\ref{cond:transf:I}, \ref{cond:transf:II}, or \ref{cond:transf:III} holds. 
  Note that each of \ref{cond:transf:I}, \ref{cond:transf:II}, and \ref{cond:transf:III} implies that $\bverti{1}$ and $\bverti{2}$ are distinct.
  
  In order to carry out the induction step we used a case distinction. Below we repeat the arguments, and supplement them 
  with illustrations. 
  
  \begin{enumerate}[label={Case~\arabic{*}:},leftmargin=*,align=right,labelsep=1ex,itemsep=0.5ex] 
    \item
      $\cverti{2}\loopsbacktotc \averti{2}$.\smallskip
        
      Since $\cverti{2}\ult \cverti{2}'$, either $\cverti{2}'=\averti{2}$ or $\averti{2} \descendsinlooptotc\cverti{2}'$. 
      Moreover,  $\sccof{\cverti{2}'}=\sccof{\cverti{2}}=\sccof{\averti{2}}$, so by Lem.~\ref{lem:loop:relations},~\ref{it:descendsinloopto:scc:loopsbackto},
      $\cverti{2}'\loopsbacktortc \averti{2}$.
      Hence,
      $
         \cvertacci{1} 
             \loopsbacktortc 
           \averti{1}
             \dloopsbackto 
           \avert
             \convdloopsbackto 
           \averti{2}
             \convloopsbacktortc 
           \cvertacci{2}
           \logand
         \lognot{( \averti{2} \redrtci{\bodylab} \averti{1} )} 
         $,
     and $\lbsminn{\cverti{1}'}<\lbsminn{\cverti{1}}$. 
     We apply the induction hypothesis to obtain 
     a bisimilar pair $w_1,w_2$ for which \ref{cond:transf:I}, \ref{cond:transf:II}, or \ref{cond:transf:III} holds. 
     In the illustration below, we drew both of the two cases in which the transition $\cverti{2} \red \cvertacci{2}$ 
     is a \loopentry\ transition, or a \bodytransition, from $\cverti{2}$.
     \begin{center}
       $
       \begin{aligned}[c]
         \scalebox{1}{\begin{tikzpicture}[scale=1,every node/.style={transform shape}]
%
\matrix[anchor=center,row sep=0.8cm,column sep=1.75cm,every node/.style={draw,thick,circle,minimum width=2.5pt,fill,inner sep=0pt,outer sep=2pt}] at (0,0) {
                   &  \node(v){};
  \\[0.25cm]
  \node(v_1){};    &  \node[draw=none,fill=none](h0){};
                                   &  \node(v_2){};  
  \\
  \node(v_11){};   &               &  \node(v_12){};
  \\[0.25cm]
  \node(v_n1){};   &               &  \node(v_n2){};
  \\
  \node(u_1){};    &  \node[draw=none,fill=none](h){}; 
                                   &  \node(u_2){};
  \\
};
\calcLength(v,h0){mylen};

\path (v) ++ (0pt,{0.25*\mylen pt}) node{$\avert$};
  \path (v) ++ ({-0.5*\mylen pt},{-0.4*\mylen pt}) node[draw,thick,circle,minimum width=2.5pt,fill,inner sep=0pt,outer sep=2pt] (v_01){};
    \draw[->,thick] (v) to (v_01);
    \draw[->>,bend right,distance={0.35*\mylen pt}] (v_01) to (v_1);
  \path (v) ++ ({0.5*\mylen pt},{-0.4*\mylen pt}) node[draw,thick,circle,minimum width=2.5pt,fill,inner sep=0pt,outer sep=2pt] (v_03){}; 
    \draw[->,thick] (v) to (v_03);
    \draw[->>,bend left,distance={0.35*\mylen pt}] (v_03) to (v_2);
   %
   %
%
\path (v_1) ++ ({-0.35*\mylen pt},{0.05*\mylen pt}) node{$\averti{1}$};
\draw[->>,out=170,in=220,distance={0.5*\mylen pt}] (v_11) to (v_1);

\path (v_1) ++ ({0.25*\mylen pt},{-0.4*\mylen pt}) node[draw,thick,circle,minimum width=2.5pt,fill,inner sep=0pt,outer sep=2pt](v_1_v_11){};
  \draw[->,thick] (v_1) to (v_1_v_11);
  \draw[->>,distance={0.25*\mylen pt},out=-30,in=20] (v_1_v_11) to (v_11);
\draw[->>,out=135,in=180,distance={1*\mylen pt}] (v_1) to (v);          

%
\draw[-,dotted,thick,shorten <={0.2*\mylen pt},shorten >={0.2*\mylen pt}] (v_11) to (v_n1);

\path (v_n1) ++ ({0.25*\mylen pt},{-0.4*\mylen pt}) node[draw,thick,circle,minimum width=2.5pt,fill,inner sep=0pt,outer sep=2pt](v_n1_u_1){};
  \draw[->,thick] (v_n1) to (v_n1_u_1);
  \draw[->>,distance={0.25*\mylen pt},out=-30,in=20] (v_n1_u_1) to (u_1);

\draw[->,thick] (v_11) to ($(v_11) + ({-0.25*\mylen pt},{-0.25*\mylen pt})$);
\draw[->,thick] (v_11) to ($(v_11) + ({0.25*\mylen pt},{-0.25*\mylen pt})$);

%
\draw[<<-,out=190,in=90,distance={0.25*\mylen pt}] (v_11) to ($(v_11) + ({-0.4*\mylen pt},{-0.3*\mylen pt})$);

\draw[->>,out=160,in=270,distance={0.25*\mylen pt}] (v_n1) to ($(v_n1) + ({-0.35*\mylen pt},{0.4*\mylen pt})$);

%
\path (v_2) ++ ({0.35*\mylen pt},{0.05*\mylen pt}) node{$\averti{2}$};
\draw[->>,out=10,in=-40,distance={0.5*\mylen pt}] (v_12) to (v_2);
\draw[->>,out=45,in=0,distance={1*\mylen pt}] (v_2) to (v);

\path (v_2) ++ ({-0.25*\mylen pt},{-0.4*\mylen pt}) node[draw,thick,circle,minimum width=2.5pt,fill,inner sep=0pt,outer sep=2pt](v_2_v_12){};
  \draw[->,thick] (v_2) to (v_2_v_12);
  \draw[->>,distance={0.25*\mylen pt},out=210,in=160] (v_2_v_12) to (v_12);

%
\draw[-,dotted,thick,shorten <={0.2*\mylen pt},shorten >={0.2*\mylen pt}] (v_12) to (v_n2);
\draw[->,thick] (v_12) to ($(v_12) + ({-0.25*\mylen pt},{-0.25*\mylen pt})$);
\draw[->,thick] (v_12) to ($(v_12) + ({0.25*\mylen pt},{-0.25*\mylen pt})$);

\draw[<<-,out=-10,in=90,distance={0.25*\mylen pt}] (v_12) to ($(v_12) + ({0.4*\mylen pt},{-0.3*\mylen pt})$); 

\draw[->>,shorten <= {0.1 *\mylen pt},shorten >={0.1 *\mylen pt}] (v_2) to node[pos=0.5,sloped]{$\small {/}$} 
                                                                          node[pos=0.15,yshift={-0.15 * \mylen pt}]{$\scriptstyle \bodylab$}(v_1);

\path (u_1) ++ ({0*\mylen pt},{-0.225*\mylen pt}) node[forestgreen]{$\cverti{1}$};
   \path (u_1) ++ ({-0.4*\mylen pt},{0.4*\mylen pt}) node[draw,thick,circle,minimum width=2.5pt,fill,inner sep=0pt,outer sep=2pt] (u'_1){};
     \path(u'_1) ++ ({-0.25*\mylen pt},0pt) node{$\colorred{\cvertacci{1}}$};

\draw[->,out=180,in=270,distance={0.25*\mylen pt},red] (u_1) to node[pos=0.55,below]{$\scriptstyle\slbs$} (u'_1); 
\draw[->>,out=90,in=180,distance={0.25*\mylen pt}] (u'_1) to (v_n1);

\path (u_2) ++ ({0.075*\mylen pt},{0.175*\mylen pt}) node[forestgreen]{$\cverti{2}$};
   \path (u_2) ++ ({0.4*\mylen pt},{0.35*\mylen pt}) node[draw,thick,circle,minimum width=2.5pt,fill,inner sep=0pt,outer sep=2pt] (u'_2_1){};
     \path(u'_2_1) ++ ({0.25*\mylen pt},0pt) node{$\colorred{\cvertacci{2}}$};
   \path (u_2) ++ ({0*\mylen pt},{-0.55*\mylen pt}) node[draw,thick,circle,minimum width=2.5pt,fill,inner sep=0pt,outer sep=2pt] (u'_2_2){};
     \path(u'_2_2) ++ ({0*\mylen pt},{-0.3*\mylen pt}) node{$\colorred{\cvertacci{2}}$};
     \draw[->>,out=-25,in=-40,distance={0.5*\mylen pt}] (u'_2_2) to (u_2);
     
\path (v_n2) ++ ({-0.25*\mylen pt},{-0.4*\mylen pt}) node[draw,thick,circle,minimum width=2.5pt,fill,inner sep=0pt,outer sep=2pt](v_n2_u_2){};
  \draw[->,thick] (v_n2) to (v_n2_u_2);
  \draw[->>,distance={0.25*\mylen pt},out=210,in=160] (v_n2_u_2) to (u_2);

\draw[->>,out=20,in=270,distance={0.25*\mylen pt}] (v_n2) to ($(v_n2) + ({0.35*\mylen pt},{0.4*\mylen pt})$);

\draw[->,thick,red] (u_2) to node[left,pos=0.3,xshift={0.05*\mylen pt}]{$\scriptstyle \loopnsteplab{\aLname} $}(u'_2_2);

\draw[->,out=0,in=270,distance={0.25*\mylen pt},red] (u_2) to node[below,pos=0.7]{$\scriptstyle \bodylab$} (u'_2_1); 
\draw[->>,out=90,in=0,distance={0.25*\mylen pt}] (u'_2_1) to (v_n2);

\draw[-,thick,magenta,densely dashed] 
  (u_1) to node[pos=0.5](mid){} 
           node[pos=0.65](left){} (u_2);

\draw[-,thick,magenta,densely dashed,out=25,in=145,distance={1.25*\mylen pt}] 
  (u'_1) to node[pos=0.5,above,sloped,black]{use ind.\ hyp.} 
            node[pos=0.562](left_1){}
                                                             (u'_2_1);
  \draw[-implies,double equal sign distance,thick] (left) to (left_1);

\draw[-,thick,magenta,densely dashed,out=-10,in=190,distance={1.5*\mylen pt}] 
  (u'_1) to node[below,pos=0.65,sloped,black]{use ind.\ hyp.} 
            node[pos=0.75](left-2){} (u'_2_2);
  \draw[-implies,double equal sign distance,thick] (left) to (left-2);   



\end{tikzpicture}}
       \end{aligned}
        $
     \end{center}   
        
      \vspace{-2.75ex}
    \item
      $\cverti{2}=\averti{2}$.
      
      \begin{enumerate}[label={Case~2.\arabic{*}:},align=right,labelsep=1ex,itemsep=0.5ex] 
        \item
          $\cverti{2}\loopnstepto{\alpha}\cverti{2}'$.
          
          Then either $\cverti{2}'=\cverti{2}$ or $\cverti{2} \descendsinloopto\cverti{2}'$. Moreover,  $\sccof{\cverti{2}'}=\sccof{\cverti{2}}$, 
          so by Lem.~\ref{lem:loop:relations},~\ref{it:descendsinloopto:scc:loopsbackto},
          $\cverti{2}'\loopsbacktortc \cverti{2}$, and hence
          $\cverti{2}'\loopsbacktortc \averti{2}$. 
          \vspace*{-.25mm}Thus we have obtained
          $
             \cvertacci{1} 
                 \loopsbacktortc 
               \averti{1}
                 \dloopsbackto 
               \avert
                 \convdloopsbackto 
               \averti{2}
                 \convloopsbacktortc 
               \cvertacci{2}
               \logand
             \lognot{( \averti{2} \redrtci{\bodylab} \averti{1} )} 
             $.
         Due to $\lbsminn{\cverti{1}'}<\lbsminn{\cverti{1}}$,
         we can apply the induction hypothesis again.
         %
         %
         \begin{equation*}
             \begin{aligned}[c]
               \scalebox{1}{\begin{tikzpicture}[scale=1,every node/.style={transform shape}]
%
\matrix[anchor=center,row sep=0.8cm,column sep=1.75cm,every node/.style={draw,thick,circle,minimum width=2.5pt,fill,inner sep=0pt,outer sep=2pt}] at (0,0) {
                   &  \node(v){};
  \\[0.25cm]
  \node(v_1){};    &  \node[draw=none,fill=none](h0){};
                                   &  \node(v_2){};  
  \\
  \node(v_11){};   &               &  \node[draw=none,fill=none](v_12){};
  \\[0.25cm]
  \node(v_n1){};   &               &  \node[draw=none,fill=none](v_n2){};
  \\
  \node(u_1){};    &  \node[draw=none,fill=none](h){}; 
                                   &  \node[draw=none,fill=none](u_2){};
  \\
};
\calcLength(v,h0){mylen};

\path (v) ++ (0pt,{0.25*\mylen pt}) node{$\avert$};
  \path (v) ++ ({-0.5*\mylen pt},{-0.4*\mylen pt}) node[draw,thick,circle,minimum width=2.5pt,fill,inner sep=0pt,outer sep=2pt] (v_01){};
    \draw[->,thick] (v) to (v_01);
    \draw[->>,bend right,distance={0.35*\mylen pt}] (v_01) to (v_1);
  \path (v) ++ ({0.5*\mylen pt},{-0.4*\mylen pt}) node[draw,thick,circle,minimum width=2.5pt,fill,inner sep=0pt,outer sep=2pt] (v_03){}; 
    \draw[->,thick] (v) to (v_03);
    \draw[->>,bend left,distance={0.35*\mylen pt}] (v_03) to (v_2);
   %
   %
%
\path (v_1) ++ ({-0.35*\mylen pt},{0.05*\mylen pt}) node{$\averti{1}$};
\draw[->>,out=170,in=220,distance={0.5*\mylen pt}] (v_11) to (v_1);

\path (v_1) ++ ({0.25*\mylen pt},{-0.4*\mylen pt}) node[draw,thick,circle,minimum width=2.5pt,fill,inner sep=0pt,outer sep=2pt](v_1_v_11){};
  \draw[->,thick] (v_1) to (v_1_v_11);
  \draw[->>,distance={0.25*\mylen pt},out=-30,in=20] (v_1_v_11) to (v_11);
\draw[->>,out=135,in=180,distance={1*\mylen pt}] (v_1) to (v);          

%
\draw[-,dotted,thick,shorten <={0.2*\mylen pt},shorten >={0.2*\mylen pt}] (v_11) to (v_n1);

\path (v_n1) ++ ({0.25*\mylen pt},{-0.4*\mylen pt}) node[draw,thick,circle,minimum width=2.5pt,fill,inner sep=0pt,outer sep=2pt](v_n1_u_1){};
  \draw[->,thick] (v_n1) to (v_n1_u_1);
  \draw[->>,distance={0.25*\mylen pt},out=-30,in=20] (v_n1_u_1) to (u_1);

\draw[->,thick] (v_11) to ($(v_11) + ({-0.25*\mylen pt},{-0.25*\mylen pt})$);
\draw[->,thick] (v_11) to ($(v_11) + ({0.25*\mylen pt},{-0.25*\mylen pt})$);

%
\draw[<<-,out=190,in=90,distance={0.25*\mylen pt}] (v_11) to ($(v_11) + ({-0.4*\mylen pt},{-0.3*\mylen pt})$);


\draw[->>,out=160,in=270,distance={0.25*\mylen pt}] (v_n1) to ($(v_n1) + ({-0.35*\mylen pt},{0.4*\mylen pt})$);    

%
\path (v_2) ++ ({0.15*\mylen pt},{0.05*\mylen pt}) node[right]{$\averti{2} = \forestgreen{\cverti{2}} = \colorred{\cvertacci{2}} $};
\draw[->,red,thick,out=-80,in=-30,distance={0.9*\mylen pt}] (v_2) to node[pos=0.5,below,yshift={0.05*\mylen pt}]{$\scriptstyle\loopnsteplab{\aLname}$} (v_2);
\draw[->>,out=45,in=0,distance={1*\mylen pt}] (v_2) to (v);

\draw[->>,shorten <= {0.1 *\mylen pt},shorten >={0.1 *\mylen pt}] (v_2) to node[pos=0.5,sloped]{$\small {/}$} 
                                                                          node[pos=0.15,yshift={-0.15 * \mylen pt}]{$\scriptstyle \bodylab$}(v_1);

\path (u_1) ++ ({0.025*\mylen pt},{-0.225*\mylen pt}) node{\forestgreen{$\cverti{1}$}};
   \path (u_1) ++ ({-0.4*\mylen pt},{0.4*\mylen pt}) node[draw,thick,circle,minimum width=2.5pt,fill,inner sep=0pt,outer sep=2pt] (u'_1){};
     \path(u'_1) ++ ({-0.25*\mylen pt},0pt) node{$\colorred{\cvertacci{1}}$};

\draw[->,out=180,in=270,distance={0.25*\mylen pt},red] (u_1) to node[below,pos=0.4,yshift={0.04*\mylen pt}]{$\scriptstyle \slbs$} (u'_1); 
\draw[->>,out=90,in=180,distance={0.25*\mylen pt}] (u'_1) to (v_n1);

\draw[-,thick,magenta,densely dashed,out=0,in=260,distance={2*\mylen pt}] 
  (u_1) to node[pos=0.5](mid){} 
           node[pos=0.65](left){} (v_2);

\draw[-,thick,magenta,densely dashed,out=25,in=215,distance={1*\mylen pt}] 
  (u'_1) to node[pos=0.45,above,sloped,black,xshift={0.05*\mylen pt},yshift={-0.05*\mylen pt}]{\small use ind.\ hyp.} 
            node[pos=0.45](mid_1){}
                                                             (v_2);
  
\draw[-implies,double equal sign distance,thick] (mid) to (mid_1);

\end{tikzpicture}}
             \end{aligned}
             \hspace*{0ex} 
             \begin{aligned}[c]
               \scalebox{1}{\begin{tikzpicture}[scale=1,every node/.style={transform shape}]
%
\matrix[anchor=center,row sep=0.8cm,column sep=1.75cm,every node/.style={draw,thick,circle,minimum width=2.5pt,fill,inner sep=0pt,outer sep=2pt}] at (0,0) {
                   &  \node(v){};
  \\[0.25cm]
  \node(v_1){};    &  \node[draw=none,fill=none](h0){};
                                   &  \node(v_2){};  
  \\
  \node(v_11){};   &               &  \node(v_12){};
  \\[0.25cm]
  \node(v_n1){};   &               &  \node(v_n2){};
  \\
  \node(u_1){};    &  \node[draw=none,fill=none](h){}; 
                                   &  \node(u_2){};
  \\
};
\calcLength(v,h0){mylen};

\path (v) ++ (0pt,{0.25*\mylen pt}) node{$\avert$};
  \path (v) ++ ({-0.5*\mylen pt},{-0.4*\mylen pt}) node[draw,thick,circle,minimum width=2.5pt,fill,inner sep=0pt,outer sep=2pt] (v_01){};
    \draw[->,thick] (v) to (v_01);
    \draw[->>,bend right,distance={0.35*\mylen pt}] (v_01) to (v_1);
  \path (v) ++ ({0.5*\mylen pt},{-0.4*\mylen pt}) node[draw,thick,circle,minimum width=2.5pt,fill,inner sep=0pt,outer sep=2pt] (v_03){}; 
    \draw[->,thick] (v) to (v_03);
    \draw[->>,bend left,distance={0.35*\mylen pt}] (v_03) to (v_2);
   %
   %
%
\path (v_1) ++ ({-0.35*\mylen pt},{0.05*\mylen pt}) node{$\averti{1}$};
\draw[->>,out=170,in=220,distance={0.5*\mylen pt}] (v_11) to (v_1);

\path (v_1) ++ ({0.25*\mylen pt},{-0.4*\mylen pt}) node[draw,thick,circle,minimum width=2.5pt,fill,inner sep=0pt,outer sep=2pt](v_1_v_11){};
  \draw[->,thick] (v_1) to (v_1_v_11);
  \draw[->>,distance={0.25*\mylen pt},out=-30,in=20] (v_1_v_11) to (v_11);
\draw[->>,out=135,in=180,distance={1*\mylen pt}] (v_1) to (v);          

%
\draw[-,dotted,thick,shorten <={0.2*\mylen pt},shorten >={0.2*\mylen pt}] (v_11) to (v_n1);

\path (v_n1) ++ ({0.25*\mylen pt},{-0.4*\mylen pt}) node[draw,thick,circle,minimum width=2.5pt,fill,inner sep=0pt,outer sep=2pt](v_n1_u_1){};
  \draw[->,thick] (v_n1) to (v_n1_u_1);
  \draw[->>,distance={0.25*\mylen pt},out=-30,in=20] (v_n1_u_1) to (u_1);

\draw[->,thick] (v_11) to ($(v_11) + ({-0.25*\mylen pt},{-0.25*\mylen pt})$);
\draw[->,thick] (v_11) to ($(v_11) + ({0.25*\mylen pt},{-0.25*\mylen pt})$);

%
\draw[<<-,out=190,in=90,distance={0.25*\mylen pt}] (v_11) to ($(v_11) + ({-0.4*\mylen pt},{-0.3*\mylen pt})$);

\draw[->>,out=160,in=270,distance={0.25*\mylen pt}] (v_n1) to ($(v_n1) + ({-0.35*\mylen pt},{0.4*\mylen pt})$);

%
\path (v_2) ++ ({0.15*\mylen pt},{0*\mylen pt}) node[right]{$\averti{2} = \forestgreen{\cverti{2}} $};
\draw[->>,out=60,in=280,distance={0.25*\mylen pt}] (v_12) to (v_2);
\draw[->>,out=45,in=0,distance={1*\mylen pt}] (v_2) to (v);

\path (v_2) ++ ({-0.25*\mylen pt},{-0.4*\mylen pt}) node[draw,thick,circle,minimum width=2.5pt,fill,inner sep=0pt,outer sep=2pt](v_2_v_12){};
  \draw[->,thick] (v_2) to (v_2_v_12);
  \draw[->>,distance={0.25*\mylen pt},out=210,in=160] (v_2_v_12) to (v_12);

\draw[->,thick,red,out=-45,in=25,distance={1.25*\mylen pt}] (v_2) to node[pos=0.5,right,xshift={-0.05*\mylen pt}]{$\scriptstyle \loopnsteplab{\aLname}$} (u_2);

%
\draw[-,dotted,thick,shorten <={0.2*\mylen pt},shorten >={0.2*\mylen pt}] (v_12) to (v_n2);
\draw[->,thick] (v_12) to ($(v_12) + ({-0.25*\mylen pt},{-0.25*\mylen pt})$);
\draw[->,thick] (v_12) to ($(v_12) + ({0.25*\mylen pt},{-0.25*\mylen pt})$);

\draw[<<-,out=-10,in=90,distance={0.25*\mylen pt}] (v_12) to ($(v_12) + ({0.4*\mylen pt},{-0.3*\mylen pt})$); 

\draw[->>,shorten <= {0.1 *\mylen pt},shorten >={0.1 *\mylen pt}] (v_2) to node[pos=0.5,sloped]{$\small {/}$} 
                                                                          node[pos=0.15,yshift={-0.15 * \mylen pt}]{$\scriptstyle \bodylab$}(v_1);

\path (u_1) ++ ({0*\mylen pt},{-0.225*\mylen pt}) node[forestgreen]{$\cverti{1}$};
   \path (u_1) ++ ({-0.4*\mylen pt},{0.4*\mylen pt}) node[draw,thick,circle,minimum width=2.5pt,fill,inner sep=0pt,outer sep=2pt] (u'_1){};
     \path(u'_1) ++ ({-0.25*\mylen pt},{0.25*\mylen pt}) node{$\colorred{\cvertacci{1}}$};

\draw[->,out=180,in=270,distance={0.25*\mylen pt},red] (u_1) to node[pos=0.55,below]{$\scriptstyle\slbs$} (u'_1); 
\draw[->>,out=90,in=180,distance={0.25*\mylen pt}] (u'_1) to (v_n1);

\path (u_2) ++ ({0*\mylen pt},{-0.25*\mylen pt}) node[forestgreen]{$\colorred{\cvertacci{2}}$};
     
%
\path (v_n2) ++ ({-0.25*\mylen pt},{-0.4*\mylen pt}) node[draw,thick,circle,minimum width=2.5pt,fill,inner sep=0pt,outer sep=2pt](v_n2_u_2){};
  \draw[->,thick] (v_n2) to (v_n2_u_2);
  \draw[->>,distance={0.25*\mylen pt},out=210,in=160] (v_n2_u_2) to (u_2);     
\draw[->>,out=60,in=280,distance={0.25*\mylen pt}] (u_2) to (v_n2);     
   
\draw[->>,out=20,in=270,distance={0.25*\mylen pt}] (v_n2) to ($(v_n2) + ({0.35*\mylen pt},{0.4*\mylen pt})$);

\draw[-,thick,magenta,densely dashed,out=110,in=215,distance={1*\mylen pt}] 
  (u_1) to node[pos=0.5](mid){} 
           node[pos=0.6](left){} (v_2);

\draw[-,thick,magenta,densely dashed,out=225,in=195,distance={1.75*\mylen pt}] 
  (u'_1) to node[pos=0.5](mid'){} 
           node[pos=0.7](left'){} 
           node[pos=0.7,below,sloped,black]{\small use ind.\ hyp.} (u_2);

\draw[-implies,double equal sign distance,thick] (mid) to (left');

\end{tikzpicture}}
             \end{aligned} 
         \end{equation*} 
         
       \item 
          $\cverti{2}\redi{\bodylab}\cverti{2}'$.\smallskip
          
          Then $\neg(\averti{2} \redrtci{\bodylab} \averti{1})$ together with
          $\averti{2}=\cverti{2}\redi{\bodylab}\cverti{2}'$ and
          $\cverti{1}' \redrtci{\bodylab} \averti{1}$
          (because $\cverti{1}'\loopsbacktortc \averti{1}$) imply $\cverti{1}'\neq \cverti{2}'$.
          We distinguish two cases.
          
          \vspace{0.5ex}
          \begin{enumerate}[label={Case~2.2.\arabic{*}:},align=right,labelsep=1ex,itemsep=0.5ex] 
            \item 
              $\cverti{2}'=\avert$.\smallskip
              
              Then $\cverti{1}'\loopsbacktortc \averti{1}\dloopsbackto v=\cverti{2}'$, i.e., $\cverti{1}'\loopsbacktotc \cverti{2}'$, 
              so we are done, because \ref{cond:transf:II} holds 
              for $\bverti{1} = \cverti{2}'$ and $\bverti{2} = \cverti{1}'$.
              
             \begin{center}   
               $
               \begin{aligned}[c]
                 \begin{tikzpicture}[scale=1,every node/.style={transform shape}]
%
\matrix[anchor=center,row sep=0.8cm,column sep=1.75cm,every node/.style={draw,thick,circle,minimum width=2.5pt,fill,inner sep=0pt,outer sep=2pt}] at (0,0) {
                   &  \node(v){};
  \\[0.25cm]
  \node(v_1){};    &  \node[draw=none,fill=none](h0){};
                                   &  \node(v_2){};  
  \\
  \node(v_11){};   &               &  \node[draw=none,fill=none](v_12){};
  \\[0.25cm]
  \node(v_n1){};   &               &  \node[draw=none,fill=none](v_n2){};
  \\
  \node(u_1){};    &  \node[draw=none,fill=none](h){}; 
                                   &  \node[draw=none,fill=none](u_2){};
  \\
};
\calcLength(v,h0){mylen};

\path (v) ++ (0pt,{0.25*\mylen pt}) node{$\avert = \colorred{\cvertacci{2}}$};
  \path (v) ++ ({-0.5*\mylen pt},{-0.4*\mylen pt}) node[draw,thick,circle,minimum width=2.5pt,fill,inner sep=0pt,outer sep=2pt] (v_01){};
    \draw[->,thick] (v) to (v_01);
    \draw[->>,bend right,distance={0.35*\mylen pt}] (v_01) to (v_1);
  \path (v) ++ ({0.5*\mylen pt},{-0.4*\mylen pt}) node[draw,thick,circle,minimum width=2.5pt,fill,inner sep=0pt,outer sep=2pt] (v_03){}; 
    \draw[->,thick] (v) to (v_03);
    \draw[->>,bend left,distance={0.35*\mylen pt}] (v_03) to (v_2);
   %
   %
%
\path (v_1) ++ ({-0.35*\mylen pt},{0.05*\mylen pt}) node{$\averti{1}$};
\draw[->>,out=170,in=220,distance={0.5*\mylen pt}] (v_11) to (v_1);

\path (v_1) ++ ({0.25*\mylen pt},{-0.4*\mylen pt}) node[draw,thick,circle,minimum width=2.5pt,fill,inner sep=0pt,outer sep=2pt](v_1_v_11){};
  \draw[->,thick] (v_1) to (v_1_v_11);
  \draw[->>,distance={0.25*\mylen pt},out=-30,in=20] (v_1_v_11) to (v_11);
\draw[->>,out=135,in=180,distance={1*\mylen pt}] (v_1) to (v);          

%
\draw[-,dotted,thick,shorten <={0.2*\mylen pt},shorten >={0.2*\mylen pt}] (v_11) to (v_n1);

\path (v_n1) ++ ({0.25*\mylen pt},{-0.4*\mylen pt}) node[draw,thick,circle,minimum width=2.5pt,fill,inner sep=0pt,outer sep=2pt](v_n1_u_1){};
  \draw[->,thick] (v_n1) to (v_n1_u_1);
  \draw[->>,distance={0.25*\mylen pt},out=-30,in=20] (v_n1_u_1) to (u_1);

\draw[->,thick] (v_11) to ($(v_11) + ({-0.25*\mylen pt},{-0.25*\mylen pt})$);
\draw[->,thick] (v_11) to ($(v_11) + ({0.25*\mylen pt},{-0.25*\mylen pt})$);

%
\draw[<<-,out=190,in=90,distance={0.25*\mylen pt}] (v_11) to ($(v_11) + ({-0.4*\mylen pt},{-0.3*\mylen pt})$);


\draw[->>,out=160,in=270,distance={0.25*\mylen pt}] (v_n1) to ($(v_n1) + ({-0.35*\mylen pt},{0.4*\mylen pt})$);    

%
\path (v_2) ++ ({0.15*\mylen pt},{0*\mylen pt}) node[right]{$\averti{2} = \forestgreen{\cverti{2}}$};
\draw[->,out=45,in=0,distance={1*\mylen pt},red] (v_2) to node[right,pos=0.25,xshift={-0.05*\mylen pt}]{$\scriptstyle \bodylab$} (v);

\draw[->>,shorten <= {0.1 *\mylen pt},shorten >={0.1 *\mylen pt}] (v_2) to node[pos=0.5,sloped]{$\small {/}$} 
                                                                          node[pos=0.15,yshift={-0.15 * \mylen pt}]{$\scriptstyle \bodylab$}(v_1);

\path (u_1) ++ ({0.025*\mylen pt},{-0.225*\mylen pt}) node{\forestgreen{$\cverti{1}$}};
   \path (u_1) ++ ({-0.4*\mylen pt},{0.4*\mylen pt}) node[draw,thick,circle,minimum width=2.5pt,fill,inner sep=0pt,outer sep=2pt] (u'_1){};
     \path(u'_1) ++ ({-0.25*\mylen pt},0pt) node{$\colorred{\cvertacci{1}}$};

\draw[->,out=180,in=270,distance={0.25*\mylen pt},red] (u_1) to node[below,pos=0.4,yshift={0.04*\mylen pt}]{$\scriptstyle \slbs$} (u'_1); 
\draw[->>,out=90,in=180,distance={0.25*\mylen pt}] (u'_1) to (v_n1);

\draw[-,thick,magenta,densely dashed,out=0,in=260,distance={2*\mylen pt}] 
  (u_1) to node[pos=0.5](mid){} 
           node[pos=0.65](left){} (v_2);

\draw[-,thick,magenta,densely dashed,out=25,in=260,distance={1.75*\mylen pt}] 
  (u'_1) to node[pos=0.45,above,sloped,black]{\small use ind.\ hyp.} 
            node[pos=0.53,right,black]{\small \ref{cond:transf:II}} 
            node[pos=0.45](mid_1){}
                                                             (v);
  
\draw[-implies,double equal sign distance,thick] (mid) to (mid_1);

\end{tikzpicture}
               \end{aligned}  
               $
             \end{center}   
             
             \medskip
              
           \item 
             $\cverti{2}'\neq \avert$.\smallskip 
              
             By Lem.~\ref{lem:loop:relations},~\ref{it:descendsinloopto:scc:loopsbackto}, $\cverti{2}'\loopsbacktotc \avert$.
             Hence, $\cverti{2}'\loopsbacktortc \avertacci{2}\dloopsbackto \avert$ for some $\avertacci{2}$.
             Since $\averti{2}=\cverti{2}\redi{\bodylab}\cverti{2}'\loopsbacktortc \avertacci{2}$ and
             $\neg(\averti{2} \redrtci{\bodylab} \averti{1})$, it follows that
             $\neg(\avertacci{2} \redrtci{\bodylab} \averti{1})$.
             So
             $ \cvertacci{1} 
                  \loopsbacktortc 
                \averti{1}
                  \dloopsbackto 
                \avert
                  \convdloopsbackto 
                \avertacci{2}
                  \convloopsbacktortc 
                \cvertacci{2}
                   \logand
                \lognot{( \averti{2} \redrtci{\bodylab} \avertacci{1} )}$.
             Due to $\lbsminn{\cverti{1}'}<\lbsminn{\cverti{1}}$,
             we can apply the induction hypothesis again.
             \begin{center}   
               $
               \begin{aligned}[c]
                 \begin{tikzpicture}[every node/.style={transform shape}]
\useasboundingbox (-3.25,2.45) rectangle (3.2,-3);
\matrix[anchor=center,row sep=0.8cm,column sep=1cm,every node/.style={draw,thick,circle,minimum width=2.5pt,fill,inner sep=0pt,outer sep=2pt}] at (0,0) {
                   &  \node(v){};
  \\[0.25cm]
  \node(v_1){};    &  \node[draw=none,fill=none](h0){};
                                   &  \node(v_2){};     &[1.75ex] & \node(v'_2){}; 
  \\
  \node(v_11){};   &               &  \node[draw=none,fill=none](v_12){};    
                                                        & & \node(v'_12){};
  \\[0.25cm]
  \node(v_n1){};   &               &  \node[draw=none,fill=none](v_n2){};
                                                        & & \node(v'_n2){};
  \\
  \node(u_1){};    &  \node[draw=none,fill=none](h){}; 
                                   &  \node[draw=none,fill=none](u_2){};
                                                        & & \node(u'_2){};
  \\
};
\calcLength(v,h0){mylen};

\path (v) ++ (0pt,{0.25*\mylen pt}) node{$\avert$};
  \path (v) ++ ({-0.5*\mylen pt},{-0.4*\mylen pt}) node[draw,thick,circle,minimum width=2.5pt,fill,inner sep=0pt,outer sep=2pt] (v_01){};
    \draw[->,thick] (v) to (v_01);
    \draw[->>,bend right,distance={0.35*\mylen pt}] (v_01) to (v_1);
  \path (v) ++ ({0.5*\mylen pt},{-0.4*\mylen pt}) node[draw,thick,circle,minimum width=2.5pt,fill,inner sep=0pt,outer sep=2pt] (v_03){}; 
    \draw[->,thick] (v) to (v_03);
    \draw[->>,bend left,distance={0.35*\mylen pt}] (v_03) to (v_2);
   %
   %
%
\path (v_1) ++ ({-0.35*\mylen pt},{0.05*\mylen pt}) node{$\averti{1}$};
\draw[->>,out=170,in=220,distance={0.5*\mylen pt}] (v_11) to (v_1);

\path (v_1) ++ ({0.25*\mylen pt},{-0.4*\mylen pt}) node[draw,thick,circle,minimum width=2.5pt,fill,inner sep=0pt,outer sep=2pt](v_1_v_11){};
  \draw[->,thick] (v_1) to (v_1_v_11);
  \draw[->>,distance={0.25*\mylen pt},out=-30,in=20] (v_1_v_11) to (v_11);
\draw[->>,out=135,in=180,distance={1*\mylen pt}] (v_1) to (v);          

%
\draw[-,dotted,thick,shorten <={0.2*\mylen pt},shorten >={0.2*\mylen pt}] (v_11) to (v_n1);

\path (v_n1) ++ ({0.25*\mylen pt},{-0.4*\mylen pt}) node[draw,thick,circle,minimum width=2.5pt,fill,inner sep=0pt,outer sep=2pt](v_n1_u_1){};
  \draw[->,thick] (v_n1) to (v_n1_u_1);
  \draw[->>,distance={0.25*\mylen pt},out=-30,in=20] (v_n1_u_1) to (u_1);

\draw[->,thick] (v_11) to ($(v_11) + ({-0.25*\mylen pt},{-0.25*\mylen pt})$);
\draw[->,thick] (v_11) to ($(v_11) + ({0.25*\mylen pt},{-0.25*\mylen pt})$);

%
\draw[<<-,out=190,in=90,distance={0.25*\mylen pt}] (v_11) to ($(v_11) + ({-0.4*\mylen pt},{-0.3*\mylen pt})$);

\draw[->>,out=160,in=270,distance={0.25*\mylen pt}] (v_n1) to ($(v_n1) + ({-0.35*\mylen pt},{0.4*\mylen pt})$);

%
\path (v_2) ++ ({0.1*\mylen pt},{0.05*\mylen pt}) node[right]{$\averti{2} = \forestgreen{\cverti{2}} $};
\draw[->>,shorten <= {0.1 *\mylen pt},shorten >={0.1 *\mylen pt}] 
  (v_2) to node[pos=0.5,sloped]{$\small {/}$} 
           node[pos=0.15,yshift={-0.125 * \mylen pt}]{$\scriptstyle \bodylab$} (v_1);            
\draw[->,out=-25,in=165,distance={1*\mylen pt},red]
  (v_2) to node[left,pos=0.3,xshift={0.075*\mylen pt}]{$\scriptstyle \bodylab$} (u'_2);                                                                          
\draw[->>,out=45,in=0,distance={1*\mylen pt}] (v_2) to (v);
\draw[->,thick] (v_2) to ($(v_2) + ({-0.15*\mylen pt},{-0.35*\mylen pt})$);
\draw[->,thick] (v_2) to ($(v_2) + ({0.15*\mylen pt},{-0.35*\mylen pt})$);

\path (u_1) ++ ({-0*\mylen pt},{-0.25*\mylen pt}) node[forestgreen]{$\cverti{1}$};
   \path (u_1) ++ ({-0.4*\mylen pt},{0.4*\mylen pt}) node[draw,thick,circle,minimum width=2.5pt,fill,inner sep=0pt,outer sep=2pt] (u'_1){};
     \path(u'_1) ++ ({-0.2*\mylen pt},{0.15*\mylen pt}) node{$\colorred{\cvertacci{1}}$};

\draw[->,out=180,in=270,distance={0.25*\mylen pt},red] (u_1) to node[below,yshift={0.04*\mylen pt}]{$\scriptstyle\slbs$} (u'_1); 
\draw[->>,out=90,in=180,distance={0.25*\mylen pt}] (u'_1) to (v_n1); 


%
\path (v'_2) ++ ({0.1*\mylen pt},{0.05*\mylen pt}) node[right]{$\avertacci{2}$};
\draw[->>,out=10,in=-40,distance={0.5*\mylen pt}] (v'_12) to (v'_2);
\draw[->>,out=75,in=0,distance={0.8*\mylen pt}] (v'_2) to (v);
\draw[->>,out=135,in=20,distance={1*\mylen pt}] 
   (v'_2) to node[pos=0.175,sloped]{$\small {/}$} 
             node[pos=0.1,yshift={-0.15 * \mylen pt}]{$\scriptstyle \bodylab$} (v_1); 

\path (v'_2) ++ ({-0.25*\mylen pt},{-0.4*\mylen pt}) node[draw,thick,circle,minimum width=2.5pt,fill,inner sep=0pt,outer sep=2pt](v'_2_v'_12){};
  \draw[->,thick] (v'_2) to (v'_2_v'_12);
  \draw[->>,distance={0.25*\mylen pt},out=210,in=160] (v'_2_v'_12) to (v'_12);

%
\draw[-,dotted,thick,shorten <={0.2*\mylen pt},shorten >={0.2*\mylen pt}] (v'_12) to (v'_n2);
\draw[->,thick] (v'_12) to ($(v'_12) + ({-0.25*\mylen pt},{-0.25*\mylen pt})$);
\draw[->,thick] (v'_12) to ($(v'_12) + ({0.25*\mylen pt},{-0.25*\mylen pt})$);

%
\draw[<<-,out=-10,in=90,distance={0.25*\mylen pt}] (v'_12) to ($(v'_12) + ({0.4*\mylen pt},{-0.3*\mylen pt})$);

%
\path (v'_n2) ++ ({-0.25*\mylen pt},{-0.4*\mylen pt}) node[draw,thick,circle,minimum width=2.5pt,fill,inner sep=0pt,outer sep=2pt](v'_n2_u'_2){};
  \draw[->,thick] (v'_n2) to (v'_n2_u'_2);
  \draw[->>,distance={0.25*\mylen pt},out=0,in=80] (v'_n2_u'_2) to (u'_2);     
     
\draw[->>,out=20,in=270,distance={0.25*\mylen pt}] (v'_n2) to ($(v'_n2) + ({0.35*\mylen pt},{0.4*\mylen pt})$);

%
\path (u'_2) ++ ({0*\mylen pt},{-0.275*\mylen pt}) node[red]{$\cvertacci{2}$};
\draw[->>,out=10,in=-40,distance={0.5*\mylen pt}] (u'_2) to (v'_n2);


\draw[-,thick,magenta,densely dashed,out=0,in=270,distance={1.5*\mylen pt}] 
  (u_1) to node[pos=0.4](mid){} (v_2);

\draw[-,thick,magenta,densely dashed,out=230,in=195,distance={1.75*\mylen pt}] 
  (u'_1) to node[pos=0.7](mid'){} 
            node[pos=0.7,below,sloped,black]{\small use ind.\ hyp.} (u'_2);

\draw[-implies,double equal sign distance,thick] (mid) to (mid');

%
%

\end{tikzpicture}
               \end{aligned}  
               $
             \end{center} 
             
          \end{enumerate}
      \end{enumerate}    
  \end{enumerate}
\end{proof}

\begin{repeatedprop}[= Proposition~\ref{prop:LEEshape:preserve:conds}]
  Let $\achart$ be a \LLEEchart.
  If a pair $\pair{\bverti{1}}{\bverti{2}}$ of vertices 
  satisfies \ref{cond:transf:I}$\!$, \ref{cond:transf:II}$\!$, or \ref{cond:transf:III}$\!$\vspace*{-0.5mm}
  with respect to a \LLEEwitness\ of $\achart$,\vspace*{-0.5mm} 
  then\/ $\connthroughin{\achart}{\bverti{1}}{\bverti{2}}$ 
  is a \mbox{\LLEEchart}.
\end{repeatedprop}  

\noindent
As background for the proof of this proposition, 
we first give examples why conditions \ref{cond:transf:I}, \ref{cond:transf:II}, and \ref{cond:transf:III} 
cannot be readily relaxed or changed. 
These examples showcase that, far from being artificial,  
the conditions \ref{cond:transf:I}, \ref{cond:transf:II}, and \ref{cond:transf:III} mark sharp borders between whether, 
on a given \LLEEwitness, a connect-through operation is possible while preserving LLEE, or not.
Thus these examples demonstrate that a further simplification 
of the case analysis provided by Proposition~7.3 is not readily possible,
with an eye towards LLEE-structure preserving connect-through operations. 
Therefore a substantial further improvement of our stepwise collapse procedure appears unlikely.

For convenience, the pictures in these examples neglect action labels on transitions. 

\begin{example}[= Example~\ref{ex:trans-I}]
  To show that in \ref{cond:transf:I} it is crucial that $\bverti{1}$ does not loop back,
  we refer back to the \LLEEwitness\ $\acharthat$ in Ex.~\ref{ex:trans-I}.
  There $\protect\lognot{(\bverti{2} \protect\redrtc \bverti{1})}$,
  but \ref{cond:transf:I} is not satisfied by the pair $w_1,w_2$ because $\bverti{1}\loopsbackto\bverthati{1}$.
  Since in $\acharthat$ the levels of loop-entry transitions that descend to $\bverti{1}$ are higher than the loop levels that descend from $\bverti{2}$, 
  the preprocessing step of transformation~I is void.
  We observed that the \connectthroughchart{\bverti{1}}{\bverti{2}} $\protect\connthroughin{\achart}{\bverti{1}}{\bverti{2}}$ on the left in Ex.~\ref{ex:trans-I} has no LLEE-witness.
  The bisimilar pair $w_1,w_2$ in $\acharthat$ progresses to the bisimilar pair $\bverthati{1},\bverthati{2}$, for which \ref{cond:transf:I} holds.
  Since $\connthroughin{\acharthat}{\bverthati{1}}{\bverthati{2}}$ on the right of Ex.~\ref{ex:trans-I}
  is obtained by applying transformation~I to this pair,
  it is guaranteed to be a LEE-witness; this will be argued in the proof of Prop.~\ref{prop:LEEshape:preserve:conds}. 
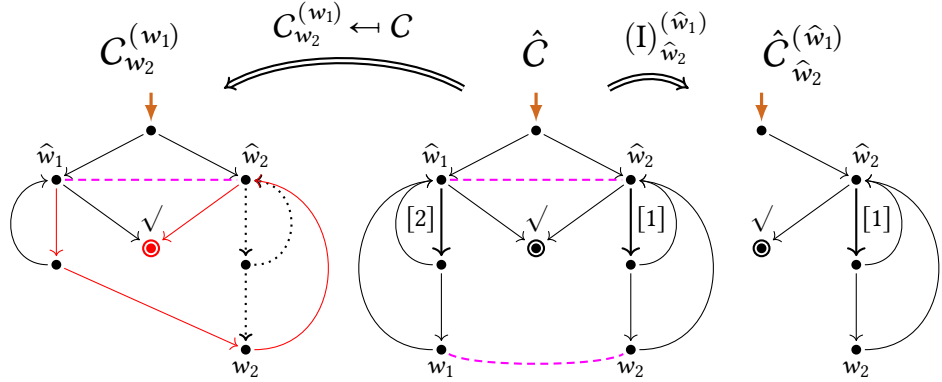
   \begin{center}
  %
\begin{center}
\begin{tikzpicture}[scale=0.875]
\matrix[anchor=center,row sep=1cm,column sep=1.15cm,every node/.style={draw,very thick,circle,minimum width=2.5pt,fill,inner sep=0pt,outer sep=2pt}] at (-5.85,0) {
  \node(0){};   & \node[draw=none,fill=none](root-anchor){};                    
                                       & \node(1){};
  \\
  \node(00){};  & \node[draw=none,fill=none](sink-anchor){};  
                                       & \node(10){};
  \\
  \node(000)[draw=none,fill=none]{}; 
                & \node[draw=none,fill=none](label-anchor){};  
                                       & \node(100){};
  \\
};
\calcLength(0,00){mylen};
\path (root-anchor) ++ (0cm,0.75cm) node[style={draw,very thick,circle,minimum width=2.5pt,fill,inner sep=0pt,outer sep=2pt}](root){};
\path (sink-anchor) ++ (0cm,0.25cm) node[style={draw,very thick,circle,minimum width=2.5pt,fill,inner sep=0pt,outer sep=2pt,red}](sink){};
\draw[thick,red] (sink) circle (0.12cm); 
%

%
%
\draw[<-,very thick,>=latex,chocolate](root) -- ++ (90:0.575cm);
\path (root) ++ (-0.15cm,1.25cm) node{\scalebox{1.45}{{$\connthroughin{\achart}{\bverti{1}}{\bverti{2}}$}}};
\draw[->] (root) to node[below,xshift=0.1cm]{} (0); 
\draw[->] (root) to node[below,xshift=-0.1cm]{} (1);
\path (0) ++ (-0.09cm,0.05cm) node[above]{$\bverthati{1}$};
\draw[->,red] (0) to node[left,pos=0.45,xshift=0.06cm]{
} (00);
\draw[->,shorten >= 2pt] (0) to (sink);
\draw[->,red] (00) to node[left,darkmagenta]{} (100);
\draw[->,distance=0.75cm,out=180,in=185] (00) to (0);
\path (1) ++ (0.15cm,0.05cm) node[above]{$\bverthati{2}$};
\draw[-{>[length=1mm,width=1.8mm]},thick,dotted] (1) to node[right,pos=0.45,xshift=-0.06cm]{
} (10);
\draw[->,shorten >= 2pt,red] (1) to (sink);
\path (sink) ++ (0cm,0.45cm) node{$\tick$};
\draw[-{>[length=1mm,width=1.8mm]},thick,dotted] (10) to (100);
\draw[-{>[length=1mm,width=1.8mm]},thick,dotted,distance=0.75cm,out=0,in=-5] (10) to (1);
\draw[->,distance=1.5cm,out=0,in=-5,red] (100) to (1);
\path (100) ++ (0cm,-0.3cm) node{$\bverti{2}$};

\draw[magenta,thick,densely dashed] (0) to (1);
\matrix[anchor=center,row sep=1cm,column sep=1.15cm,every node/.style={draw,very thick,circle,minimum width=2.5pt,fill,inner sep=0pt,outer sep=2pt}] at (0,0) {
  \node(C_0){};   & \node[draw=none,fill=none](C_root-anchor){};
                                       & \node(C_1){};
  \\
  \node(C_00){};  & \node[draw=none,fill=none](C_sink-anchor){};  
                                       & \node(C_10){};
  \\
  \node(C_000){}; & \node[draw=none,fill=none](C_label-anchor){};  
                                       & \node(C_100){};
  \\
};
\path (C_root-anchor) ++ (0cm,0.75cm) node[style={draw,very thick,circle,minimum width=2.5pt,fill,inner sep=0pt,outer sep=2pt}](C_root){};
\path (C_root) ++ (0cm,1.25cm) node{\scalebox{1.45}{{$\acharthat$}}};
\path (C_sink-anchor) ++ (0cm,0.25cm) node[style={draw,very thick,circle,minimum width=2.5pt,fill,inner sep=0pt,outer sep=2pt}](C_sink){};
\draw[thick] (C_sink) circle (0.12cm); 
\path (C_sink) ++ (0cm,0.45cm) node{$\tick$};
%
%
%
%
\draw[<-,very thick,>=latex,chocolate](C_root) -- ++ (90:0.575cm);
\draw[->] (C_root) to node[below,xshift=0.1cm]{} (C_0);
\draw[->] (C_root) to node[below,xshift=-0.1cm]{} (C_1);
\draw[->,thick] (C_0) to node[left,pos=0.45,xshift=0.06cm]{$\loopsteplabof{2}$} (C_00);
\draw[->,shorten >= 2pt] (C_0) to (C_sink);
\path (C_0) ++ (-0.09cm,0.05cm) node[above]{$\bverthati{1}$};
\draw[->] (C_00) to node[left]{} (C_000);
\draw[->,distance=0.75cm,out=180,in=185] (C_00) to (C_0);
\draw[->,distance=1.5cm,out=180,in=185] (C_000) to (C_0);
\path (C_000) ++ (0cm,-0.3cm) node{$\bverti{1}$};
\draw[->,thick] (C_1) to node[right,pos=0.45,xshift=-0.06cm]{$\loopsteplabof{1}$} (C_10);
\draw[->,shorten >= 2pt] (C_1) to (C_sink);
\path (C_1) ++ (0.15cm,0.05cm) node[above]{$\bverthati{2}$};
\draw[->] (C_10) to (C_100);
\draw[->,distance=0.75cm,out=0,in=-5] (C_10) to (C_1);
\draw[->,distance=1.5cm,out=0,in=-5] (C_100) to (C_1);
\path (C_100) ++ (0cm,-0.3cm) node{$\bverti{2}$};
\draw[magenta,thick,densely dashed] (C_0) to (C_1);
\draw[magenta,thick,densely dashed,bend right,distance=0.4cm,looseness=1] (C_000) to (C_100);

\matrix[anchor=center,row sep=1cm,column sep=1.15cm,every node/.style={draw,very thick,circle,minimum width=2.5pt,fill,inner sep=0pt,outer sep=2pt}] at (3.45,0) {
  \node[draw=none,fill=none](Cw1hatw2hat_0){};   & \node[draw=none,fill=none](Cw1hatw2hat_root-anchor){};                    
                                       & \node(Cw1hatw2hat_1){};
  \\
  \node[draw=none,fill=none](Cw1hatw2hat_00){};  & \node[draw=none,fill=none](Cw1hatw2hat_sink-anchor){};  
                                       & \node(Cw1hatw2hat_10){};
  \\
  \node[draw=none,fill=none](Cw1hatw2hat_000)[draw=none,fill=none]{}; 
                & \node[draw=none,fill=none](Cw1hatw2hat_label-anchor){};  
                                       & \node(Cw1hatw2hat_100){};
  \\
};
\path (Cw1hatw2hat_root-anchor) ++ (0cm,0.75cm) node[style={draw,very thick,circle,minimum width=2.5pt,fill,inner sep=0pt,outer sep=2pt}](Cw1hatw2hat_root){};
\path (Cw1hatw2hat_sink-anchor) ++ (0cm,0.25cm) node[style={draw,very thick,circle,minimum width=2.5pt,fill,inner sep=0pt,outer sep=2pt}](Cw1hatw2hat_sink){};
\draw[thick] (Cw1hatw2hat_sink) circle (0.12cm); 
\path (Cw1hatw2hat_sink) ++ (0cm,0.45cm) node{$\tick$};
%
%
%
\draw[<-,very thick,>=latex,chocolate](Cw1hatw2hat_root) -- ++ (90:0.575cm);
\draw[->] (Cw1hatw2hat_root) to node[below,xshift=-0.1cm]{} (Cw1hatw2hat_1);
\path (Cw1hatw2hat_root) ++ (0.65cm,1.15cm) node{\scalebox{1.45}{{$\connthroughin{\acharthat}{\bverthati{1}}{\bverthati{2}}$}}};
\draw[->,thick] (Cw1hatw2hat_1) to node[right,pos=0.45,xshift=-0.06cm]{$\loopsteplabof{1}$} (Cw1hatw2hat_10);
\draw[->,shorten >= 2pt] (Cw1hatw2hat_1) to (Cw1hatw2hat_sink);
\path (Cw1hatw2hat_1) ++ (0.15cm,0.05cm) node[above]{$\bverthati{2}$};
\draw[->] (Cw1hatw2hat_10) to (Cw1hatw2hat_100);
\draw[->,distance=0.75cm,out=0,in=-5] (Cw1hatw2hat_10) to (Cw1hatw2hat_1);
\draw[->,distance=1.5cm,out=0,in=-5] (Cw1hatw2hat_100) to (Cw1hatw2hat_1);
\path (Cw1hatw2hat_100) ++ (0cm,-0.3cm) node{$\bverti{2}$};

\draw[-implies,thick,double equal sign distance, bend right,distance={1.75*\mylen pt},
               shorten <= 1cm,shorten >= 1cm
               ] (C_root) to node[above,pos=0.5,yshift={0.15*\mylen pt}] {\scalebox{1.25}{$\connthroughin{\achart}{\bverti{1}}{\bverti{2}}\! \mapsfrom \achart$}}  
                                                                         (root) ;

\draw[-implies,thick,double equal sign distance, bend left,distance={1.25*\mylen pt},
               shorten <= 1cm,shorten >= 1cm
               ] (C_root) to node[above,pos=0.6,yshift={0.15*\mylen pt}]{\scalebox{1.25}{$(\text{\nf I})^{(\bverthati{1})}_{\bverthati{2}}$}} (Cw1hatw2hat_root) ;

\end{tikzpicture}
\end{center}
  \end{center}
\end{example}

To avoid the creation of body step cycles in transformation II, it would seem expedient to connect transitions to $\bverti{2}$ through to $\bverti{1}$, 
since \ref{cond:transf:II}, $\bverti{2}\loopsbacktotc\bverti{1}$, rules out the existence of a path $\bverti{1}\,\sredtci{\bodylab}\,\bverti{2}$ in $\acharthat$.
(Instead, transitions to $\bverti{1}$ are connected through to $\bverti{2}$, and resulting body step cycles are eliminated by turning the body transitions at $\bverthati{2}$ into loop-entry transitions.)
However, connecting transitions to $\bverti{2}$ through to $\bverti{1}$ may produce a chart for which no LLEE-witness exists.

\begin{example}[= Example~\ref{ex:trans-II}]
  For the \LLEEchart~$\achart$ with \LLEEwitness~$\acharthat$ below in the middle, the \connectthroughchart{\bverti{2}}{\bverti{1}} $\connthroughin{\achart}{\bverti{2}}{\bverti{1}}$ on the left does not have a \LLEEwitness:
  it has no loop subchart, because from each of its three vertices an infinite path starts that does not return to this vertex. From $\bverthati{2}$ this path, drawn in red, cycles between $\cvert$ and $\bverti{1}$.
  Transformation~II applied to the pair $\bverti{1},\bverti{2}$ (instead of $\bverti{2},\bverti{1}$) in $\acharthat$ yields
  the \entrybodylabeling~$\connthroughin{\acharthat}{\bverti{1}}{\bverti{2}}$ 
  for the \connectthroughchart{\bverti{1}}{\bverti{2}}
  with additionally $\bverthati{2}\,\sredi{\bodylab}\,\bverti{2}$ turned into $\bverthati{2}\,\sredi{\loopnsteplab{2}}\,\bverti{2}$. 
  Since the pair $\bverti{1},\bverti{2}$ satisfies \ref{cond:transf:II}, the proof of Prop.~\ref{prop:LEEshape:preserve:conds} 
  guarantees that this \entrybodylabeling, drawn on the right, is a \LLEEwitness.
  \begin{center}
  %
\begin{center}
\begin{tikzpicture}
\matrix[anchor=center,row sep=1cm,column sep=1cm,every node/.style={draw,very thick,circle,minimum width=2.5pt,fill,inner sep=0pt,outer sep=2pt}] at (0,0) {
  \node(w1_Cw2w1){};                          & &[0.75cm] & & \node(w1_C){};     & &[0.75cm] & & \node(label_Cw1w2)[draw=none,fill=none]{};                                             
  \\
  \node(w2hat_Cw2w1){};                       & &         & & \node(w2hat_C){};  &           & & & \node(w2hat_Cw1w2){};     
  \\
  \node(u_Cw2w1){};                           & &         & & \node(u_C){};      &           & & & \node(u_Cw1w2){};    
  \\
  \node(label_Cw2w1)[draw=none,fill=none]{};  & &         & & \node(w2_C){};     &           & & & \node(w2_Cw1w2){};
  \\
};

\draw[-implies,thick,double equal sign distance, 
               shorten <= 1.15cm,shorten >= 1.55cm
               ] (w2hat_C) to node[above,pos=0.4425,yshift={0*\mylen pt}]{\scalebox{1.25}{$\connthroughin{\achart}{\bverti{2}}{\bverti{1}}\! \mapsfrom \achart$}} (w2hat_Cw2w1) ;

\draw[-implies,thick,double equal sign distance, 
               shorten <= 1.9cm,shorten >= 1.4cm
               ] (w2hat_C) to node[above,pos=0.55,yshift={0*\mylen pt}]{\scalebox{1.25}{$(\text{\nf II})^{(\bverti{1})}_{\bverti{2}}$}} (w2hat_Cw1w2) ;

\draw[<-,very thick,>=latex,chocolate](w1_Cw2w1) -- ++ (90:0.5cm);
\path (w1_Cw2w1) ++ (0.35cm,0.175cm) node{$\bverti{1}$};
\draw[->](w1_Cw2w1) to node[right,xshift=-0.025cm,yshift=0.05cm]{
} (w2hat_Cw2w1);
\draw[->,distance=1.5cm,out=-5,in=5,color=red] (w1_Cw2w1) to node[above,xshift=-0cm,yshift=0.15cm,
color=red,
pos=0.275]{
} (u_Cw2w1);
\path (w2hat_Cw2w1) ++ (0.325cm,0cm) node{$\bverthati{2}$};
\draw[->,color=red
](w2hat_Cw2w1) to node[right,xshift=-0.025cm,yshift=0.05cm]{
} (u_Cw2w1);
\draw[->,distance=0.75cm,out=165,in=185,shorten <= 4.5pt] (w2hat_Cw2w1) to (w1_Cw2w1);
\path (u_Cw2w1) ++ (0cm,-0.275cm) node{$\cvert$};
\draw[->,distance=0.75cm,out=175,in=185] (u_Cw2w1) to (w2hat_Cw2w1);
\draw[->,distance=1.5cm,out=175,in=185
,color=red
] (u_Cw2w1) to (w1_Cw2w1);
\path(label_Cw2w1) ++ (0cm,0cm) node{\scalebox{1.45}{$\connthroughin{\achart}{\bverti{2}}{\bverti{1}}$}};

\draw[<-,very thick,>=latex,chocolate](w1_C) -- ++ (90:0.5cm);
\path (w1_C) ++ (0.35cm,0.175cm) node{$\bverti{1}$};
\draw[->,thick](w1_C) to node[right,xshift=-0.025cm,yshift=0.05cm]{$\loopnsteplab{2}$} (w2hat_C);
\draw[->,thick,distance=1.5cm,out=-5,in=5] (w1_C) to node[above,xshift=-0cm,yshift=0.15cm,pos=0.275]{$\loopnsteplab{2}$} (u_C);
\path (w2hat_C) ++ (0.325cm,0cm) node{$\bverthati{2}$};
\draw[->,thick](w2hat_C) to node[right,xshift=-0.025cm,yshift=0.05cm]{$\loopnsteplab{1}$} (u_C);
\draw[->,distance=0.75cm,out=165,in=185,shorten <= 4.5pt] (w2hat_C) to (w1_C);
\path (u_C) ++ (0.175cm,-0.25cm) node{$\cvert$};
\draw[->](u_C) to (w2_C);
\draw[->,distance=0.75cm,out=175,in=185] (u_C) to (w2hat_C);
\path (w2_C) ++ (0.35cm,-0.05cm) node{$\bverti{2}$};
\draw[->,distance=1.5cm,out=175,in=185] (w2_C) to (w2hat_C);
\path (w2_C) ++ (1.05cm,0.65cm) node(label_C){\scalebox{1.45}{$\acharthat$}};

\draw[magenta,thick,densely dashed,bend right,distance=1cm,looseness=1] (w1_C) to (w2_C);

%
%

\path (w2hat_Cw1w2) ++ (0cm,0.3cm) node{$\bverthati{2}$};
\draw[->,thick](w2hat_Cw1w2) to node[right,xshift=-0.025cm,yshift=0.05cm]{$\loopnsteplab{1}$} (u_Cw1w2);
\draw[->,thick,distance=1.5cm,out=-5,in=5](w2hat_Cw1w2) to node[above,xshift=0cm,yshift=0.15cm,pos=0.275]{$\loopnsteplab{2}$} (w2_Cw1w2);
\path (u_Cw1w2) ++ (0.25cm,-0cm) node{$\cvert$};
\draw[->] (u_Cw1w2) to (w2_Cw1w2);
\draw[->,distance=0.75cm,out=175,in=185] (u_Cw1w2) to (w2hat_Cw1w2);
\draw[<-,very thick,>=latex,chocolate](w2_Cw1w2) -- ++ (270:0.5cm);
\path (w2_Cw1w2) ++ (0.325cm,-0.3cm) node{$\bverti{2}$};
\draw[->,distance=1.5cm,out=175,in=185] (w2_Cw1w2) to (w2hat_Cw1w2);
\path(label_Cw1w2) ++ (0cm,0cm) node{\scalebox{1.45}{$\connthroughin{\acharthat}{\bverti{1}}{\bverti{2}}$}};

\end{tikzpicture}\label{fig:repeatedex:prop:transformation:II}
\end{center}
  %
  %
  \end{center}  
\end{example}

The following example shows that for transformation III it is essential to select a bisimilar pair $w_1, w_2$ where $w_1$ \emph{directly} loops back to $v$.

\begin{example}[= Example~\ref{ex:trans-III}]
  In the \LLEEwitness~$\acharthat$ below in the middle, $w_1,w_2\loopsbacktotc v$, and there is no body step path from $w_2$ to $w_1$, 
  but \ref{cond:transf:III} does not hold for the pair $w_1,w_2$ because $\lognot(w_1\dloopsbackto v)$.
  \vspace*{-.5mm}All loop-entry transitions from $v$ have the same loop label, so the preprocessing step of transformation III is void.
  The \connectthroughchart{\bverti{1}}{\bverti{2}} $\connthroughin{\achart}{\bverti{1}}{\bverti{2}}$ on the left does not have a \LLEEwitness.
  Namely, the transition from $\bverthati{2}$ can be declared a loop-entry transition, 
  and after its removal also two transitions from $v$ can be declared loop-entry transitions,
  leading to the removal of the five transitions that are depicted as dotted arrows. 
  The remaining chart (of solid arrows) however has no further loop subchart,
  because from each of its vertices an infinite path starts that does not return to this vertex.
  The bisimilar pair $w_1,w_2$ progresses to the bisimilar pair $\bverthati{1},\bverthati{2}$ in $\acharthat$, for which \ref{cond:transf:III} holds because $\bverthati{1} \dloopsbackto \avert
  \convloopsbackto \bverthati{2}$ and $\lognot(\bverthati{2}\,\sredrtci{\bodylab}\,\bverthati{1})$.
  Transformation III applied to this pair yields
  the \entrybodylabeling\ $\protect\connthroughin{\acharthat}{\bverthati{1}}{\bverthati{2}}$ on the right. In the proof of
  Prop.~\ref{prop:LEEshape:preserve:conds} it is argued that this is guaranteed to be a \LLEEwitness.
  The remaining two bisimilar pairs can be eliminated by one or two further applications of transformation III.
  
  
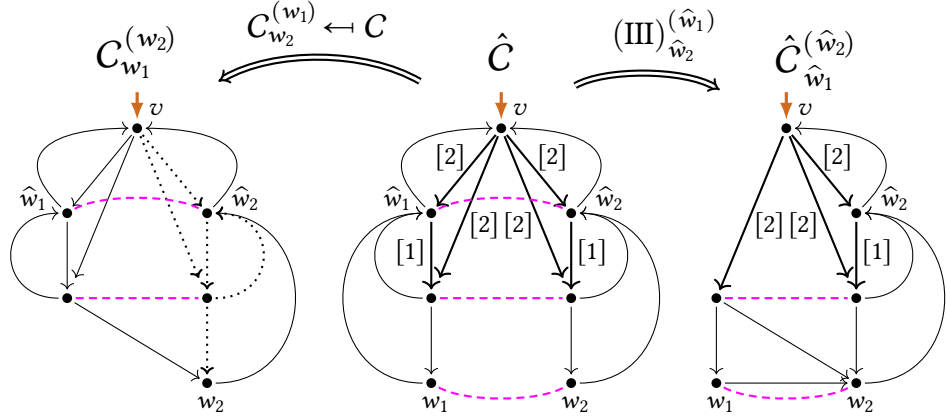
\begin{center}
  %
%
\begin{tikzpicture}
\matrix[anchor=center,row sep=1cm,column sep=0.8cm,every node/.style={draw,very thick,circle,minimum width=2.5pt,fill,inner sep=0pt,outer sep=2pt}] at (0,0) {
                        &  \node(v_Cw1w2){}; &                       &[1.25cm] &                   & \node(v_C){}; &                   &[0.2cm] &                           & \node(v_Chatw1hatw2){}; &
  \\
  \node(hatw1_Cw1w2){}; &                    & \node(hatw2_Cw1w2){}; &        & \node(hatw1_C){}; &               & \node(hatw2_C){}; &       &                           &                         & \node(hatw2_Chatw1hatw2){};
  \\
  \node(u1_Cw1w2){};    &                    & \node(u2_Cw1w2){};    &        & \node(u1_C){};    &               & \node(u2_C){};    &       & \node(u1_Chatw1hatw2){};  &                         & \node(u2_Chatw1hatw2){};
  \\
  \node[draw=none,fill=none]{};
                        &                    & \node(w2_Cw1w2){};    &        & \node(w1_C){};    &               & \node(w2_C){};    &       & \node(w1_Chatw1hatw2){};  &                         & \node(w2_Chatw1hatw2){};
  \\
};
\calcLength(hatw1_C,u1_C){mylen}; 
%
\draw[<-,very thick,>=latex,chocolate](v_Cw1w2) -- ++ (90:{0.425*\mylen pt}); 
\path (v_Cw1w2) ++ ({0.225*\mylen pt},{0.225*\mylen pt}) node{$\avert$};
\path (v_Cw1w2) ++ ({0*\mylen pt},{0.9*\mylen pt}) node{\scalebox{1.45}{$\connthroughin{\achart}{\bverti{2}}{\bverti{1}}$}};
\path (hatw1_Cw1w2) ++ ({-0.4*\mylen pt},{0.175*\mylen pt}) node{$\bverthati{1}$};
\draw[->] (v_Cw1w2) to node[above,pos=0.7,xshift={-0.05*\mylen pt},yshift={0.05*\mylen pt}
]{
} (hatw1_Cw1w2);
\draw[->,shorten >= 5pt
] (v_Cw1w2) to node[right,pos=0.58,xshift={-0.025*\mylen pt},yshift={0.00*\mylen pt}]{
} (u1_Cw1w2);
\draw[-{>[length=1mm,width=1.8mm]},thick,dotted] (v_Cw1w2) to node[above,pos=0.7,xshift={0.05*\mylen pt},yshift={0.05*\mylen pt}]{
} (hatw2_Cw1w2);
\draw[-{>[length=1mm,width=1.8mm]},thick,dotted,shorten >= 5pt] (v_Cw1w2) to node[left,pos=0.58,xshift={0.025*\mylen pt},yshift={0.00*\mylen pt}]{
} (u2_Cw1w2);
\draw[->
] (hatw1_Cw1w2) to node[left,pos=0.45,xshift={0.05*\mylen pt}]{
} (u1_Cw1w2);
\draw[->,distance={0.9*\mylen pt},out=135,in=180] (hatw1_Cw1w2) to (v_Cw1w2);
\draw[->] (u1_Cw1w2) to (w2_Cw1w2);
\draw[->,distance={0.75*\mylen pt},out=180,in=180
] (u1_Cw1w2) to (hatw1_Cw1w2);
\path (hatw2_Cw1w2) ++ ({0.45*\mylen pt},{0.175*\mylen pt}) node{$\bverthati{2}$};
\draw[-{>[length=1mm,width=1.8mm]},thick,dotted] (hatw2_Cw1w2) to node[right,pos=0.45,xshift={-0.05*\mylen pt}]{
} (u2_Cw1w2);
\draw[->,distance={0.9*\mylen pt},out=45,in=0
] (hatw2_Cw1w2) to (v_Cw1w2);
\draw[-{>[length=1mm,width=1.8mm]},thick,dotted] (u2_Cw1w2) to (w2_Cw1w2);
\draw[-{>[length=1mm,width=1.8mm]},thick,dotted,distance={0.75*\mylen pt},out=0,in=0] (u2_Cw1w2) to (hatw2_Cw1w2);
\path (w2_Cw1w2) ++ ({0.05*\mylen pt},{-0.25*\mylen pt}) node{$\bverti{2}$};
\draw[->,distance={1.25*\mylen pt},out=0,in=0
] (w2_Cw1w2) to (hatw2_Cw1w2);

\draw[magenta,thick,densely dashed,bend left,distance={0.35*\mylen pt},looseness=1] (hatw1_Cw1w2) to (hatw2_Cw1w2);
\draw[magenta,thick,densely dashed] (u1_Cw1w2) to (u2_Cw1w2);

\draw[<-,very thick,>=latex,chocolate](v_C) -- ++ (90:{0.425*\mylen pt});   
\path (v_C) ++ ({0.225*\mylen pt},{0.225*\mylen pt}) node{$\avert$};
\path (v_C) ++ ({0*\mylen pt},{0.9*\mylen pt}) node{\scalebox{1.45}{${\acharthat}$}};
\path (hatw1_C) ++ ({-0.4*\mylen pt},{0.175*\mylen pt}) node{$\bverthati{1}$};
\draw[->,thick] (v_C) to node[above,pos=0.7,xshift={-0.05*\mylen pt},yshift={0.05*\mylen pt}]{$\loopnsteplab{2}$} (hatw1_C);
\draw[->,thick,shorten >= 5pt] (v_C) to node[right,pos=0.58,xshift={-0.025*\mylen pt},yshift={0.00*\mylen pt}]{$\loopnsteplab{2}$} (u1_C);
\draw[->,thick] (v_C) to node[above,pos=0.7,xshift={0.05*\mylen pt},yshift={0.05*\mylen pt}]{$\loopnsteplab{2}$}  (hatw2_C);
\draw[->,thick,shorten >= 5pt] (v_C) to node[left,pos=0.58,xshift={0.025*\mylen pt},yshift={0.00*\mylen pt}]{$\loopnsteplab{2}$} (u2_C);
\draw[->,thick] (hatw1_C) to node[left,pos=0.45,xshift={0.05*\mylen pt}]{$\loopnsteplab{1}$} (u1_C);
\draw[->,distance={0.9*\mylen pt},out=135,in=180] (hatw1_C) to (v_C);
\draw[->] (u1_C) to (w1_C);
\draw[->,distance={0.75*\mylen pt},out=180,in=180] (u1_C) to (hatw1_C);
\path (w1_C) ++ ({0.05*\mylen pt},{-0.25*\mylen pt}) node{$\bverti{1}$};
\draw[->,distance={1.25*\mylen pt},out=180,in=180] (w1_C) to (hatw1_C);

\path (hatw2_C) ++ ({0.45*\mylen pt},{0.175*\mylen pt}) node{$\bverthati{2}$};
\draw[->,thick] (hatw2_C) to node[right,pos=0.45,xshift={-0.05*\mylen pt}]{$\loopnsteplab{1}$} (u2_C);
\draw[->,distance={0.9*\mylen pt},out=45,in=0] (hatw2_C) to (v_C);
\draw[->] (u2_C) to (w2_C);
\draw[->,distance={0.75*\mylen pt},out=0,in=0] (u2_C) to (hatw2_C);
\path (w2_C) ++ ({0.05*\mylen pt},{-0.25*\mylen pt}) node{$\bverti{2}$};
\draw[->,distance={1.25*\mylen pt},out=0,in=0] (w2_C) to (hatw2_C);

\draw[magenta,thick,densely dashed,bend left,distance={0.35*\mylen pt},looseness=1] (hatw1_C) to (hatw2_C);
\draw[magenta,thick,densely dashed] (u1_C) to (u2_C);
\draw[magenta,thick,densely dashed,bend right,distance={0.35*\mylen pt},looseness=1] (w1_C) to (w2_C);

\draw[<-,very thick,>=latex,chocolate](v_Chatw1hatw2) -- ++ (90:{0.425*\mylen pt}); 
\path (v_Chatw1hatw2) ++ ({0.225*\mylen pt},{0.225*\mylen pt}) node{$\avert$};
\path (v_Chatw1hatw2) ++ ({0.35*\mylen pt},{0.8*\mylen pt}) node{\scalebox{1.45}{$\connthroughin{\acharthat}{\bverthati{2}}{\bverthati{1}}$}};
\draw[->,thick] (v_Chatw1hatw2) to node[above,pos=0.7,xshift={0.05*\mylen pt},yshift={0.05*\mylen pt}]{$\loopnsteplab{2}$} (hatw2_Chatw1hatw2);
\draw[->,thick,shorten >= 5pt] (v_Chatw1hatw2) to node[left,pos=0.58,xshift={0.025*\mylen pt},yshift={0.00*\mylen pt}]{$\loopnsteplab{2}$} (u2_Chatw1hatw2);
\draw[->,thick] (v_Chatw1hatw2) to node[right,pos=0.58,xshift={-0.025*\mylen pt},yshift={0.00*\mylen pt}]{$\loopnsteplab{2}$} (u1_Chatw1hatw2);
\draw[->] (u1_Chatw1hatw2) to (w1_Chatw1hatw2);
\draw[->] (u1_Chatw1hatw2) to (w2_Chatw1hatw2);
\path (w1_Chatw1hatw2) ++ ({0.05*\mylen pt},{-0.25*\mylen pt}) node{$\bverti{1}$};
\draw[->] (w1_Chatw1hatw2) to (w2_Chatw1hatw2);
\path (hatw2_Chatw1hatw2) ++ ({0.45*\mylen pt},{0.175*\mylen pt}) node{$\bverthati{2}$};
\draw[->,thick] (hatw2_Chatw1hatw2) to node[right,pos=0.45,xshift={-0.05*\mylen pt}]{$\loopnsteplab{1}$} (u2_Chatw1hatw2);
\draw[->,distance={0.9*\mylen pt},out=45,in=0] (hatw2_Chatw1hatw2) to (v_Chatw1hatw2);
\draw[->] (u2_Chatw1hatw2) to (w2_Chatw1hatw2);
\draw[->,distance={0.75*\mylen pt},out=0,in=0] (u2_Chatw1hatw2) to (hatw2_Chatw1hatw2);
\path (w2_Chatw1hatw2) ++ ({0.05*\mylen pt},{-0.25*\mylen pt}) node{$\bverti{2}$};
\draw[->,distance={1.25*\mylen pt},out=0,in=0] (w2_Chatw1hatw2) to (hatw2_Chatw1hatw2);

\draw[magenta,thick,densely dashed] (u1_Chatw1hatw2) to (u2_Chatw1hatw2);
\draw[magenta,thick,densely dashed,bend right,distance={0.35*\mylen pt},looseness=1] (w1_Chatw1hatw2) to (w2_Chatw1hatw2);

\draw[-implies,thick,double equal sign distance, bend right,distance={1.75*\mylen pt},
               shorten <= {1*\mylen pt},shorten >= {1*\mylen pt},
               ] (v_C) to node[above,pos=0.505,yshift={0.15*\mylen pt}] {\scalebox{1.25}{$\connthroughin{\achart}{\bverti{1}}{\bverti{2}}\! \mapsfrom \achart$}} 
                                                                         (v_Cw1w2) ;

\draw[-implies,thick,double equal sign distance, bend left,distance={1.3*\mylen pt},
               shorten <= {0.9*\mylen pt},shorten >= {0.8*\mylen pt},
               ] (v_C) to node[above,pos=0.58,yshift={0.15*\mylen pt}]{\scalebox{1.25}{$(\text{\nf III})^{(\bverthati{1})}_{\bverthati{2}}$}} (v_Chatw1hatw2) ;
\end{tikzpicture}
  %
  %
  %

  \end{center}
\end{example}

The following example shows \ref{cond:transf:III} cannot be weakened
by dropping $\lognot{(\bverti{2} \redrtci{\bodylab} \bverti{1})}$.

\begin{example}
  For the \LLEEwitness~$\acharthat$ \vspace*{-.5mm}below in the middle,
  $\bverti{1} \dloopsbackto \avert
  \convloopsbacktotc \bverti{2}$,
  but there is a body step path from $\bverti{2}$ to $\bverti{1}$. 
  The \connectthroughchart{\bverti{1}}{\bverti{2}} $\connthroughin{\achart}{\bverti{1}}{\bverti{2}}$ on the left does not have a \LLEEwitness,
  because from each of its vertices an infinite path starts that does not return to it.
  The bisimilar pair $\bverti{1},\bverti{2}$ in $\acharthat$
  progresses to the bisimilar pair $\avert,\bverthati{2}$,
  to which \vspace*{-.25mm}transformation~II is applicable because \ref{cond:transf:II} holds: $\bverthati{2}\loopsbackto\avert$.
  In the resulting \LLEEwitness~$\connthroughin{\acharthat}{\avert}{\bverthati{2}}$, second to the right,
  \ref{cond:transf:III} holds for the pair $\bverti{1},\bverti{2}$ because $\bverti{1} \dloopsbackto \bverthati{2} \convloopsbackto \bverti{2}$ and $\lognot(\bverti{2}\,\sredrtci{\bodylab}\,\bverti{1})$.
  Applying transformation III to this pair results in the \LLEEwitness\ on the right.
   
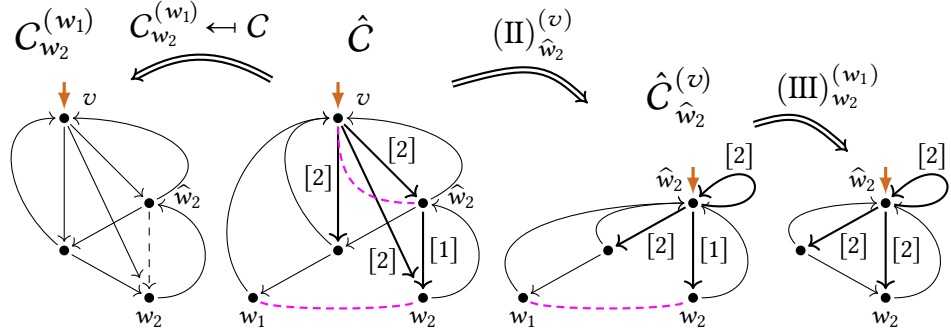
\begin{center}
%
%
\begin{tikzpicture}
\matrix[anchor=center,row sep=0.5cm,column sep=1cm,every node/.style={draw,very thick,circle,minimum width=2.5pt,fill,inner sep=0pt,outer sep=2pt}] {
  \node(v_Cw1w2){}; &                       &[0.25cm]                  & \node(v_C){};  & \node[draw=none,fill=none](dummy_C){};  &[0.2cm]  & \node[draw=none,fill=none](v_Cvw2hat){};  
  \\
  \\
                    & \node(w2hat_Cw1w2){}; &                          &                & \node(w2hat_C){};                       &                       &                     & \node(w2hat_Cvw2hat){};  &[0.3cm]                     & \node(w2hat_C'w1w2){};                            
  \\
  \node(u_Cw1w2){}; &                       &                          & \node(u_C){};  &                                         &                       & \node(u_Cvw2hat){}; &                          &         \node(u_C'w1w2){}; &
  \\
                    & \node(w2_Cw1w2){};    &           \node(w1_C){}; &                & \node(w2_C){};                          &  \node(w1_Cvw2hat){}; &                     & \node(w2_Cvw2hat){};     &                            & \node(w2_C'w1w2){};
  \\    
};
\calcLength(dummy_C,w2hat_C){mylen};   
%
  
\draw[<-,very thick,>=latex,chocolate](v_Cw1w2) -- ++ (90:{0.425*\mylen pt});      
\path (v_Cw1w2) ++ ({-0.1*\mylen pt},{1*\mylen pt}) node{\scalebox{1.45}{$\connthroughin{\achart}{\bverti{1}}{\bverti{2}}$}};
\path (v_Cw1w2) ++ ({0.3*\mylen pt},{0.225*\mylen pt}) node{$\avert$};
\draw[->] (v_Cw1w2) to (u_Cw1w2);
\draw[->] (v_Cw1w2) to (w2hat_Cw1w2);  
\draw[->,shorten >= 5pt] (v_Cw1w2) to (w2_Cw1w2); 
\path (w2hat_Cw1w2) ++ ({0.45*\mylen pt},{0.075*\mylen pt}) node{$\bverthati{2}$};
\draw[->] (w2hat_Cw1w2) to (u_Cw1w2);
\draw[->,dashed] (w2hat_Cw1w2) to (w2_Cw1w2);  
\draw[->,out=27,in=0,distance={1*\mylen pt}] (w2hat_Cw1w2) to (v_Cw1w2);
\draw[->] (u_Cw1w2) to (w2_Cw1w2); 
\draw[->,out=140,in=180,distance={0.8*\mylen pt}] (u_Cw1w2) to (v_Cw1w2);
\path (w2_Cw1w2) ++ ({0*\mylen pt},{-0.25*\mylen pt}) node{$\bverti{2}$};
\draw[->,out=0,in=-10,distance={0.75*\mylen pt}] (w2_Cw1w2) to (w2hat_Cw1w2);

\draw[<-,very thick,>=latex,chocolate](v_C) -- ++ (90:{0.425*\mylen pt});    
\path (v_C) ++ ({0*\mylen pt},{1*\mylen pt}) node[xshift={0.275*\mylen pt}]{\scalebox{1.45}{$\acharthat$}}; 
\path (v_C) ++ ({0.3*\mylen pt},{0.225*\mylen pt}) node{$\avert$};
\draw[->,thick] (v_C) to node[right,pos=0.45,yshift={0.08*\mylen pt}]{$\loopnsteplab{2}$} (w2hat_C);
\draw[->,thick] (v_C) to node[left,pos=0.45,xshift={0.065*\mylen pt}]{$\loopnsteplab{2}$} (u_C);  
\draw[->,thick,shorten >= 5pt] (v_C) to node[left,pos=0.825,xshift={0.025*\mylen pt}]{$\loopnsteplab{2}$} (w2_C); 
\path (w2hat_C) ++ ({0.45*\mylen pt},{0.075*\mylen pt}) node{$\bverthati{2}$};
\draw[->] (w2hat_C) to (u_C);
\draw[->,thick] (w2hat_C) to node[right,xshift={-0.05*\mylen pt}]{$\loopnsteplab{1}$} (w2_C);  
\draw[->,out=27,in=0,distance={1*\mylen pt}] (w2hat_C) to (v_C);
\draw[->] (u_C) to (w1_C); 
\draw[->,out=140,in=180,distance={0.8*\mylen pt}] (u_C) to (v_C);
\path (w1_C) ++ ({0*\mylen pt},{-0.25*\mylen pt}) node{$\bverti{1}$};
\draw[->,out=130,in=180,distance={0.95*\mylen pt}] (w1_C) to (v_C);
\path (w2_C) ++ ({0*\mylen pt},{-0.25*\mylen pt}) node{$\bverti{2}$};
\draw[->,out=0,in=-10,distance={0.75*\mylen pt}] (w2_C) to (w2hat_C);
\draw[magenta,thick,densely dashed,out=270,in=180,distance={0.65*\mylen pt},looseness=1] (v_C) to (w2hat_C);
\draw[magenta,thick,bend right,distance={0.2*\mylen pt},looseness=2,densely dashed] (w1_C) to (w2_C);

\draw[<-,very thick,>=latex,chocolate](w2hat_Cvw2hat) -- ++ (90:{0.425*\mylen pt});   
\path (w2hat_Cvw2hat) ++ ({-0.15*\mylen pt},{1.25*\mylen pt}) node{\scalebox{1.45}{$\connthroughin{\acharthat}{\avert}{\bverthati{2}}$}};
\path (w2hat_Cvw2hat) ++ ({-0.275*\mylen pt},{0.275*\mylen pt}) node{$\bverthati{2}$};
\draw[->,thick] (w2hat_Cvw2hat) to node[below,pos=0.35,yshift={-0.025*\mylen pt}]{$\loopnsteplab{2}$} (u_Cvw2hat);
\draw[->,thick] (w2hat_Cvw2hat) to node[right,xshift={-0.05*\mylen pt}]{$\loopnsteplab{1}$} (w2_Cvw2hat); 
\draw[->,thick,out=-10,in=45,distance={1*\mylen pt}] (w2hat_Cvw2hat) to node[above,pos=0.725,yshift={-0.05*\mylen pt}]{$\loopnsteplab{2}$} (w2hat_Cvw2hat);
\draw[->] (u_Cvw2hat) to (w1_Cvw2hat); 
\draw[->,out=130,in=180,distance={0.5*\mylen pt}] (u_Cvw2hat) to (w2hat_Cvw2hat);
\path (w1_Cvw2hat) ++ ({0*\mylen pt},{-0.25*\mylen pt}) node{$\bverti{1}$};
\draw[->,out=130,in=180,distance={1*\mylen pt}] (w1_Cvw2hat) to (w2hat_Cvw2hat);
\path (w2_Cvw2hat) ++ ({0*\mylen pt},{-0.25*\mylen pt}) node{$\bverti{2}$};
\draw[->,out=0,in=-20,distance={0.75*\mylen pt}] (w2_Cvw2hat) to (w2hat_Cvw2hat);  
\draw[magenta,thick,bend right,distance={0.2*\mylen pt},looseness=2,densely dashed] (w1_Cvw2hat) to (w2_Cvw2hat);

\draw[<-,very thick,>=latex,chocolate](w2hat_C'w1w2) -- ++ (90:{0.425*\mylen pt});
\path (w2hat_C'w1w2) ++ ({-0.275*\mylen pt},{0.275*\mylen pt}) node{$\bverthati{2}$};
\draw[->,thick] (w2hat_C'w1w2) to node[below,pos=0.35,yshift={-0.025*\mylen pt}]{$\loopnsteplab{2}$} (u_C'w1w2);
\draw[->,thick] (w2hat_C'w1w2) to node[right,xshift={-0.05*\mylen pt}]{$\loopnsteplab{2}$} (w2_C'w1w2); 
\draw[->,thick,out=-10,in=45,distance={1*\mylen pt}] (w2hat_C'w1w2) to node[above,pos=0.725,yshift={-0.05*\mylen pt}]{$\loopnsteplab{2}$} (w2hat_C'w1w2);
\draw[->] (u_C'w1w2) to (w2_C'w1w2); 
\draw[->,out=130,in=180,distance={0.5*\mylen pt}] (u_C'w1w2) to (w2hat_C'w1w2);
\path (w2_C'w1w2) ++ ({0*\mylen pt},{-0.25*\mylen pt}) node{$\bverti{2}$};
\draw[->,out=0,in=-20,distance={0.75*\mylen pt}] (w2_C'w1w2) to (w2hat_C'w1w2);

  
\draw[-implies,thick,double equal sign distance, bend right,distance={1.5*\mylen pt},
               shorten <= {0.8*\mylen pt},shorten >= {0.8*\mylen pt},
               ] (v_C) to node[above,pos=0.505,yshift={0.125*\mylen pt}] {\scalebox{1.25}{$\connthroughin{\achart}{\bverti{1}}{\bverti{2}}\! \mapsfrom \achart$}} 
                                                                         (v_Cw1w2) ;

\draw[-implies,thick,double equal sign distance, bend left,distance={2.25*\mylen pt},
               shorten <= {1.3*\mylen pt},shorten >= {1.6*\mylen pt},
               ] (v_C) to node[above,pos=0.5,yshift={0.2*\mylen pt}]{\scalebox{1.25}{$(\text{\nf II})^{(\avert)}_{\bverthati{2}}$}} (w2hat_Cvw2hat) ; 

\draw[-implies,thick,double equal sign distance,bend left,distance={1.25*\mylen pt},looseness=2,
               shorten <= {0.75*\mylen pt},shorten >= {0.6*\mylen pt},
               ] ($(w2hat_Cvw2hat) + ({0*\mylen pt},{0.7*\mylen pt})$) 
                   to node[above,pos=0.7,yshift={0.25*\mylen pt}]{\scalebox{1.25}{$(\text{\nf III})^{(\bverti{1})}_{\bverti{2}}$}} 
                 ($(w2hat_C'w1w2) + ({0*\mylen pt},{0.2*\mylen pt})$); 
\end{tikzpicture}   
%
%
  \end{center}
\end{example}


\begin{proof}[Supplement for the proof of Proposition~\ref{prop:LEEshape:preserve:conds}]
  Let $\acharthat$ be a LLEE-chart. 
  For vertices $\bverti{1}$, $\bverti{2}$ such that \ref{cond:transf:I}, \ref{cond:transf:II}, or \ref{cond:transf:III} holds,
  transformation~I, II, or III, \vspace*{-.25mm}respectively, produces an \entrybodylabeling~$\connthroughin{\acharthat}{\bverti{1}}{\bverti{2}}$.
  In the article submission we have proved for transformation~I that it is a \LLEEwitness.
  Here we do the same for transformations II and III. 
  
  We recall that in the proof in the article submission we have shown that 
  it suffices to show that each of the transformations
  produces, before the final clean-up step, an \entrybodylabeling\ that satisfies the LLEE-conditions 
  with the exception of possible violations of the loop property~\ref{loop:1} in \ref{LLEEw:2}\ref{LLEEw:2a}.

  \smallskip
  
\begin{description}

  \item[Transformation~II:]
    We argue the correctness of transformation~II. 
    Consider vertices~$\bverti{1},\bverti{2}$ such that \ref{cond:transf:II} holds,
    that is, $\bverti{2} \loopsbacktotc \bverti{1}$. 
    Let $\bverthati{2}$ be the $\sdloopsbackto$\nb-pre\-de\-cessor of $\bverti{1}$ in the $\sdloopsbackto$\nb-chain from $\bverti{2}$ to $\bverti{1}$,
    i.e., $\bverti{2} \loopsbacktortc \bverthati{2} \dloopsbackto \bverti{1}$. 
    
    As for the transformations~I and III it suffices to show,
    in view of the alleviation of the proof obligation at the start of the proof on page~\pageref{alleviation:prf:prop:LEEshape:preserve:conds},
    that the intermediate result $\acharthatdacc$ of transformation~II
    before the clean-up step satisfies the LLEE-witness properties, except for possible violations of \ref{loop:1}.
    By the definition of transformation~II, $\acharthatdacc$ results
    by performing the adaptation step~\ref{labels:II} to the chart $\acharthatacc \defdby \connthroughin{\acharthat}{\bverti{1}}{\bverti{2}}$
    that arises from $\acharthat$ by connecting $\bverti{1}$ through to $\bverti{2}$. 
    
    To prove that \ref{LLEEw:1}, and the part concerning~\ref{loop:2} for \ref{LLEEw:2}\ref{LLEEw:2a} is satisfied for~$\acharthatdacc$, 
    it suffices to show that the transformed chart does not contain a cycle of body transitions.
    At first, the step of connecting $\bverti{1}$ through to $\bverti{2}$ in $\acharthat$ may introduce a body step cycle in $\acharthatacc = \connthroughin{\acharthat}{\bverti{1}}{\bverti{2}}$.  
    But every such cycle is removed in the subsequent level adaptation step \ref{labels:II}. 
    Namely, each body step cycle introduced in $\acharthatacc$
    must stem from a transition $\cvert \redi{\bodylab} \bverti{1}$ (which is redirected to $\bverti{2}$ in $\acharthatacc$)
    and a path $\bverti{2} \redrtci{\bodylab} \cvert$ in $\acharthat$, 
    for some $\cvert\neq\bverti{1}$.
    Since $\bverti{2} \loopsbacktortc \bverthati{2} \dloopsbackto \bverti{1}$, 
    by Lem.\ \ref{lem:loopsbackto:channel}, 
    the path $\bverti{2} \redrtci{\bodylab} \cvert \redi{\bodylab} \bverti{1}$ in $\acharthat$ must visit $\bverthati{2}$. 
    Since all body transitions from $\bverthati{2}$ are turned into loop-entry transitions in step~\ref{labels:II}, 
    the body step cycle $\bverti{2} \redrtci{\bodylab} \cvert \redi{\bodylab} \bverti{2}$ in $\acharthatacc$ that was introduced 
    in the connect-through step, is after step~\ref{labels:II} no longer a body step cycle in $\acharthatdacc$.
  
    
    Now we prove that \ref{LLEEw:2}\ref{LLEEw:2b} is preserved by the two steps from $\acharthat$ via $\achartacc = \connthroughin{\acharthat}{\bverti{1}}{\bverti{2}}$
    to $\acharthatdacc$.  
    Every path 
    $\cvert 
       \:\scomprewrels{\sredtavoidsvi{\cvert}{\loopnsteplab{\aLname}}}{\,\scomprewrels{\sredtavoidsvrtci{\cvert}{\bodylab}}{\sredi{\loopnsteplab{\bLname}}}}$ 
    in $\acharthatdacc$ with $\cvert\neq\bverti{1},\bverti{2}$ arises by a, possibly empty, combination of
    the following three kinds of modifications in the first two transformation steps:
    \begin{itemize}
      \item[(i)]
        A transition to $\bverti{1}$ was redirected to $\bverti{2}$ in the connect-through step. 
      \item[(ii)]
        The loop-entry transition at the beginning of the path is from $\bverthati{2}$ and was a body transition before step~\ref{labels:II}, 
        meaning that $\cvert=\bverthati{2}$ and $\aLname=\cLname$.
        (Recall that $\cLname$ is a loop name of maximum loop level among the loop-entries at $\bverti{1}$ in $\acharthat$.)
      \item[(iii)]
        The loop-entry transition at the end of the path is from $\bverthati{2}$ and was a body transition before step~\ref{labels:II}, 
        meaning that $\bLname=\cLname$.
    \end{itemize}
    This gives $2^3=8$ possibilities. Of these, three possibilities are void: if all three adaptations are not the case, 
    the path is already present in $\acharthat$, and so $\ll{\aLname}>\ll{\bLname}$ is guaranteed; 
    (ii) and (iii) together cannot hold, because then the path would return to $\cvert=\bverthati{2}$,
    which it cannot, because all of its steps avoid $\cvert$ as target.
    We now show that in the remaining five cases always $\ll{\aLname}>\ll{\bLname}$.
    Since $\bverti{2}\loopsbacktotc\bverti{1}$,
    there is a path $\bverti{1} \comprewrels{\sredtavoidsvi{\bverti{1}}{\loopnsteplab{\dLname}}}{\sredtavoidsvrtci{\bverti{1}}{\bodylab}} \bverti{2}$ in $\acharthat$. 
    By definition of $\cLname$, $\ll{\cLname}\geq\ll{\dLname}$.\vspace*{-1.5mm}
    \begin{itemize}
      \item[A] Let \vspace*{-1mm}only (i) hold: 
        there are paths $\cvert \comprewrels{\sredtavoidsvi{\cvert}{\loopnsteplab{\aLname}}}{\sredtavoidsvrtci{\cvert}{\bodylab}} \bverti{1}$ 
        and $\bverti{2} \:\scomprewrels{\sredtavoidsvrtci{\cvert}{\bodylab}}{\sredi{\loopnsteplab{\bLname}}}$ in $\acharthat$
        (which do not visit $\bverthati{2}$).
        Then there is a path 
        \vspace*{-2.25mm}%
        $\cvert \comprewrels{\sredtavoidsvi{\cvert}{\loopnsteplab{\aLname}}}{\sredtavoidsvrtci{\cvert}{\bodylab}} \bverti{1} \:\sredi{\loopnsteplab{\cLname}}$ in $\acharthat$, 
        so $\ll{\aLname}>\ll{\cLname}$. We distinguish two cases.
        \vspace{1.5mm}
      
        {\sc Case 1:} 
          The \vspace*{-.5mm}path $\bverti{2} \:\scomprewrels{\sredtavoidsvrtci{\cvert}{\bodylab}}{\sredi{\loopnsteplab{\bLname}}}$ visits $\bverti{1}$. 
          Then there is a path $\bverti{1} \:\scomprewrels{\sredtavoidsvrtc{\cvert}}{\sredi{\loopnsteplab{\bLname}}}$ in $\acharthat$.
          So $\cvert
          \comprewrels{\sredtavoidsvi{\cvert}{\loopnsteplab{\aLname}}}{\sredtavoidsvrtc{\cvert}}
          \bverti{1}
          \scomprewrels{\sredtavoidsvrtci{\cvert}{\bodylab}}{\sredi{\loopnsteplab{\bLname}}}$ in $\acharthat$.
          So by \ref{LLEEw:2}\ref{LLEEw:2b}, $\ll{\aLname}>\ll{\bLname}$.
          \vspace{1.5mm}
      
        {\sc Case 2:} The path $\bverti{2} \:\scomprewrels{\sredtavoidsvrtci{\cvert}{\bodylab}}{\sredi{\loopnsteplab{\bLname}}}$ does not visit $\bverti{1}$.
        \vspace*{-1.25mm}Then there is a path 
          $\bverti{1} 
            \comprewrels{\sredtavoidsvi{\bverti{1}}{\loopnsteplab{\dLname}}}{\sredtavoidsvrtci{\bverti{1}}{\bodylab}}  
           \bverti{2} 
             \:\scomprewrels{\sredtavoidsvrtci{\bverti{1}}{\bodylab}}{\sredi{\loopnsteplab{\bLname}}}$ in $\acharthat$, 
           so $\ll{\dLname}>\ll{\bLname}$.
           Hence $\ll{\aLname}>\ll{\cLname}\geq\ll{\dLname}>\ll{\bLname}$.
        \vspace{1.5mm}

      \item[B]
        Let only (ii) hold. Then $\cvert=\bverthati{2}$, $\aLname=\cLname$, and there is a path 
        $\bverthati{2} \;\scomprewrels{\sredtavoidsvtci{\bverthati{2},\bverti{1}}{\bodylab}}{\sredi{\loopnsteplab{\bLname}}}$ in $\acharthat$.
        \vspace*{-.25mm}As $\bverthati{2}\dloopsbackto\bverti{1}$, 
        there is a path
        \mbox{$\bverti{1} \comprewrels{\sredtavoidsvi{\bverti{1}}{\loopnsteplab{\dLname}}}{\sredtavoidsvrtci{\bverti{1}}{\bodylab}} \bverthati{2}$} in $\acharthat$.
        Hence $\bverti{1} 
           \comprewrels{\sredtavoidsvi{\bverti{1}}{\loopnsteplab{\dLname}}}{\sredtavoidsvrtci{\bverti{1}}{\bodylab}}
         \bverthati{2}
           \:\comprewrels{\sredtavoidsvtci{\bverti{1}}{\bodylab}}{\sredi{\loopnsteplab{\bLname}}}$ in $\acharthat$, 
        so $\ll{\dLname}>\ll{\bLname}$.
        Hence $\ll{\aLname}=\ll{\cLname}\geq\ll{\dLname}>\ll{\bLname}$.
      \vspace{1.5mm}
      
      \item[C]
        Let only (iii) hold. \vspace*{-1mm}Then $\bLname=\cLname$, and  
        $\cvert 
           \comprewrels{\sredtavoidsvi{\cvert,\bverti{1}}{\loopnsteplab{\aLname}}}{\sredtavoidsvrtci{\cvert,\bverti{1}}{\bodylab}}
         \bverthati{2}$ with $\cvert\neq\bverti{1}$ is a path in in $\acharthat$.
        Since $\bverthati{2}\dloopsbackto\bverti{1}$ and $\cvert\neq \bverti{1}$, it follows that $\lognot{(\bverthati{2}\dloopsbackto\cvert)}$.
        So in view of 
        the path $\cvert \comprewrels{\sredtavoidsvi{\cvert}{\loopnsteplab{\aLname}}}{\sredtavoidsvrtci{\cvert}{\bodylab}} \bverthati{2}$, 
        \vspace*{-1mm}there is no path $\bverthati{2} \redrtci{\bodylab} \cvert$ in $\acharthat$.
        Since $\bverthati{2}\dloopsbackto\bverti{1}$, there is a path $\bverthati{2} \redrtci{\bodylab} \bverti{1}$ in $\acharthat$, 
        which by the previous observation is of the form $\bverthati{2} \redtavoidsvrtci{\cvert}{\bodylab} \bverti{1}$.
        \vspace*{-1mm}Hence there is 
        a path $\cvert 
                  \comprewrels{\sredtavoidsvi{\cvert}{\loopnsteplab{\aLname}}}{\sredtavoidsvrtci{\cvert}{\bodylab}}
                \bverthati{2}
                  \redtavoidsvrtci{\cvert}{\bodylab}
                \bverti{1}
                  \:\sredi{\loopnsteplab{\cLname}}$ in $\acharthat$, so $\ll{\aLname}>\ll{\cLname}=\ll{\bLname}$.
      \vspace{1.5mm}
      
      \item[D] 
        Let only (i) and (ii) hold, meaning $\cvert=\bverthati{2}$, $\aLname=\cLname$, and
        there are paths 
        $\bverthati{2} \redtavoidsvtc{\bverthati{2}} \bverti{1}$ 
        and $\bverti{2} \:\scomprewrels{\sredtavoidsvrtci{\bverthati{2}}{\bodylab}}{\sredi{\loopnsteplab{\bLname}}}$ in $\acharthat$.
       \vspace{-.25mm}Since $\bverti{2} \loopsbacktortc \bverthati{2} \dloopsbacktotc \bverti{1}$, 
        and $\cvert=\bverthati{2}$ implies $\bverti{2}\neq\bverthati{2}$, by Lem.\ \ref{lem:loopsbackto:channel}, 
        \vspace*{-.25mm}the path $\bverti{2} \:\scomprewrels{\sredtavoidsvrtci{\bverthati{2}}{\bodylab}}{\sredi{\loopnsteplab{\bLname}}}$ cannot visit $\bverti{1}$.
        Hence $\bverti{1} \comprewrels{\sredtavoidsvi{\bverti{1}}{\loopnsteplab{\dLname}}}{\sredtavoidsvrtci{\bverti{1}}{\bodylab}}
                \bverti{2}
                 \:\scomprewrels{\sredtavoidsvrtci{\bverti{1}}{\bodylab}}{\sredi{\loopnsteplab{\bLname}}}$ in $\acharthat$.
        So $\ll{\dLname}>\ll{\bLname}$.
        Hence $\ll{\aLname}=\ll{\cLname}\geq\ll{\dLname}>\ll{\bLname}$.
      \vspace{1.5mm}
      
      \item[E]
        Let only (i) and (iii) hold. \vspace*{-.25mm}Then $\bLname=\cLname$, and 
        $\cvert \comprewrels{\sredtavoidsvi{\cvert}{\loopnsteplab{\aLname}}}{\sredtavoidsvrtci{\cvert}{\bodylab}} \bverti{1}$ 
        and $\bverti{2} \redtavoidsvrtci{\cvert}{\bodylab} \bverthati{2}$ are paths in $\acharthat$.
        Since $\cvert 
           \comprewrels{\sredtavoidsvi{\cvert}{\loopnsteplab{\aLname}}}{\sredtavoidsvrtci{\cvert}{\bodylab}} 
         \bverti{1}
           \:\sredi{\loopnsteplab{\cLname}}$ in $\acharthat$,
         $\ll{\aLname}>\ll{\cLname}=\ll{\bLname}$.\vspace*{-1mm}
      \end{itemize}
      We conclude that in all five cases, $\acharthatdacc$ satisfies \ref{LLEEw:2}\ref{LLEEw:2b}.\vspace{1.5mm}
      
       Finally we argue that part~\ref{loop:3} of \ref{LLEEw:2}\ref{LLEEw:2a} holds for $\acharthatdacc$, i.e.,
      there are no \txtdescendsinloopto\ paths of the form 
      \vspace*{-0.25mm}
      $\cvert 
         \comprewrels{\sredtavoidsvi{\cvert}{\loopnsteplab{\aLname}}}{\sredtavoidsvrtci{\cvert}{\bodylab}}
        \tick$ 
      in $\acharthatdacc$.  
      We can use part of the argumentation employed for demonstrating \ref{LLEEw:2}\ref{LLEEw:2b} \vspace*{-.25mm}above.
      It was demonstrated in particular that for every \txtdescendsinloopto\ path 
      $\cvert 
         \comprewrels{\sredtavoidsvi{\cvert}{\loopnsteplab{\aLname}}}{\sredtavoidsvrtci{\cvert}{\bodylab}}
        \dvert$ 
      in $\acharthatdacc$,
      there is a \txtdescendsinloopto\ path
      $\cverttilde 
         \comprewrels{\sredtavoidsvi{\cverttilde}{\loopnsteplab{\cLname}}}{\sredtavoidsvrtci{\cverttilde}{\bodylab}}
        \dvert$
      with\vspace*{-1mm} the same target~$\dvert$ in $\acharthat$. 
      From this it follows that if a \txtdescendsinloopto\ path in $\acharthatdacc$ had $\tick$ as target,
      then there were a \txtdescendsinloopto\ path already in $\acharthat$ that had $\tick$ as target,
      violating \ref{loop:3} for the LLEE-chart $\acharthat$.
      Hence $\acharthatdacc$ must satisfy \ref{loop:3}. 
      
      We conclude that the result of transformation~II is  a \LLEEchart.\vspace{2mm}

  \item[Transformation~III:]

  To show the correctness of transformation~III, consider vertices~$\bverti{1}$ and $\bverti{2}$ such that \ref{cond:transf:III} holds.
  Let $\avert$ be such that $\bverti{1} \dloopsbackto \avert \convloopsbacktotc \bverti{2}$.
  We show that its intermediate result $\connthroughin{\acharthat}{\bverti{1}}{\bverti{2}}$ before the clean-up step
  satisfies the LLEE-witness properties, except for possible violations of \ref{loop:1}.
 
  First we show that \ref{LLEEw:2}\ref{LLEEw:2b} is preserved
  by both the level adaptation and the connect-through step.
  A violation arising by the first step, i.e., in $\acharthatacc$, would involve a path
  \vspace*{-0.5mm}%
  $\cvert \comprewrels{\sredtavoidsvi{\cvert}{\loopnsteplab{\aLname}}}{\sredtavoidsvrtci{\cvert}{\bodylab}} \avert
          \redi{\loopnsteplab{\bLname}}$ in $\acharthat$
  where $\bLname$ is increased to a loop label $\cLname$ of maximum level among all loop-entries at $\avert$. 
  But in this way no violation can arise, since there was already a path 
  $\cvert \comprewrels{\sredtavoidsvi{\cvert}{\loopnsteplab{\aLname}}}{\sredtavoidsvrtci{\cvert}{\bodylab}} \avert
          \redi{\loopnsteplab{\cLname}}$ in $\acharthat$, so
  $\ll{\aLname}>\ll{\cLname}\geq\ll{\bLname}$.
  
  Now we exclude violations of \ref{LLEEw:2}\ref{LLEEw:2b}
  in the connect-through step, \vspace*{-.5mm}by showing that in $\connthroughin{\acharthat}{\bverti{1}}{\bverti{2}}$, $\ll{\aLname}>\ll{\bLname}$ for all newly created paths
  \vspace*{-1mm}%
  $\cvert  \comprewrels{\sredtavoidsvi{\cvert}{\loopnsteplab{\aLname}}}{\,\scomprewrels{\sredtavoidsvrtci{\cvert}{\bodylab}}{\sredi{\loopnsteplab{\bLname}}}}$ 
  with $\cvert\neq\bverti{1}$ that stem from paths
  $\cvert  \comprewrels{\sredtavoidsvi{\cvert}{\loopnsteplab{\aLname}}}{\sredtavoidsvrtci{\cvert}{\bodylab}} \bverti{1}$ 
  and $\bverti{2} \comprewrels{\sredtavoidsvrtci{\cvert}{\bodylab}}{\sredi{\loopnsteplab{\bLname}}}$ 
  in $\acharthatacc$.
  As $\bverti{2}\loopsbacktotc\avert$, there is a path
  $\avert \comprewrels{\sredtavoidsvi{\avert}{\loopnsteplab{\cLname}}}{\sredtavoidsvrtci{\avert}{\bodylab}} \bverti{2}$ 
  in $\acharthatacc$. We distinguish two cases.
  
  \noindent
  {\sc Case 1:} $\cvert=\avert$. \vspace*{-2mm}Then, by the level adaptation step, $\aLname=\cLname$. Since $u=v$, 
  there is a path $\avert \comprewrels{\sredtavoidsvi{\avert}{\loopnsteplab{\cLname}}}{\sredtavoidsvrtci{\avert}{\bodylab}} \bverti{2}
                          \:\scomprewrels{\sredtavoidsvrtci{\avert}{\bodylab}}{\sredi{\loopnsteplab{\bLname}}}$ in $\acharthatacc$. 
  By \ref{LLEEw:2}\ref{LLEEw:2b} for $\acharthatacc$, $\ll{\cLname}>\ll{\bLname}$.\vspace{1.5mm}
  
  \noindent
  {\sc Case 2:} $\cvert\neq\avert$. Since $\bverti{1}\dloopsbackto\avert$, 
  there is a path $\bverti{1} \redtci{\bodylab} \avert$ in $\acharthat$ and thus in $\acharthatacc$. Suppose, toward a contradiction, that this path visits $\cvert$.
  Then
  \vspace*{-1.5mm}
  $\cvert \comprewrels{\sredtavoidsvi{\cvert}{\loopnsteplab{\aLname}}}{\sredtavoidsvrtci{\cvert}{\bodylab}} \bverti{1}
          \redtci{\bodylab} \cvert$, 
  so $\bverti{1}\loopsbackto\cvert$ in $\acharthatacc$ and thus in $\acharthat$. 
  Then $\bverti{1}\dloopsbackto\avert$ and $\cvert\neq\avert$ imply $\avert\loopsbackto\cvert$, 
  which together with $\cvert\,\sredtci{\bodylab}\,\avert$ yields a body step cycle between $\cvert$ and $\avert$ in $\acharthat$. 
  This contradicts that \ref{LLEEw:1} holds in $\acharthat$.
  Therefore $\bverti{1} \redtavoidsvtci{\cvert}{\bodylab} \avert$ in $\acharthatacc$.
  We consider two cases.
  
  \noindent
  {\sc Case 2.1:} $\bverti{2} \comprewrels{\sredtavoidsvrtci{\cvert}{\bodylab}}{\sredi{\loopnsteplab{\bLname}}}$ visits $\avert$, 
      so 
    \vspace*{-.75mm}%
    $\avert \:\scomprewrels{\sredtavoidsvrtci{\cvert}{\bodylab}}{\sredi{\loopnsteplab{\bLname}}}$ in $\acharthatacc$.
    Then 
    $\cvert \comprewrels{\sredtavoidsvi{\cvert}{\loopnsteplab{\aLname}}}{\sredtavoidsvrtci{\cvert}{\bodylab}} \bverti{1}
            \redtavoidsvtci{\cvert}{\bodylab} \avert
            \:\scomprewrels{\sredtavoidsvrtci{\cvert}{\bodylab}}{\sredi{\loopnsteplab{\bLname}}}$ in $\acharthatacc$. 
    By  \ref{LLEEw:2}\ref{LLEEw:2b} for $\acharthatacc$, $\ll{\aLname}>\ll{\bLname}$.
   \vspace{1mm}
  
  \noindent
  {\sc Case 2.2:} 
    \vspace{-1mm}$\bverti{2} \:\scomprewrels{\sredtavoidsvrtci{\cvert}{\bodylab}}{\sredi{\loopnsteplab{\bLname}}}$ does not visit $\avert$.
    Then since $\bverti{2} \loopsbacktotc \avert$ implies $\avert \descendsinlooptotc \bverti{2}$,
    there is \vspace*{-.75mm}a path 
    $\avert 
       \comprewrels{\sredtavoidsvi{\avert}{\loopnsteplab{\cLname}}}{\sredtavoidsvrtci{\avert}{\bodylab}} 
     \dverti{k} 
       \comprewrels{\sredtavoidsvi{\dverti{k}}{\loopnsteplab{\dLnamei{k}}}}{\sredtavoidsvrtci{\dverti{k}}{\bodylab}}
         \cdots\,
     \dverti{1}    
       \comprewrels{\sredtavoidsvi{\dverti{1}}{\loopnsteplab{\dLnamei{1}}}}{\sredtavoidsvrtci{\dverti{1}}{\bodylab}}   
     \bverti{2} 
       \:\comprewrels{\redtavoidsvrtc{\avert}}{\sredi{\loopnsteplab{\bLname}}}$ in $\acharthatacc$, for \vspace*{-.75mm}some $k\geq 0$.
  Since also 
  $\cvert \comprewrels{\sredtavoidsvi{\cvert}{\loopnsteplab{\aLname}}}{\sredtavoidsvrtci{\cvert}{\bodylab}} \bverti{1}
          \redtavoidsvtci{\cvert}{\bodylab} \avert \,\sredi{\loopnsteplab{\cLname}}$ in $\acharthatacc$, 
  by \ref{LLEEw:2}\ref{LLEEw:2b},
  $\ll{\aLname} >\ll{\cLname} > \ll{\dLnamei{k}} > \cdots > \ll{\dLnamei{1}} > \ll{\bLname}$.
  So $\ll{\aLname} > \ll{\bLname}$.
  
  To verify \ref{LLEEw:1} together with part \ref{loop:2} of \ref{LLEEw:2}\ref{LLEEw:2a} for $\connthroughin{\acharthat}{\bverti{1}}{\bverti{2}}$,
  it suffices to show that $\connthroughin{\acharthat}{\bverti{1}}{\bverti{2}}$ does not contain body step cycles. 
  This can be verified analogously as for transformation~I.
  That is, under the assumption of a body step cycle we can construct a path $\bverti{2} \redtci{\bodylab} \bverti{1}$ in $\acharthat$, 
  which contradicts \ref{cond:transf:III} (as it contradicted \ref{cond:transf:I}).  
  
  To show part \ref{loop:3} of \ref{LLEEw:2}\ref{LLEEw:2a} for $\connthroughin{\acharthat}{\bverti{1}}{\bverti{2}}$, 
  we can use part of the argumentation employed above for proving \ref{LLEEw:2}\ref{LLEEw:2b}.
  It was demonstrated in particular that for every \txtdescendsinloopto\ path 
  \vspace*{-1.25mm}%
  $\cvert 
     \comprewrels{\sredtavoidsvi{\cvert}{\loopnsteplab{\aLname}}}{\sredtavoidsvrtci{\cvert}{\bodylab}}
    \dvert$ 
  in $\acharthatdacc$
  there is a \txtdescendsinloopto\ path 
  \vspace*{-.5mm}%
  $\cverttilde
     \comprewrels{\sredtavoidsvi{\cverttilde}{\loopnsteplab{\cLname'}}}{\sredtavoidsvrtci{\cverttilde}{\bodylab}}
    \dvert$
  with the same target~$\dvert$ in $\acharthat$.
  This entails that if a \txtdescendsinloopto\ path in $\acharthatdacc$ had $\tick$ as target,
  then there were a \txtdescendsinloopto\ path in $\acharthat$ with $\tick$ as target,
  contradicting \ref{loop:3} for the \LLEEwitness~$\acharthat$.
  Hence $\acharthatdacc$ must satisfy part~\ref{loop:3}~of~\ref{LLEEw:2}\ref{LLEEw:2a}.
  
  We conclude that the result of transformation~III is again a \LLEEwitness. 
\end{description}  
\end{proof}
 

\end{document}